\documentclass{lmcs} 
\pdfoutput=1

\usepackage{lastpage}
\lmcsdoi{21}{4}{28}
\lmcsheading{}{\pageref{LastPage}}{}{}%
{Sep.~27,~2024}{Dec.~12,~2025}{}

\keywords{Probabilistic Programming, Term Rewriting, Evaluation Strategies, Modularity}

\usepackage{hyperref}
\usepackage{amsmath, amssymb, amsthm}
\usepackage{thmtools}
\usepackage{thm-restate}
\PassOptionsToPackage{hyphens}{url}
\usepackage{lineno}
\usepackage{booktabs}
\usepackage{tikz}
\usetikzlibrary{shapes,calc,arrows,automata,decorations.pathmorphing,backgrounds,arrows.meta,shapes.geometric,shapes.multipart,shapes.misc}
\usepackage[all,cmtip]{xy}
\usepackage{mymatrix}
\RequirePackage{makecell}

\pgfdeclarelayer{edgelayer}
\pgfdeclarelayer{nodelayer}
\pgfsetlayers{background,edgelayer,nodelayer,main}
\tikzstyle{none}=[inner sep=0mm]
\tikzstyle{moveBlock}=[fill=white, draw=black, shape=rectangle]
\tikzstyle{target}=[fill=white, draw=black, shape=circle]

\tikzstyle{dotHead}=[dotted, ->]
\tikzstyle{dotWithoutHead}=[dotted, -]
\tikzstyle{dashHead}=[dashed,->]
\tikzstyle{dashWithoutHead}=[dashed,-]
\tikzstyle{arrow}=[->]

\RequirePackage{makecell}

\usepackage{xifthen}
\usepackage{xspace}
\usepackage[noline,noend]{algorithm2e}
\usepackage[capitalize,nameinlink]{cleveref}

\theoremstyle{definition}
\newtheorem{counterexample}[thm]{Counterexample}

\usepackage{stackengine}
\usepackage[utf8]{inputenc}
\usepackage[inline]{enumitem}
\setlist[enumerate,1]{label=(\roman*), wide=0pt, leftmargin=*}

\renewcommand{\emptyset}{\varnothing}

\newenvironment{myproofof}[1]{%
	\begin{proof}[Proof of #1]%
}{\end{proof}%
}

\newenvironment{myproofsketch}{
  \begin{proof}[Proof Sketch]%
}{\end{proof}%
}

\usepackage{mathtools}

\usepackage{stmaryrd}
\usepackage{etoolbox}
\usepackage{mymatrix}
\usepackage{multirow}
\usepackage{bm}
\usepackage[english]{babel}
\usepackage{accents}
\usepackage{microtype}
\usepackage{nicefrac,xfrac}
\usepackage{centernot}
\usepackage{array,longtable}
\newcolumntype{C}{>{$}c<{$}}  
\newcolumntype{R}{>{$}r<{$}}  
\newcolumntype{L}{>{$}l<{$}}  
\usepackage{wrapfig}
\usepackage{subcaption}
\makeatletter         
\def\figurecaption#1#2{\noindent\hangindent 40pt
                       \hbox to 36pt {\small\sl #1 \hfil}
                       \ignorespaces {\small #2}}

\long\def\@makecaption#1#2{
  \vskip 10pt 
  \settowidth{\@tempdima}{#2}
  \ifdim\@tempdima>0pt
       \setbox\@tempboxa\hbox{#1: #2}
     \else
       \setbox\@tempboxa\hbox{#1 #2}
   \fi
   \ifdim \wd\@tempboxa >\hsize               
       \begin{list}{#1:}{
       \settowidth{\labelwidth}{#1:}
       \setlength{\leftmargin}{\labelwidth}
       \addtolength{\leftmargin}{\labelsep}
        }\item #2 \end{list}\par   
     \else                                    
       \hbox to\hsize{\hfil\box\@tempboxa\hfil}  
   \fi}
\makeatother

\makeatletter
\let\xx@thm\@thm
\AtBeginDocument{\let\@thm\xx@thm}
\makeatother

\usepackage{cite}
\usepackage[numbers]{natbib}
\usepackage{doi}

\newcommand{\makeproof}[2]{}
\newcommand{\paper}[1]{}
\newcommand{\report}[1]{#1} 

\allowdisplaybreaks

\hypersetup{
    pdftitle={From Innermost to Full Probabilistic Term Rewriting},
    colorlinks=true,
    linkcolor=blue,
    citecolor=olive,
    filecolor=magenta,
    urlcolor=cyan
}

\makeatletter \newcommand*\bigcdot{\mathpalette\bigcdot@{.4}}
\newcommand*\bigcdot@[2]{\mathbin{\vcenter{\hbox{\scalebox{#2}{$\m@th#1\bullet$}}}}}
\makeatother

\renewcommand{\emptyset}{\varnothing}

\newcommand{\disabledcomment}[1]{}
\newcommand{\oldcomment}[1]{}

\newcommand{\tool}[1]{\textsf{#1}}

\newcommand{\dontprint}[1]{}

\renewcommand{\epsilon}{\varepsilon}

\newcommand{\IN}{\mathbb{N}}
\newcommand{\IS}{\mathbb{S}}
\newcommand{\IR}{\mathbb{R}}

\newcommand{\R}{\mathcal{R}}

\newcommand{\PP}{\mathcal{P}}
\newcommand{\TT}{\mathcal{T}}
\newcommand{\BB}{\mathcal{B}}

\newcommand{\aprove}{\textsf{AProVE}}
\newcommand{\muterm}{\textsf{MuTerm}}
\newcommand{\natt}{\textsf{NaTT}}
\newcommand{\ttttwo}{\textsf{T\kern-0.15em\raisebox{-0.55ex}T\kern-0.15emT\kern-0.15em\raisebox{-0.55ex}2}}
\newcommand{\tct}{\textsf{TcT}}
\newcommand{\ceta}{\textsf{CeTA}}

\crefname{defi}{Def.}{Def.}
\crefname{exa}{Ex.}{Ex.}
\crefname{counterexample}{Counterex.}{Counterex.}
\crefname{appendix}{App.}{App.}
\crefname{ex}{Ex.}{Ex.}
\crefname{thm}{Thm.}{Thm.}
\crefname{lem}{Lemma}{Lemmas}
\crefname{rem}{Rem.}{Rem.}
\crefname{section}{Section}{Sections}
\crefname{subsection}{Section}{Sections}
\crefname{subsubsection}{Section}{Sections}
\crefname{line}{Line}{Lines}
\crefname{cor}{Cor.}{Cor.}
\crefname{figure}{Fig.}{Fig.}
\crefname{enumi}{}{}
\crefname{algorithm}{Alg.}{Alg.}

\renewcommand{\emptyset}{\varnothing}

\newcommand{\IE}{\ensuremath{\mathbb{E}}}

\newcommand{\F}[1]{\mathfrak{#1}}
\newcommand{\C}[1]{\mathcal{#1}}

\makeatletter
\def\moverlay{\mathpalette\mov@rlay}
\def\mov@rlay#1#2{\leavevmode\vtop{%
   \baselineskip\z@skip \lineskiplimit-\maxdimen
   \ialign{\hfil$\m@th#1##$\hfil\cr#2\crcr}}}
\newcommand{\charfusion}[3][\mathord]{
    #1{\ifx#1\mathop\vphantom{#2}\fi
        \mathpalette\mov@rlay{#2\cr#3}
      }
    \ifx#1\mathop\expandafter\displaylimits\fi}
\makeatother

\newcommand{\Var}{\mathcal{V}}

\newcommand{\pos}{\mathrm{Pos}}

\newcommand{\TSet}[2]{\mathcal{T}(#1,#2)}

\newcommand{\VSet}{\mathcal{V}}

\renewcommand{\O}{\mathcal{O}}

\newcommand{\FDist}{\operatorname{FDist}}
\newcommand{\Supp}{\operatorname{Supp}}

\newcommand{\rootsym}{\operatorname{root}}

\newcommand{\rt}{\operatorname{root}}
\newcommand{\edh}{\operatorname{edh}}
\newcommand{\edl}{\operatorname{edl}}
\newcommand{\capterm}{\operatorname{Abs}}
\newcommand{\caphterm}{\operatorname{A}}
\newcommand{\caphbterm}{\operatorname{B}}
\newcommand{\caphcterm}{\operatorname{C}}

\newcommand{\ts}{\mathsf{s}}
\renewcommand{\O}{\mathcal{O}}
\renewcommand{\O}{\mathcal{O}}
\newcommand{\tz}{\mathsf{0}}
\newcommand{\tf}{\mathsf{f}}
\newcommand{\tg}{\mathsf{g}}
\renewcommand{\th}{\mathsf{h}}
\newcommand{\ta}{\mathsf{a}}
\newcommand{\tb}{\mathsf{b}}
\newcommand{\tc}{\mathsf{c}}
\newcommand{\td}{\mathsf{d}}
\newcommand{\te}{\mathsf{e}}

\newcommand{\tq}{\mathsf{q}}

\newcommand{\trw}{\mathsf{rw}}
\newcommand{\tcons}{\mathsf{cons}}
\newcommand{\tnode}{\mathsf{node}}
\newcommand{\tenc}{\mathsf{enc}}
\newcommand{\targenc}{\mathsf{argenc}}
\newcommand{\tnil}{\mathsf{nil}}

\newcommand{\cv}{\operatorname{cv}}
\newcommand{\bv}{\operatorname{bv}}
\newcommand{\dv}{\operatorname{dv}}

\newcommand{\dheight}{\operatorname{dh}}
\newcommand{\runtime}{\operatorname{rc}}
\newcommand{\eruntime}{\operatorname{erc}}
\newcommand{\dertime}{\operatorname{dc}}
\newcommand{\edertime}{\operatorname{edc}}

\newcommand{\ctroot}{\F{r}}
\newcommand{\ctleaf}{\operatorname{Leaf}}

\newcommand{\ctdepth}{\operatorname{d}}

\newcommand{\ctlevelTwo}{\C{L}_{2}}
\newcommand{\ctlevelTwowithborder}{\F{L}_2}

\newcommand{\SN}[1][]{%
\ifx\relax#1\relax
  \mathtt{SN}%
\else
  \mathtt{SN}_{#1}%
\fi
}
\newcommand{\wSN}[1][]{%
\ifx\relax#1\relax
  \mathtt{WN}%
\else
  \mathtt{WN}_{#1}%
\fi
}
\newcommand{\SNf}{\SN[\fto_{\R}]}
\newcommand{\SNi}{\SN[\ito_{\R}]}
\newcommand{\SNli}{\SN[\lito_{\R}]}
\newcommand{\SNs}{\SN[\sto_{\R}]}
\newcommand{\wSNf}{\wSN[\fto_{\R}]}
\newcommand{\wSNi}{\wSN[\ito_{\R}]}
\newcommand{\wSNs}{\wSN[\sto_{\R}]}
\newcommand{\wSNli}{\wSN[\lito_{\R}]}
\newcommand{\AST}[1][]{%
  \ifx\relax#1\relax
    \mathtt{AST}%
  \else
    \mathtt{AST}_{#1}%
  \fi
}
\newcommand{\wAST}[1][]{%
  \ifx\relax#1\relax
    \mathtt{wAST}%
  \else
    \mathtt{wAST}_{#1}%
  \fi
}
\newcommand{\ASTf}{\AST[\fto_{\PP}]}
\newcommand{\ASTi}{\AST[\ito_{\PP}]}

\newcommand{\ASTs}{\AST[\sto_{\PP}]}
\newcommand{\wASTf}{\wAST[\fto_{\PP}]}
\newcommand{\PAST}[1][]{%
  \ifx\relax#1\relax
    \mathtt{PAST}%
  \else
    \mathtt{PAST}_{#1}%
  \fi
}
\newcommand{\wPAST}[1][]{%
  \ifx\relax#1\relax
    \mathtt{wPAST}%
  \else
    \mathtt{wPAST}_{#1}%
  \fi
}
\newcommand{\PASTf}{\PAST[\fto_{\PP}]}
\newcommand{\PASTi}{\PAST[\ito_{\PP}]}

\newcommand{\PASTs}{\PAST[\sto_{\PP}]}
\newcommand{\wPASTf}{\wPAST[\fto_{\PP}]}
\newcommand{\SAST}[1][]{%
  \ifx\relax#1\relax
    \mathtt{SAST}%
  \else
    \mathtt{SAST}_{#1}%
  \fi
}
\newcommand{\SASTf}{\SAST[\fto_{\PP}]}
\newcommand{\SASTi}{\SAST[\ito_{\PP}]}

\newcommand{\SASTs}{\SAST[\sto_{\PP}]}
\newcommand{\PSN}[1][]{%
  \ifx\relax#1\relax
    \mathtt{PSN}%
  \else
    \mathtt{PSN}_{#1}%
  \fi
}
\newcommand{\wPSN}[1][]{%
  \ifx\relax#1\relax
    \mathtt{wPSN}%
  \else
    \mathtt{wPSN}_{#1}%
  \fi
}
\newcommand{\PSNf}{\PSN[\fto_{\PP}]}
\newcommand{\PSNi}{\PSN[\ito_{\PP}]}
\newcommand{\PSNli}{\PSN[\lito_{\PP}]}
\newcommand{\PSNs}{\PSN[\sto_{\PP}]}
\newcommand{\wPSNf}{\wPSN[\fto_{\PP}]}

\makeatletter
\NewDocumentCommand{\dparrow}{+O{} +O{0.5cm}}{%
\begin{tikzpicture}[baseline=-0.5ex] {
\node[inner sep=0](@1) at (0,0) {};
\node[inner sep=0](@2) at (#2,0) {};
\draw [arrows={-Triangle[open]},shorten >= 1pt,shorten <= 1pt](@1)--(@2) node[pos=.5,above,inner sep=1pt] {\ensuremath{#1}};}
\end{tikzpicture}\xspace
}

\NewDocumentCommand{\myto}{+O{} +O{0.5cm}}{%
\begin{tikzpicture}[baseline=-0.5ex] {
\node[inner sep=0](@1) at (0,0) {};
\node[inner sep=0](@2) at (#2,0) {};
\draw [arrows={-to}](@1)--(@2) node[pos=.5,above,inner sep=1pt] {\ensuremath{#1}};}
\end{tikzpicture}\xspace
}

\NewDocumentCommand{\paraarrow}{+O{} +O{0.4cm}}{%
\begin{tikzpicture}[baseline=-0.5ex] {
\node[inner sep=0](@1) at (0,0) {};
\node[inner sep=0](@2) at (#2,0) {};
\node[inner sep=0](@3) at (0.07,0) {};
\draw [arrows={-to}](@1)--(@2) node[pos=.5,above,inner sep=1pt] {\ensuremath{#1}};
\draw [arrows={-to}](@1)--(@3);}
\end{tikzpicture}\xspace
}

\NewDocumentCommand{\paradparrow}{+O{} +O{0.4cm}}{%
\begin{tikzpicture}[baseline=-0.5ex] {
\node[inner sep=0](@1) at (0,0) {};
\node[inner sep=0](@2) at (#2,0) {};
\node[inner sep=0](@3) at (0.07,0) {};
\draw [arrows={-Triangle[open]}](@1)--(@2) node[pos=.5,above,inner sep=1pt] {\ensuremath{#1}};
\draw [arrows={-to}](@1)--(@3);}
\end{tikzpicture}\xspace
}

\newcommand{\oset}[2]{%
  {\mathop{#2}\limits^{\vbox to 1\ex@{\kern-\tw@\ex@
   \hbox{\scriptsize #1}\vss}}}}

\newcommand{\osetthree}[2]{%
  {\mathop{#2}\limits^{\vbox to 3\ex@{\kern-\tw@\ex@
   \hbox{\scriptsize #1}\vss}}}}

\newcommand{\osetfive}[2]{%
  {\mathop{#2}\limits^{\vbox to 5\ex@{\kern-\tw@\ex@
   \hbox{\scriptsize #1}\vss}}}}

\newcommand{\osetminus}[2]{%
  {\mathop{#2}\limits^{\vbox to -2\ex@{\kern-\tw@\ex@
   \hbox{\scriptsize #1}\vss}}}}
\makeatother

\newcommand{\NF}{\mathtt{NF}}

\newcommand{\liftArr}[3]{
  \mathrel{
    \xliftto{#1}
    \!\!{}^{#2}_{#3}
  }
}

\newcommand{\tored}[3]{
  \mathrel{
    \xhookrightarrow{{}_{\scriptstyle #1}}
    \!\!{}^{#2}_{#3}
  }
}

\newcommand{\fto}{\mathrel{\smash{\stackrel{\raisebox{2pt}{$\scriptscriptstyle \mathbf{f}\:$}}{\smash{\rightarrow}}}}}
\newcommand{\ito}{\mathrel{\smash{\stackrel{\raisebox{2pt}{$\scriptscriptstyle \mathbf{i}\:$}}{\smash{\rightarrow}}}}}
\newcommand{\lito}{\mathrel{\smash{\stackrel{\raisebox{2pt}{$\scriptscriptstyle \mathbf{li}\:$}}{\smash{\rightarrow}}}}}
\newcommand{\sto}{\mathrel{\smash{\stackrel{\raisebox{2pt}{$\scriptscriptstyle s\:$}}{\smash{\rightarrow}}}}}
\newcommand{\sprimeto}{\mathrel{\smash{\stackrel{\raisebox{2pt}{$\scriptscriptstyle s'\:$}}{\smash{\rightarrow}}}}}

\newcommand{\snotrootto}{\mathrel{\smash{\stackrel{\raisebox{3.4pt}{\tiny $s\:\lnot \varepsilon$}}{\smash{\longrightarrow}}}}}

\newcommand{\ftoleft}{\mathrel{\smash{\stackrel{\raisebox{2pt}{$\scriptscriptstyle \:\mathbf{f}$}}{\smash{\leftarrow}}}}}

\newcommand{\topara}{\mathrel{\paraarrow}}

\newcommand{\ftopara}{\mathrel{\smash{\stackrel{\raisebox{2pt}{$\scriptscriptstyle \mathbf{f}\:$}}{\smash{\paraarrow}}}}}
\newcommand{\itopara}{\mathrel{\smash{\stackrel{\raisebox{2pt}{$\scriptscriptstyle \mathbf{i}\:$}}{\smash{\paraarrow}}}}}
\newcommand{\stopara}{\mathrel{\smash{\stackrel{\raisebox{2pt}{$\scriptscriptstyle s\:$}}{\smash{\paraarrow}}}}}

\newcommand{\leftR}{\mathrel{\prescript{}{\R}{\ftoleft}}}
\newcommand{\leftRStar}{\mathrel{\prescript{*}{\R}{\ftoleft}}}
\newcommand{\leftROne}{\mathrel{\prescript{}{\R_1}{\ftoleft}}}

\newcommand{\ftor}{\mathrel{\fto_{\R}}}
\newcommand{\itor}{\mathrel{\ito_{\R}}}
\newcommand{\litor}{\mathrel{\lito_{\R}}}

\newcommand{\itorExOne}{\mathrel{\ito_{\R_1}}}

\newcommand{\itoRd}{\mathrel{\ito_{\R_{\td}}}}

\newcommand{\itos}{\mathrel{\ito_{\PP}}}

\newcommand{\xliftto}[1]{\mathrel{\smash{\stackrel{\raisebox{3.4pt}{\tiny $#1\:$}}{\smash{\liftto}}}}}
\newcommand{\xleftliftto}[1]{\mathrel{\smash{\stackrel{\raisebox{3.4pt}{\tiny $\:#1$}}{\smash{\leftliftto}}}}}
\newcommand{\liftto}{\rightrightarrows}
\newcommand{\iliftto}{\xliftto{\mathbf{i}}}
\newcommand{\fliftto}{\xliftto{\mathbf{f}}}
\newcommand{\liliftto}{\xliftto{\mathbf{li}}}

\newcommand{\leftliftto}{\mathrel{\ooalign{\raisebox{2pt}{$\leftarrow$}\cr\hfil\raisebox{-2pt}{$\leftarrow$}\hfil}}}

\newcommand{\paraliftto}{\mathrel{\ooalign{\raisebox{2pt}{$\paraarrow$}\cr\hfil\raisebox{-2pt}{$\paraarrow$}\hfil}}}

\newcommand{\fparaliftto}{\mathrel{\smash{\stackrel{\raisebox{3.4pt}{\tiny $\mathbf{f}\:$}}{\smash{\paraliftto}}}}}
\newcommand{\iparaliftto}{\mathrel{\smash{\stackrel{\raisebox{3.4pt}{\tiny $\mathbf{i}\:$}}{\smash{\paraliftto}}}}}

\theoremstyle{plain} 

\begin{document}

\title[From Innermost to Full Probabilistic Term Rewriting]{From Innermost to Full Probabilistic Term Rewriting: 
  Almost-Sure Termination,
  Complexity, and Modularity}
\titlecomment{{\lsuper*}This is a revised and extended journal version of our earlier conference paper \cite{FoSSaCS2024}.}
\thanks{Funded by the DFG Research Training Group 2236 UnRAVeL}

\author[J.-C.~Kassing]{Jan-Christoph Kassing\lmcsorcid{0009-0001-9972-2470}}
\author[J.~Giesl]{J\"urgen Giesl\lmcsorcid{0000-0003-0283-8520}}
\address{RWTH Aachen University, Aachen, Germany}	
\email{\{kassing,giesl\}@cs.rwth-aachen.de}  




\begin{abstract}
  \noindent There are many evaluation strategies for term rewrite systems, 
but automatically proving termination or analyzing complexity is usually easiest for innermost rewriting.
Several syntactic criteria exist when innermost termination implies (full) termination 
or when runtime complexity and innermost runtime complexity coincide.
We adapt these criteria to the probabilistic setting, e.g., we show when it suffices to
analyze almost-sure termination w.r.t.\ innermost rewriting in order to prove
(full) almost-sure termination of probabilistic term rewrite systems.
These criteria can be applied for both termination and complexity analysis in the probabilistic setting.
We implemented and evaluated our new contributions in the tool \aprove.
Moreover, we also use our new results to
investigate the modularity of probabilistic termination properties.

\end{abstract}

\maketitle


\section{Introduction}\label{sec-introduction}

Termination
and complexity analysis are among the main tasks in program verification, and 
techniques and tools to analyze termination or complexity of
term rewrite systems (TRSs) automatically have been studied for decades. 
While a direct application of classical reduction orderings is often too weak,
these orderings can be used successfully within the \emph{dependency pair} (DP) framework
for termination~\cite{arts2000termination,giesl2006mechanizing} and for innermost runtime complexity~\cite{noschinski2013analyzing}. 
Moreover, the framework of~\cite{avanzini_combination_2016} uses \emph{weak}
dependency pairs in order to analyze derivational and runtime complexity.
These frameworks allow for modular termination and complexity proofs
by decomposing the original problem into sub-problems which can
then be analyzed independently using different techniques.
Thus, DPs are used in essentially all current termination and complexity
tools  for TRSs
(e.g., \aprove{}~\cite{JAR-AProVE2017}, \muterm{}~\cite{gutierrez_mu-term_2020},
\natt{}~\cite{natt_sys_2014},
\tct{}~\cite{avanzini_tct_2016},
\ttttwo{}~\cite{ttt2_sys}).
To allow certification of proofs with DPs, they have been formalized in
several proof assistants (e.g., in \tool{Rocq} (formerly \tool{Coq})~\cite{Contejean07,Blanqui11},
\tool{Isabelle}~\cite{ceta_sys}, and recently together with the size-change
principle \cite{DBLP:conf/popl/LeeJB01,DBLP:journals/aaecc/ThiemannG05,DBLP:conf/cav/ManoliosV06}
in \tool{PVS} \cite{almeida_formalizing_2020, munoz2023FormalVerificationTermination}), and there exist several corresponding certification tools
for termination and complexity proofs with DPs (e.g., \ceta{}~\cite{ceta_sys}).

\emph{Probabilistic} programs are used to describe randomized
algorithms and probability distributions, with applications in many areas, 
see, e.g.,~\cite{Gordon14}.
To use TRSs also for such programs, \emph{probabilistic term rewrite systems} (PTRSs) were introduced in~\cite{BournezRTA02,bournez2005proving,avanzini2020probabilistic}.
In the probabilistic setting, there are several notions of ``termination''.
For example, a program is \emph{almost-surely terminating} ($\mathtt{AST}$) 
if the probability for termination is $1$.
Another interesting property is \emph{positive almost-sure termination} ($\mathtt{PAST}$)~\cite{DBLP:conf/mfcs/Saheb-Djahromi78,bournez2005proving}, which means
that the expected length of every evaluation is finite.
Finally, \emph{strong} or \emph{bounded almost-sure termination}
($\mathtt{SAST}$)~\cite{fuTerminationNondeterministicProbabilistic2019,avanzini2020probabilistic} 
requires that, for every 
configuration $t$ (e.g., for every term), there is a finite bound on the expected
lengths of all evaluations starting in~$t$. Thus, if there is a start configuration
$t$ which non-deterministically leads to evaluations of arbitrary finite length, then the
program can be $\mathtt{PAST}$, but not $\mathtt{SAST}$. 
In general,
$\mathtt{SAST}$ implies $\mathtt{PAST}$, and $\mathtt{PAST}$ implies $\mathtt{AST}$, 
but the converse directions do not hold.

We recently developed an adaption of the DP framework for $\mathtt{AST}$~\cite{JPK60} and an adaption for
innermost $\mathtt{AST}$~\cite{kassinggiesl2023iAST,FLOPS2024} (i.e., $\mathtt{AST}$ restricted to rewrite sequences where one only
evaluates at innermost positions), which allows us to benefit from a similar modularity
as in the non-probabilistic setting.
However, the DP framework for innermost $\mathtt{AST}$ is 
substantially more powerful than the one for $\mathtt{AST}$.
Indeed, also in the non-probabilistic setting, 
innermost termination is usually substantially easier to prove than (full) termination,
i.e., termination w.r.t.\ any possible evaluation strategy,
see, e.g.,~\cite{arts2000termination,giesl2006mechanizing}. 
The same holds for non-probabilistic complexity analysis, where the DP framework
of~\cite{noschinski2013analyzing} is restricted to innermost rewriting and 
the framework of~\cite{avanzini_combination_2016}
is considerably more powerful for innermost than for full rewriting.
To lift innermost termination and complexity proofs to full rewriting,
in the non-probabilistic setting there exist several sufficient criteria 
which ensure that innermost termination implies full termination~\cite{Gramlich1995AbstractRB} 
and that innermost runtime complexity coincides with full runtime complexity~\cite{frohn_analyzing_nodate}.

Up to now no such results were known in the probabilistic setting. 
Our paper presents the first sufficient criteria for PTRSs which ensure that, e.g., $\mathtt{AST}$
coincides for full and innermost rewriting, and we also show similar results for
other rewrite strategies like \emph{leftmost-innermost} rewriting.
We focus on criteria that can be checked automatically, so we can combine our results
with the DP framework for proving innermost $\mathtt{AST}$ of PTRSs~\cite{kassinggiesl2023iAST,FLOPS2024}.  
In this way, we obtain a technique that (if applicable) can
prove $\mathtt{AST}$ for \emph{full} rewriting automatically
and is significantly more powerful than the corresponding DP framework for full $\mathtt{AST}$~\cite{JPK60}.

Our criteria also hold for $\mathtt{PAST}$, $\mathtt{SAST}$, and expected complexity.
Similar to $\AST$, very recently we developed an adaption of the DP framework to prove innermost $\mathtt{SAST}$ 
and upper bounds on the innermost expected runtime complexity~\cite{kassing2025DependencyPairsExpected}.
Again, we can use our results on the relation between innermost and full rewriting to
apply this adaption of the DP framework for arbitrary evaluation strategies.
A corresponding evaluation can be found in~\cite{kassing2025DependencyPairsExpected}.
Except for applying polynomial or matrix orderings directly on the whole
PTRS~\cite{avanzini2020probabilistic}, currently there do not
exist any other automatic techniques for $\mathtt{PAST}$, $\mathtt{SAST}$, or expected
complexity of PTRSs.

As a corollary of our results in the probabilistic setting, 
we also develop the first results relating derivational and innermost 
derivational complexity for ordinary non-probabilistic TRSs.
The difference between derivational and runtime complexity is that 
runtime complexity only considers start terms where a defined function symbol (i.e., an
``algorithm'')
is applied
to arguments built with constructor symbols (i.e., to ``data''),
while derivational complexity allows arbitrary start terms.

There exist numerous results  on the
modularity of termination, confluence, and completeness of TRSs
in the non-probabilistic setting, see, e.g., \cite{Gramlich1995AbstractRB,
  gramlich2012ModularityTermRewriting, Toyama87, toyama1987ChurchRosserPropertyDirect}.
Based on our novel criteria,
we develop the first modularity
results for probabilistic termination w.r.t.\ different evaluation strategies,
 i.e., 
we investigate whether $\mathtt{AST}$, $\mathtt{PAST}$, or $\mathtt{SAST}$ are preserved
for unions of PTRSs.
Additionally, we also study preservation of $\mathtt{AST}$, $\mathtt{PAST}$, or $\mathtt{SAST}$
under signature extensions, which can be seen as a special case of modularity, i.e., a specific union of PTRSs, 
where the second PTRS only contains trivially terminating rewrite rules over the new signature.
We show that while $\mathtt{AST}$ and $\mathtt{SAST}$ are preserved under signature extensions, this does not hold for $\mathtt{PAST}$,
implying that any sound and complete proof technique for $\mathtt{PAST}$ of PTRSs
has to take the specific signature into account.
Related to these results, we show that $\mathtt{PAST}$ and $\mathtt{SAST}$ coincide
for a broad class of PTRSs.
For example, if the signature contains at
least one function symbol of arity greater than $1$,
then there is no difference between
$\mathtt{PAST}$ and $\mathtt{SAST}$ for finite PTRSs.

\vspace*{-.2cm}

\subsection*{Structure}\label{sec:structure}

We start with preliminaries on term rewriting in \cref{Preliminaries}.
Then we recapitulate PTRSs based on~\cite{bournez2005proving,diazcaro_confluence_2018,avanzini2020probabilistic,Faggian2019ProbabilisticRN,kassinggiesl2023iAST} 
in \cref{Probabilistic Term Rewriting}.
In \cref{Relating AST and its Restricted Forms} we show that the properties of~\cite{Gramlich1995AbstractRB}
that ensure equivalence of innermost and full termination do not suffice in the probabilistic setting and extend them accordingly.
In particular, we show that innermost and full $\mathtt{AST}$ coincide for PTRSs that are
non-overlapping and linear.
This result also holds for $\mathtt{PAST}$, $\mathtt{SAST}$, and the expected runtime complexity, as well as for strategies like
leftmost-innermost evaluation.
In \cref{Improving Applicability} we try to weaken the linearity
requirement in order to prove
full $\mathtt{AST}$ for larger classes of PTRSs.
The implementation of our criteria in the tool \aprove{} is evaluated in \cref{Evaluation}.
Afterwards, in \cref{Modularity} we analyze the
modularity of all these (full and innermost)
termination properties for PTRSs.
We discuss related work on the
verification of probabilistic programs in \Cref{Related Work}.
Finally, we conclude in \cref{Conclusion} and
refer to App.~\ref{appendix} for all missing proofs.\footnote{To ease readability,
for those proofs which require larger technical constructions, we only give proof sketches
in the main part of the paper and present the corresponding full technical proofs in
App.~\ref{appendix}.}

\vspace*{-.2cm}

\subsection*{Novel Contributions of the Paper}\label{sec:novel_results}

The current paper extends our earlier conference paper~\cite{FoSSaCS2024} by:

\begin{itemize}
    \item All results concerning $\mathtt{SAST}$ and the novel relations between
      $\mathtt{SAST}$ and $\mathtt{PAST}$ (the whole \Cref{sect:Relating PAST and SAST} 
      as well as all results on $\mathtt{SAST}$ from \Cref{Relating AST and its Restricted
      Forms,Improving Applicability}).
    \item The new result that $\mathtt{PAST}$ is not closed under signature extensions
      (\Cref{thm:Sig-PAST}), while $\mathtt{AST}$ and $\mathtt{SAST}$ are (\Cref{signature-extensions-AST-SAST}).
    \item The novel definitions of expected derivational and runtime complexity
      (\Cref{def:expected-derivational-complexity,def:expected-runtime-complexity}) and all corresponding results
      (\Cref{properties-eq-AST-iAST-1-complex}, \ref{properties-eq-iAST-liAST-complex}, and
      \ref{properties-eq-AST-iAST-3-complex}).
    \item The corollaries for non-probabilistic derivational complexity
      (\Cref{properties-eq-Der-iDer-1,properties-eq-iDer-liDer-complex}).
    \item The whole \cref{Modularity} concerning the analysis of modularity.
    \item Numerous additional explanations, examples, and remarks.
    \item More details on the proofs of the main theorems (including central lemmas like
      \Cref{lemma-eq-AST-iAST-1}, \ref{lemma-eq-iAST-liAST},
      and \ref{lemma-eq-AST-iAST-3}), in addition to the full proofs in\paper{
      \Cref{REPORT}}\report{ App.\ \ref{appendix}}.
    \item A new counterexample (Counterex.\ \ref{example:simultaneous-rewriting-bug}) showing a mistake 
      in our earlier conference paper \cite{FoSSaCS2024}.
    \item An improved implementation and evaluation which combines the contributions
      of the current paper with the DP framework for full $\mathtt{AST}$  from \cite{JPK60}
      (which had not yet been developed at the time of our conference paper
      \cite{FoSSaCS2024}), see   \pagebreak[3] \Cref{Evaluation}.
\end{itemize}

\section{Preliminaries}\label{Preliminaries}

For any relation $\to \; \subseteq A \times A$ on some set $A$ and $n \in \IN$, we define $\to^{n}$ as 
$\to^{0} \;=\; \{(a,a) \mid a \in A\}$ and $\to^{n+1} \;=\; \to^{n} \circ \to$, 
where ``$\circ$'' denotes composition of relations, and define $\to^* = \bigcup_{n\in \IN}\to^{n}$, 
i.e., $\to^*$ is the \emph{reflexive and transitive closure} of $\to$.
Let $\NF_{\to}$ denote the set of all elements that are in \emph{normal form} w.r.t.\ $\to$, 
i.e., for all $a \in \NF_{\to}$ there is no $b \in A$ with $a \to b$. 

We assume familiarity with term rewriting~\cite{baader_nipkow_1999}, but recapitulate
the notions that are needed for this work.
We write $\TSet{\Sigma}{\VSet}$ for the set of all \emph{terms} over a (possibly infinite) countable 
set of \emph{function symbols} $\Sigma = \biguplus_{k \in \IN} \Sigma_k$ and a (possibly infinite) countable set of \emph{variables} $\VSet$, 
and $\TT$ if the specific sets $\Sigma$ and $\VSet$ are irrelevant or clear from the context.
To be precise, $\TSet{\Sigma}{\VSet}$ is the smallest set with $\VSet \subseteq \TSet{\Sigma}{ \VSet}$, 
and if $f \in \Sigma_k$ and $t_1, \dots, t_k \in \TSet{\Sigma}{ \VSet}$ then $f(t_1,\dots,t_k) \in \TSet{\Sigma}{ \VSet}$.
We say that a function symbol has \emph{arity} $k \in \IN$ if $f \in \Sigma_k$, 
i.e., the arity denotes how many arguments a function symbol takes.
A \emph{substitution} is a function $\sigma:\VSet \to \TT$ with $\sigma(x) = x$ for all but finitely many $x \in \VSet$, 
and we often write $x\sigma$ instead of $\sigma(x)$.
Substitutions homomorphically extend to terms: If $t=f(t_1,\dots,t_k)\in \TT$ then $t \sigma
= f(t_1\sigma,\dots,t_k \sigma)$.
For a term $t \in \TT$, the set of \emph{positions} $\pos(t)$ 
is the smallest subset of $\IN^*$ satisfying $\varepsilon \in \pos(t)$, 
and if $t=f(t_1,\dots,t_k)$ then for all $1 \leq i \leq k$ and all $\pi \in \pos(t_i)$ we have $i.\pi \in \pos(t)$.
If $\pi \in \pos(t)$ then $t|_{\pi}$ denotes the subterm starting at position $\pi$
and $t[r]_{\pi}$ denotes the term that results from replacing the subterm $t|_{\pi}$ at position $\pi$ with the term $r \in \TT$.
A \emph{context} $C$ is a term containing
exactly one occurrence of 
a special constant symbol $\square$, 
denoting a hole.
The constant $\square$ is used to indicate where one can insert a term, 
i.e., if $\square$ is at position $\pi$ of $C$, then $C[t]$ is a shorthand notation for $C[t]_{\pi}$.

A \emph{rewrite rule} $\ell \to r$ is a pair of terms  $\ell, r \in \TT$ 
such that $\VSet(r) \subseteq \VSet(\ell)$ and $\ell \notin \VSet$, 
where $\VSet(t)$ denotes the set of all variables occurring in $t \in \TT$.
A \emph{term rewrite system} (TRS) is a (possibly infinite) countable set of rewrite rules.
As an example, consider the TRS $\R_{\td}$ that doubles a natural number (represented by the terms $\ts$ and $\tz$) 
with the rewrite rules $\td(\ts(x)) \to \ts(\ts(\td(x)))$ and $\td(\tz) \to \tz$.
A TRS $\R$ induces a \emph{rewrite relation} ${\fto_{\R}} \subseteq \TT
\times \TT$ on terms where $s \fto_{\R} t$ holds if there is a position $\pi \in \pos(s)$, 
a rule $\ell \to r \in \R$, and a substitution $\sigma$ such that $s|_{\pi}=\ell\sigma$ and $t = s[r\sigma]_{\pi}$.
Here, $\mathbf{f}$ stands for ``full rewriting''\footnote{In the literature, 
one usually simply writes $\to_{\R}$ instead. 
We use $\fto_{\R}$ here to clearly indicate the corresponding rewrite strategy.}
as we did not fix any specific strategy yet.
A rewrite step $s \fto_{\R} t$ is an \emph{innermost} rewrite step (denoted $s \itor t$) if
all proper subterms of the used \emph{redex} $\ell\sigma$ are in normal form w.r.t.\ $\R$, i.e.,
the proper subterms of $\ell\sigma$
do not contain redexes themselves and thus, they cannot be reduced with $\fto_\R$.
For example, we have $\td(\ts(\td(\ts(\tz)))) \itoRd \td(\ts(\ts(\ts(\td(\tz)))))$.
Let $\NF_{\R}$ denote the set of all terms that are in normal form w.r.t.\ $\fto_{\R}$.

Let $<$ be the \emph{prefix ordering} on positions and let $\leq$ be its reflexive closure.
Two positions $\tau$ and $\pi$ are \emph{parallel} if both $\tau \not\leq \pi$ and $\pi \not\leq \tau$ hold.
For two parallel positions $\tau$ and $\pi$ we define
$\tau \prec \pi$ if we have $i < j$ for the unique $i,j$ such that $\chi.i \leq \tau$ and $\chi.j \leq \pi$,
where $\chi$ is the longest common prefix of $\tau$ and $\pi$.
An innermost rewrite step $s \itor t$ at position $\pi$ is
\emph{leftmost} (denoted $s \litor t$) if $s$ does not contain
any redex at a position
$\tau$ with $\tau \prec \pi$.

In this paper, we will consider innermost rewriting ($\mathbf{i}$), leftmost-innermost rewriting ($\mathbf{li}$),
and full rewriting ($\mathbf{f}$).
Since every leftmost-innermost rewrite step is an innermost rewrite step, 
and every innermost rewrite step is a step w.r.t.\ full rewriting, one directly obtains
$\litor \;\subseteq\; \itor \;\subseteq\; \ftor$ for every TRS $\R$.
Let $\IS = \{\mathbf{f}, \mathbf{i}, \mathbf{li}\}$ be the set of these three
rewrite strategies.

\subsection{Termination}
Let $\to \;\subseteq\; \TT \times \TT$ be a relation on terms.
We call $\to$ \emph{strongly normalizing}
or $\SN[\to]$ for short if $\to$ is well founded. A TRS
$\R$ is \emph{terminating} if we have $\SNf$,
and $\R$ is
\emph{innermost terminating} or \emph{leftmost-innermost terminating}
if we have $\SNi$ or $\SNli$, respectively.
If every term $t \in \TT$ has a normal form w.r.t.\ $\to$
(i.e., we have $t \to^* t'$ where $t' \in \NF_{\to}$)
then we call $\to$ \emph{weakly normalizing} (denoted by $\wSN[\to]$).

Two terms $s,t \in \TT$ are \emph{joinable} via $\R$ 
if there exists a $w \in \TT$ such that $s \fto_{\R}^* w \leftRStar t$.
Two rules $\ell_1 \to r_1, \ell_2 \to r_2 \in \R$ with renamed variables such that
$\VSet(\ell_1) \cap \VSet(\ell_2) = \emptyset$ are \emph{overlapping}
if there exists a non-variable position $\pi$ of $\ell_1$ such that $\ell_1|_{\pi}$ and
$\ell_2$ are unifiable, i.e., there exists a substitution $\sigma$ such that $\ell_1|_{\pi} \sigma = \ell_2 \sigma$.
If $(\ell_1 \to r_1)
= (\ell_2 \to r_2)$,  then we require that $\pi \neq \varepsilon$.
$\R$ is \emph{non-overlapping} (NO) if it has no overlapping rules.
As an example, the TRS $\R_{\td}$ is non-overlapping.
A TRS is \emph{left-linear} (LL) (\emph{right-linear}, RL) if
every variable occurs at most once in the left-hand side  (right-hand side) of a rule.
A TRS is \emph{linear} if it is both left- and right-linear.
A TRS is \emph{non-erasing} (NE)
if in every rule,
all variables of the left-hand side also occur
in the right-hand side.

Next, we recapitulate the relations
between $\SNf$, $\SNi$, $\SNli$, and $\wSNf$ in the non-probabilistic
setting.
Obviously, the stronger notion always implies the weaker one, e.g., $\SNf$ implies $\SNi$
and  $\wSNi$ implies $\wSNf$
since $\itor \;\subseteq\; \ftor$.
The interesting question is for which classes of TRSs the weaker and the stronger notion
are equivalent.\footnote{Note that to this end, we do not have to consider
 $\wSNli$ and $\wSNi$. The reason is that 
   when
 analyzing under which conditions
   $\wSNf$ implies $\SNf$, we also know that under these conditions  we have $\wSNs \Longrightarrow \SNs$
 for all $s \in \IS$, since $\wSNs \Longrightarrow\wSNf$ and $\SNf \Longrightarrow \SNs$ hold.
 Moreover, we have
  $\SNs \Longrightarrow \wSNs$  for all $s \in \IS$.}
We start with the relation between $\SNf$ and $\SNi$.

\begin{counterexample}[Toyama's Counterexample~\cite{Toyama87}]\label{example:diff-SN-vs-iSN-toyama}
    For the TRS $\R_1$ with the rules $\tf(\ta,\tb,x) \to \tf(x,x,x)$, $\tg \to \ta$, and $\tg \to \tb$,
    we do not have $\SN[\fto_{\R_1}]$ due to the infinite rewrite sequence $\tf(\ta, \tb, \tg) \fto_{\R_1} \tf(\tg, \tg, \tg) \fto_{\R_1}
    \tf(\ta, \tg, \tg) \fto_{\R_1} \tf(\ta, \tb, \tg) \fto_{\R_1} \ldots$
    But the only innermost rewrite sequences starting with $\tf(\ta, \tb, \tg)$ are
    $\tf(\ta, \tb, \tg) \itorExOne \tf(\ta, \tb, \ta) \itorExOne \tf(\ta, \ta, \ta)$ and
    $\tf(\ta, \tb, \tg) \itorExOne \tf(\ta, \tb, \tb) \itorExOne \tf(\tb, \tb, \tb)$,
    i.e., both of them reach normal forms in the end.
    Thus, $\SN[\ito_{\R_1}]$ holds as
    we have to rewrite the inner $\tg$ before we can use the $\tf$-rule. 
\end{counterexample}

The first property known to ensure equivalence of $\SNf$ and $\SNi$ is orthogonality,
which was already shown in~\cite{ODonnell77}.
A TRS is \emph{orthogonal} (OR) if it is non-overlapping and left-linear.

\begin{restatable}[From $\SNi$ to ${\SNf}$ (1),~\cite{ODonnell77}]{thm}{SNvsiSN1}\label{properties-eq-SN-iSN-1}
    If a TRS $\R$ is OR, then:
    \begin{align*}
        \SNf &\Longleftrightarrow \SNi\!
    \end{align*}
\end{restatable}

Then, in~\cite{Gramlich1995AbstractRB}, it was shown that one can remove the requirement
of left-linearity.

\begin{restatable}[From $\SNi$ to ${\SNf}$ (2),~\cite{Gramlich1995AbstractRB}]{thm}{SNvsiSN2}\label{properties-eq-SN-iSN-2}
    If a TRS $\R$ is NO, then:
    \begin{align*}
        \SNf &\Longleftrightarrow \SNi\!
    \end{align*}
\end{restatable}
 
Moreover,~\cite{Gramlich1995AbstractRB} refined
\Cref{properties-eq-SN-iSN-2} further.
A TRS $\R$ is an \emph{overlay system} (OS) if
its rules may only overlap at the root position, i.e.,
$\pi = \varepsilon$.
For instance,  $\R_1$ from Counterex.\ 
\ref{example:diff-SN-vs-iSN-toyama} is an overlay system.
Furthermore, a TRS is \emph{locally confluent} (or \emph{weakly Church-Rosser}, abbreviated WCR) 
if for all terms $s,t_1,t_2 \in \TT$ such that $t_1 \leftR s \fto_{\R} t_2$ the terms $t_1$ and $t_2$ are joinable.
$\R_1$ is \emph{not} WCR, as we have $\ta \leftROne \tg \fto_{\R_1} \tb$, but $\ta$ and $\tb$ are not joinable.
If a TRS has both of these properties, then $\SNi$ and $\SNf$ are again equivalent.

\begin{restatable}[From $\SNi$ to ${\SNf}$ (3),~\cite{Gramlich1995AbstractRB}]{thm}{SNvsiSN3}\label{properties-eq-SN-iSN-3}
    If a TRS $\R$ is OS and WCR, then:
    \begin{align*}
        \SNf &\Longleftrightarrow \SNi\!
    \end{align*}
\end{restatable}

\Cref{properties-eq-SN-iSN-3} is stronger than \Cref{properties-eq-SN-iSN-2} 
as every non-overlapping TRS is a locally confluent overlay
system~\cite{knuth_simple_1970}.

Next, we recapitulate the results on the
relation between $\wSNf$ and $\SNf$.

\begin{counterexample}\label{example:diff-SN-vs-WN}
    Consider the TRS $\R_2$ with the rules $\tf(x) \to \tb$ and $\ta \to \tf(\ta)$.
    We do not have $\SN[\fto_{\R_2}]$ since we can always rewrite the inner $\ta$ to get 
    $\ta \fto_{\R_2} \tf(\ta) \fto_{\R_2} \tf(\tf(\ta)) \fto_{\R_2} \ldots$.
    On the other hand, $\wSNf$ holds since we can also rewrite the outer $\tf(\ldots)$ 
    before we use the $\ta$-rule twice, resulting in the term $\tb$, which is a normal form.
    For the TRS $\R_3$ with the rules $\tf(\ta) \to \tb$ and $\ta \to \tf(\ta)$, the situation is similar.
\end{counterexample}

The TRS $\R_2$ from Counterex.\ \ref{example:diff-SN-vs-WN} is erasing and
$\R_3$ is overlapping.
For TRSs with neither of those two properties, $\SNf$ and $\wSNf$ are equivalent.

\begin{restatable}[From $\wSNf$ to ${\SNf}$,~\cite{Gramlich1995AbstractRB}]{thm}{SNvsWN}\label{properties-SN-vs-WN}
    If a TRS $\R$ is NO and NE, then:
    \begin{align*}
        \SNf &\Longleftrightarrow \wSNf\!
    \end{align*}
\end{restatable}

Finally, we look at the relation between rewrite strategies that use an ordering 
for parallel redexes like leftmost-innermost rewriting compared to innermost rewriting.
It turns out that such an ordering does not interfere with termination at all.

\begin{restatable}[From $\SNli$ to $\SNi$,~\cite{KRISHNARAO2000141}]{thm}{iSNvsliSN}\label{properties-iSN-vs-liSN}
    For all TRSs $\R$ we have:
    \begin{align*}
        \SNi &\Longleftrightarrow \SNli\!
    \end{align*}
\end{restatable}

The relations between the different termination properties for non-probabilistic TRSs (given in
\Cref{properties-eq-SN-iSN-3},~\ref{properties-SN-vs-WN},
and~\ref{properties-iSN-vs-liSN}) are summarized in \Cref{fig:relations_SN}.

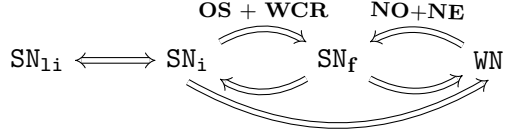
\begin{figure}
	\centering
	\begin{tikzpicture}
        \tikzstyle{adam}=[rectangle,thick,draw=black!100,fill=white!100,minimum size=4mm]
        \tikzstyle{empty}=[rectangle,thick,minimum size=4mm]
        
        \node[empty] at (0, -1)  (ast) {$\SNf$};
        \node[empty] at (-2.8, -1)  (iast) {$\SNi$};
        \node[empty] at (2.8, -1)  (wast) {$\wSNf$};
        \node[empty] at (-5.6, -1)  (liast) {$\SNli$};

        \draw (iast) edge[implies-implies,double equal sign distance] (liast);
        \draw (ast) edge[-implies,double equal sign distance, bend left] (iast);
        \draw (iast) edge[-implies,double equal sign distance, bend left] node[sloped, anchor=center,above] {\scriptsize \textbf{OS} + \textbf{WCR}} (ast);
        \draw (ast) edge[-implies,double equal sign distance, bend right] (wast);
        \draw (wast) edge[-implies,double equal sign distance, bend right] node[sloped, anchor=center,above] {\scriptsize \textbf{NO}+\textbf{NE}} (ast);
        \draw (iast.south) edge[-implies,double equal sign distance, bend right] (wast.south);

    \end{tikzpicture}
    \caption{Relations between the different termination properties for TRSs}\label{fig:relations_SN}
\end{figure}

\subsection{Complexity}\label{subsection-Complexity}

Next, we recapitulate known results regarding the complexity of TRSs under different
rewrite strategies.
There are two standard notions of complexity used in term rewriting: \emph{derivational} and \emph{runtime
complexity} \cite{DBLP:conf/rta/HofbauerL89,DBLP:conf/cade/HirokawaM08}.
For any $M \subseteq  \IN \cup \{\omega\}$, $\sup M$ denotes the least upper bound of $M$,
where $\sup \emptyset = 0$.
For a relation ${\to} \subseteq \TT \times \TT$,
the \emph{derivation height}
$\dheight_{\to}:\TT \to
 \IN \cup \{ \omega\}$
of a term $t \in \TT$
is $\dheight_{\to}(t) = \sup\{m \mid \exists t' \in \TT: t \to^{m} t'\}$, 
i.e., the length of the longest $\to$-rewrite sequence starting with $t$.
Then, the \emph{derivational complexity} 
$\dertime_{\to} : \IN \to \IN \cup \{\omega\}$ of $\to$ is defined as 
$\dertime_{\to}(n) = \sup\{\dheight_\to(t) \mid t \in \TT, |t| \leq n\}$.
Here, the \emph{size} $|t|$ of a term $t$ is the number of occurrences of
function symbols and variables in~$t$, i.e., 
we have $|x| = 1$ for $x \in \VSet$ and $|f(t_1, \ldots, t_k)| = 1 + \sum_{i=1}^{k} |t_i|$. 
So $\dertime_{\to}(n)$ denotes the length of the longest $\to$-rewrite
sequence starting with an arbitrary term of at most size~$n$.
For a TRS $\R$ and a strategy $s \in \IS$, the \emph{derivational complexity} of $\R$ w.r.t.\ $s$
is $\dertime_{\sto_{\R}}$.

In contrast, for \emph{runtime complexity} one restricts the start terms to be \emph{basic}.
For a TRS $\R$, we decompose its signature $\Sigma = \Sigma_{C} \uplus \Sigma_{D}$ such
that $f \in \Sigma_D$ if $f = \rootsym(\ell)$ for some rule $\ell \to r \in \R$. 
The symbols in $\Sigma_{C}$ and $\Sigma_{D}$ are called \emph{constructors} and
\emph{defined symbols}, respectively. 
A term $t \in \TT$ is \emph{basic} if $t = f(t_1, \ldots, t_k)$ such that $f \in \Sigma_D$ 
and $t_i \in \TSet{\Sigma_C}{\VSet}$ for all $1 \leq i \leq k$, and the set of all basic terms is denoted by $\TT_{\BB}$.
The \emph{runtime complexity} of a relation ${\to} \subseteq \TT \times \TT$ w.r.t.\ a TRS $\R$ maps any $n \in \IN$ 
to the length of the longest $\to$-rewrite sequence starting with a basic term $t \in \TT_{\BB}$ with $|t| \leq n$. 
So $\runtime_{\to, \R} : \IN \to \IN \cup \{\omega\}$ is defined as 
$\runtime_{\to, \R}(n) = \sup\{\dheight_\to(t) \mid t \in \TT_{\BB}, |t| \leq n\}$, where
$\TT_{\BB}$ are  the basic terms w.r.t.\ $\R$.
For a TRS $\R$ and a strategy $s \in \IS$, the \emph{runtime complexity} of $\R$ w.r.t.\ $s$
is $\runtime_{\sto_{\R},\R}$, where we often omit the additional index $\R$ for
readability, i.e., $\runtime_{\sto_{\R},\R} = \runtime_{\sto_{\R}}$.

We are not aware of any non-trivial
classes of TRSs where innermost and full derivational complexity coincide,
i.e., where we have $\dertime_{\fto_{\R}}(n) = \dertime_{\ito_{\R}}(n)$ for
 all $n \in \IN$ and
all TRSs $\R$ from that class.
However, for runtime complexity, there exist sufficient criteria for TRSs
$\R$ which imply 
$\runtime_{\fto_{\R}} = \runtime_{\ito_{\R}}$~\cite{frohn_analyzing_nodate}.
To relate innermost and full runtime complexity, it does not suffice to require 
that the rules are non-overlapping, but we also have to make sure that we cannot
duplicate redexes during an evaluation, as shown by the following counterexample from~\cite{frohn_analyzing_nodate}.

\begin{counterexample}\label{example:runtime-vs-iruntime}
    Let $\ts^n(\tz)$ abbreviate $\underbrace{\ts(\ldots \ts(}_{n\,\text{times}}\!\tz)\ldots)$
    and let $\R_4$ consist of the rules:

    \vspace*{-.6cm}
    \begin{minipage}[t]{7cm}
        \begin{align*}
            \tf(\tz,y) &\to y\\
            \tf(\ts(x),y) &\to \tf(x,\tnode(y,y))\!
        \end{align*}
    \end{minipage}
    \begin{minipage}[t]{7cm}
        \begin{align*}
            \tg(x) &\to \tf(x,\ta)\\
            \ta &\to \tb\!
        \end{align*}
    \end{minipage}
    \vspace*{.2cm}

    \noindent
    Note that $\R_4$ is non-overlapping.
    For the basic term $\tg(\ts^n(\tz))$ of size $n + 2$ we have $\tg(\ts^n(\tz)) \fto_{\R_4} \tf(\ts^n(\tz), \ta)$.
    Next, one could apply the second $\tf$-rule repeatedly and 
    obtain a term corresponding to a full binary tree of height $n$ whose
    exponentially many leaves all correspond to the symbol $\ta$.
    Finally, these leaves can all be reduced to $\tb$ in $2^n$ steps,
    hence $\runtime^{\R_4}_{\mathbf{f}}(n) \in \Theta(2^n)$.
    On the other hand, any basic term of size $n$ only leads to innermost rewrite sequences of length $\O(n)$,
    as for example in $\tf(\ts^n(\tz), \ta)$ we have to rewrite the inner $\ta$ before we
    are able to duplicate it.
    Thus, we obtain $\runtime^{\R_4}_{\mathbf{i}}(n) \in \Theta(n)$.
\end{counterexample}

The property that no redexes are duplicated during a rewrite sequence 
that starts with a basic term is called \emph{spareness}~\cite{frohn_analyzing_nodate}.
For a TRS $\R$, a rewrite step using the rule $\ell \to r \in \R$ and the substitution
$\sigma$ is \emph{spare} 
if $\sigma(x) \in \NF_{\R}$ for every $x \in \Var$ that occurs more than once in $r$.
A $\fto_\R$-rewrite sequence is spare if each of its $\fto_\R$-steps is spare.
$\R$ is spare (SP) if each $\fto_\R$-rewrite sequence that starts with a basic term $t \in \TT_{\BB}$ is spare.
In general, it is undecidable whether a TRS is spare.
However, there exist computable sufficient conditions for
spareness, see~\cite{frohn_analyzing_nodate}.

Similar to the results regarding termination, for the equivalence of
     $\runtime_{\fto_{\R}}$ and $\runtime_{\ito_{\R}}$
it again suffices to require
overlay systems instead of non-overlapping TRSs.
However, while we have to require spareness, in contrast to \Cref{properties-eq-SN-iSN-3}
we do not need local confluence.

\begin{restatable}[From $\runtime_{\ito_{\R}}$ to $\runtime_{\fto_{\R}}$,~\cite{frohn_analyzing_nodate}]{thm}{irunvsrun}\label{properties-irun-vs-run}
    If a TRS $\R$ is OS and SP, then:
    \begin{align*}
        \runtime_{\fto_{\R}} &= \runtime_{\ito_{\R}}\!
    \end{align*}
\end{restatable}

Counterex.\
\ref{example:runtime-vs-iruntime} also works for derivational complexity, 
showing that duplication of redexes needs to be prohibited 
for equivalence of $\dertime_{\fto_{\R}}$ and $\dertime_{\ito_{\R}}$ as well.
However, note that \emph{spareness} requires basic start terms, 
so an analogous theorem for derivational complexity would have to require a stronger
property,
e.g., right-linearity of $\R$.
We will present such a result as a corollary of our analysis of probabilistic rewriting
in \Cref{Relating AST and its Restricted Forms} (\Cref{properties-eq-Der-iDer-1}).
Finally,
the relation between the runtime complexities for 
rewrite strategies
that use an ordering for parallel redexes has not been studied so far.
Here,  our probabilistic analysis will also yield a corresponding
corollary (\Cref{properties-eq-iDer-liDer-complex}).

\section{Probabilistic Term Rewriting}\label{Probabilistic Term Rewriting}

In \Cref{sect:Probabilistic Term Rewriting}, we define \emph{probabilistic
TRSs}~\cite{bournez2005proving,avanzini2020probabilistic,kassinggiesl2023iAST}. 
In \Cref{sect:Probabilistic Notion of Termination}, we re\-ca\-pi\-tu\-late
 notions of termination 
in the probabilistic setting, like almost-sure termination ($\AST$), positive almost-sure termination ($\PAST$),
strong almost-sure termination ($\SAST$), and we introduce a novel definition
of \emph{expected derivational/runtime complexity}.
Then in \Cref{sect:Relating PAST and SAST}, we
present new results on the
relation between $\PAST$ and $\SAST$.

\subsection{Probabilistic Term Rewriting}\label{sect:Probabilistic Term Rewriting}

In contrast to TRSs, a PTRS has finite\footnote{The restriction to
\emph{finite} multi-distributions allows us to simplify the handling of PTRSs in the proofs. However, we conjecture that most of our results also hold for PTRSs with infinite countable multi-distributions.}
multi-distributions on the right-hand sides of its rewrite rules.\footnote{A different form of probabilistic rewrite rules was proposed in \textsf{PMaude}~\cite{pmaude_2006}, where numerical extra variables in right-hand sides of rules are instantiated according to a probability distribution.}
A finite \emph{multi-distribution} $\mu$ on a set $A \neq \emptyset$ is a finite multiset of pairs $(p:a)$, where $0 < p \leq 1$ is a probability and $a \in A$, such that $\sum _{(p:a) \in \mu}p = 1$.
$\FDist(A)$ is the set of all finite multi-distributions on $A$.
For $\mu\in\FDist(A)$, its \emph{support} is the multiset $\Supp(\mu)=\{a \mid (p:a)\in\mu$ for some $p\}$.
A \emph{probabilistic rewrite rule} is a pair $(\ell \to \mu) \in \TT \times \FDist(\TT)$ such that $\ell \not\in \VSet$ and $\VSet(r) \subseteq \VSet(\ell)$ for every $r \in \Supp(\mu)$.
A \emph{probabilistic TRS} (PTRS) is a (possibly infinite) countable set $\PP$ of probabilistic rewrite rules.
Similar to TRSs, the PTRS $\PP$ induces a \emph{rewrite relation} ${\fto_{\PP}} \subseteq \TT \times \FDist(\TT)$ where $s \fto_{\PP} \{p_1:t_1, \ldots, p_k:t_k\}$  if there is a position $\pi$, a rule $\ell \to \{p_1:r_1, \ldots, p_k:r_k\} \in \PP$, and a substitution $\sigma$ such that $s|_{\pi}=\ell\sigma$ and $t_j = s[r_j\sigma]_{\pi}$ for all $1 \leq j \leq k$.
We call $s \fto_{\PP} \mu$ an \emph{innermost} rewrite step (denoted $s \ito_{\PP} \mu$)
if all proper subterms of the used redex $\ell\sigma$ are in normal form w.r.t.\ $\PP$.
We have $s \lito_{\PP} \mu$ if the rewrite step $s \ito_{\PP} \mu$ at position $\pi$ is
leftmost (i.e., there is no redex at a position $\tau$ with $\tau \prec \pi$).
For example, the PTRS $\PP_{\trw}$ with the only rule $\tg \to
\{\nicefrac{1}{2}:\tc(\tg,\tg), \; \nicefrac{1}{2}:\tz\}$ corresponds to a symmetric
random walk on the number of $\tg$-symbols in a term.

Many properties of TRSs from \Cref{Preliminaries} can be lifted to PTRSs in a
straightforward way:
A PTRS $\PP$ is right-linear (non-erasing) iff the TRS
$\{\ell \to r \mid \ell \to \mu \in \PP, r \in \Supp(\mu)\}$
has the respective property.
Moreover, all properties that just consider the left-hand sides, e.g., normal forms, left-linearity,
being
non-overlapping, orthogonality,  and being an overlay system,
can be lifted to PTRSs directly as well, since their rules again only have a single
left-hand side.

\subsection{Probabilistic Notions of Termination}\label{sect:Probabilistic Notion of Termination}

Next, we introduce the different notions of probabilistic termination.
In this section, we regard an arbitrary \emph{probabilistic (term) relation} $\to \; \subseteq \; \TT \times \FDist(\TT)$
and let $\NF_{\to}$ again be the set of all normal forms for $\to$.
As in~\cite{diazcaro_confluence_2018,avanzini2020probabilistic,Faggian2019ProbabilisticRN,kassinggiesl2023iAST}, we
\emph{lift} $\to \; \subseteq \TT \times \FDist(\TT)$ to a rewrite relation $\liftto \; \subseteq \FDist(\TT) \times \FDist(\TT)$ between multi-distributions in order to track all probabilistic rewrite sequences (up to non-determinism) at once.
For any $0 < p \leq 1$ and any $\mu \in \FDist(\TT)$, let $p \cdot \mu = \{ (p\cdot q:a) \mid (q:a) \in \mu \}$.

\begin{defi}[Lifting]\label{def:lifting}
    The \emph{lifting} ${\liftto} \subseteq \FDist(\TT) \times \FDist(\TT)$ of a relation ${\to} \subseteq \TT \times \FDist(\TT)$ is the smallest relation with:
	\begin{itemize}
	    \item[$\bullet$] If $t \in \TT$ is in normal form w.r.t.~$\rightarrow$, then $\{1: t\} \liftto \{1:t\}$.
		\item[$\bullet$] If $t \to \mu$, then $\{1: t\} \liftto \mu$.
		\item[$\bullet$] If  for all $1 \leq j \leq k$ there are $\mu_j, \nu_j \in
                  \FDist(\TT)$ with $\mu_j \liftto \nu_j$ and $0 < p_j
                  \leq 1$ with $\sum_{1 \leq j \leq k} p_j = 1$, then $\bigcup_{1 \leq j
                    \leq k} p_j \cdot \mu_j \liftto \bigcup_{1 \leq j \leq k} p_j \cdot
                  \nu_j$. 
	\end{itemize}
\end{defi}
For a PTRS $\PP$, we write $\fliftto_\PP$, $\iliftto_{\PP}$, and $\liliftto_{\PP}$ for
the liftings of $\fto_{\PP}$, $\ito_{\PP}$, and $\lito_{\PP}$, respectively.

\begin{exa}\label{example:PTRS-random-walk-lifting-sequence}
  We obtain the following $\liftArr{\mathbf{f}}{}{\PP_{\trw}}$-rewrite sequence (which is also a
  $\liftArr{\mathbf{i}}{}{\PP_{\trw}}$-rewrite sequence, but not a $\liftArr{\mathbf{li}}{}{\PP_{\trw}}$-rewrite sequence).
    {\small
	\[
        \begin{array}{ll}
            &\{1:\tg\}\\
	        \liftArr{\mathbf{f}}{}{\PP_{\trw}}&\{\nicefrac{1}{2}:\tc(\tg,\tg), \nicefrac{1}{2}:\tz\}\\
            \liftArr{\mathbf{f}}{}{\PP_{\trw}}& \{\nicefrac{1}{4}:\tc(\tc(\tg,\tg),\tg),\nicefrac{1}{4}:\tc(\tz,\tg), \nicefrac{1}{2}:\tz \}\\
            \liftArr{\mathbf{f}}{}{\PP_{\trw}}& \{ \nicefrac{1}{8}:\tc(\tc(\tg,\tg),\tc(\tg,\tg)), \nicefrac{1}{8}:\tc(\tc(\tg,\tg),\tz), \nicefrac{1}{8}:\tc(\tz,\tc(\tg,\tg)), \nicefrac{1}{8}:\tc(\tz,\tz), \nicefrac{1}{2}:\tz \}
        \end{array}
    \]}
\end{exa}

Another way to track all possible rewrite sequences with their corresponding probabilities
is to lift $\to$ to \emph{rewrite sequence trees (RSTs)}~\cite{FLOPS2024}.
The nodes $v$ of an $\to$-RST are labeled by pairs $(p_v:t_v)$ of a
probability $p_v$ and a term $t_v$, where
the root is always labeled with the probability $1$.
For each node $v$ with the successors $w_1, \ldots, w_k$, the edge relation represents a
step with the relation $\to$,
i.e., $t_v \to \{\tfrac{p_{w_1}}{p_v}:t_{w_1}, \ldots, \tfrac{p_{w_k}}{p_v}:t_{w_k}\}$.
For an $\to$-RST $\F{T}$,
let $N^{\F{T}}$ denote the set of nodes and $\ctleaf^{\F{T}}$ denote the set of leaves.
We say that $\F{T}$ is \emph{fully evaluated} if for every $v \in \ctleaf^{\F{T}}$
the corresponding term $t_v$ is a normal form w.r.t.\ $\to$, i.e., $t_v \in \NF_{\to}$.
In \Cref{fig:RST_example} one can see the $\fto_{\PP_{\trw}}$-RST for the $\liftArr{\mathbf{f}}{}{\PP_{\trw}}$-rewrite sequence from \Cref{example:PTRS-random-walk-lifting-sequence}.
Note that the normal forms remain in each multi-distribution of a
$\liftArr{\mathbf{f}}{}{\PP_{\trw}}$-rewrite sequence,  
but they are leaves of the corresponding $\fto_{\PP_{\trw}}$-RST.\@

\begin{figure}
	\centering
    \begin{tikzpicture}
        \tikzstyle{adam}=[thick,draw=black!100,fill=white!100,minimum size=4mm, shape=rectangle split, rectangle split parts=2,rectangle split horizontal]
        \tikzstyle{empty}=[rectangle,thick,minimum size=4mm]
        
        \node[empty] at (-6, 0)  (m1) {$\mu_0:$};
        \node[adam] at (0, 0)  (a) {$1$\nodepart{two}$\tg$};

        \node[empty] at (-6, -1)  (m2) {$\mu_1:$};
        \node[adam] at (-2, -1)  (b) {$\nicefrac{1}{2}$\nodepart{two}$\tc(\tg,\tg)$};
        \node[adam,label=below:{$\quad \NF_{\PP_{\trw}}$}] at (2, -1)  (c) {$\nicefrac{1}{2}$\nodepart{two}$\tz$};

        \node[empty] at (-6, -2)  (m3) {$\mu_2:$};
        \node[adam] at (-4, -2)  (d) {$\nicefrac{1}{4}$\nodepart{two}$\tc(\tc(\tg,\tg),\tg)$};
        \node[adam] at (0, -2)  (e) {$\nicefrac{1}{4}$\nodepart{two}$\tc(\tz,\tg)$};

        \node[empty] at (-5, -3)  (f) {$\ldots$};
        \node[empty] at (-3, -3)  (g) {$\ldots$};
        \node[empty] at (-1, -3)  (h) {$\ldots$};
        \node[empty] at (1, -3)  (i) {$\ldots$};
        
        \draw (a) edge[->] (b);
        \draw (a) edge[->] (c);
        \draw (b) edge[->] (d);
        \draw (b) edge[->] (e);
        \draw (d) edge[->] (f);
        \draw (d) edge[->] (g);
        \draw (e) edge[->] (h);
        \draw (e) edge[->] (i);
    \end{tikzpicture}
	  \caption{Corresponding RST for the $\fliftto_{\PP_{\trw}}$-rewrite sequence in \Cref{example:PTRS-random-walk-lifting-sequence}.}\label{fig:RST_example}
\end{figure}

To express the concept of almost-sure termination, one has to determine the probability for normal forms in a multi-distribution.

\begin{defi}[$|\mu|_{\to}$, $|\mu|_{\PP}$]\label{def:prob-abs-value}
    Let $\mu \in \FDist(\TT)$.
    For a relation $\to \; \subseteq \TT \times \FDist(\TT)$, let $|\mu|_{\to} = \sum_{(p:t) \in \mu, t \in \NF_{\to}} p$,
    and for a PTRS $\PP$, let $|\mu|_{\PP} = \sum_{(p:t) \in \mu, t \in \NF_{\PP}} p$.
\end{defi}

\begin{exa}\label{example:prob-abs-value}
    Consider $\{ \nicefrac{1}{8}:\tc(\tc(\tg,\tg),\tc(\tg,\tg)), \nicefrac{1}{8}:\tc(\tc(\tg,\tg),\tz), \nicefrac{1}{8}:\tc(\tz,\tc(\tg,\tg)),\linebreak \nicefrac{1}{8}:\tc(\tz,\tz), \nicefrac{1}{2}:\tz \} = \mu$ from \Cref{example:PTRS-random-walk-lifting-sequence}.
    Then $|\mu|_{\PP_{\trw}} = \nicefrac{1}{8} + \nicefrac{1}{2}  = \nicefrac{5}{8}$,
    since $\tc(\tz,\tz)$ and $\tz$ are both normal forms w.r.t.\ $\PP_{\trw}$.
\end{exa}

\begin{defi}[Almost-Sure Termination~\cite{avanzini2020probabilistic}]\label{def:ptrs-innermost-term-innermost-AST}
    Let $\to \; \subseteq \TT \times \FDist(\TT)$
    and $\vec{\mu} = (\mu_n)_{n \in \IN}$ be an infinite
    $\liftto$-rewrite
    sequence, i.e., $\mu_n \liftto \mu_{n+1}$ for all $n \in \IN$. We say that
    $\vec{\mu}$ \emph{converges with probability} $\lim\limits_{n \to \infty}|\mu_n|_{\to}$.
    The relation $\to$ is \emph{almost-surely terminating} (denoted $\AST[\to]$) 
    if $\lim\limits_{n \to \infty} |\mu_n|_{\to} = 1$
    holds for every infinite $\liftto$-rewrite sequence $(\mu_n)_{n \in \IN}$.   
    We say that $\to$ is \emph{weakly $\AST$} (denoted $\wAST[\to]$) if for every term
    $t$ there exists an infinite $\liftto$-rewrite sequence $(\mu_n)_{n \in \IN}$ 
    with $\lim\limits_{n \to \infty} |\mu_n|_{\to} = 1$ and $\mu_0 = \{1:t\}$.
\end{defi}

For the definition of $\wAST$ recall that
by \Cref{def:lifting} every term (even normal forms) 
can start infinite $\liftto$-rewrite sequences,
as we keep normal forms in $\liftto$-steps. 

Equivalently, one can also define $\AST[\to]$ (and $\wAST[\to]$) via $\to$-RSTs.
For any $\to$-RST $\F{T}$ we define its 
\emph{convergence probability} $|\F{T}| = \sum_{v \in \ctleaf^{\F{T}}} p_v$.
Then $\AST[\to]$ holds iff for all $\to$-RSTs $\F{T}$ we have $|\F{T}| = 1$.
Moreover, $\wAST[\to]$ holds iff for every term $t$ there exists a fully evaluated $\to$-RST
$\F{T}$ whose root is labeled with $(1:t)$ such that $|\F{T}| = 1$.

\begin{exa}\label{example:rw-is-AST} 
    For every infinite extension $(\mu_n)_{n \in \IN}$ of the $\liftArr{\mathbf{f}}{}{\PP_{\trw}}$-rewrite sequence in
    \cref{example:PTRS-random-walk-lifting-sequence}, we have $\lim\limits_{n \to \infty}
    |\mu_n|_{\PP_{\trw}} = 1$. Indeed, we have $\AST[\fto_{\PP_{\trw}}]$ 
    and thus also $\AST[\ito_{\PP_{\trw}}]$, $\AST[\lito_{\PP_{\trw}}]$, and $\wAST[\fto_{\PP_{\trw}}]$.
\end{exa}

Next, we define \emph{positive} almost-sure termination, which 
considers the \emph{expected derivation length} $\edl(\vec{\mu})$ of a rewrite sequence $\vec{\mu}$,
i.e., the expected number of steps until one reaches a normal form.
For positive almost-sure termination, we require that the expected derivation length 
of every possible rewrite sequence is finite.
In the following definition, $(1 - |\mu_n|_{\to})$ is the probability of terms 
that are \emph{not} in normal form w.r.t.\ $\to$ after the $n$-th step.

\begin{defi}[Positive Almost-Sure Termination, $\edl$~\cite{avanzini2020probabilistic}]\label{def:expected-derivation-length-PAST}
    Let $\to \; \subseteq \TT \times \FDist(\TT)$ and $\vec{\mu} = (\mu_n)_{n \in \IN}$ be an infinite 
    $\liftto$-rewrite sequence.
    By $\edl(\vec{\mu}) = \sum_{n = 0}^{\infty} (1 - |\mu_n|_{\to})$
    we denote the \emph{expected derivation length} of $\vec{\mu}$.
    The relation $\to$ is \emph{positively almost-surely terminating} (denoted $\PAST[\to]$)
    if $\edl(\vec{\mu})$ is finite for every infinite 
    $\liftto$-rewrite sequence $\vec{\mu} = (\mu_n)_{n \in \IN}$
    starting with a single term, i.e., $\mu_0 = \{1:t\}$ with $t \in \TT$.
    We say that $\to$ is \emph{weakly $\PAST$} (denoted $\wPAST[\to]$) if for every
    term $t$ there exists an infinite $\liftto$-rewrite
    sequence $\vec{\mu} = (\mu_n)_{n \in \IN}$ such that $\edl(\vec{\mu})$ is finite and $\mu_0 = \{1:t\}$.
\end{defi}

In terms of $\to$-RSTs, we define the \emph{expected derivation length} of an
$\to$-RST $\F{T}$ to be $\edl(\F{T}) = \sum_{v \in N^{\F{T}} \setminus \ctleaf^{\F{T}}} p_v$.
Then, we have $\PAST[\to]$ iff $\edl(\F{T})$ is finite
for every $\to$-RST $\F{T}$.
Similarly, we have $\wPAST[\to]$ iff for every term $t$ there exists a fully evaluated $\to$-RST $\F{T}$
whose root is labeled with $(1:t)$ such that $\edl(\F{T})$ is finite.

\begin{rem}
    For every $\liftto$-rewrite sequence $\vec{\mu}$ that converges with probability 1, 
    we also have $\edl(\vec{\mu}) = \sum_{n = 0}^{\infty} (1 - |\mu_n|_{\to})
    = \sum_{n = 1}^{\infty} n \cdot (|\mu_n|_{\to} - |\mu_{n-1}|_{\to})$, 
    where $|\mu_n|_{\to} - |\mu_{n-1}|_{\to}$ denotes the probability that we reach a normal form in the $n$-th step.
    This is due to the correspondence between a probability mass
    function $f:\IN \to [0,1]$ and its distribution function $F:\IN \to
    [0,1]$ w.r.t.\ the expected value, namely $\IE(f) = \sum_{n = 1}^{\infty} n \cdot f(n)
    = \sum_{n = 1}^{\infty} (1 - F(n))$.
    With $f_{\vec{\mu}}(0) = |\mu_0|_{\to}$ and $f_{\vec{\mu}}(n) = |\mu_n|_{\to} -
    |\mu_{n-1}|_{\to}$ for all $n > 0$, we get $F_{\vec{\mu}}(n) = |\mu_n|_{\to}$, and
    hence we obtain the above equation for $\edl(\vec{\mu})$.    
    Note that $f_{\vec{\mu}}$ is only a probability mass function if $\lim_{n \to \infty} |\mu_{n}|_{\to} = 1$,
    because then $\sum_{n=0}^{\infty} f_{\vec{\mu}}(n) 
    = |\mu_0|_{\to} + \sum_{n=1}^{\infty} (|\mu_n|_{\to} - |\mu_{n-1}|_{\to})
    = \lim_{n \to \infty} |\mu_{n}|_{\to} = 1$.
\end{rem}

It is well known that $\PAST[\to]$ implies $\AST[\to]$, but not vice versa.

\begin{exa}
    For every infinite extension $\vec{\mu} = (\mu_n)_{n \in \IN}$ of the $\liftArr{\mathbf{f}}{}{\PP_{\trw}}$-rewrite sequence in
    \cref{example:PTRS-random-walk-lifting-sequence}, the expected derivation length
    $\edl(\vec{\mu})$ is infinite, hence $\wPAST[\fto_{\PP_{\trw}}]$ does not hold, and
        $\PAST[\fto_{\PP_{\trw}}]$, $\PAST[\ito_{\PP_{\trw}}]$, or
    $\PAST[\lito_{\PP_{\trw}}]$ do not hold either.
\end{exa}

Next, we define \emph{strong almost-sure termination}~\cite{fuTerminationNondeterministicProbabilistic2019,avanzini2020probabilistic},
which is even stricter than $\PAST$ in  case of non-determinism.
It requires a finite bound on the expected derivation lengths of all rewrite
sequences with the same start term.
For a term $t \in \TT$, the \emph{expected derivation height} $\edh_\to(t)$ 
considers all $\liftArr{}{}{}$-rewrite sequences that start with $\{1:t\}$.

\begin{defi}[Strong Almost-Sure Termination, $\edh$~\cite{avanzini2020probabilistic}]\label{def:expected-derivation-height-SAST}
    For every relation $\to \; \subseteq \TT \times \FDist(\TT)$ we define 
    the \emph{expected derivation height} of a term $t \in \TT$ by 
    $\edh_\to(t) = \sup\{\edl(\vec{\mu}) \mid \vec{\mu}= (\mu_n)_{n \in \IN}
    \text{ is a
    $\liftArr{}{}{}$-rewrite sequence with } \mu_0 = \{1:t\}\}$.    
    We say that $\to$ is \emph{strongly almost-surely terminating} ($\SAST[\to]$)
    if $\edh_\to(t)$ is finite for all $t \in \TT$.
\end{defi}

In terms of $\to$-RSTs, we have $\SAST[\to]$ iff 
$\sup\{\edl(\F{T}) \mid \F{T}$ is an $\to$-RST whose root is labeled with $(1:t)\}$
is finite for all $t \in \TT$.

Note that in contrast to $\wAST$ and $\wPAST$, we did not define any notion of \emph{weak} $\SAST$. The
reason is that the definition of weak forms of termination  always only requires the
\emph{existence} of some suitable rewrite sequence. But the definition of strong almost-sure termination
imposes a requirement on \emph{all} rewrite sequences. Thus, it is not clear how to
obtain a useful definition for ``weak $\SAST$'' that differs from $\wPAST$.

$\PAST[\to]$ and $\SAST[\to]$ are already defined 
in terms of the expected derivation length and height,
but they consider arbitrary start terms of arbitrary size.
As in the non-probabilistic setting, one can also regard
a function that maps the size of the start term to
the expected derivation length,
leading to the novel notion of the \emph{expected derivational complexity}.

\begin{defi}[Expected Derivational Complexity, $\edertime_{\to}$]\label{def:expected-derivational-complexity}
    For a relation $\to \; \subseteq \TT \times \FDist(\TT)$, 
    we define its \emph{expected derivational complexity} 
    $\edertime_{\to} : \IN \to \IN \cup \{\omega\}$ as
    $\edertime_{\to}(n) = \sup\{\edh_\to(t) \mid t \in \TT, |t| \leq n\}$.
\end{defi}

Note that as in the non-probabilistic setting,
this definition uses the expected derivation height of a term,
and not the expected derivation length of a rewrite sequence.
Hence, the expected derivational complexity $\edertime_{\to}$ relates to
$\SAST[\to]$ more than to $\PAST[\to]$.
Indeed, we may have both $\PAST[\to]$ and $\edertime_{\to} \in
\Theta(\omega)$, see 
Counterex.\ \ref{example:Past-vs-Sast-finite} in the next section.
Here, $\edertime_{\to} \in \Theta(\omega)$ 
means that there is some $n \in \IN$ such that $\edertime_{\to}(n) = \omega$,
and hence, $\edertime_{\to}(n') = \omega$ for all $n' \geq n$.
Furthermore, we have the following easy observation regarding $\edertime_{\to}$ and $\SAST[\to]$,
where $\edertime_{\to} \in o(\omega)$ means that for all $n \in \IN$ we have $\edertime_{\to}(n) < \omega$.

\begin{lem}[Relation between $\edertime$ and $\SAST$]\label{lemma:SAST-via-eDer}
    Let $\to \; \subseteq \TT \times \FDist(\TT)$ where $\TT =\TSet{\Sigma}{\VSet}$ 
    for a finite signature $\Sigma$.
    Then $\edertime_{\to} \in o(\omega)$ iff $\SAST[\to]$.
\end{lem}

\begin{proof}
    If we do not have $\SAST[\to]$, then $\edertime_{\to} \in \Theta(\omega)$
    follows directly from the definition.
    On the other hand, if $\Sigma$ is finite, then $\edertime_{\to}(n) = \sup\{\edh_\to(t) \mid t \in \TT, |t| \leq n\}$ where $\{t \in \TT \mid |t| \leq n\}$ is
    finite (up to variable renamings), so that the supremum
    is equal to $\edh_\to(t)$ for a term $t \in \TT$ with $|t| \leq n$.
    Hence, if $\edertime_{\to} \in \Theta(\omega)$, then there exists some
    $n \in \IN$ and a term $t \in \TT$ such that $|t| \leq n$ and $\edh_{\to}(t) =
    \omega$, such that $\SAST[\to]$ does not hold.
\end{proof}

Finally, for expected runtime complexity we additionally require basic start terms again.

\begin{defi}[Expected Runtime Complexity, $\eruntime_{\to,\PP}$]\label{def:expected-runtime-complexity}
    For a relation $\to \; \subseteq \TT \times \FDist(\TT)$, the \emph{expected runtime complexity} $\eruntime_{\to, \PP} : \IN \to \IN \cup \{\omega\}$ w.r.t.\ a PTRS $\PP$ is 
    $\eruntime_{\to, \PP}(n) = \sup\{\edh_\to(t) \mid t \in \TT_{\BB}, |t| \leq n\}$,
    where
$\TT_{\BB}$ denotes the basic terms w.r.t.\ $\PP$.
\end{defi}

When considering the expected runtime complexity
for a rewrite relation like $\fto_{\PP}$, we again omit the additional index $\PP$, 
i.e., $\eruntime_{\fto_{\PP}, \PP} = \eruntime_{\fto_{\PP}}$.

\subsection{Relating Positive and Strong Almost-Sure Termination}\label{sect:Relating PAST and SAST}

We now present novel results on the
relation between $\PASTs$ and $\SASTs$.
Similar to the implication from $\PAST[\to]$ to $\AST[\to]$,
it is well known that $\SAST[\to]$ implies $\PAST[\to]$,
and this implication is again strict~\cite{avanzini2020probabilistic}.
While the counterexample for $\SASTf = \PASTf$
in~\cite{avanzini2020probabilistic} uses infinitely many rules, the following new example
shows that even for PTRSs with \emph{finitely} many rules, 
$\SASTf$ and $\PASTf$ are not equivalent.

\begin{counterexample}\label{example:Past-vs-Sast-finite}
    Consider the PTRS $\PP_{\mathsf{unary}}$ with the rules

    \vspace*{-.3cm}
    \begin{minipage}[t]{7cm}
        \begin{align*}
            \tf(x) &\to \{ \nicefrac{1}{2}:\tf(\ts(x)), \nicefrac{1}{2}:\tb\}\\
            \tf(x) &\to \{1:\tg(x)\}\\
            \tg(\tz) &\to \{1:\tz\}\\            
            \tg(\ts(x)) &\to \{1:\tg_1(\ts(x))\}\\
            \tg_1(\ts(x)) &\to \{1:\th(\ts(x))\}
        \end{align*}
    \end{minipage}
    \begin{minipage}[t]{7cm}
        \begin{align*}
            \th(\ts(x)) &\to \{1:\tq(\th(x))\}\\
            \th(\tz) &\to \{1:\ta(\ta(\ta(\ta(\tz))))\}\\
            \tq(\ta(x)) &\to \{1: \tq_1(x)\} \\
            \tq_1(x) &\to \{1: \tq_2(x)\} \\
            \tq_2(x) & \to  \{1:\ta(\ta(\ta(\ta(\tq(x)))))\}
        \end{align*}
    \end{minipage}
    \vspace*{.2cm}

    \noindent
    and the following rewrite sequence:
    \[
    \begin{array}{lllll}
    \tg(\ts^n(\tz)) &\to_{\PP_{\mathsf{unary}}}^2&    \th(\ts^n(\tz))  &\to_{\PP_{\mathsf{unary}}}^{n} & \tq^n(\th(\tz))\\
                    &\to_{\PP_{\mathsf{unary}}}         &\tq^n(\ta^4(\tz))             &\to_{\PP_{\mathsf{unary}}}^{3 * 4} & \tq^{n-1}(\ta^{4^2}(\tq(\tz)))\\
                    & \to_{\PP_{\mathsf{unary}}}^{3 * 4^2} &\tq^{n-2}(\ta^{4^3}(\tq^2(\tz)))
    &\to_{\PP_{\mathsf{unary}}}^{3 * 4^3}& \ldots\\
    & \to_{\PP_{\mathsf{unary}}}^{3 * 4^{n-1}} &\tq(\ta^{4^n}(\tq^{n-1}(\tz))) &
    \to_{\PP_{\mathsf{unary}}}^{3 * 4^{n}} &\ta^{4^{n+1}}(\tq^{n}(\tz))
    \end{array}
    \]
    To ease readability, here we wrote $\to_{\PP_{\mathsf{unary}}}$ instead of
    $\fliftto_{\PP_{\mathsf{unary}}}$ and we 
    omitted the multi-distributions, as all the used rules have trivial probabilities, 
    i.e., they are of the form $\ell \to \{1:r\}$ for some $\ell, r \in \TT$.
    Thus, for every $n \in \IN$ with $n > 0$, the term $\tg(\ts^n(\tz))$ has an expected
    derivation height of
    $2 + n + 1 + 3 * \sum_{i=1}^n 4^i = n + 3 * \sum_{i=0}^n 4^i  = n + 3 * \tfrac{4^{n+1} - 1}{3} = 4^{n+1} - 1 + n$.

    Next, consider the possible derivations of $\tf(\tz)$. 
    If we never use the rule from $\tf$ to $\tg$, then the expected derivation length is 
    $\tfrac{1}{2} \cdot 1 + \tfrac{1}{4} \cdot 2 + \tfrac{1}{8} \cdot 3 + \ldots =
    \sum_{i=1}^\infty \tfrac{i}{2^i} = 2$
    and if we never use the probabilistic $\tf$-rule, then the derivation length is $2$ as well.
    Otherwise, if we use the probabilistic
    $\tf$-rule in the first $k \in \IN$ steps with $k > 0$ and the rule from $\tf$ to $\tg$ in the $(k+1)$-th step, then 
    the expected derivation length is
    $(\sum_{i=1}^k \tfrac{i}{2^i}) +
    \tfrac{1}{2^k} +
    \tfrac{4^{k+1}-1+k}{2^k}
    = (\sum_{i=1}^k \tfrac{i}{2^i}) + 2^{k+2} + \frac{k}{2^k}$.
    Thus, for every rewrite sequence starting with $\tf(\tz)$, the expected derivation length is finite. 
    However, since $k$ can be any number, for any $k >0$, there is a rewrite sequence starting with 
    $\tf(\tz)$ whose expected derivation 
    length is $(\sum_{i=1}^k \tfrac{i}{2^i}) +
    2^{k+2} + \frac{k}{2^k} \geq 2^{k+2}$. 
    In other words, although every rewrite sequence starting with 
    $\tf(\tz)$ has finite expected derivation length, the supremum over the lengths
    for all these rewrite sequences is infinite.    
    Thus, we do not have $\SAST[\fto_{\PP_{\mathsf{unary}}}]$,
    as $\tf(\tz)$ has infinite expected derivation height, 
    but $\PAST[\fto_{\PP_{\mathsf{unary}}}]$ holds,
    as every rewrite sequence starting with 
    $\tf(\tz)$ has finite expected derivation length
    and a similar argument holds for every other start term.
\end{counterexample}

The PTRS in Counterex.\ \ref{example:Past-vs-Sast-finite} is of a very specific form:
Its signature contains only unary symbols and constants, 
and the probabilistic rule $\tf(x) \to \{\nicefrac{1}{2}:\tf(\ts(x)),
\nicefrac{1}{2}:\tb\}$ is erasing.
If we remove one of these properties, then
$\PAST[\fto_{\PP_{\mathsf{unary}}}]$ does not hold either.

\begin{exa}\label{example:Past-binary-function}
    Reconsider $\PP_{\mathsf{unary}}$ from Counterex.\ \ref{example:Past-vs-Sast-finite}.
    If we extend the signature by a binary symbol $\tc$, then
    $\PAST[\fto_{\PP_{\mathsf{unary}}}]$ does not hold anymore, as we can start with the term $\tc(\tf(\tz),
    \tf(\tz))$ and consider the following rewrite \pagebreak[3] sequence.

\vspace*{-.15cm}
    
    {\small
	  \[
        \begin{array}{ll}
            &\{1:\tc(\tf(\tz), \tf(\tz))\}\\
	        \liftArr{\mathbf{f}}{}{\PP_{\mathsf{unary}}}&\{\nicefrac{1}{2}:\underline{\textcolor{red}{\tc(\tb, \tf(\tz))}}, \nicefrac{1}{2}:\tc(\tf(\ts(\tz)), \tf(\tz))\}\\
	        \liftArr{\mathbf{f}}{}{\PP_{\mathsf{unary}}}&\{\hspace*{1.2cm}\ldots\hspace*{1.2cm}, \nicefrac{1}{4}:\underline{\textcolor{red}{\tc(\tb, \tf(\tz))}}, \nicefrac{1}{4}:\tc(\tf(\ts^2(\tz)), \tf(\tz))\}\\
	        \liftArr{\mathbf{f}}{}{\PP_{\mathsf{unary}}}&\{\hspace*{1.2cm}\ldots\hspace*{1.2cm}, \hspace*{1.2cm}\ldots\hspace*{1.2cm}, \nicefrac{1}{8}:\underline{\textcolor{red}{\tc(\tb, \tf(\tz))}}, \nicefrac{1}{8}:\tc(\tf(\ts^3(\tz)), \tf(\tz))\}\\
        \end{array}
    \]}

    \noindent
    Here, we use the term $\tf(\tz)$ in $\tc$'s first argument
    to create infinitely many copies 
    of the term $\tf(\tz)$ in $\tc$'s second argument
    using the probabilistic $\tf$-rule. 
    The underlined terms are not yet normal forms, as they contain the subterm $\tf(\tz)$, and will be evaluated further. 
    As in Counterex.\ \ref{example:Past-vs-Sast-finite}, for each $k > 0$, 
    there is a rewrite sequence starting with 
    $\tf(\tz)$ whose expected derivation length is at least $2^{k+2}$.
    We can choose $k = 1$ for the first underlined term, $k = 2$ for the second, etc.,
    leading to a final expected derivation length of at least 
    $\sum_{i=1}^\infty \tfrac{1}{2^i} \cdot 2^{i+2} = \sum_{i=1}^\infty 4$, which diverges to infinity.
    Hence, $\PAST[\fto_{\PP_{\mathsf{unary}}}]$  does not hold
    over the extended signature.
    In particular, this shows that in contrast to ordinary termination for non-probabilistic TRSs, 
    $\PAST$ is not preserved under extensions of the signature.

    Similarly, if we consider $\PP_{\mathsf{unary}}^{'}$ with the non-erasing rule
    $\tf(x) \to \{\nicefrac{1}{2}:\tf(\ts(x)), \nicefrac{1}{2}:\tb(x)\}$ using
    a unary function symbol $\tb$
    instead of $\tf(x) \to \{\nicefrac{1}{2}:\tf(\ts(x)), \nicefrac{1}{2}:\tb\}$, then we
    obtain
    the following rewrite sequence:
  {\small
	\[
        \begin{array}{ll}
            &\{1:\tf(\tf(\tz))\}\\
	        \liftArr{\mathbf{f}}{}{\PP_{\mathsf{unary}}}&\{\nicefrac{1}{2}:\underline{\textcolor{red}{\tb(\tf(\tz))}}, \nicefrac{1}{2}:\tf(\ts(\tf(\tz)))\}\\
	        \liftArr{\mathbf{f}}{}{\PP_{\mathsf{unary}}}&\{\hspace*{1cm}\ldots\hspace*{1cm}, \nicefrac{1}{4}:\underline{\textcolor{red}{\tb(\ts(\tf(\tz)))}}, \nicefrac{1}{4}:\tf(\ts^2(\tf(\tz)))\}\\
	        \liftArr{\mathbf{f}}{}{\PP_{\mathsf{unary}}}&\{\hspace*{1cm}\ldots\hspace*{1cm}, \hspace*{1.2cm}\ldots\hspace*{1.2cm}, \nicefrac{1}{8}:\underline{\textcolor{red}{\tb(\ts^2(\tf(\tz)))}}, \nicefrac{1}{8}:\tf(\ts^3(\tf(\tz)))\}\\
        \end{array}
    \]}

    \noindent
    Again, we can extend this sequence to an infinite rewrite sequence with an infinite expected derivation length.
    Hence, $\PAST[\fto_{\PP_{\mathsf{unary}}}]$ does not hold either.
\end{exa}

Counterex.\ \ref{example:Past-vs-Sast-finite} and \Cref{example:Past-binary-function} 
show that $\PASTs$ is not preserved under signature extensions for any strategy $s \in \IS$,
as the first rewrite sequence given in \Cref{example:Past-binary-function} is in fact a leftmost-innermost rewrite sequence.

\begin{restatable}[Signature Extensions for $\PASTs$]{thm}{SigPAST}\label{thm:Sig-PAST}
    Let $s \in \IS$.  
    There exists a PTRS $\PP$ and signatures $\Sigma, \Sigma'$ with $\Sigma \subset \Sigma'$ 
    such that $\PASTs$ holds over the signature $\Sigma$,
    but not over $\Sigma'$.
\end{restatable}

In contrast,
when analyzing modularity of $\ASTs$ and $\SASTs$
in \Cref{Modularity}, we will show
that they are closed under signature extensions (see \Cref{signature-extensions-AST-SAST}).

The core idea of both examples in \Cref{example:Past-binary-function} 
is that in the limit we can reach an infinite multi-distribution
whose support contains an infinite number of 
terms like $\tf(\tz)$ with unbounded expected derivation height.
A PTRS that allows such sequences is said to \emph{admit infinite splits}.

\begin{defi}[Infinite Splits]\label{def:infinite-splitting}
    A PTRS $\PP$ \emph{admits infinite splits}
    if for every term $t \in \TT$ there exists a term $t' \in \TT$
    and an $\fto_{\PP}$-RST whose root is labeled with
    $(1:t')$ such that there are infinitely many leaves labeled with
    terms that have $t$ as a subterm.
\end{defi}

\begin{exa}
  $\PP_{\mathsf{unary}}$ over the signature containing a function symbol of arity $\geq 2$
  and the non-erasing PTRS $\PP_{\mathsf{unary}}^{'}$ both admit infinite splits.
  However, $\PP_{\mathsf{unary}}$ over the signature 
  $\{\tf, \tb, \ts, \tz, \tg, \tg_1, \th, \ta, \tq, \tq_1, \tq_2\}$
  does not admit infinite splits.
\end{exa}

As shown by the following theorem,
admitting infinite splits implies
that $\PASTf$ is the same as $\SASTf$.

\begin{restatable}[Equivalence of $\PASTf$ and $\SASTf$ via Infinite Splits]{thm}{PASTvsSASTInfiniteSplits}\label{thm:PAST-vs-SAST-inf-splits}
    If a PTRS $\PP$ admits infinite splits, then:
    \begin{align*}
        \PASTf &\Longleftrightarrow \SASTf\!
    \end{align*}
\end{restatable}

\begin{proof}
    We only have to prove ``$\Longrightarrow$'', and do this via contraposition.
    Assume that $\SASTf$ does not hold. 
    Then there exists a term $t \in \TT$ such that for every $n \in \IN$ 
    there exists an $\fto_{\PP}$-RST $\F{T}_{n}^{t}$ whose root is labeled with
    $(1:t)$ such that $\edl(\F{T}_{n}^{t}) \geq n$.
    We now construct a single $\fto_{\PP}$-RST $\F{T}^{\infty}$ with $\edl(\F{T}^{\infty}) = \infty$,
    which implies that $\PASTf$ does not hold either.
    Since $\PP$ admits infinite splits, there exists a term $t' \in \TT$
    and an $\fto_{\PP}$-RST $\F{T}$ whose root is labeled with $(1:t')$ 
    such that there are infinitely many leaves whose corresponding terms have $t$ as a subterm.
    Let $v$ be a leaf in $\F{T}$ with $t_v^{\F{T}} = C[t]$ for some context $C$
    and let $n_v \in \IN$ be such that $p_v^{\F{T}} \geq \frac{1}{2^{n_v}}$.
    Then we can replace the leaf by the $\fto_{\PP}$-RST $\F{T}_{2^{n_v}}^{C[t]}$. Here, 
    $\F{T}_{2^{n_v}}^{C[t]}$ is the same tree as $\F{T}_{2^{n_v}}^{t}$ where in addition we have the context
    $C$ around every term of every node.
    Hence,
    $\edl(\F{T}_{2^{n_v}}^{C[t]}) =  \edl(\F{T}_{2^{n_v}}^{t}) \geq 2^{n_v}$.
    Let $\F{T}^{\infty}$ be the $\fto_{\PP}$-RST that results from performing this
    replacement
    for every leaf $v$ in $\F{T}$ 
    that contains $t$ as a subterm.
    Then, for $\F{T}^{\infty}$ we have 
    \begin{align*}\textstyle  
    \edl(\F{T}^{\infty}) &\textstyle\geq \sum_{v \in \ctleaf^{\F{T}} \land\;
t_v^{\F{T}} = C[t] \text{ for some context } C} \; p_v^{\F{T}} \cdot
    \edl(\F{T}_{2^{n_v}}^{C[t]})\\
    &\textstyle\geq \sum_{v \in \ctleaf^{\F{T}} \land\;
      t_v^{\F{T}} = C[t]} \; \frac{1}{2^{n_v}} \cdot 2^{n_v} \; = \; \sum_{v \in \ctleaf^{\F{T}} \land\;
      t_v^{\F{T}} = C[t]} \; 1 \; = \;\infty \tag*{\qedhere}
\end{align*} 
    
\end{proof}

\begin{rem}
    \Cref{thm:PAST-vs-SAST-inf-splits}
    can also be adapted to strategies like innermost or leftmost-innermost rewriting,
    but then one has to ensure that
    in the infinitely many leaves of the RST in \Cref{def:infinite-splitting}, 
    the subterm $t$ can be used as the next
    redex according to the respective strategy.
\end{rem}

As an application of \Cref{thm:PAST-vs-SAST-inf-splits} we give two  syntactical 
criteria that ensure that $\PASTf$ 
is equivalent to $\SASTf$ for a given PTRS $\PP$, where both criteria are
very easy to check automatically.
The first one (illustrated by the first PTRS in \Cref{example:Past-binary-function})
states that if we consider only PTRSs with finitely many rules,
then the existence of a function symbol of arity at least $2$ suffices for equivalence
of $\PASTf$ and $\SASTf$.
Thus, this novel observation shows that
for almost all finite PTRSs in practice, there is no difference between 
$\PASTf$ and $\SASTf$.

\begin{restatable}[Equivalence of $\PASTf$ and $\SASTf$ (1)]{thm}{PASTvsSASTOne}\label{thm:PAST-vs-SAST-1}
    If a PTRS $\PP$ has only finitely many rules and the corresponding signature contains a function symbol of at least arity $2$, then:
    \begin{align*}
        \PASTf &\Longleftrightarrow \SASTf\!
    \end{align*}
\end{restatable}

\begin{proof}
    Let $\PP$ contain only finitely many rules 
    and assume for a contradiction that we have $\PASTf$ but not $\SASTf$.
    Then there exists a term $t'$ and an $\fto_{\PP}$-RST whose root is labeled with
    $(1:t')$ which has infinitely many leaves.
    (If there were only finitely many leaves for every term, then we would
    either have $\SASTf$
    if there exists no infinite path, or we would not have $\PASTf$ if there exists an infinite path.)
    Let $t \in \TT$ be an arbitrary term, and $\tc$ be a function symbol of arity $\geq 2$.
    We can now construct an $\fto_{\PP}$-RST that starts with $(1:\tc(t',t, \ldots, t))$ 
    such that there are infinitely many leaves labeled with
    terms that have $t$ as a subterm.
    To do so, one simply restricts rewriting to the first argument of $\tc$.
    Note that this construction is only possible because we assume that there is a
    function symbol of at least arity 2, 
    such that both $t$ and $t'$ can be subterms of the same term $\tc(t',t, \ldots, t)$.
    This shows that $\PP$ admits infinite splits, and thus, we have
    $\SASTf$ by \Cref{thm:PAST-vs-SAST-inf-splits},
    which is our desired contradiction.
\end{proof}

The second sufficient criterion for equivalence of $\PASTf$ 
and $\SASTf$ (illustrated by the second PTRS in \Cref{example:Past-binary-function})
requires specific forms of ``non-erasing loops''.

\begin{restatable}[Equivalence of $\PASTf$ and $\SASTf$ (2)]{thm}{PASTvsSASTTwo}\label{thm:PAST-vs-SAST-2}
  If there exists a probabilistic rule $\ell \to \{p_1: C[\ell \sigma], p_2:s, \ldots\}$
  such that there is a variable
    $x \in \VSet(\ell)$ with $x \in \Var(x\sigma) \cap \Var(s)$, then:
    \begin{align*}
        \PASTf &\Longleftrightarrow \SASTf\!
    \end{align*}
\end{restatable}

\begin{proof}
    Let $t \in \TT$ be an arbitrary term, and let $\delta$ be a substitution such that $x \delta = t$ for a variable $x \in \VSet(x \sigma) \cap \VSet(s)$.
    We can now construct an $\fto_{\PP}$-RST that starts with $(1: \ell \delta)$ 
    such that there are infinitely many leaves labeled with
    terms that have $t$ as a subterm.
    To do so, one rewrites the redex $\ell\delta$ with the rule $\ell \to \{p_1: C[\ell \sigma], p_2:s, \ldots\}$,
    leading to a new leaf labeled with the term $s \delta$ (containing the subterm $t$ since
    $x \in \Var(s)$)
    and a new node labeled with the term $C[\ell \sigma]\delta = C\delta[\ell \sigma\delta]$, where we can rewrite the redex $\ell\sigma\delta$ again.
    This in turn will lead to a leaf labeled with $s\sigma\delta$ (which again contains $t$
    since $x \in \VSet(x \sigma) \cap \VSet(s)$ implies $x \in \VSet(s\sigma)$)
    and a new node labeled with a term containing $\ell\sigma^2\delta$, etc.
    Hence, $\PP$ admits infinite splits and the theorem is implied by \Cref{thm:PAST-vs-SAST-inf-splits}.
\end{proof}

\Cref{thm:PAST-vs-SAST-2} can also be extended so that the loop does not consist of a single rewrite step,
but we can have arbitrary many steps during the loop.

\section{Relating Variants of Probabilistic Termination and Expected Complexity}\label{Relating AST and its Restricted Forms}

Our goal is to relate the different probabilistic termination properties 
($\AST$, $\PAST$, and $\SAST$)
and the expected complexity of full rewriting
to the respective properties
of innermost rewriting (\Cref{subsection-iast-fast}), weak termination (\Cref{subsection-wast-fast}), and leftmost-innermost rewriting (\Cref{subsection-liast-fast}).
More precisely, we want to find properties of a PTRS $\PP$ which are suitable for automated
checking and which guarantee that, e.g., $\ASTs \Longleftrightarrow \AST[\sprimeto_{\PP}]$ for $s, s' \in \IS$.
Then, for example, we can use existing tools that analyze $\ASTi$ in order to prove $\ASTf$,
if we have successfully checked the properties that guarantee equivalence
 of $\ASTf$ and $\ASTi$.
Let $\PSN \in \{\AST,\PAST,\SAST\}$. 
As most of our results hold for all these three termination properties, 
we use $\PSN$ (\emph{probabilistic strong normalization}) to
refer to all of them. Similarly, we use
$\wPSN \in \{\wAST,\wPAST\}$.
Clearly, we have to require at least the same properties as in the
non-probabilistic setting, as every TRS $\R$ can be transformed into a PTRS $\PP$
by replacing every rule $\ell \to r$ with $\ell \to \{1:r\}$.
Then for every $s \in \IS$, we have $\SNs$ iff $\ASTs$ 
and $\wSNf$ iff $\wASTf$.

The following subsections are all structured as follows: 
We first give examples to show why the criteria from the non-probabilistic setting
do not carry over to the probabilistic setting.
Then, we explain the criteria needed in the probabilistic setting, 
state the corresponding theorem, and usually
give a lemma to explain the main proof idea.
We also use these lemmas to reason about expected complexity.

\subsection{From $\PSNi$ to $\PSNf$}\label{subsection-iast-fast}

We start by analyzing the relation between innermost and full rewriting.
The following example shows that \Cref{properties-eq-SN-iSN-1} does not carry over to the probabilistic setting,
i.e., orthogonality is not sufficient to ensure that $\PSNi$ implies $\PSNf$.

\begin{counterexample}[Orthogonality Does Not Suffice]\label{example:diff-AST-vs-iAST-dup}
    Consider the orthogonal PTRS $\PP_1$ with the two rules:

    \vspace*{-.5cm}
    \begin{minipage}[t]{7cm}
        \begin{align*}
            \tg &\to \{\nicefrac{3}{4}:\td(\tg), \nicefrac{1}{4}:\tz\}
        \end{align*}
    \end{minipage}
    \begin{minipage}[t]{7cm}
        \begin{align*}
            \td(x) &\to \{1:\tc(x,x)\}
        \end{align*}
    \end{minipage}
    \vspace*{.2cm}

    \noindent
    We do not have $\AST[\fto_{\PP_1}]$ (hence also neither
     $\PAST[\fto_{\PP_1}]$ nor $\SAST[\fto_{\PP_1}]$), because
    $\{1:\tg\} \fliftto_{\PP_1}^2\linebreak \{\nicefrac{3}{4}:\tc(\tg,\tg), \nicefrac{1}{4}:\tz\}$,
    which corresponds to a random walk biased towards non-termination (since $\tfrac{3}{4} > \tfrac{1}{4}$).

    However, the $\td$-rule can only duplicate normal forms in innermost evaluations.
    To see that we have $\SAST[\ito_{\PP_1}]$ (hence
    $\PAST[\ito_{\PP_1}]$ and $\AST[\ito_{\PP_1}]$), consider the
    only possible innermost rewrite sequence  $\vec{\mu}$ starting with $\{1:\tg\}$:
    {
    \[
        \{1:\tg\} \iliftto_{\PP_1} \{\nicefrac{3}{4}:\td(\tg), \nicefrac{1}{4}:\tz\}
        \iliftto_{\PP_1} \{(\nicefrac{3}{4})^2:\td(\td(\tg)), \;
        \nicefrac{1}{4} \cdot \nicefrac{3}{4}:\td(\tz), \; \nicefrac{1}{4}:\tz\} \iliftto_{\PP_1} \ldots
    \]}
    We can also view this rewrite sequence as an $\ito_{\PP_1}$-RST:\@

    
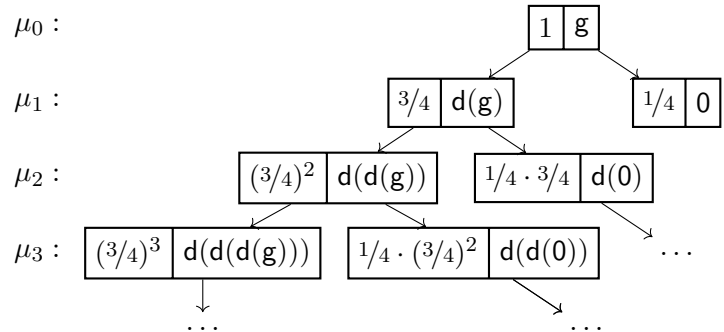
\medskip
    
    \begin{center}
        \begin{tikzpicture}
            \tikzstyle{adam}=[thick,draw=black!100,fill=white!100,minimum size=4mm, shape=rectangle split, rectangle split parts=2,rectangle split horizontal]
            \tikzstyle{empty}=[rectangle,thick,minimum size=4mm]

            \node[empty] at (-7, 2)  (a) {$\mu_0:$};
            \node[adam] at (0, 2)  (1) {$1$ \nodepart{two} $\tg$};

            \node[empty] at (-7, 1)  (b) {$\mu_1:$};
            \node[adam] at (-1.5, 1)  (11) {$\nicefrac{3}{4}$\nodepart{two}$\td(\tg)$};
            \node[adam] at (1.5, 1)  (12) {$\nicefrac{1}{4}$\nodepart{two}$\tz$};

            \node[empty] at (-7, 0)  (c) {$\mu_2:$};
            \node[adam] at (-3, 0)  (111) {$(\nicefrac{3}{4})^2$\nodepart{two}$\td(\td(\tg))$};
            \node[adam] at (0, 0)  (112) {$\nicefrac{1}{4} \cdot \nicefrac{3}{4}$\nodepart{two}$\td(\tz)$};

            \node[empty] at (-7, -1)  (d) {$\mu_3:$};
            \node[adam] at (-4.8, -1)  (1111) {$(\nicefrac{3}{4})^3$\nodepart{two}$\td(\td(\td(\tg)))$};
            \node[adam] at (-1.2, -1)  (1112) {$\nicefrac{1}{4} \cdot (\nicefrac{3}{4})^2$\nodepart{two}$\td(\td(\tz))$};
            \node[empty] at (1.5, -1)  (1121) {$\ldots$};

            \node[empty] at (-4.8, -2)  (11111) {$\ldots$};
            \node[empty] at (0.3, -2)  (11121) {$\ldots$};

            \draw (1) edge[->] (11);
            \draw (1) edge[->] (12);
            \draw (11) edge[->] (111);
            \draw (11) edge[->] (112);
            \draw (111) edge[->] (1111);
            \draw (111) edge[->] (1112);
            \draw (112) edge[->] (1121);
            \draw (1112) edge[->] (11121);
            \draw (1111) edge[->] (11111);
            \draw (1112) edge[->] (11121);
        \end{tikzpicture}
    \end{center}
    The branch to the right that starts with $\tz$ stops after $0$ innermost steps, the
    branch that starts with $\td(\tz)$ stops after $1$ innermost step, the branch that starts with $\td(\td(\tz))$ stops after $2$ innermost steps, and so on.
    So if we start with the term $\td^n(\tz)$, then we reach a normal form after
    $n$ steps, and we reach $\td^n(\tz)$ after $n+1$ steps
    from the initial term $\tg$.
    For every $k \in \IN$ we have
    $|\mu_{2\cdot k + 1}|_{\PP_1} = |\mu_{2 \cdot k + 2}|_{\PP_1} = \sum_{n=0}^{k} \nicefrac{1}{4} \cdot (\nicefrac{3}{4})^n$
    and thus   
    \[ 
        \begin{array}{rclcl}
            \edl(\vec{\mu}) &=& \sum_{n = 0}^{\infty} (1 - |\mu_n|_{\PP_1}) 
            &=&
            1 + 2 \cdot  \sum_{k \in \IN} (1 - |\mu_{2\cdot k + 1}|_{\PP_1}) \\
            &=&
            1 + 2 \cdot  \sum_{k \in \IN} (1 - \sum_{n=0}^{k} \nicefrac{1}{4} \cdot (\nicefrac{3}{4})^n) 
            &=&
            1 + 2 \cdot  \sum_{k \in \IN} (\nicefrac{3}{4})^{k+1}\\
            &=&
            (2 \cdot  \sum_{k \in \IN} (\nicefrac{3}{4})^{k}) -1 
            &=& 
            7
        \end{array}
        \]
    Due to innermost rewriting, there is no non-determinism in this sequence,
    i.e., when starting with   $\{1:\tg\}$, there is no rewrite sequence with higher expected 
    derivation length.
    Thus, we also have $\edh_{\ito_{\PP_1}}(\tg) = 7$.
    Analogously, in all other innermost rewrite sequences, the $\td$-rule can also only duplicate normal
    forms. Thus, all terms have finite expected derivation height w.r.t.\ innermost rewriting. 
    Therefore, $\SAST[\ito_{\PP_1}]$ and thus, also $\PAST[\ito_{\PP_1}]$ and $\AST[\ito_{\PP_1}]$ hold.
    The latter can also be proved automatically by our implementation of the probabilistic DP
    framework for $\ASTi$~\cite{kassinggiesl2023iAST,FLOPS2024} in \aprove{}.
\end{counterexample}

To construct a counterexample, we exploited the fact that $\PP_1$ is not right-linear, 
which allows us to duplicate the redex $\tg$ repeatedly during the rewrite sequence. 
Similar to the complexity analysis in the non-probabilistic setting (see
\Cref{subsection-Complexity}),  
we need to prohibit the duplication of redexes.
Indeed, requiring right-linearity prevents this kind of duplication and yields our desired result.

\begin{restatable}[From $\PSNi$ to $\PSNf$]{thm}{ASTAndIASTPropertyOne}\label{properties-eq-AST-iAST-1}
  If a PTRS $\PP$ is OR and RL (i.e., NO and linear), then:
\begin{align*}
        \PSNf &\Longleftrightarrow \PSNi\!
    \end{align*}
\end{restatable}

For the proof of
\Cref{properties-eq-AST-iAST-1}, we prove the following lemma which directly implies
\Cref{properties-eq-AST-iAST-1}.

\begin{restatable}[From Innermost to Full Rewriting]{lem}{ASTAndIASTLemmaOne}\label{lemma-eq-AST-iAST-1}
  If a PTRS $\PP$ is OR and RL (i.e., NO and linear), then for every infinite $\fliftto_{\PP}$-rewrite sequence $\vec{\mu} = (\mu_n)_{n \in \IN}$ there exists an infinite $\iliftto_{\PP}$-rewrite sequence $\vec{\nu} = (\nu_n)_{n \in \IN}$ such that 
    \begin{enumerate}
        \item $\lim\limits_{n \to \infty}|\mu_n|_{\PP} \geq \lim\limits_{n \to \infty}|\nu_n|_{\PP}$
        \item $\edl(\vec{\mu}) \leq \edl(\vec{\nu})$
    \end{enumerate}
\end{restatable}

\paper{For reasons of space, here we only give a proof sketch. 
As mentioned, all missing complete proofs can be found in~\cite{REPORT}.}
\report{As mentioned, all missing complete proofs can be found in App.\ \ref{appendix}.}

\smallskip

\begin{myproofsketch}
    The proofs for all lemmas in this section follow a similar structure.
    We always iteratively replace rewrite steps by ones that use the desired strategy and
    ensure that this neither increases
    the probability of convergence
    nor decreases the expected derivation length.
    For this replacement, we lift the corresponding construction from the non-probabilistic to the
    probabilistic setting. However, this cannot be done directly but instead, we have to
    regard the limit of a sequence of transformation steps.

    Let $\PP$ be a PTRS that is non-overlapping, linear, 
    and there exists an infinite $\fliftto_{\PP}$-rewrite sequence $\vec{\mu} = (\mu_n)_{n \in \IN}$ such that $\lim_{n \to \infty} |\mu_n|_{\PP} = c$ for some $c \in \IR$ with $0 \leq c < 1$.
    Our goal is to transform this sequence into an innermost sequence that converges with at most probability $c$. 
    If the sequence is not yet an innermost one, then in 
    $(\mu_n)_{n \in \IN}$ at least one rewrite step is performed with a redex that is not an innermost redex.
    Since $\PP$ is non-overlapping, we can replace a first such non-innermost rewrite step with an innermost rewrite step using
    a similar construction as in the non-probabilistic setting.
    In this way, we result in a rewrite sequence $\vec{\mu}^{(1)} = (\mu^{(1)}_n)_{n \in \IN}$ with 
    $\lim_{n \to \infty} |\mu^{(1)}_n|_{\PP} = \lim_{n \to \infty} |\mu_n|_{\PP} = c$.
    Here, linearity is needed to ensure that the probability of convergence
    does not increase during this replacement.
    We can then repeat this replacement for every non-innermost rewrite step, i.e., we again
    replace a first non-innermost rewrite step in $(\mu^{(1)}_n)_{n \in \IN}$ to obtain
    $(\mu^{(2)}_n)_{n \in \IN}$ with  the same convergence probability, etc.
    In the end, the limit of all these rewrite sequences $\vec{\mu}^{(\infty)} = \lim_{i \to \infty}
    (\mu^{(i)}_n)_{n \in \IN}$ is an innermost rewrite sequence that converges with
    probability at most $c$.

    Regarding the expected derivation length, we can use exactly the same construction,
    as this also guarantees that in each step, $\vec{\mu}^{(1)}$ does not only  
    converge with the same probability as $\vec{\mu}$, but we also have 
    $\edl(\vec{\mu}^{(1)}) \geq \edl(\vec{\mu})$ and $\edl(\vec{\mu}^{(i+1)}) \geq \edl(\vec{\mu}^{(i)})$ for all $i > 0$.
    So in the end, the limit of all these rewrite sequences 
    $\vec{\mu}^{(\infty)}$ is an innermost rewrite sequence 
    with $\edl(\vec{\mu}^{(\infty)}) \geq \edl(\vec{\mu})$.
\end{myproofsketch}

\medskip

\begin{myproofof}{\Cref{properties-eq-AST-iAST-1}}
    We only need to prove the ``$\Longleftarrow$'' direction.
    Let us first consider $\ASTf$.
    Assume that $\PP$ is orthogonal, right-linear, and that $\ASTf$ does not hold.
    Then, there exists an infinite $\fliftto_{\PP}$-rewrite sequence $\vec{\mu} = (\mu_n)_{n \in \IN}$ such that $\lim\limits_{n \to \infty}|\mu_n|_{\PP} < 1$.
    By \Cref{lemma-eq-AST-iAST-1} we obtain an infinite $\iliftto_{\PP}$-rewrite sequence
    $\vec{\nu} = (\nu_n)_{n \in \IN}$ such that $\lim\limits_{n \to \infty}|\nu_n|_{\PP}
    \leq \lim\limits_{n \to \infty}|\mu_n|_{\PP} < 1$, and thus,
    $\ASTi$ does not hold either.
    The argument is completely analogous for $\PASTf$, just reasoning
    about the expected derivation length.

    For $\SASTf$ we have two cases: 
    If there exists a single $\fliftto_{\PP}$-rewrite sequence $\vec{\mu} = (\mu_n)_{n \in \IN}$ 
    such that $\edl(\vec{\mu}) = \omega$, then we proceed as for $\PASTf$.
    Otherwise, there exists an infinite set
    $\{ \vec{\mu}^{(i)} \mid i \in \IN \}$ of
    $\fliftto_{\PP}$-rewrite sequences $\vec{\mu}^{(i)}$
    with the same initial multi-distribution
    such that $\sup \{\edl(\vec{\mu}^{(i)}) \mid i \in \IN \} = \omega$.
    For each of these rewrite sequences we apply \Cref{lemma-eq-AST-iAST-1} as before such that we obtain  \emph{innermost} rewrite sequences $\vec{\nu}^{(i)}$ with $\edl(\vec{\nu}^{(i)}) \geq \edl(\vec{\mu}^{(i)})$ for all $i \in \IN$.
    Thus, $\{ \vec{\nu}^{(i)} \mid i \in \IN \}$
    is an infinite set of innermost rewrite sequences
    with the same initial multi-distribution
    such that $\sup \{\edl(\vec{\nu}^{(i)}) \mid i \in \IN \} = \omega$, 
    which proves that $\SASTi$ does not hold either.
\end{myproofof}

One may wonder whether we can remove the left-linearity requirement from \Cref{properties-eq-AST-iAST-1}, 
as in the non-probabilistic setting.
It turns out that this is not possible.

\begin{counterexample}[Left-Linearity Cannot be Removed]\label{example:diff-AST-vs-iAST-left-lin}
    Consider the PTRS $\PP_2$ with the rules:
    
    \vspace*{-.5cm}
    \begin{minipage}[t]{7cm}
        \begin{align*}
            \tf(x,x) &\to \{1:\tf(\ta,\ta)\}
        \end{align*}
    \end{minipage}
    \begin{minipage}[t]{7cm}
        \begin{align*}
            \ta &\to \{\nicefrac{1}{2}:\tb, \nicefrac{1}{2}:\tc\}
        \end{align*}
    \end{minipage}
    \vspace*{.2cm}

    \noindent
    We do not have $\AST[\fto_{\PP_2}]$ (hence also neither $\PAST[\fto_{\PP_2}]$ nor $\SAST[\fto_{\PP_2}]$), 
    since $\{1:\tf(\ta,\ta)\} \fliftto_{\PP_2} \{1:\tf(\ta,\ta)\} \fliftto_{\PP_2} \ldots$ 
    is an infinite rewrite sequence that converges with probability $0$.
    However, we have $\SAST[\ito_{\PP_2}]$ (and hence, $\PAST[\ito_{\PP_2}]$ and $\AST[\ito_{\PP_2}]$) since the corresponding innermost sequence has the form $\{1:\tf(\ta,\ta)\} \iliftto_{\PP_2} \{\tfrac{1}{2}:\tf(\tb,\ta),\tfrac{1}{2}:\tf(\tc,\ta)\} \iliftto_{\PP_2} \{\tfrac{1}{4}:\tf(\tb,\tb), \tfrac{1}{4}:\tf(\tb,\tc), \tfrac{1}{4}:\tf(\tc,\tb), \tfrac{1}{4}:\tf(\tc,\tc)\} \iliftto_{\PP_2} \ldots$.
    Here, the last distribution contains two normal forms $\tf(\tb,\tc)$ and $\tf(\tc,\tb)$ that did not occur in the previous rewrite sequence,
    so that the expected derivation length of this $\iliftto_{\PP_2}$-rewrite sequence
    is
       $2 + 3 \cdot \sum_{i=1}^{\infty} (\nicefrac{1}{2})^i = 5$.
    Since all innermost rewrite sequences keep on adding such normal forms after a constant
    number of steps for each start term, $\edh_{\ito_{\PP_2}}(t)$ is finite for each $t \in \TT$ (again, $\AST[\ito_{\PP_2}]$ can be shown automatically by \aprove{}).
    Note that adding the requirement of being non-erasing would not help to get rid of the
    left-linearity requirement, 
    as shown by the PTRS $\PP_3$ which results from $\PP_2$ by replacing the
    $\tf$-rule with $\tf(x,x) \to \{1:\td(\tf(\ta,\ta), x)\}$.
\end{counterexample}

The problem here is that although we rewrite both occurrences of $\ta$ with the same rewrite rule,
the two $\ta$-symbols are replaced by two different terms (each with a probability $> 0$).
This would be impossible in the non-probabilistic setting.

Next, one could try to adapt \cref{properties-eq-SN-iSN-3} to the probabilistic setting
(when requiring linearity in addition).
So one could investigate whether $\PSNi$ implies $\PSNf$ for PTRSs
that are linear locally confluent
overlay systems. A PTRS $\PP$ is \emph{locally confluent} if 
for all multi-distributions $\mu, \mu_1, \mu_2$ such that $\mu_1 \;
{\prescript{}{\PP}{\xleftliftto{\mathbf{f}}}} \; \mu 
\fliftto_{\PP} \mu_2$, there exists a multi-distribution $\mu'$ such that 
 $\mu_1 \fliftto_{\PP}^* \mu' \;
 {\prescript{*}{\PP}{\xleftliftto{\mathbf{f}}}} \; \mu_2$, see~\cite{diazcaro_confluence_2018}. Note that in
contrast to the probabilistic setting, there are non-overlapping PTRSs that are not
locally confluent
(e.g., the variant $\PP_2'$ of $\PP_2$ that consists of the rules
 $\tf(x,x) \to \{1:\td\}$ and $\ta \to \{\nicefrac{1}{2}:\tb, \nicefrac{1}{2}:\tc\}$,
since we have 
 $\{1:\td\} \; {\prescript{}{\PP_2'}{\xleftliftto{\mathbf{f}}}} \; \{1:\tf(\ta,\ta)\}
\fliftto_{\PP_2'} \{\nicefrac{1}{2}:\tf(\tb,\ta), \nicefrac{1}{2}:\tf(\tc,\ta)\}$ and the two
resulting multi-distributions are not joinable).
Whether every linear, non-overlapping PTRS is locally confluent is an open problem,
thus, it is open  whether an adaption of \cref{properties-eq-SN-iSN-3} would
subsume \Cref{properties-eq-AST-iAST-1} as in the non-probabilistic setting.

In contrast to the proof of \cref{properties-eq-SN-iSN-1},
the proof of \cref{properties-eq-SN-iSN-3} relies on a minimality requirement for
the used redex.
In the non-probabilistic setting,
whenever a term $t$ starts an infinite rewrite sequence, 
then there exists a  \emph{minimal} infinite rewrite sequence beginning with $t$, 
where one only reduces redexes whose proper subterms are terminating.
 However, such minimal infinite sequences do not always exist in the probabilistic setting.

\begin{exa}[No Minimal Infinite Rewrite Sequence for $\ASTf$ and $\PASTf$]\label{example:minimality-in-positions-AST}
    Reconsider the PTRS $\PP_1$ from Counterex.\ \ref{example:diff-AST-vs-iAST-dup}, 
    where $\AST[\fto_{\PP_1}]$ does not hold.
    However, there is no minimal rewrite sequence with convergence probability $< 1$.
    If we always rewrite the proper subterm $\tg$ of the redex $\td(\tg)$, then this 
    yields a rewrite sequence that converges with probability $1$,
    like the $\iliftto_{\PP_1}$-rewrite sequence in
    Counterex.\ \ref{example:diff-AST-vs-iAST-dup}.
    Hence, a rewrite sequence $\vec{\mu}$ with convergence probability $< 1$ would have to
    eventually use the $\td$-rule on a term of the form $\td(t)$ where $t$ contains
    $\tg$.
    But then $\vec{\mu}$ is not minimal since $\tg$ itself starts a rewrite sequence
    with convergence probability $< 1$.

    To simplify the argumentation for $\PAST[\fto_{\PP_1}]$,
    let us replace the rule $\td(x) \to \{1:\tc(x,x)\}$ by $\td(\tg) \to
    \{1:\tc(\tg,\tg)\}$
    such that we can only duplicate $\tg$
    and no other symbols.
    Then the same argumentation as for $\AST[\fto_{\PP_1}]$ above
    shows that there is also no minimal
    non-$\PAST[\fto_{\PP_1}]$
    sequence, i.e., no minimal rewrite sequence $\vec{\mu}$ with $\edl(\vec{\mu}) =
    \infty$.
    Again, if we 
    always rewrite the proper subterm $\tg$ if possible, then the expected derivation length of this sequence
    would be finite.
    Thus, we have to rewrite $\td$ eventually, although its argument contains
    $\tg$. However, then the rewrite sequence is not minimal.
\end{exa}

It remains open whether one can also adapt
\cref{properties-eq-SN-iSN-3} to the probabilistic setting (e.g., if one can replace
non-overlappingness in \Cref{properties-eq-AST-iAST-1}
by the requirement of locally confluent overlay systems). 
There are two main difficulties when trying to adapt the proof of
\Cref{properties-eq-SN-iSN-3}
to PTRSs.
First, the minimality requirement cannot be imposed in the probabilistic setting, as
discussed above.
In the non-probabilistic setting, this requirement is needed  to ensure that any subterm
below a position that was reduced in the original (minimal) infinite rewrite sequence
is terminating.
Second, the original proof of \cref{properties-eq-SN-iSN-3} uses Newman's Lemma~\cite{Newman42}  
which states that local confluence implies confluence for strongly normalizing terms $t$, and thus
it implies that $t$ has a unique normal form.
Local confluence and adaptions of the unique normal form property for the probabilistic
setting have
been studied in~\cite{diazcaro_confluence_2018,Faggian2019ProbabilisticRN}, who concluded that obtaining an
analogous statement to Newman's Lemma for PTRSs that are $\AST$ would be very difficult.
The reason is that one cannot use well-founded induction on the length of a rewrite
sequence of a PTRS that is $\AST$, 
since these rewrite sequences may be infinite.

We can also use \Cref{lemma-eq-AST-iAST-1} for expected complexity, leading to the following result.

\begin{restatable}[From Innermost to Full Expected Complexity]{thm}{ASTAndIASTPropertyOneComplex}\label{properties-eq-AST-iAST-1-complex}
    If a PTRS $\PP$ is OR and RL (i.e., NO and linear), then:
    \[\edertime_{\fto_{\PP}} = \edertime_{\ito_{\PP}} \hspace*{1cm} \text{ and } \hspace*{1cm} \eruntime_{\fto_{\PP}} = \eruntime_{\ito_{\PP}}\!\]
\end{restatable}

\begin{proof}
  Similar to the proof of \Cref{properties-eq-AST-iAST-1}, this
  is a direct consequence of \Cref{lemma-eq-AST-iAST-1}.
\end{proof}

While it is open whether an adaption to locally confluent overlay systems would work for $\PSNi$,
we can give a counterexample
to show that one cannot weaken the requirement of NO to overlay systems for expected complexity, i.e., \Cref{properties-eq-AST-iAST-1-complex}
does not hold for linear overlay systems in general. In contrast, 
in the non-probabilistic setting
we have $\runtime_{\fto_{\R}} = \runtime_{\ito_{\R}}$ for all right-linear
overlay systems (since right-linearity implies spareness), see 
\Cref{properties-irun-vs-run}.

\begin{exa}\label{ORinsteadofOS}
    Consider the PTRS $\PP_4$ with the six rules:

    \vspace*{-.5cm}
    \begin{minipage}[t]{7cm}
        \begin{align*}
            \tf(x) &\to \{\nicefrac{1}{2}:\tg(x), \nicefrac{1}{2}:\th(x)\}\\
            \tg(\tb) &\to \{1: \tf(\ta)\}\\
            \th(\tc) &\to \{1: \tf(\ta)\}\!
        \end{align*}
    \end{minipage}
    \begin{minipage}[t]{7cm}
        \begin{align*}
            \td &\to \{1: \tf(\ta)\}\\
            \ta &\to \{1: \tb\}\\
            \ta &\to \{1: \tc\}\!
        \end{align*}
    \end{minipage}

    \noindent
    $\PP_4$ is a linear overlay system.
    Moreover,
    $\eruntime_{\fto_{\PP}} \in \Theta(\omega)$
    due to the infinite $\fliftto_{\PP_4}$-rewrite sequence 
    $\{1:\td\} \fliftto_{\PP_4} \{1:\tf(\ta)\} \fliftto_{\PP_4} \{\nicefrac{1}{2}:\tg(\ta), \nicefrac{1}{2}:\th(\ta)\} \fliftto_{\PP_4} \{\nicefrac{1}{2}:\tg(\tb), \nicefrac{1}{2}:\th(\tc)\} \fliftto_{\PP_4} \{\nicefrac{1}{2}:\tf(\ta), \nicefrac{1}{2}:\tf(\ta)\}$
    that converges with probability $0$.
    But for $\iliftto_{\PP_4}$-rewrite sequences we have to rewrite the argument
    $\ta$ of $\tf(\ta)$ first, leading to a normal form with a chance of $\nicefrac{1}{2}$ after two steps.
    Hence, $\eruntime_{\ito_{\PP}} \in \O(1)$.
    For expected derivational complexity, we get $\edertime_{\fto_{\PP}} \in \Theta(\omega)$ and $\edertime_{\ito_{\PP}} \in \O(n)$ via similar arguments.
    Note that we have $\edertime_{\ito_{\PP}} \in \O(n)$ and not $\edertime_{\ito_{\PP}} \in \O(1)$, 
    since for derivational complexity we allow start terms like $\tf^n(\ta)$
    where $n$ steps are needed to reach a multi-distribution that also contains normal
    forms.
\end{exa}

Clearly $\PP_4$ is not locally confluent.
Similar as for $\PSNi$, 
it remains open whether there is an adaption of \Cref{properties-eq-AST-iAST-1-complex} for locally
confluent overlay systems.

In \Cref{Improving Applicability}, we will show 
that for left-linear, non-overlapping, and \emph{spare} PTRSs $\PP$ we still
have $\eruntime_{\fto_{\PP}} = \eruntime_{\ito_{\PP}}$, see
\Cref{properties-eq-AST-iAST-3-complex}. In other words, 
RL can be weakened to SP.

Note that \Cref{properties-eq-AST-iAST-1-complex} of course also holds
for PTRSs with trivial probabilities, i.e., 
if all rules have the form $\ell \to \{1:r\}$.
Thus, we can use
\Cref{properties-eq-AST-iAST-1-complex} to obtain the first result on the
relation between full and innermost derivational complexity in the non-probabilistic
setting.

\begin{restatable}[From Innermost to Full Complexity]{cor}{DerAndIDerPropertyOneComplex}\label{properties-eq-Der-iDer-1}
  If
  a TRS $\R$ is OR and RL (i.e., NO and linear), then:
  \[\dertime_{\fto_{\R}} = \dertime_{\ito_{\R}}\]
\end{restatable}

\subsection{From $\wPSNf$ to $\PSNf$}\label{subsection-wast-fast}

Next, we investigate $\wPSNf$. Since $\PSNi$ implies $\wPSNf$, 
we essentially have the same problems as for $\PSNi$, i.e.,
in addition to non-overlappingness, we need linearity. This can be seen in
Counterex.\ \ref{example:diff-AST-vs-iAST-dup} and \ref{example:diff-AST-vs-iAST-left-lin},
as
for $i \in \{1,3\}$ we have
$\PSN[\ito_{\PP_i}]$ (and hence $\wPSN[\fto_{\PP_i}]$) but not $\PSN[\fto_{\PP_i}]$,
while $\PP_1$ and $\PP_3$ are non-overlapping and non-erasing, but not linear.
Furthermore, we need non-erasingness as we did in the non-probabilistic setting for the
same reasons, see Counterex.~\ref{example:diff-SN-vs-WN}.

\begin{restatable}[From $\wPSNf$ to $\PSNf$]{thm}{ASTvswAST}\label{properties-AST-vs-wAST}
    If a PTRS $\PP$ is NO, linear, and NE, then
    \begin{align*}
        \PSNf &\Longleftrightarrow \wPSNf\!
    \end{align*}
\end{restatable}

\subsection{From $\PSNli$ to $\PSNf$}\label{subsection-liast-fast}

Finally, we look at leftmost-innermost rewriting
as an example for a rewrite strategy that uses an ordering for parallel redexes. 
In contrast to the non-probabilistic setting, it turns out that $\PSNli$ and $\PSNi$ are not equivalent in general.
The following counterexample is similar to Counterex.\ \ref{example:diff-AST-vs-iAST-left-lin}, which illustrated that $\PSNf$ and $\PSNi$ are not equivalent without left-linearity.

\begin{counterexample}\label{example:liAST-vs-iAST}
    Consider the PTRS $\PP_5$ with the five rules:
    
    \vspace*{-.4cm}
    \begin{minipage}[t]{7cm}
        \vspace*{.3cm}
        \begin{align*}
            \ta &\to \{1:\tc_1\}\\
            \ta &\to \{1:\tc_2\}
        \end{align*}
    \end{minipage}
    \begin{minipage}[t]{7cm}
        \begin{align*}
            \tb &\to \{\nicefrac{1}{2}:\td_1, \nicefrac{1}{2}:\td_2\}\\
            \tf(\tc_1,\td_1) &\to \{1:\tf(\ta,\tb)\}\\
            \tf(\tc_2,\td_2) &\to \{1:\tf(\ta,\tb)\}
        \end{align*}
    \end{minipage}
    \vspace*{.2cm}

    \noindent
    We do not have $\AST[\ito_{\PP_5}]$ (hence also neither $\PAST[\ito_{\PP_5}]$ nor $\SAST[\ito_{\PP_5}]$), 
    since there exists the infinite rewrite sequence
    $\{1:\tf(\ta,\tb)\} \iliftto_{\PP_5} \{\nicefrac{1}{2}:\tf(\ta,\td_1),
    \nicefrac{1}{2}:\tf(\ta,\td_2)\} \iliftto_{\PP_5} \{\nicefrac{1}{2}:\tf(\tc_1,\td_1),
    \nicefrac{1}{2}:\tf(\tc_2,\td_2)\} \iliftto_{\PP_5} \{\nicefrac{1}{2}:\tf(\ta,\tb),
    \nicefrac{1}{2}:\tf(\ta,\tb)\} \iliftto_{\PP_5} \ldots$, which converges with probability $0$.
    It first ``splits'' the term $\tf(\ta,\tb)$ with the $\tb$-rule, and then applies one of the two different $\ta$-rules to each of the resulting terms.
    In contrast, when applying a leftmost-innermost rewrite strategy, we have to decide which $\ta$-rule to use before we split the term with the $\tb$-rule.
    For example, we have $\{1:\tf(\ta,\tb)\} \liliftto_{\PP_5} \{1:\tf(\tc_1,\tb)\} \liliftto_{\PP_5} \{\nicefrac{1}{2}:\tf(\tc_1,\td_1), \nicefrac{1}{2}:\tf(\tc_1,\td_2)\}$.
    Here, the second term $\tf(\tc_1,\td_2)$ is a normal form.
    Since all leftmost-innermost rewrite sequences keep on adding such normal forms after
    a certain number of steps for each start term, 
    we have $\SAST[\lito_{\PP_5}]$ (and hence, $\PAST[\lito_{\PP_5}]$ and $\AST[\lito_{\PP_5}]$).
\end{counterexample}

The counterexample above can easily be adapted to variants of innermost rewriting that impose different orders on parallel redexes like, e.g., \emph{rightmost}-innermost rewriting.

However, $\PSNli$ and $\PSNi$ are again equivalent for non-overlapping PTRSs $\PP$.
For such PTRSs, at most one rule can be used to rewrite at a given position, which prevents
the problem illustrated in Counterex.\ \ref{example:liAST-vs-iAST}.

\begin{restatable}[From $\PSNli$ to $\PSNi$]{thm}{iASTvsliAST}\label{properties-iAST-vs-liAST}
    If a PTRS $\PP$ is NO, then
    \begin{align*}
        \PSNi &\Longleftrightarrow \PSNli\!
    \end{align*}
\end{restatable}

For the proof of \Cref{properties-iAST-vs-liAST}, we use the following lemma which
immediately implies \Cref{properties-iAST-vs-liAST}.

\begin{restatable}[From Leftmost-Innermost to Innermost Rewriting]{lem}{IASTAndLIASTLemma}\label{lemma-eq-iAST-liAST}
    If a PTRS $\PP$ is NO, then for every infinite $\iliftto_{\PP}$-rewrite sequence $\vec{\mu} = (\mu_n)_{n \in \IN}$ there exists an infinite $\liliftto_{\PP}$-rewrite sequence $\vec{\nu} =
    (\nu_n)_{n \in \IN}$ such that
    \begin{enumerate}
        \item $\lim\limits_{n \to \infty}|\mu_n|_{\PP} \geq \lim\limits_{n \to \infty}|\nu_n|_{\PP}$
        \item $\edl(\vec{\mu}) \leq \edl(\vec{\nu})$
    \end{enumerate}
\end{restatable}

\begin{figure}[t]
	\centering
	\begin{tikzpicture}
        \tikzstyle{adam}=[rectangle,thick,draw=black!100,fill=white!100,minimum size=4mm]
        \tikzstyle{empty}=[rectangle,thick,minimum size=4mm]

        \node[empty] at (0, -1)  (ast) {$\PSNf$};
        \node[empty] at (-3, -1)  (iast) {$\PSNi$};
        \node[empty] at (3, -1)  (weak-AST) {$\wPSNf$};
        \node[empty] at (-6, -1)  (liast) {$\PSNli$};

        \draw (iast) edge[-implies,double equal sign distance, bend left] (liast);
        \draw (liast) edge[-implies,double equal sign distance, bend left] node[sloped, anchor=center,above] {\scriptsize \textbf{NO}} (iast);
        \draw (ast) edge[-implies,double equal sign distance, bend left] (iast);
        \draw (iast) edge[-implies,double equal sign distance, bend left] node[sloped, anchor=center,above] {\scriptsize \textbf{NO}+\textbf{LL}+\textbf{RL}} (ast);
        \draw (ast) edge[-implies,double equal sign distance, bend right] (weak-AST);
        \draw (weak-AST) edge[-implies,double equal sign distance, bend right] node[sloped, anchor=center,above] {\scriptsize \textbf{NO}+\textbf{LL}+\textbf{RL}+\textbf{NE}} (ast);
        \draw (iast) edge[-implies,double equal sign distance, bend right, out=310, in=235, looseness=0.5] (weak-AST);
        \draw (liast) edge[-implies,double equal sign distance, bend right, out=310,
          in=250, looseness=0.5] (weak-AST.south);
    \end{tikzpicture}
    \caption{Relations between the different termination properties for PTRSs}\label{fig:relations_PSN}
\end{figure}
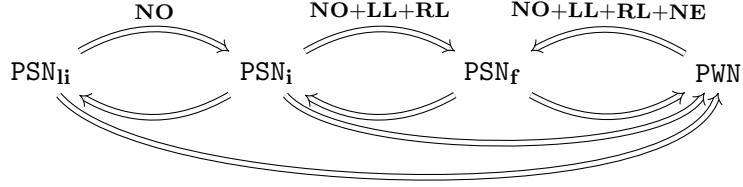

The relations between the different notions of $\mathtt{AST}$,
$\mathtt{PAST}$, and $\mathtt{SAST}$ 
of PTRSs for different rewrite strategies (given in
\Cref{properties-eq-AST-iAST-1},~\ref{properties-AST-vs-wAST}, and~\ref{properties-iAST-vs-liAST})
are summarized in \Cref{fig:relations_PSN}.

For expected complexity, we obtain the following result from \Cref{lemma-eq-iAST-liAST}.

\begin{restatable}[From Leftmost-Innermost to Innermost Expected Complexity]{thm}{iASTAndlIASTPropertyOneComplex}\label{properties-eq-iAST-liAST-complex}
    If a PTRS $\PP$ is NO, then:
    \[\edertime_{\lito_{\PP}} = \edertime_{\ito_{\PP}} \hspace*{1cm} \text{ and } \hspace*{1cm} \eruntime_{\lito_{\PP}} = \eruntime_{\ito_{\PP}}\!\]
\end{restatable}

As shown by $\PP_5$ from Counterex.\ \ref{example:liAST-vs-iAST}, again
we cannot change the requirement of non-overlapping PTRSs to overlay systems, 
since $\PP_5$ is an overlay system where
$\PSN[\lito_{\PP_5}]$ holds, but $\PSN[\ito_{\PP_5}]$ does not.
The relations between the different expected complexities of PTRSs (given in
\Cref{properties-eq-AST-iAST-1-complex,properties-eq-iAST-liAST-complex})
are summarized in \Cref{fig:relations_ecomplex}. Here, arrows ``$\implies$'' stand for
``$\geq$''.

\begin{figure}
  \centering
  \begin{subfigure}{0.5\textwidth}\label{fig:relations_eruntime}
    \centering
      \begin{tikzpicture}
            \tikzstyle{adam}=[rectangle,thick,draw=black!100,fill=white!100,minimum size=4mm]
            \tikzstyle{empty}=[rectangle,thick,minimum size=4mm]

            \node[empty] at (0, -1)  (ast) {$\eruntime_{\fto_{\PP}}$};
            \node[empty] at (-3, -1)  (iast) {$\eruntime_{\ito_{\PP}}$};
            \node[empty] at (-6, -1)  (liast) {$\eruntime_{\lito_{\PP}}$};

            \draw (iast) edge[-implies,double equal sign distance, bend left] node[sloped, anchor=center,below] {\scriptsize $\leq$} (liast);
            \draw (liast) edge[-implies,double equal sign distance, bend left] node[sloped, anchor=center,above] {\scriptsize \textbf{NO}} node[sloped, anchor=center,below] {\scriptsize $\geq$} (iast);
            \draw (ast) edge[-implies,double equal sign distance, bend left] node[sloped, anchor=center,below] {\scriptsize $\leq$} (iast);
            \draw (iast) edge[-implies,double equal sign distance, bend left] node[sloped, anchor=center,above] {\scriptsize \textbf{NO}+\textbf{LL}+\textbf{RL}} node[sloped, anchor=center,below] {\scriptsize $\geq$} (ast);
        \end{tikzpicture}\vspace*{.3cm}
        \subcaption{Runtime Complexity}
    \end{subfigure}%
    \begin{subfigure}{0.5\textwidth}\label{fig:relations_edertime}
      \centering
        \begin{tikzpicture}
            \tikzstyle{adam}=[rectangle,thick,draw=black!100,fill=white!100,minimum size=4mm]
            \tikzstyle{empty}=[rectangle,thick,minimum size=4mm]

            \node[empty] at (0, -1)  (ast) {$\edertime_{\fto_{\PP}}$};
            \node[empty] at (-3, -1)  (iast) {$\edertime_{\ito_{\PP}}$};
            \node[empty] at (-6, -1)  (liast) {$\edertime_{\lito_{\PP}}$};

            \draw (iast) edge[-implies,double equal sign distance, bend left] node[sloped, anchor=center,below] {\scriptsize $\leq$} (liast);
            \draw (liast) edge[-implies,double equal sign distance, bend left] node[sloped, anchor=center,above] {\scriptsize \textbf{NO}} node[sloped, anchor=center,below] {\scriptsize $\geq$} (iast);
            \draw (ast) edge[-implies,double equal sign distance, bend left] node[sloped, anchor=center,below] {\scriptsize $\leq$} (iast);
            \draw (iast) edge[-implies,double equal sign distance, bend left] node[sloped, anchor=center,above] {\scriptsize \textbf{NO}+\textbf{LL}+\textbf{RL}} node[sloped, anchor=center,below] {\scriptsize $\geq$} (ast);
        \end{tikzpicture}\vspace*{.2cm}
        \subcaption{Derivational Complexity}
    \end{subfigure}
    \caption{Relations for expected complexity}\label{fig:relations_ecomplex}
  \end{figure}
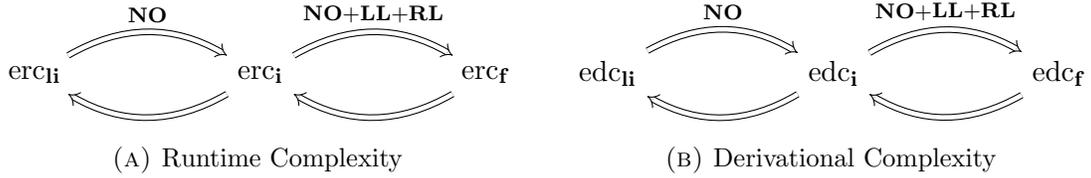

For the non-probabilistic setting,
\Cref{properties-eq-iAST-liAST-complex} implies
the following corollary.

\begin{restatable}[From Leftmost-Innermost to Innermost Complexity]{cor}{iSNAndliSNPropertyOneComplex}\label{properties-eq-iDer-liDer-complex}
  If a TRS $\R$ is NO, then:
  \[\dertime_{\lito_{\R}} = \dertime_{\ito_{\R}} \hspace*{1cm} \text{ and } \hspace*{1cm} \runtime_{\lito_{\R}} = \runtime_{\ito_{\R}}\!\]
\end{restatable}

\section{Improving Applicability}\label{Improving Applicability}

In this section, we present several ideas to improve the applicability of \cref{properties-eq-AST-iAST-1}  
which relates $\PSNf$ and $\PSNi$ --- the most interesting theorem in practice.
Recall that according to \cref{properties-eq-AST-iAST-1},
$\PSNi$ implies $\PSNf$ for PTRSs $\PP$ that are NO, LL, and RL.
In \Cref{sec:simultaneous rewriting}, we first discuss attempts in order
to remove the requirement of 
left-linearity by modifying the rewrite relation to \emph{simultaneous rewriting}.
Here, we correct an inaccuracy in our previous conference paper \cite{FoSSaCS2024}, where
we claimed that simultaneous rewriting allows us to remove
the left-linearity requirement. However, we now give a
counterexample which shows that this improvement does not work in general.
Then in \Cref{sec:spareness} we target right-linearity, and show that it can be
weakened to spareness if one only considers rewrite sequences that start with basic terms,
as in the definition of runtime complexity.

\subsection{Attempting to Remove Left-Linearity by Simultaneous Rewriting}\label{sec:simultaneous rewriting}

To try to remove the left-linearity requirement, one might consider
the
simultaneous reduction of several copies of identical redexes. 
For a PTRS $\PP$, this results in the notion of \emph{simultaneous rewriting}, denoted $\ftopara_{\PP}$. 
The following example illustrates the idea of handling non-left-linear PTRSs by
applying the same rewrite rule at parallel positions simultaneously. 

\begin{exa}[Simultaneous Rewriting]\label{example:simultaneous-simultaneous-rewriting}
    Reconsider the PTRS $\PP_2$ from Counterex.\ \ref{example:diff-AST-vs-iAST-left-lin} 
    with the rules $\tf(x,x) \to \{1:\tf(\ta,\ta)\}$ and 
    $\ta \to \{\nicefrac{1}{2}:\tb, \nicefrac{1}{2}:\tc\}$, 
    where we have $\PSN[\ito_{\PP_2}]$, but not $\PSN[\fto_{\PP_2}]$.
    Our new rewrite relation ${\paraliftto_{\PP_2}}$ allows us to reduce several copies
    of the same redex simultaneously, so that we get
    \[\{1:\tf(\ta,\ta)\}
    \iparaliftto_{\PP_2} \{\tfrac{1}{2}:\tf(\tb,\tb),\tfrac{1}{2}:\tf(\tc,\tc)\}
    \iparaliftto_{\PP_2} \{\nicefrac{1}{2}:\tf(\ta,\ta),
    \nicefrac{1}{2}:\tf(\ta,\ta)\} \iparaliftto_{\PP_2} \ldots,\]
    i.e., this $\,\iparaliftto_{\PP_2}$-rewrite sequence converges with probability 0 and thus,
    we do \emph{not} have $\AST[\itopara_{\PP_2}]$ 
    (hence, also neither $\PAST[\itopara_{\PP_2}]$ nor $\SAST[\itopara_{\PP_2}]$). 
    Note that we simultaneously reduced both occurrences of $\ta$ in the first step.
\end{exa}

\begin{defi}[Simultaneous Rewriting]\label{def:PTRS}
    Let $\PP$ be a PTRS\@.
    A term $s$ rewrites \emph{simultaneously} to a multi-distribution $\mu = \{p_1:t_1, \ldots,
    p_k:t_k\}$ (denoted $s \ftopara_{\PP} \mu$) if there is a non-empty set of
    parallel positions $\Pi = \{\pi_1, \ldots, \pi_m\} \subseteq \pos(s)$, a rule $\ell \to \{p_1:r_1, \ldots, p_k:r_k\} \in \PP$, 
    and a substitution $\sigma$ such that $s|_{\pi_i}=\ell\sigma$ for all $1 \leq i \leq m$
    and $t_j = s[r_j\sigma]_{\pi_1}\ldots[r_j\sigma]_{\pi_m}$ for all $1 \leq j \leq k$.
    The simultaneous rewrite step is \emph{innermost} (denoted $s \itopara_\PP \mu$) 
    if all proper subterms of the redex $\ell\sigma$ are in normal form w.r.t.\ $\PP$.
\end{defi}

Clearly, if the set of positions $\Pi$ in \Cref{def:PTRS} is a singleton, then the
resulting simultaneous rewrite step is an ``ordinary'' probabilistic rewrite step, i.e.,
\pagebreak[3]
${\fto_{\PP}} \subseteq {\ftopara_{\PP}}$ and
${\itos} \subseteq {\itopara_{\PP}}$.

\begin{restatable}[From $\topara_{\PP}$ to $\to_{\PP}$]{cor}{rewriteRelationRelation}\label{rewrite-relation-relation-cor}
    Let $\PP$ be a PTRS. Then, $\PSN[\stopara_{\PP}]$ implies $\PSNs$ for $s \in \{\mathbf{f}, \mathbf{i}\}$.
\end{restatable}

However, the converse of \cref{rewrite-relation-relation-cor} does not hold.
\cref{example:simultaneous-simultaneous-rewriting} shows that $\PSNi$  
does not imply $\PSN[\itopara_{\PP}]$, and
the following example shows the same for $\PSNf$ and $\PSN[\ftopara_{\PP}]$.

\begin{exa}\label{example:simultaneous-simultaneous-rewriting-full}
    Consider the PTRS $\overline{\PP}_2$ with the three rules:
    
    \vspace*{-.5cm}
    \begin{minipage}[t]{7cm}
        \begin{align*}
            \tf(\tb,\tb) &\to \{1:\tf(\ta,\ta)\}\\
            \tf(\tc,\tc) &\to \{1:\tf(\ta,\ta)\}
        \end{align*}
    \end{minipage}
    \begin{minipage}[t]{7cm}
        \vspace*{.2cm}
        \begin{align*}
            \ta &\to \{\nicefrac{1}{2}:\tb, \nicefrac{1}{2}:\tc\}
        \end{align*}
    \end{minipage}

\vspace*{.2cm}
    
    \noindent
    We have $\SAST[\fto_{\overline{\PP}_2}]$ (hence, also $\PAST[\fto_{\overline{\PP}_2}]$ and $\AST[\fto_{\overline{\PP}_2}]$).
    But as in \Cref{example:simultaneous-simultaneous-rewriting}, we obtain
    $\{1:\tf(\ta,\ta)\}
    \iparaliftto_{\overline{\PP}_2} \{\tfrac{1}{2}:\tf(\tb,\tb),\tfrac{1}{2}:\tf(\tc,\tc)\}
    \iparaliftto_{\overline{\PP}_2} \{\nicefrac{1}{2}:\tf(\ta,\ta),
    \nicefrac{1}{2}:\tf(\ta,\ta)\}$,
    i.e., there are rewrite sequences with
    $\,\iparaliftto_{\overline{\PP}_2}$, and thus, also with
    $\,\fparaliftto_{\overline{\PP}_2}$, that converge with probability 0. 
    Hence, $\AST[\itopara_{\overline{\PP}_2}]$
    does not hold, and therefore, 
    $\AST[\ftopara_{\overline{\PP}_2}]$, 
    $\PAST[\ftopara_{\overline{\PP}_2}]$, and
 $\SAST[\ftopara_{\overline{\PP}_2}]$ do not hold either.
\end{exa}

Note that this kind of simultaneous rewriting is different from the ``ordinary'' parallelism used for non-probabilistic rewriting, which is typically denoted by $\fto_{||}$. 
There, one may reduce multiple parallel redexes in a single rewrite step.
To be precise, a term $s$ rewrites \emph{parallel} to a multi-distribution $\mu$ (denoted $s \fto_{||} \mu$) w.r.t.\ a PTRS $\PP$
if there is a non-empty set of parallel redex positions $\Pi = \{\pi_1, \ldots, \pi_m\}
\subseteq \pos(s)$, 
such that $\{1:s\} \fliftto_{\PP,\pi_1} \ldots \fliftto_{\PP,\pi_m} \mu$.
Here, $\fliftto_{\PP,\pi}$ denotes the relation $\fliftto_{\PP}$, where we rewrite at position $\pi$ in each term in the multi-distribution.
This is always possible, since the positions are parallel positions of the
redexes in the start term $s$.

The two differences between simultaneous rewriting and parallel rewriting are that 
while both of them
allow the reduction of
multiple redexes,  simultaneous rewriting  ``merges'' the
corresponding terms in the multi-distributions that result from rewriting the several redexes. 
Because of this merging, 
we only allow the simultaneous reduction of \emph{equal} redexes, whereas ``ordinary'' parallel rewriting allows the simultaneous reduction of arbitrary parallel redexes,
which is the second difference.
For example, for $\PP_2$ from Counterex.\ \ref{example:diff-AST-vs-iAST-left-lin} we have
$\{1:\tf(\ta,\ta)\} \iparaliftto_{\PP_2}
\{\tfrac{1}{2}:\tf(\tb,\tb),\tfrac{1}{2}:\tf(\tc,\tc)\}$, whereas
with ordinary parallel rewriting we would obtain $\{1:\tf(\ta,\ta)\} \iliftto_{|| \PP_2}
\{\tfrac{1}{4}:\tf(\tb,\tb),\tfrac{1}{4}:\tf(\tb,\tc),\tfrac{1}{4}:\tf(\tc,\tb),\tfrac{1}{4}:\tf(\tc,\tc)\}$.

Even though simultaneous rewriting solves the issue of Counterex.\ \ref{example:diff-AST-vs-iAST-left-lin}, 
the following (more involved) counterexample shows that even with this refined rewrite relation,
we cannot remove left-linearity from \cref{properties-eq-AST-iAST-1}. In other words, there
exist PTRSs $\PP$ that are non-overlapping and right-linear, but where
$\PSN[\itopara_{\PP}]$ does not imply $\PSNf$.

\begin{counterexample}\label{example:simultaneous-rewriting-bug}
  Consider the PTRS
    $\PP_6$ with the following seven rules,
    where $s = \th(\tf(\ta,\ta))$:
    
    \vspace*{-.5cm}
    \begin{minipage}[t]{7cm}
        \vspace*{.25cm}
        \begin{align*}
            \th(x) &\to \{\nicefrac{1}{2}:\th_1(x), \nicefrac{1}{2}:\th_2(x)\}\\
            \ta &\to \{\nicefrac{1}{2}:\tb, \nicefrac{1}{2}:\tc\}\\
            \tf(x,x) &\to \{1:x\}\\
        \end{align*}
    \end{minipage}
    \begin{minipage}[t]{7cm}
        \begin{align*}
            \th_1(\tb) &\to \{1:s\}\\
            \th_1(\tc) &\to \{1:s\}\\
            \th_2(\tf(\tb,\tc)) &\to \{1:\td(s,s,s)\}\\
            \th_2(\tf(\tc,\tb)) &\to \{1:\te(s,s)\}\\
        \end{align*}
    \end{minipage}

    \noindent
    Let us take a closer look at the following
    $\fto_{\PP_6}$-RST 
    that starts with the term $s$.

    \medskip

    \begin{center}
        \scriptsize 
        \begin{tikzpicture}
            \tikzstyle{adam}=[thick,draw=black!100,fill=white!100,minimum size=4mm, shape=rectangle split, rectangle split parts=2,rectangle split horizontal]
            \tikzstyle{empty}=[rectangle,thick,minimum size=4mm]

            \node[adam] at (0, 1)  (11) {$1$\nodepart{two}$\th(\tf(\ta, \ta))$};

            \node[adam] at (-4, 0)  (111) {$\nicefrac{1}{2}$\nodepart{two}$\th_1(\tf(\ta, \ta))$};
            \node[adam] at (4, 0)  (112) {$\nicefrac{1}{2}$\nodepart{two}$\th_2(\tf(\ta, \ta))$};

            \node[adam] at (-4, -1)  (1111) {$\nicefrac{1}{2}$\nodepart{two}$\th_1(\ta)$};
            \node[adam] at (2, -1)  (1121) {$\nicefrac{1}{4}$\nodepart{two}$\th_2(\tf(\tb, \ta))$};
            \node[adam] at (6, -1)  (1122) {$\nicefrac{1}{4}$\nodepart{two}$\th_2(\tf(\tc, \ta))$};

            \node[adam] at (-5, -2)  (11111) {$\nicefrac{1}{4}$\nodepart{two}$\th_1(\tb)$};
            \node[adam] at (-3, -2)  (11112) {$\nicefrac{1}{4}$\nodepart{two}$\th_1(\tc)$};
            \node[adam] at (0, -2)  (11211) {$\nicefrac{1}{8}$\nodepart{two}$\th_2(\tf(\tb, \tb))$};
            \node[adam] at (2.5, -2)  (11212) {$\nicefrac{1}{8}$\nodepart{two}$\th_2(\tf(\tb, \tc))$};
            \node[adam] at (5.5, -2)  (11221) {$\nicefrac{1}{8}$\nodepart{two}$\th_2(\tf(\tc, \tb))$};
            \node[adam] at (8, -2)  (11222) {$\nicefrac{1}{8}$\nodepart{two}$\th_2(\tf(\tc, \tc))$};

            \node[adam] at (-5, -3)  (111111) {$\nicefrac{1}{4}$\nodepart{two}$s$};
            \node[adam] at (-3, -3)  (111121) {$\nicefrac{1}{4}$\nodepart{two}$s$};
            \node[adam] at (0, -3)  (112111) {$\nicefrac{1}{8}$\nodepart{two}$\th_2(\tb)$};
            \node[adam] at (2.5, -3)  (112121) {$\nicefrac{1}{8}$\nodepart{two}$\td(s,s,s)$};
            \node[adam] at (5.5, -3)  (112211) {$\nicefrac{1}{8}$\nodepart{two}$\te(s,s)$};
            \node[adam] at (8, -3)  (112221) {$\nicefrac{1}{8}$\nodepart{two}$\th_2(\tc)$};

            \draw (11) edge[->] (111);
            \draw (11) edge[->] (112);
            \draw (111) edge[->] (1111);
            \draw (112) edge[->] (1121);
            \draw (112) edge[->] (1122);
            \draw (1111) edge[->] (11111);
            \draw (1111) edge[->] (11112);
            \draw (1121) edge[->] (11211);
            \draw (1121) edge[->] (11212);
            \draw (1122) edge[->] (11221);
            \draw (1122) edge[->] (11222);
            \draw (11111) edge[->] (111111);
            \draw (11112) edge[->] (111121);
            \draw (11211) edge[->] (112111);
            \draw (11212) edge[->] (112121);
            \draw (11221) edge[->] (112211);
            \draw (11222) edge[->] (112221);
        \end{tikzpicture}
    \end{center}

    \medskip

    We do not have $\AST[\fto_{\PP_6}]$ 
    (hence, also neither $\PAST[\fto_{\PP_6}]$ nor $\SAST[\fto_{\PP_6}]$).
    Starting with $s$, we can either obtain $s$ again with a chance of $\nicefrac{1}{2}$,
    end in a normal form with a chance of $\nicefrac{1}{4}$,
    rewrite to two occurrences of $s$ with a chance of $\nicefrac{1}{8}$,
    or rewrite to three occurrences of $s$ with a chance of $\nicefrac{1}{8}$.
    Overall, this results in a random walk biased towards non-termination, that is not
    $\AST$
    (i.e., this RST can be extended to an RST whose convergence probability is $< 1$).

    The crucial observation is that the order of rule application differs between the subtrees:
    In the left subtree starting at $\th_1(\tf(\ta, \ta))$, 
    we want to use the non-left-linear $\tf$-rule first.
    In the right subtree starting at $\th_2(\tf(\ta, \ta))$,
    we want to use the $\ta$-rules first to create the terms $\tf(\tb, \tc)$ and $\tf(\tc, \tb)$.
    So essentially, we have to decide on the rewrite strategy after rewriting the top $\th$-symbol.
    This is impossible in innermost evaluations and cannot be solved by simultaneous rewriting either.
    For innermost simultaneous rewriting, we would have to decide on whether we want to rewrite 
    both occurrences of $\ta$ in the first step or just one.
    In both cases, we obtain an RST
    with finite expected derivation length. Hence, 
 we have $\SAST[\itopara_{\PP_6}]$ 
    (which implies $\PAST[\itopara_{\PP_6}]$ and $\AST[\itopara_{\PP_6}]$).

More precisely, if we decide to rewrite both $\ta$-symbols simultaneously at the start, we obtain
 the following $\itopara_{\PP_6}$-RST.

\medskip
 
    \begin{center}
        \scriptsize 
        \begin{tikzpicture}
            \tikzstyle{adam}=[thick,draw=black!100,fill=white!100,minimum size=4mm, shape=rectangle split, rectangle split parts=2,rectangle split horizontal]
            \tikzstyle{empty}=[rectangle,thick,minimum size=4mm]

            \node[adam] at (0, 1)  (11) {$1$\nodepart{two}$\th(\tf(\ta, \ta))$};

            \node[adam] at (-3, 0)  (111) {$\nicefrac{1}{2}$\nodepart{two}$\th(\tf(\tb, \tb))$};
            \node[adam] at (3, 0)  (112) {$\nicefrac{1}{2}$\nodepart{two}$\th(\tf(\tc,
              \tc))$};

            \node[adam] at (-3, -1)  (1111) {$\nicefrac{1}{2}$\nodepart{two}$\th(\tb)$};
            \node[adam] at (3, -1)  (1121) {$\nicefrac{1}{2}$\nodepart{two}$\th(\tc)$};

            \node[adam] at (-4.5, -2)  (11111) {$\nicefrac{1}{4}$\nodepart{two}$\th_1(\tb)$};
            \node[adam] at (-1.5, -2)  (11112) {$\nicefrac{1}{4}$\nodepart{two}$\th_2(\tb)$};
            \node[adam] at (1.5, -2)  (11211) {$\nicefrac{1}{4}$\nodepart{two}$\th_1(\tc)$};
            \node[adam] at (4.5, -2)  (11212) {$\nicefrac{1}{4}$\nodepart{two}$\th_2(\tc)$};

            \node[adam] at (-4.5, -3)  (111111) {$\nicefrac{1}{4}$\nodepart{two}$s$};
            \node[adam] at (1.5, -3)  (112111) {$\nicefrac{1}{4}$\nodepart{two}$s$};

            \draw (11) edge[->] (111);
            \draw (11) edge[->] (112);
            \draw (111) edge[->] (1111);
            \draw (112) edge[->] (1121);

            \draw (1111) edge[->] (11111);
            \draw (1111) edge[->] (11112);
            
            \draw (1121) edge[->] (11211);
            \draw (1121) edge[->] (11212);

            \draw (11111) edge[->] (111111);
            \draw (11211) edge[->] (112111);
        \end{tikzpicture}
    \end{center}

\medskip
    
Overall,  this results in a random walk biased towards termination,
because we reach a normal form with a chance of $\nicefrac{1}{2}$ and the term
$s$ with a chance of $\nicefrac{1}{2}$,
    but we cannot increase the number of occurrences of $s$.

    If we rewrite both $\ta$-symbols individually, then we  obtain
 the following $\itopara_{\PP_6}$-RST.

\medskip
 
    \begin{center}
        \tiny 
        \begin{tikzpicture}
            \tikzstyle{adam}=[thick,draw=black!100,fill=white!100,minimum size=4mm, shape=rectangle split, rectangle split parts=2,rectangle split horizontal]
            \tikzstyle{empty}=[rectangle,thick,minimum size=4mm]

            \node[adam] at (0, 1)  (11) {$1$\nodepart{two}$\th(\tf(\ta, \ta))$};

            \node[adam] at (-4, 0)  (111) {$\nicefrac{1}{2}$\nodepart{two}$\th(\tf(\tb, \ta))$};
            \node[adam] at (4, 0)  (112) {$\nicefrac{1}{2}$\nodepart{two}$\th(\tf(\tc, \ta))$};

            \node[adam] at (-6, -1)  (1111) {$\nicefrac{1}{4}$\nodepart{two}$\th(\tf(\tb, \tb))$};
            \node[adam] at (-2.15, -1)  (1112) {$\nicefrac{1}{4}$\nodepart{two}$\th(\tf(\tb, \tc))$};
            \node[adam] at (2.15, -1)  (1121) {$\nicefrac{1}{4}$\nodepart{two}$\th(\tf(\tc, \tb))$};
            \node[adam] at (6, -1)  (1122) {$\nicefrac{1}{4}$\nodepart{two}$\th(\tf(\tc, \tc))$};

            \node[adam] at (-6, -2)  (11111) {$\nicefrac{1}{4}$\nodepart{two}$\th(\tb)$};
            \node[adam] at (-3.2, -2)  (11121) {$\nicefrac{1}{8}$\nodepart{two}$\th_1(\tf(\tb, \tc))$};
            \node[adam] at (-1.1, -2)  (11122) {$\nicefrac{1}{8}$\nodepart{two}$\th_2(\tf(\tb, \tc))$};

            \node[adam] at (1.1, -2)  (11211) {$\nicefrac{1}{8}$\nodepart{two}$\th_1(\tf(\tc, \tb))$};
            \node[adam] at (3.2, -2)  (11212) {$\nicefrac{1}{8}$\nodepart{two}$\th_2(\tf(\tc, \tb))$};
            \node[adam] at (6, -2)  (11221) {$\nicefrac{1}{8}$\nodepart{two}$\th(\tc)$};
            
            \node[adam] at (-6.8, -3)  (111111) {$\nicefrac{1}{8}$\nodepart{two}$\th_1(\tb)$};
            \node[adam] at (-5.2, -3)  (111112) {$\nicefrac{1}{8}$\nodepart{two}$\th_2(\tb)$};
            \node[adam] at (-1.1, -3)  (111221) {$\nicefrac{1}{8}$\nodepart{two}$\td(s,s,s)$};
            \node[adam] at (3.2, -3)  (112121) {$\nicefrac{1}{8}$\nodepart{two}$\te(s,s)$};
            \node[adam] at (5.2, -3)  (112211) {$\nicefrac{1}{8}$\nodepart{two}$\th_1(\tc)$};
            \node[adam] at (6.8, -3)  (112212) {$\nicefrac{1}{8}$\nodepart{two}$\th_2(\tc)$};

            \node[adam] at (-6.8, -4)  (1111111) {$\nicefrac{1}{8}$\nodepart{two}$s$};
            \node[adam] at (5.2, -4)  (1122111) {$\nicefrac{1}{8}$\nodepart{two}$s$};

            \draw (11) edge[->] (111);
            \draw (11) edge[->] (112);
            
            \draw (111) edge[->] (1111);
            \draw (111) edge[->] (1112);
            
            \draw (112) edge[->] (1121);
            \draw (112) edge[->] (1122);
            
            \draw (1111) edge[->] (11111);
            \draw (1112) edge[->] (11121);
            \draw (1112) edge[->] (11122);
            \draw (1121) edge[->] (11211);
            \draw (1121) edge[->] (11212);
            \draw (1122) edge[->] (11221);

            \draw (11111) edge[->] (111111);
            \draw (11111) edge[->] (111112);

            \draw (11122) edge[->] (111221);
            \draw (11221) edge[->] (112211);
            \draw (11122) edge[->] (111221);
            \draw (11212) edge[->] (112121);

            \draw (11221) edge[->] (112211);
            \draw (11221) edge[->] (112212);
      
            \draw (111111) edge[->] (1111111);
            \draw (112211) edge[->] (1122111);
        \end{tikzpicture}
    \end{center}

    \medskip
    
    Again, this results in a random walk biased towards termination,
because starting with $s$, we can either obtain $s$ again with a chance of $\nicefrac{1}{4}$,
    end in a normal form with a chance of $\nicefrac{1}{2}$,
    rewrite to two occurrences of $s$ with a chance of $\nicefrac{1}{8}$,
    or rewrite to three occurrences of $s$ with a chance of $\nicefrac{1}{8}$.
So every $\itopara_{\PP_6}$-RST has a finite expected derivation length and we
obtain $\SAST[\itopara_{\PP_6}]$ 
    (which implies $\PAST[\itopara_{\PP_6}]$ and
$\AST[\itopara_{\PP_6}]$).
\end{counterexample}

So this counterexample refutes 
\cite[Thm.\ 27]{FoSSaCS2024} where it was claimed that
$\PSN[\itopara_{\PP}]$  implies $\PSNf$ for $\PSN \in  \{\AST,\PAST\}$.
The problem is that the proof of 
\cite[Thm.\ 27]{FoSSaCS2024}  was
limited to terms with at most two nested defined symbols.
For example, every term in the rewrite sequence of
Counterex.\ \ref{example:diff-AST-vs-iAST-left-lin}
has at most two nested defined symbols, e.g., in the term $\tf(\ta,\ta)$
the defined symbol $\tf$ is above the defined symbol $\ta$, 
but there is no further defined symbol above $\tf$.
However, Counterex.\ \ref{example:simultaneous-rewriting-bug} 
uses a term $\th(\tf(\ta,\ta))$ with three nested defined symbols.
As mentioned in Counterex.\ \ref{example:simultaneous-rewriting-bug}, 
we have to decide on the rewrite strategy after rewriting the symbol $\th$ at the root position,
and cannot fix an innermost evaluation strategy.
Due to the non-left-linear $\tf$-rule, depending on
whether one first rewrites $\tf$ or $\ta$, one obtains
completely different distributions even though the PTRS is non-overlapping.

Note that for $\wPSN$,
the direction of the implication in \Cref{rewrite-relation-relation-cor} is reversed, 
since $\wPSN$ requires that for each start term, there \emph{exists} an infinite rewrite sequence satisfying a certain property, 
whereas $\PSN$ requires that \emph{all} infinite rewrite sequences
satisfy a certain property.
Thus, if there exists an infinite $\liftto_\PP$-rewrite sequence that, e.g., converges with probability $1$ 
(showing that  $\wASTf$ holds),
then this is also a valid $\paraliftto_\PP$-rewrite sequence that converges with probability $1$
(showing that $\wAST[\ftopara_{\PP}]$ holds).

\begin{restatable}[From $\wPSNf$ to \protect{$\wPSN[\ftopara_{\PP}]$}]{cor}{rewriteRelationRelationWeakAST}\label{rewrite-relation-relation-cor-weak-AST}
    Let $\PP$ be a PTRS. Then, $\wPSNf$ implies $\wPSN[\ftopara_{\PP}]$.
\end{restatable}

For $\wPSN$, Counterex.\ \ref{example:simultaneous-rewriting-bug} shows that 
simultaneous rewriting cannot be used to remove the requirement of left-linearity
from
\Cref{properties-AST-vs-wAST} either.
The following Counterex.\ \ref{example:problem-sim-rewriting-wAST}
shows that we can even give a counterexample that only considers terms 
where the number of nested defined symbols is at most $2$.

\begin{counterexample}\label{example:problem-sim-rewriting-wAST}
    Consider the non-left-linear PTRS $\PP_7$ with the two rules:

    \vspace*{-.5cm}
    \begin{minipage}[t]{7cm}
        \begin{align*}
            \tg &\to \{\nicefrac{3}{4}:\td(\tg,\tg), \nicefrac{1}{4}:\tz\}
        \end{align*}
    \end{minipage}
    \begin{minipage}[t]{7cm}
        \begin{align*}
            \td(x,x) &\to \{1:x\}
        \end{align*}
    \end{minipage}
    \vspace*{.2cm}

    \noindent
    We do not have $\AST[\fto_{\PP_7}]$ (hence, not $\PAST[\fto_{\PP_7}]$ either), 
    as we have $\{1:\tg\} \fliftto_{\PP_7} \{\nicefrac{3}{4}:\td(\tg,\tg), \nicefrac{1}{4}:\tz\}$,
    which corresponds to a random walk biased towards non-termination if we never use the $\td$-rule (since $\tfrac{3}{4} > \tfrac{1}{4}$).
    However, if we always use the $\td$-rule directly after the $\tg$-rule, then we
    essentially end up with a PTRS whose only rule is $\tg \to \{\nicefrac{3}{4}:\tg,
    \nicefrac{1}{4}:\tz\}$, which corresponds to flipping a biased coin until heads comes up.
    This proves $\wPAST[\fto_{\PP_7}]$ and hence, also $\wAST[\fto_{\PP_7}]$.
    As $\PP_7$ is non-overlapping, right-linear, and non-erasing,
    this shows that a variant of \Cref{properties-AST-vs-wAST} without the requirement of
    left-linearity would need more than just moving to simultaneous rewriting.
\end{counterexample}

\subsection{Weakening Right-Linearity to Spareness}\label{sec:spareness}

In the non-probabilistic setting,
runtime complexity is easier to analyze than derivational complexity
because of the restriction to basic start terms.
In particular, this restriction also
allows us to use notions like spareness such that full runtime complexity can be analyzed
via innermost runtime complexity, see \Cref{properties-irun-vs-run}.
Similarly, in the probabilistic setting
we can also require spareness instead of right-linearity,
if we only consider basic start terms.
To adapt spareness to PTRSs $\PP$, a rewrite step
using the rule $\ell \to \mu \in \PP$ and the substitution
$\sigma$ is called \emph{spare}
if $\sigma(x) \in \NF_{\PP}$ for every $x \in \Var$ that occurs more than once in
some $r\in\Supp(\mu)$.
A $\fliftto_\PP$-rewrite sequence is spare if each of its $\fliftto_\PP$-steps is spare.
$\PP$ is spare if each $\fliftto_\PP$-rewrite sequence that starts with $\{1:t\}$
for a basic term $t \in \TT_{\BB}$ is spare.

\begin{exa}\label{example:spareness}
    Consider the PTRS $\PP_8$ with the two rules:

    \vspace*{-.5cm}
    \begin{minipage}[t]{7cm}
        \begin{align*}
            \tg &\to \{\nicefrac{3}{4}:\td(\tz), \nicefrac{1}{4}:\tg\}
        \end{align*}
    \end{minipage}
    \hspace*{.5cm}
  \begin{minipage}[t]{4cm}
        \begin{align*}
            \td(x) &\to \{1:\tc(x,x)\}
        \end{align*}
    \end{minipage}

    \vspace*{.3cm}

    \noindent
    It is similar to the PTRS $\PP_1$ from Counterex.\ \ref{example:diff-AST-vs-iAST-dup}, but we
    exchanged
    the symbols $\tg$ and $\tz$ in the right-hand side of the $\tg$-rule.
    This PTRS is orthogonal but duplicating due to the $\td$-rule.
    However, in any rewrite sequence that starts with $\{1:t\}$
    for a basic term $t$ we can only duplicate constructor symbols but no terms with
    defined symbols.
    Hence, $\PP_8$ is spare.
\end{exa}

If a PTRS $\PP$ is spare, and we start with a basic term, then we will only duplicate normal
forms with our duplicating rules.
This means that the duplicating rules do not influence the (expected) runtime and, more importantly
for $\mathtt{AST}$, the probability of convergence.
This leads to the following theorem, which
weakens the requirement of RL to SP in \Cref{properties-eq-AST-iAST-1},
where ``\emph{starting in $\TT_{\BB}$}'' means 
that one only considers rewrite sequences that start with $\{1:t\}$ for a term $t \in
\TT_{\BB}$, where $\TT_{\BB}$ is again the set of all basic terms w.r.t.\ $\PP$.

\begin{restatable}[From $\PSNi$ Starting in $\TT_{\BB}$ to $\PSNf$ Starting in $\TT_{\BB}$]{thm}{ASTAndIASTPropertyThree}\label{properties-eq-AST-iAST-3}
    If a PTRS $\PP$ is OR and SP, then:
    \begin{align*}
        \PSNf \text{ starting in } \TT_{\BB} &\Longleftrightarrow \PSNi \text{ starting in } \TT_{\BB}\!
    \end{align*}
\end{restatable}

For the proof of \Cref{properties-eq-AST-iAST-3} we use the following lemma.

\begin{restatable}[From Innermost to Full Rewriting Starting in $\TT_{\BB}$]{lem}{ASTAndIASTLemmaThree}\label{lemma-eq-AST-iAST-3}
    If a PTRS $\PP$ is OR and SP, then for every infinite $\fliftto_{\PP}$-rewrite sequence 
    $\vec{\mu} = (\mu_n)_{n \in \IN}$ that starts with a basic term
    there exists an infinite $\iliftto_{\PP}$-rewrite sequence 
    $\vec{\nu} = (\nu_n)_{n \in \IN}$ that starts with a basic term
    such that 
    \begin{enumerate}
        \item $\lim\limits_{n \to \infty}|\mu_n|_{\PP} \geq \lim\limits_{n \to \infty}|\nu_n|_{\PP}$
        \item $\edl(\vec{\mu}) \leq \edl(\vec{\nu})$
    \end{enumerate}
\end{restatable}

The requirement of
basic start terms is a real restriction for orthogonal PTRSs,
i.e., there exist very simple orthogonal PTRSs $\PP$, where $\PSNs$ starting in $\TT_{\BB}$ holds, 
but not $\PSNs$ in general.
One can see this already in the non-probabilistic setting for the TRS $\R$ consisting of
the two rules $\tf(\tg(\ta)) \to \tf(\tg(\ta))$ and $\tg(\tb) \to \tb$,
where we do not have $\SNs$, but $\SNs$ starting in $\TT_{\BB}$ for every $s \in \IS$.

The following example shows that \Cref{properties-eq-AST-iAST-3} does not hold for arbitrary start terms.

\begin{counterexample}\label{example:basic-terms-problem}
    Consider the PTRS $\PP_9$ with the two rules:

    \vspace*{-.5cm}
    \begin{minipage}[t]{7cm}
        \begin{align*}
            \tg &\to \{\nicefrac{3}{4}:\ts(\tg), \nicefrac{1}{4}:\tz\}
        \end{align*}
    \end{minipage}
  \begin{minipage}[t]{7cm}
        \begin{align*}
            \tf(\ts(x)) &\to \{1:\tc(\tf(x),\tf(x))\}
        \end{align*}
    \end{minipage}
  
    \vspace*{.3cm}

    \noindent
    This PTRS behaves similarly to $\PP_1$ (see Counterex.\ \ref{example:diff-AST-vs-iAST-dup}).
    We do not have $\AST[\fto_{\PP_9}]$ (and hence, also
    neither $\PAST[\fto_{\PP_9}]$ nor $\SAST[\fto_{\PP_9}]$), 
    as we have $\{1:\tf(\tg)\} \liftto_{\PP_9}^2 \{\nicefrac{3}{4}:\tc(\tf(\tg),\tf(\tg)), \nicefrac{1}{4}:\tf(\tz)\}$,
    which corresponds to a random walk biased towards non-termination (since $\tfrac{3}{4} > \tfrac{1}{4}$).

    However, the only basic terms for this PTRS are $\tg$ and $\tf(t)$ for terms $t$ that
    do not contain $\tg$ or $\tf$.
    A sequence starting with $\tg$ corresponds to flipping a biased coin and a sequence
    starting with $\tf(t)$ will clearly terminate. 
    Hence, we have $\SAST[\fto_{\PP_9}]$  (and thus, also
    $\PAST[\fto_{\PP_9}]$ and $\AST[\fto_{\PP_9}]$)
    starting in $\TT_{\BB}$.
    Furthermore, note that we have $\SAST[\ito_{\PP_9}]$ (and thus, also
    $\PAST[\ito_{\PP_9}]$ and $\AST[\ito_{\PP_9}]$) for arbitrary start terms,
    analogous to $\PP_1$.
    Since $\PP_9$ is OR and SP, this shows that \Cref{properties-eq-AST-iAST-3} cannot be extended to $\PSNf$ in general.
\end{counterexample}

For expected complexity, we obtain a result that is analogous to \Cref{properties-eq-AST-iAST-3}.
Since we require basic start terms, the following theorem
only holds for expected runtime complexity
and not for expected derivation complexity.

\begin{restatable}[From Innermost to Full Expected Runtime Complexity]{thm}{ASTAndIASTSparePropertyTwoComplex}\label{properties-eq-AST-iAST-3-complex}
  If a PTRS $\PP$ is OR and SP, then:
  \[\eruntime_{\fto_{\PP}} = \eruntime_{\ito_{\PP}}\]
\end{restatable}
Thus, this theorem weakens the requirement of RL in
\Cref{properties-eq-AST-iAST-1-complex} to SP (recall that right-linearity 
implies spareness).
A corresponding corollary for the non-probabilistic setting would be subsumed by
\Cref{properties-irun-vs-run},
which ensures  $\runtime_{\fto_{\R}} = \runtime_{\ito_{\R}}$ for all 
spare overlay systems $\R$. 
 As shown by \Cref{ORinsteadofOS},  OS is not enough in the probabilistic setting, but one
 needs OR.

One may wonder whether \Cref{properties-eq-AST-iAST-3} can nevertheless be used in order
to prove $\PSNf$ for a PTRS $\PP$ on \emph{all} terms (instead of just
basic start terms) by using a suitable transformation from
$\PP$ to another PTRS $\PP'$ such that $\PSNf$ holds for
all terms iff
$\PSN[\fto_{\PP'}]$ holds when starting with
 basic terms.
In~\cite{fuhs2019transformingdctorc}, a transformation was presented
that extends any (non-probabilistic) TRS $\R$ by generator rules $\C{G}(\R)$ such that the
derivational complexity of $\R$ 
is the same as the runtime complexity of $\R \cup \C{G}(\R)$,
where $\C{G}(\R)$ are considered to be \emph{relative} rules whose rewrite steps do not ``count'' for the complexity.
This transformation can indeed be reused to move from 
$\PSNf$ \emph{starting in $\TT_{\BB}$} to $\PSNf$ on
arbitrary terms.

The idea of the transformation is to introduce
a new constructor symbol $\tcons_f$ for every defined symbol $f \in \Sigma_D$,
and to introduce a new defined
symbol $\tenc_{f}$ for every function symbol $f \in \Sigma$.
As an example for $\PP_8$ from \Cref{example:spareness},
then instead of starting with the non-basic term $\tc(\tg,\tf(\tg))$,
we start with the basic term $\tenc_{\tc}(\tcons_{\tg},\tcons_{\tf}(\tcons_{\tg}))$, its \emph{basic variant}.
The new defined symbol $\tenc_{\tc}$ is used to first build the term $\tc(\tg,\tf(\tg))$
at the beginning of the rewrite sequence, i.e., it converts
all occurrences of $\tcons_{f}$ for $f \in \Sigma_D$ back into the defined symbol $f$, and then we can proceed as if we started with the term $\tc(\tg,\tf(\tg))$ directly.
For this conversion, we need another new defined symbol $\targenc$ that iterates through the term and
replaces all new constructors $\tcons_{f}$ by the original defined symbol $f$. 
Thus, we define the generator rules as in~\cite{fuhs2019transformingdctorc} 
(just with trivial probabilities in the right-hand sides $\ell \to \{1:r\}$), 
since we do not need any probabilities during this initial construction of the original start term.

\begin{defi}[Generator Rules $\C{G}(\PP)$]\label{def:generator-rules}
    Let $\PP$ be a PTRS over the signature $\Sigma$. 
    Its \emph{generator rules}
    $\C{G}(\PP)$ are the following set of rules
      \allowdisplaybreaks%
    \begin{align*}
       & \{\tenc_f(x_1, \ldots, x_k) \to \{1: f(\targenc(x_1), \ldots, \targenc(x_k))\} \mid f \in \Sigma\}\\
         \cup  \;&\{\targenc(\tcons_f(x_1, \ldots, x_k)) \to \{1: f(\targenc(x_1), \ldots, \targenc(x_k))\} \mid f \in \Sigma_D\}\\
         \cup \; & \{\targenc(f(x_1, \ldots, x_k)) \to \{1: f(\targenc(x_1), \ldots, \targenc(x_k))\} \mid f \in \Sigma_C\},
    \end{align*}
    
    \noindent
    where $x_1, \ldots, x_k$ are pairwise different variables and where the function symbols $\targenc$, $\tcons_f$, and $\tenc_f$ are fresh (i.e., they do not occur in $\PP$). 
    Moreover, we define $\Sigma_{\C{G}(\PP)} = \{\tenc_f \mid f \in \Sigma\} \cup \{\targenc\} \cup
    \{\tcons_f \mid f \in \Sigma_D\}$. 
\end{defi}

\begin{exa}\label{undefinedLink}
   For the PTRS $\PP_9$ from Counterex.\ \ref{example:basic-terms-problem},
   we obtain the following generator rules $\C{G}(\PP_9)$:

   \vspace*{-.2cm}

   \begin{longtable}{rcl}
        $\tenc_{\tg}$ & $\to$ &$\{1: \tg\}$\\
        $\tenc_{\tf}(x_1)$ & $\to$ &$\{1: \tf(\targenc(x_1))\}$\\
        $\tenc_{\tc}(x_1, x_2)$  &$\to$& $\{1: \tc(\targenc(x_1), \targenc(x_2))\}$\\
        $\tenc_{\ts}(x_1)$ & $\to$ &$\{1: \ts(\targenc(x_1))\}$\\
        $\tenc_{\tz}$ & $\to$& $\{1: \tz\}$\\
        $\targenc(\tcons_{\tg})$ & $\to$ &$\{1: \tg\}$\\
        $\targenc(\tcons_{\tf}(x_1))$ &$\to$ &$\{1: \tf(\targenc(x_1))\}$\\
        $\targenc(\tc(x_1, x_2))$ & $\to$ &$\{1: \tc(\targenc(x_1), \targenc(x_2))\}$\\
        $\targenc(\ts(x_1))$ & $\to$ &$\{1: \ts(\targenc(x_1))\}$\\
        $\targenc(\tz)$ & $\to$ &$\{1: \tz\}$
    \end{longtable}
 \end{exa}

As mentioned, using the  symbols $\tcons_f$ and $\tenc_f$, as in~\cite{fuhs2019transformingdctorc} every term over 
$\Sigma$ can be transformed into a basic term over $\Sigma \cup \Sigma_{\C{G}(\PP)}$.

The following lemma shows that by adding the generator rules, one can indeed reduce the
problem of proving  $\AST$  on all terms to
$\AST$ starting in $\TT_{\BB}$.

\begin{restatable}[From all Terms to Basic Terms (1)]{lem}{BasicASTToAST}\label{lemma:spareness-AST-proof-1}
    For any PTRS $\PP$ we have:
    \begin{align*}
        \ASTf \Longleftrightarrow \AST[\fto_{\PP \cup \C{G}(\PP)}] \text{ starting in } \TT_{\BB}.\!
    \end{align*}
\end{restatable}

If one extends the definition of $\PAST$ by
relative rules, then
similar results should also be possible for $\PASTf$,
$\SASTf$, and both
expected derivational and runtime
complexity.
When considering the additional generator rules as ordinary instead of relative rewrite rules, 
we obtain an upper bound on the expected complexity and sufficient conditions for $\PASTf$ and $\SASTf$.

\begin{restatable}[From all Terms to Basic Terms (2)]{lem}{BasicPASTToPAST}\label{lemma:spareness-AST-proof-2}
    For any PTRS $\PP$ we have:
    \begin{align*}
        \PASTf \Longleftarrow \PAST[\fto_{\PP \cup \C{G}(\PP)}] \text{ starting in } \TT_{\BB}\\
        \SASTf \Longleftarrow \SAST[\fto_{\PP \cup \C{G}(\PP)}] \text{ starting in } \TT_{\BB}\\
    \end{align*}
    \vspace*{-1.1cm}
    \[\edertime_{\fto_{\PP}} \leq \eruntime_{\fto_{\PP \cup \C{G}(\PP)}}\]
\end{restatable}

However,  even if $\PP$ is spare, the PTRS
$\PP \cup \C{G}(\PP)$ is not guaranteed to be spare, although the generator rules
themselves are right-linear. 
The problem is that the generator rules include  a rule like
$\tenc_{\tf}(x_1) \to \{1: \tf(\targenc(x_1))\}$ where
a defined symbol $\targenc$ occurs below the duplicating symbol $\tf$ on the right-hand
side.
Indeed, while $\PP_9$ is spare,
$\PP_9 \cup \C{G}(\PP_9)$ is not. For example, when starting with the basic term
$\tenc_{\tf}(\ts(\tcons_{\tg}))$, \pagebreak[3] we have
\[
\begin{array}{r@{\;\;}l@{\;\;}l}
\{1:\tenc_{\tf}(\ts(\tcons_{\tg}))\} &\fliftto_{\C{G}(\PP_9)}^2&
\{1:\tf(\ts(\targenc(\tcons_{\tg})))\}\\
&\fliftto_{\PP_9}&
\{1:\tc(\tf(\targenc(\tcons_{\tg})), \tf(\targenc(\tcons_{\tg})))\}
\end{array}
\]
where the last step is not spare.
In general, $\PP \cup \C{G}(\PP)$ is guaranteed to be spare if $\PP$ is right-linear.
So we could modify \Cref{properties-eq-AST-iAST-3} into a theorem which states that
$\ASTf$ holds for all terms iff $\AST[\ito_{\PP \cup \C{G}(\PP)}]$ holds when starting in
$\TT_{\BB}$ (and thus, for all terms) 
for orthogonal and right-linear PTRSs $\PP$.
However, this theorem would be subsumed by \Cref{properties-eq-AST-iAST-1}, where we 
already showed the equivalence of $\ASTf$ and $\ASTi$ if $\PP$ is orthogonal and right-linear. Indeed,
our goal in \Cref{properties-eq-AST-iAST-3} was to find a weaker requirement than right-linearity.
Hence, such a transformational approach to move from $\ASTf$ on all start terms to $\ASTf$
starting in $\TT_{\BB}$ does not seem viable for \Cref{properties-eq-AST-iAST-3}.

\section{Implementation and Evaluation}\label{Evaluation}

We implemented our new criteria for the equivalence of
$\ASTi$ and
$\ASTf$
in our termination prover \textsf{AProVE}~\cite{JAR-AProVE2017}. 
For every PTRS, one can indicate
whether one wants to analyze its termination behavior for all start terms or only for basic start terms.
\aprove's main technique for termination analysis of PTRSs $\PP$ is the probabilistic DP
framework from~\cite{kassinggiesl2023iAST,FLOPS2024} to prove $\ASTi$,
and its adaption for $\ASTf$ from~\cite{JPK60},
which is however strictly less powerful than the framework for $\ASTi$.
The general idea of the DP framework is
a \emph{divide-and-conquer} approach where \emph{processors} are used
to replace a termination problem by several new sub-problems that 
are easier to analyze than the original problem.
These processors are applied repeatedly 
until all sub-problems are solved.
Since different techniques can be used to analyze the different sub-problems,
this results in a \emph{modular} approach for almost-sure termination analysis.

For our evaluation, we developed a
``\tool{new}'' version of \tool{AProVE}, where we consider
\Cref{properties-eq-AST-iAST-1,properties-eq-AST-iAST-3}, 
and afterwards apply the DP framework as described above.
More precisely,
if one wants to analyze $\ASTf$ for a PTRS $\PP$, 
the ``\tool{new}'' version of \tool{AProVE} first tries to prove that
the conditions of \Cref{properties-eq-AST-iAST-1} are satisfied if one
regards arbitrary start terms or that
the conditions of \Cref{properties-eq-AST-iAST-3} are satisfied if one only wants to consider 
basic start terms.
If this succeeds, then it uses the probabilistic DP framework for $\ASTi$, 
which then implies $\ASTf$.\footnote{Currently, we only use the switch from full to innermost rewriting as a preprocessing step
before applying the DP framework.
As shown in~\cite{JPK60}, 
a corresponding modular processor \emph{within} the DP framework
(which requires the criteria of our theorems
not for the whole PTRS, but just for specific sub-problems
within the termination proof) would be unsound.}
If these theorems cannot be applied, then \textsf{AProVE}
tries to prove $\ASTf$ via the DP framework for
$\ASTf$~\cite{JPK60}, or using a direct application of polynomial orderings~\cite{kassinggiesl2023iAST}.
So in contrast to the experiments in our conference paper
\cite{FoSSaCS2024}, where we only used the
direct application of polynomial orderings to prove $\ASTf$, in this case, ``\tool{new}'' now
also uses the DP framework for $\ASTf$ from~\cite{JPK60}.

We compare ``\tool{new}'' with the
``\tool{old}'' version of \tool{AProVE} that does not apply any of the theorems of
the current paper. Thus, it directly uses the DP framework for $\ASTf$ from~\cite{JPK60}
or the direct application of polynomial orderings~\cite{kassinggiesl2023iAST}.

Note that for $\ASTf$ w.r.t.\ basic start terms,
\Cref{properties-eq-AST-iAST-3}
generalizes \Cref{properties-eq-AST-iAST-1}, since right-linearity implies
spareness.
To show the advantage of \Cref{properties-eq-AST-iAST-3}
compared to \Cref{properties-eq-AST-iAST-1}, we also experimented with the variant 
``\tool{old} + \Cref{properties-eq-AST-iAST-1}'' where we always used
\Cref{properties-eq-AST-iAST-1} instead of 
\Cref{properties-eq-AST-iAST-3}, even when only proving
$\ASTf$ for basic start terms.

We used the benchmark set of 130 PTRSs from~\cite{JPK60},
where \aprove{} can prove $\ASTi$ for 109 of them.
\Cref{table:results} shows for how many of these 130 PTRSs the respective strategy 
allows \tool{AProVE} to conclude $\ASTf$.
  
\setcounter{table}{0}
\begin{table}
  \centering
  \begin{tabular}{lccc}
  \toprule
  & \tool{old} & \tool{old} + \Cref{properties-eq-AST-iAST-1} & \tool{new} \\ 
  \midrule
  Arbitrary Terms & 52 & 57 & 57 \\ 
  Basic Terms & 59 & 65 & 70 \\ 
  \bottomrule
  \end{tabular}
  \caption{Results}\label{table:results}
\end{table}

The increase in the number of solved examples from ``\tool{old}'' to ``\tool{new}''
shows that even with dedicated strategies and techniques like the DP framework
for $\ASTf$ \cite{JPK60},
it is still beneficial to use the techniques from the current paper to move from
$\ASTf$ to $\ASTi$ whenever possible, since 
$\ASTi$ is significantly easier to prove than $\ASTf$.
Moreover, the table shows that the weakening of right-linearity to spareness (i.e., the
generalization of \Cref{properties-eq-AST-iAST-1} to \Cref{properties-eq-AST-iAST-3}) indeed
increases power when considering basic start terms.
As shown in the experiments of \cite{kassing2025DependencyPairsExpected}, our 
\Cref{properties-eq-AST-iAST-1,properties-eq-AST-iAST-3} are also very useful in order to move from
 $\SASTf$ vs.\ $\SASTi$, i.e., they often allow us to use the DP framework for
$\SASTi$ and expected innermost complexity analysis from \cite{kassing2025DependencyPairsExpected}
in order to prove $\SASTf$ and to analyze expected runtime complexity for full rewriting.

For details on our experiments, our collection of examples, and for instructions on how to run our implementation
in \textsf{AProVE} via its \emph{web interface} or locally, we refer to:
\begin{center}
  \url{https://aprove-developers.github.io/InnermostToFullAST/}
\end{center}

\section{Modularity}\label{Modularity}

In this section, we investigate the modularity of probabilistic notions of termination.
We will see that as in the non-probabilistic setting~\cite{Gramlich1995AbstractRB},
in contrast to innermost probabilistic rewriting,
full
probabilistic rewriting is not modular in general.
However, our results on the relation between innermost and full probabilistic
rewriting from the previous sections will allow us to obtain
modularity results for full probabilistic rewriting as well.
A property is called \emph{modular} if it is preserved for certain unions of PTRSs.
Modularity is not only interesting from a theoretical point of view, but in practice it is
also very important to know
whether one can split a huge PTRS into smaller parts
such that the property of interest is preserved.
Then, one can analyze these parts independently of each other.
As remarked earlier, this is also the main idea of the DP framework,
and in fact, we already benefit from
similar modularity results within the DP framework itself.
The fewer conditions are required for modularity of a notion of termination,
the stronger the corresponding DP framework is expected to be.

We will study two different forms of
unions, namely \emph{disjoint unions} (\Cref{Sect:Disjoint Unions}), where both PTRSs
do not share any function symbols, and \emph{shared constructor unions}
of PTRSs
(\Cref{Sect:Shared Constructor Unions}),
which may have common constructor symbols, but whose defined symbols are
disjoint.\footnote{The dependency pair
framework for  $\ASTi$
in \aprove{} is
already capable of 
splitting 
   disjoint unions and shared constructor unions into sub-problems that can
   be analyzed independently.
   However, in this section our goal is 
to analyze the modularity of different
   probabilistic notions of termination in general. Thus, this allows the use of these
   modularity results also outside the DP framework.}
In future work, one may extend this analysis to \emph{hierarchical unions},
which are unions of PTRSs where the first PTRS may
contain defined symbols 
of the second one, but not vice versa. Moreover, in this section we only consider the strategies
$s \in \{ \mathbf{f}, \mathbf{i} \}$ (but it is not difficult to show
that our  modularity results for innermost rewriting also hold for
leftmost-innermost rewriting). 

In \Cref{Sect:Signature Extensions}, we will also investigate signature extensions.
We have already seen in \Cref{thm:Sig-PAST} that $\PASTs$ is not closed under signature extensions for any $s \in \IS$.
Based on our modularity results, we will now
show that both $\ASTs$ and $\SASTs$ are closed under signature extensions.

To distinguish the function symbols of different
PTRSs $\PP$,
in the following we write $\Sigma_{D}^{\PP}$, $\Sigma_{C}^{\PP}$, and $\Sigma^{\PP}$ 
for the defined symbols, constructor symbols, and all
function symbols occurring in the rules of $\PP$, respectively.

\subsection{Disjoint Unions}\label{Sect:Disjoint Unions}

We first consider unions of systems that do not share any function symbols, i.e., 
we consider two PTRSs $\PP^{(1)}$ and $\PP^{(2)}$ such that $\Sigma^{\PP^{(1)}} \cap \Sigma^{\PP^{(2)}} = \emptyset$.
In the non-probabilistic setting,
\cite{Gramlich1995AbstractRB} showed that innermost termination is modular for disjoint unions.
This result can be lifted to $\ASTi$ and $\SASTi$.
We first investigate $\ASTi$, as illustrated by the following example.

\begin{exa}\label{example:modularity-innermost-disjoint-AST}
    Consider the PTRS $\PP_{10} = \PP_{10}^{(1)} \cup \PP_{10}^{(2)}$ given by
    
    \vspace*{-.5cm}
    \begin{minipage}[t]{7cm}
        \begin{align*}
            \PP_{10}^{(1)} : \tf(x) &\to \{\nicefrac{1}{2}:\tf(x), \nicefrac{1}{2}:\ta\}\!
        \end{align*}
    \end{minipage}
    \begin{minipage}[t]{7cm}
        \begin{align*}
            \PP_{10}^{(2)} : \tg(x) &\to \{\nicefrac{1}{2}:\tg(x), \nicefrac{1}{2}:\tb\}\!
        \end{align*}
    \end{minipage}
    \vspace*{.1cm}

    \noindent
    $\PP_{10}^{(1)}$ and $\PP_{10}^{(2)}$ both correspond to a fair coin flip, 
    where one terminates when obtaining heads.
    Hence, for both systems we have $\AST[\fto_{\PP_{10}^{(1)}}]$ and $\AST[\fto_{\PP_{10}^{(2)}}]$, 
    and thus also $\AST[\ito_{\PP_{10}^{(1)}}]$ and $\AST[\ito_{\PP_{10}^{(2)}}]$.
    Furthermore, $\Sigma^{\PP_{10}^{(1)}} \cap \Sigma^{\PP_{10}^{(2)}} = \emptyset$,
    i.e., $\PP_{10}$ is a disjoint union.
    When reducing a term like
    $\tf(\tg(x))$ which
    contains symbols from both systems, 
    then we first reduce the innermost redex $\tg(x)$ until we reach a normal form.
    Due to the innermost strategy, we cannot rewrite at the position of $\tf$ beforehand.
    This reduction only uses one of the two systems, namely $\PP_{10}^{(2)}$, hence it
    converges with probability $1$. 
    Then, we use the next innermost redex, which will be $\tf(\tb)$,
    using only rules of $\PP_{10}^{(1)}$ until we reach a normal form,
    where symbols from $\Sigma^{\PP_{10}^{(2)}}$ do not influence the reduction.
    The reason is that all subterms below or at positions of symbols from $\Sigma^{\PP_{10}^{(2)}}$
    are in normal form, and there is no symbol from $\Sigma^{\PP_{10}^{(2)}}$ above the
    $\tf$ at the root position.
    Again, this reduction converges with probability $1$.
    Thus, in the end, our reduction starting with $\tf(\tg(x))$ also converges with probability $1$,
    and the same holds for arbitrary start terms, which implies $\AST[\fto_{\PP_{10}}]$.
\end{exa}

In \Cref{example:modularity-innermost-disjoint-AST}, we considered the term
$\tf(\tg(x))$ 
where we swap once
between a symbol $\tf$ from $\Sigma^{\PP_{10}^{(1)}}$ and a symbol $\tg$ from $\Sigma^{\PP_{10}^{(2)}}$ on the path
from the root to the ``leaf''
of the term. 
In the proof of \Cref{modularity-iAST-disjoint}
we lift the argumentation of
\Cref{example:modularity-innermost-disjoint-AST} 
to arbitrary terms via induction.
This proof idea was also used by \cite{Gramlich1995AbstractRB} in the non-probabilistic
setting to show
the modularity of innermost termination for disjoint unions.

\begin{restatable}[Modularity of $\ASTi$ for Disjoint Unions]{thm}{ModiASTDisjoint}\label{modularity-iAST-disjoint}
  Let $\PP^{(1)}$ and $\PP^{(2)}$ be PTRSs with
  $\Sigma^{\PP^{(1)}} \cap \Sigma^{\PP^{(2)}} = \emptyset$. Then we have:
    \begin{align*}
        \AST[\ito_{\PP^{(1)} \cup \PP^{(2)}}]
        &\; \Longleftrightarrow \; \AST[\ito_{\PP^{(1)}}] \; \text{ and } \; \AST[\ito_{\PP^{(2)}}]\!
    \end{align*}
\end{restatable}

\begin{myproofsketch}
    The direction ``$\Longrightarrow$'' is trivial and thus, we only prove
    ``$\Longleftarrow$''. Assume that we have $\AST[\ito_{\PP^{(1)}}]$ and $\AST[\ito_{\PP^{(2)}}]$.

    For $\AST[\ito_{\PP^{(1)} \cup \PP^{(2)}}]$, it suffices to regard only rewrite sequences that start
    with multi-distributions of the form $\{1:t\}$ (see
    \Cref{lemma:PTRS-AST-single-start-term} in App.~\ref{appendix} for a proof).
    Thus, we show by
    structural induction on the term structure that for every 
    $t \in \TSet{\Sigma^{\PP^{(1)}} \cup \Sigma^{\PP^{(2)}}}{\VSet}$, all
    $\iliftto_{\PP^{(1)} \cup \PP^{(2)}}$-rewrite sequences starting with $\{1:t\}$
    converge with probability 1.

    If $t \in \VSet$, then $t$ is in normal form.
    If $t$ is a constant, then w.l.o.g.\ let $t \in \PP^{(1)}$.
    Since we have $\AST[\ito_{\PP^{(1)}}]$, 
    $t$ cannot start an infinite $\iliftto_{\PP^{(1)} \cup \PP^{(2)}}$-rewrite sequence
    that converges with probability~$<1$.

    In the induction step
    we have $t = f(q_1, \ldots, q_k)$.
    Due to the innermost evaluation strategy, we can only rewrite at the root position
    if every proper subterm is in normal form. Thus, we first only consider rewrite steps
    below the root.
    By the induction hypothesis, every
    infinite $\iliftto_{\PP^{(1)} \cup \PP^{(2)}}$-rewrite sequence that starts with some $\{1:q_i\}$ 
    converges with probability $1$, and hence, every
    infinite $\iliftto_{\PP^{(1)} \cup \PP^{(2)}}$-rewrite sequence that starts with
    $\{1:t\}$ converges with probability $1$ as well if it does not 
    perform rewrite steps at the root position.
    However, in the probabilistic setting, this observation 
    requires a quite complex approximation of the convergence probability.
    The reason is that if we rewrite a term $q_i$ to $\{p_1:q_{i,1}, \ldots, p_m:q_{i,m}\}$, then we obtain
    a distribution $\{p_1:f(q_1, \ldots, q_{i,1}, \ldots, q_k), \ldots, p_m:f(q_1, \ldots, q_{i,m}, \ldots, q_k)\}$.
    Now, the terms $q_j$ with $j \neq i$ occur multiple times in this distribution, and we may use different rules to rewrite them.
    Hence, the order in which we rewrite the different $q_i$ matters and cannot be
    chosen arbitrarily (as seen in Counterex.\ \ref{example:liAST-vs-iAST}).
    This
    duplication of the terms $q_j$ 
    is also the reason why $\PASTi$ is not modular for disjoint systems, see \Cref{example:Past-binary-function}.

    In the second step, we also allow rewrite steps at the root position. W.l.o.g., let the root symbol
    $f$ of $t$ be from $\Sigma^{\PP^{(1)}}$.
    Before performing a rewrite step at the root,
    we can replace all maximal (i.e., topmost)
    subterms of $t$ with root symbols from $\Sigma^{\PP^{(2)}}$ 
    by fresh variables (using the same variable for the same subterm).\footnote{Note that this
    replacement is only performed  in the induction step at the root,
    and not for every swap between symbols from $\Sigma^{\PP^{(1)}}$ and $\Sigma^{\PP^{(2)}}$
    on the path from the root to the leaves of $t$.}
    Since $f$ is the root symbol (i.e., there
    is no symbol at a position above $f$),
    and all proper subterms are in normal form,
    this replacement does not change the convergence probability. After the replacement,
    we only have function symbols from $\Sigma^{\PP^{(1)}}$, and thus the convergence
    probability of the resulting rewrite sequence is 1 since $\AST[\ito_{\PP^{(1)}}]$ holds.
\end{myproofsketch}

In contrast to $\ASTi$, 
$\PASTi$ cannot be modular due to the potential extension of the
signature (recall that by \Cref{thm:Sig-PAST}, $\PASTi$ is not closed under
signature extensions).

\begin{counterexample}\label{example:modularity-innermost-disjoint-PAST}
    Consider the PTRS $\PP_{\mathsf{unary}}$ from Counterex.\ \ref{example:Past-vs-Sast-finite}, 
    and a PTRS $\PP_{11}$ containing a binary function symbol $\tc$ such that
    $\PAST[\ito_{\PP_{11}}]$ and $\Sigma^{\PP_{\mathsf{unary}}} \cap \Sigma^{\PP_{11}} = \emptyset$ hold. 
    For example, $\PP_{11}$ could consist of
    the only rule $\tc(\td,\td) \to \{1:\tc(\te,\te)\}$.
    As explained in \Cref{example:Past-binary-function}, due to the signature extension by a
    binary function symbol, we do not have $\PAST[\ito_{\PP_{\mathsf{unary}} \cup \PP_{11}}]$, 
    while both $\PAST[\ito_{\PP_{\mathsf{unary}}}]$ and $\PAST[\ito_{\PP_{11}}]$ hold.
\end{counterexample}

Finally, we consider $\SASTi$.
To prove that
$\SASTi$ is modular for
disjoint unions,
we have to show that
the expected derivation height of any term $t$ is finite.
However, after rewriting $t$'s proper subterms to normal forms, as we did in \Cref{example:modularity-innermost-disjoint-AST} 
and in the induction proof of \Cref{modularity-iAST-disjoint}, 
we may end up with infinitely many different terms. All their expected derivation heights have
to be considered in order to compute
the expected derivation height of $t$. 

\begin{exa}\label{example:SAST-modularity-proof-problems}
    Consider the PTRSs $\PP_{12}^{(1)}$ and $\PP_{12}^{(2)}$ with 

    \vspace*{-.4cm}
    \begin{minipage}[t]{7cm}
        \begin{align*}
            \PP_{12}^{(1)}: \tf(\ts(x),y) &\to \{1:\tf(x,y)\}\\
            \tf(x,\ts(y)) &\to \{1:\tf(x,y)\}\\
            \ta &\to \{\nicefrac{1}{2}:\tz, \nicefrac{1}{2}:\ts(\ta)\}
        \end{align*}
    \end{minipage}
    \begin{minipage}[t]{7cm}
        \begin{align*}
            \PP_{12}^{(2)}: \tg(x) &\to \{1:x\}
        \end{align*}
    \end{minipage}

    \smallskip
    
    \noindent
    Clearly, we have both $\SAST[\ito_{\PP_{12}^{(1)}}]$ and
    $\SAST[\ito_{\PP_{12}^{(2)}}]$.
    Now consider the term $t = \tf(\tg(\ta), \tg(\ta))$.
    Due to the innermost strategy,  we have to
    rewrite its proper subterms first. When proceeding in a similar way as in the
    induction proof of \Cref{modularity-iAST-disjoint}, 
    then one
    would first construct bounds on the expected derivation heights of the proper subterms,
    and then use them to obtain a bound on the expected derivation height of the whole
    term $t$.
    However, reducing $t$'s proper subterms can create  
    infinitely many different terms,
    i.e., all terms of the form $\tf(\ts^n(\tz), \ts^m(\tz))$ for any $n, m \in \mathbb{N}$ can
    be reached with a certain probability.
    Since there is no finite supremum on the derivation height of
    $\tf(\ts^n(\tz), \ts^m(\tz))$ for all $n, m \in \mathbb{N}$,
    one would have to take the individual probabilities for reaching
    the terms $\tf(\ts^n(\tz), \ts^m(\tz))$  into account in order to prove that the expected
    derivation height of $t$ is indeed finite. 
\end{exa}

Instead, we use an easier argument to show
that any term like $t$ has finite expected derivation height.
Recall that in the induction proof of \Cref{modularity-iAST-disjoint}, in the induction step we first 
rewrite below the root position until every proper subterm is in normal form. Afterwards, if the root
symbol of $t$ is from $\Sigma^{\PP^{(1)}}$, then we replace all maximal subterms of $t$
with root symbols from $\Sigma^{\PP^{(2)}}$ by fresh variables. This results in a
term $t'$ over $\Sigma^{\PP^{(1)}}$ which is considered for the remaining derivation. 
As shown in \Cref{example:SAST-modularity-proof-problems}, there may be infinitely many such
terms $t'$, e.g., in \Cref{example:SAST-modularity-proof-problems}, $t'$ can be 
any term of the form $\tf(\ts^n(\tz), \ts^m(\tz))$. However, this infinite set of terms
can be over-approximated using the following
finite abstraction.

For the root symbol $\tf \in \Sigma^{\PP^{(1)}}$ of $t = \tf(\tg(\ta), \tg(\ta))$, the normal
forms reachable from  $\tg(\ta)$ can be over-approximated by considering the
normal forms reachable from the
argument $\ta$ of $\tg  \in \Sigma^{\PP^{(2)}}$ (because function symbols like $\tg$ may have ``collapsing
rules'' which return their arguments) or
by fresh variables (which represent possible normal forms
that start with symbols from
$\Sigma^{\PP^{(2)}}$). Thus, instead of considering the
rewrite steps at symbols
from $\Sigma^{\PP^{(1)}}$
in  $t$, instead we can consider all rewrite steps for the terms from the multiset $\{\tf(\ta, \ta),
\tf(x, \ta), \tf(\ta, y), \tf(x,y)\}$. This multiset is called
the \emph{disjoint union abstraction} of $t$ for $\Sigma^{\PP^{(1)}}$.
Note that all terms in this disjoint union abstraction are indeed from $\TSet{\Sigma^{\PP^{(1)}}}{\VSet}$.
To also capture the possibility that  the two occurrences of
 $\tg(\ta)$ in $t$ might reach the same normal form that starts with a symbol from
$\Sigma^{\PP^{(2)}}$, the disjoint union abstraction of $t$ also contains $\tf(x,x)$ where
we identify the variables $x$ and $y$.

Similarly, instead of considering the rewrite steps at symbols
from $\Sigma^{\PP^{(2)}}$ in $t$, we consider the two arguments of $\tf$ (i.e., the two
occurrences of $\tg(\ta)$) 
with roots from $\Sigma^{\PP^{(2)}}$
where each occurrence of the subterm $\ta \in \Sigma^{\PP^{(1)}}$ is replaced by a fresh variable 
(to represent possible normal forms that start with symbols from $\Sigma^{\PP^{(1)}}$). 
Thus, the disjoint union abstraction of $t$ for $\Sigma^{\PP^{(2)}}$ contains $\tg(x)$ and $\tg(y)$.

So instead of an induction proof as for the modularity of
$\ASTi$,\footnote{The proof idea
that we use for the modularity of $\SASTi$ for disjoint unions could also have been
used to prove the  modularity of $\ASTi$ for disjoint unions
(\Cref{modularity-iAST-disjoint}). However, for \Cref{modularity-iAST-disjoint} we used an
induction proof instead to have it similar to the
original proof of \cite{Gramlich1995AbstractRB} in the non-probabilistic setting. Moreover, such an induction
proof will also be required when showing the modularity of $\ASTi$
for shared constructor unions (\Cref{modularity-iAST-constructor}).}
for the modularity proof of $\SASTi$,
we replace the start term $t$ by all terms in the
disjoint union abstractions for both
$\Sigma^{\PP^{(1)}}$ and  $\Sigma^{\PP^{(2)}}$. This is a finite multiset of $K$ terms for some
$K \in \IN$. Since every term in this abstraction is either from 
$\TSet{\Sigma^{\PP^{(1)}}}{\VSet}$ or from
$\TSet{\Sigma^{\PP^{(2)}}}{\VSet}$, they all have a finite expected derivation height. Hence, if $C_{\max} \in \IN$
is the maximal expected derivation height of all these terms, then $K \cdot C_{\max}$ is a
(finite) bound on the expected derivation height of $t$.

The following definition introduces the \emph{disjoint union abstraction} formally. Here,
$\caphterm_1(t)$ is the multiset where all topmost subterms of $t$ with root from
$\Sigma^{\PP^{(2)}}$ are replaced by fresh variables or by
the abstractions of their subterms. The multiset $\capterm_1(t)$ results from
$\caphterm_1(t)$ by identifying any possible combination of the variables in the terms of $\caphterm_1(t)$.

\begin{defi}[Disjoint Union Abstraction]\label{def:term-abstraction}
    Let $\PP^{(1)}, \PP^{(2)}$ be PTRSs with $\Sigma^{\PP^{(1)}} \cap \Sigma^{\PP^{(2)}} = \emptyset$.
    For any $d \in \{1,2\}$ and any $t \in \TSet{\Sigma^{\PP^{(1)}} \cup
    \Sigma^{\PP^{(2)}}}{\VSet}$, 
    $\caphterm_d(t)$ and $\capterm_d(t)$ are multisets of terms from
    $\TSet{\Sigma^{\PP^{(d)}}}{\VSet}$, which are defined as follows. 
    {\small
    \[
        \begin{array}{ll}
          \caphterm_d(y) \hspace*{-.15cm}& = \{x\} \text{, if } y \in \VSet,
\text{  where  $x$  is always a new fresh variable }  \\
        	\caphterm_d(f(t_1, \ldots, t_k)) \hspace*{-.15cm}& = \{f(q_1, \ldots, q_k) \mid  q_1 \in
                \caphterm_d(t_1), \ldots,
q_k \in
                \caphterm_d(t_k)
                \} \text{, if } f \in \Sigma^{\PP^{(d)}}\\ 
        	\caphterm_d(f(t_1, \ldots, t_k)) \hspace*{-.15cm}& =  \{x\} \cup \caphterm_d(t_1) \cup
                \ldots \cup \caphterm_d(t_k)\text{,  otherwise, where
                    $x$  is always a new fresh variable }\!
    \end{array}
        \]}

    \noindent
    So $\caphterm_d(t)$ is always a linear term, i.e., it never contains multiple occurrences of the
    same variable.
 
    For any function $\varphi: X \to X$ with $X \subseteq \VSet$, 
    let $\sigma_{\varphi}$ be the substitution that replaces
    every variable $x \in X$ by $\varphi(x)$
    and leaves all other variables unchanged, i.e., $\sigma_{\varphi}(x) =
    \varphi(x)$ if $x \in X$ and $\sigma_X(x) = x$ otherwise. 
    Then we define
    \[ \capterm_d(t) = \{ \sigma_{\varphi}(q) \mid q \in \caphterm_d(t), \varphi: \VSet(q) \to \VSet(q) \}\]
    The \emph{disjoint union abstraction} of $t$ is the multiset $\capterm_1(t) \cup \capterm_2(t)$.
\end{defi}

\begin{exa}\label{example:disjoint-union-abstraction}
    Reconsider $\PP_{12}^{(1)}$ and $\PP_{12}^{(2)}$ from
    \Cref{example:SAST-modularity-proof-problems}. Here, 
    we obtain
    \[ \begin{array}{ll}
      \caphterm_1(\tf(\tg(\ta), \tg(\ta))) & = \{\tf(\ta, \ta),
      \tf(x, \ta), \tf(\ta, y), \tf(x,y)\}\\
      \capterm_1(\tf(\tg(\ta), \tg(\ta))) & = \{\tf(\ta, \ta),
      \tf(x, \ta), \tf(\ta, y),
      \tf(x,y), \tf(y,x), \tf(x,x), \tf(y,y)
      \}\\
      \caphterm_2(\tf(\tg(\ta), \tg(\ta))) & = \{x', \tg(x), \tg(y)\}\\
      \capterm_2(\tf(\tg(\ta), \tg(\ta))) & = \{x', \tg(x), \tg(y)\}\!
    \end{array}\]
\end{exa}

The following lemma states the
two most important properties regarding the disjoint union abstraction.

\begin{lem}[Properties of $\capterm_d$]
   Let
  $\PP^{(1)}, \PP^{(2)}$ be PTRSs with $\Sigma^{\PP^{(1)}} \cap \Sigma^{\PP^{(2)}} =
   \emptyset$,
  $t \in \TSet{\Sigma^{\PP^{(1)}} \cup
    \Sigma^{\PP^{(2)}}}{\VSet}$, 
and $d \in \{1,2\}$.
  Then:
  \begin{enumerate}
    \item $\capterm_d(t)$ is finite
    \item $\capterm_d(t) \subseteq \TSet{\Sigma^{\PP^{(d)}}}{\VSet}$
        \end{enumerate}
\end{lem}

\begin{proof}
  Simple proof by induction on the term structure.
\end{proof}

With the disjoint union abstraction, we can prove the modularity of
$\SASTi$ for disjoint unions.

\begin{restatable}[Modularity of $\SASTi$ for Disjoint Unions]{thm}{ModiSASTDisjoint}\label{modularity-iSAST-disjoint}
  Let $\PP^{(1)}$ and $\PP^{(2)}$ be PTRSs with
  $\Sigma^{\PP^{(1)}} \cap \Sigma^{\PP^{(2)}} = \emptyset$. Then we have:
    \begin{align*}
        \SAST[\ito_{\PP^{(1)} \cup \PP^{(2)}}]
        &\; \Longleftrightarrow \; \SAST[\ito_{\PP^{(1)}}]  \; \text{ and } \; \SAST[\ito_{\PP^{(2)}}]\!
    \end{align*}
\end{restatable}

\begin{myproofsketch}
    The direction ``$\Longrightarrow$'' is trivial and thus, we only prove ``$\Longleftarrow$''.
    So let $\PP = \PP^{(1)} \cup \PP^{(2)}$ where $\SAST[\ito_{\PP^{(1)}}]$ and $\SAST[\ito_{\PP^{(2)}}]$.
    Let $\F{T}$ be an arbitrary $\ito_{\PP}$-RST that starts with $(1:t)$.
    We show that $\edl(\F{T})$ is bounded by some constant 
    which does not depend on $\F{T}$ but just on $t$.
    This proves that $\SAST[\ito_{\PP}]$ holds.

    Since we have $\SAST[\ito_{\PP^{(d)}}]$,
    the expected derivation length of
    all $\ito_{\PP^{(d)}}$-RSTs with $d \in \{1,2\}$ that start with a term $q
    \in \TSet{\Sigma^{\PP^{(d)}}}{\VSet}$
    is bounded by some constant $C_q < \omega$.
    Thus, since $|\capterm_{1}(t) \cup \capterm_{2}(t)| = K \in \IN$ is finite, 
    there is a $C_{\max} < \omega$ such that for all $q \in
    \capterm_{1}(t) \cup \capterm_{2}(t)$ we have $\edl(\F{T}') \leq C_{\max}$ for every
    $\ito_{\PP^{(d)}}$-RST $\F{T}'$ that starts with $(1:q)$.
    Hence, we obtain $\edl(\F{T}) \leq K \cdot C_{\max}$.
\end{myproofsketch}

\begin{exa}\label{example:modularity-innermost-disjoint-SAST}
    Let us illustrate the notation in the previous proof sketch by
    applying it to the PTRS from
    \Cref{example:SAST-modularity-proof-problems}. 
    If we reconsider the start term $\tf(\tg(\ta), \tg(\ta))$,
    then $\capterm_{1}(\tf(\tg(\ta), \tg(\ta))) = \{\tf(\ta, \ta),
    \tf(x, \ta), \tf(\ta, y),
    \tf(x,y), \tf(y,x), \tf(x,x), \tf(y,y) \}$
    and $\capterm_{2}(\tf(\tg(\ta),\tg(\ta)))\!=\!\{x', \tg(x), \tg(y)\}$, as in~\Cref{example:disjoint-union-abstraction}.
    Moreover, we obtain $C_{\max} = C_{\tf(\ta,\ta)} = 2 + 2 + 2 \cdot \sum_{n = 1}^{\infty} (\nicefrac{1}{2})^{n+1} \cdot n = 2 + 2 + 2 \cdot 1 = 6$, 
    where $C_{\tf(\ta,\ta)}$ is the bound on the expected derivation height of the term $\tf(\ta,\ta)$.
    (The reason is that reaching a normal form from the subterm $\ta$ needs $2$ steps in expectation, 
    and then for each generated $\ts$ we have one additional step,
    where each
    $\ta$ generates 
    $n$ $\ts$-symbols with probability $(\nicefrac{1}{2})^{n+1}$.)
    Hence, the overall bound on the expected derivation height for $\tf(\tg(\ta), \tg(\ta))$ is
    $K \cdot C_{\max}  = 10 \cdot 6 = 60$.
    In fact, the actual expected derivation height for $\tf(\tg(\ta), \tg(\ta))$ is $6 + 2 = 8$
    ($2$ steps for the two $\tg$-symbols, and in expectation $6$ steps for the term $\tf(\ta,\ta)$), 
    but for the proof any finite bound suffices.
\end{exa}

For full rewriting, it is well known that
termination is already not modular in the non-probabilistic setting.

\begin{counterexample}\label{example:modularity-full-disjoint}
    Reconsider the TRS $\R_1$ from Counterex.\ \ref{example:diff-SN-vs-iSN-toyama}.
    This TRS is the disjoint union of $\R_1^{(1)} = \{\tf(\ta,\tb,x) \to \tf(x,x,x)\}$ and
    $\R_1^{(2)} = \{\tg \to \ta, \; \tg \to \tb\}$.
    Both $\R_1^{(1)}$ and $\R_1^{(2)}$ are terminating, but the disjoint union $\R_1$ is not.
\end{counterexample}

However, one can reuse our results
from \Cref{Relating AST and its Restricted Forms} to obtain the following corollary.

\begin{restatable}[Modularity of $\PSNf$ for Disjoint Unions]{cor}{ModASTDisjoint}\label{modularity-AST-disjoint}
    Let $\PP^{(1)}$ and $\PP^{(2)}$ be
    PTRSs with
    $\Sigma^{\PP^{(1)}} \cap \Sigma^{\PP^{(2)}} = \emptyset$
    that are NO and linear. Then we have:
    \begin{align*}
        \AST[\fto_{\PP^{(1)} \cup \PP^{(2)}}]
        &\; \Longleftrightarrow \; \AST[\fto_{\PP^{(1)}}] \; \text{ and } \; \AST[\fto_{\PP^{(2)}}]\\
        \SAST[\fto_{\PP^{(1)} \cup \PP^{(2)}}]
        &\; \Longleftrightarrow \; \SAST[\fto_{\PP^{(1)}}]  \; \text{ and } \; \SAST[\fto_{\PP^{(2)}}]\!
    \end{align*}
\end{restatable}

\subsection{Shared Constructor Unions}\label{Sect:Shared Constructor Unions}

Now we consider unions of PTRSs that may share constructor symbols, i.e., 
we consider two PTRSs $\PP^{(1)}$ and $\PP^{(2)}$ such that $\Sigma^{\PP^{(1)}}_{D} \cap \Sigma^{\PP^{(2)}}_{D} = \emptyset$,
called \emph{shared constructor unions}.
Again, we study innermost rewriting first.

In the non-probabilistic setting, innermost termination is also modular for shared
constructor unions \cite{Gramlich1995AbstractRB}. 
However,
$\PASTi$ was already not modular w.r.t.\ disjoint unions, 
so this also holds for shared constructor unions.
Moreover, $\SASTi$ also turns out to be not modular anymore for shared
constructor unions.

\begin{counterexample}\label{example:modularity-innermost-constructor-SAST}
    Consider the PTRS $\PP_{13} = \PP_{13}^{(1)} \cup \PP_{13}^{(2)}$ with the
    rules

    \vspace*{-.4cm}
    \begin{minipage}[t]{7cm}
        \begin{align*}
            \PP_{13}^{(1)}: \tf(\tc(x,y)) &\to \{1:\tc(\tf(x),\tf(y))\}\\
            \tf(\tz) &\to \{1:\tz\}\!
        \end{align*}
    \end{minipage}
    \begin{minipage}[t]{7cm}
        \begin{align*}
            \PP_{13}^{(2)}: \tg(x) &\to \{\nicefrac{1}{2}:\tg(\td(x)), \nicefrac{1}{2}:x \}\\
            \td(x) &\to \{1:\tc(x,x)\}\!
        \end{align*}
    \end{minipage}
    \vspace*{.3cm}

    \noindent
    While $\PP_{13}^{(1)}$ and $\PP_{13}^{(2)}$ do not have any common defined symbols, they
    share the constructor $\tc$.     
    We do not have $\PAST[\ito_{\PP_{13}}]$
    (and thus, not $\SAST[\ito_{\PP_{13}}]$),
    as the infinite $\iliftto_{\PP_{13}}$-rewrite sequence $(\mu_n)_{n \in \IN}$
    depicted in the following $\ito_{\PP_{13}}$-RST has an infinite expected derivation length.

    \medskip
    
    
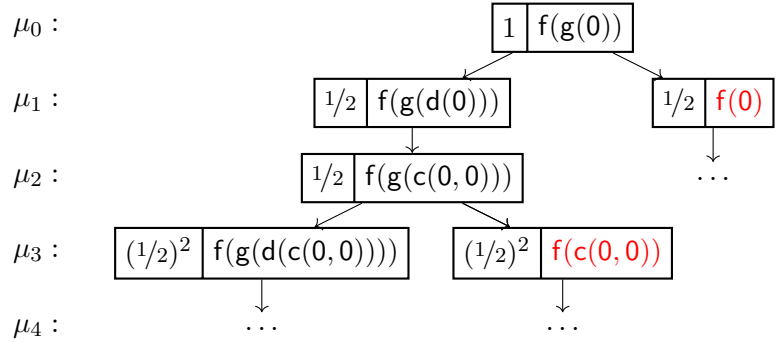
\begin{center}
      \begin{tikzpicture}
          \tikzstyle{adam}=[thick,draw=black!100,fill=white!100,minimum size=4mm, shape=rectangle split, rectangle split parts=2,rectangle split horizontal]
          \tikzstyle{empty}=[rectangle,thick,minimum size=4mm]

          \node[empty] at (-7, 2)  (a) {$\mu_0:$};
          \node[adam] at (0, 2)  (1) {$1$ \nodepart{two} $\tf(\tg(\tz))$};

          \node[empty] at (-7, 1)  (b) {$\mu_1:$};
          \node[adam] at (-2, 1)  (11) {$\nicefrac{1}{2}$\nodepart{two}$\tf(\tg(\td(\tz)))$};
          \node[adam] at (2, 1)  (12)
               {$\nicefrac{1}{2}$\nodepart{two}$\underline{\textcolor{red}{\tf(\tz)}}$}; 

          \node[empty] at (-7, 0)  (c) {$\mu_2:$};
          \node[adam] at (-2, 0)  (111) {$\nicefrac{1}{2}$\nodepart{two}$\tf(\tg(\tc(\tz,\tz)))$};
          \node[empty] at (2, 0)  (121) {$\ldots$};

          \node[empty] at (-7, -1)  (d) {$\mu_3:$};
          \node[adam] at (-4, -1)  (1111) {$(\nicefrac{1}{2})^2$\nodepart{two}$\tf(\tg(\td(\tc(\tz,\tz))))$};
          \node[adam] at (0, -1)  (1112) {$(\nicefrac{1}{2})^2$\nodepart{two}$\underline{\textcolor{red}{\tf(\tc(\tz,\tz))}}$};

          \node[empty] at (-7, -2)  (d) {$\mu_4:$};
          \node[empty] at (-4, -2)  (11111) {$\ldots$};
          \node[empty] at (0, -2)  (11121) {$\ldots$};

          \draw (1) edge[->] (11);
          \draw (1) edge[->] (12);
          \draw (12) edge[->] (121);
          \draw (11) edge[->] (111);
          \draw (111) edge[->] (1111);
          \draw (111) edge[->] (1112);
          \draw (111) edge[->] (1112);
          \draw (1111) edge[->] (11111);
          \draw (1112) edge[->] (11121);
      \end{tikzpicture}
    \end{center}

\medskip
    
    \noindent
    For any $n \in \IN$,
    each underlined term $\tf(\tc^n(\tz,\tz))$
    in the tree above can start a reduction of at least length $2^{n}$, where
    $\tc^n(\tz,\tz)$ corresponds to 
    the full binary tree of height $n$ 
    with $\tc$ in inner nodes and $\tz$ in the leaves.
    Hence, the term $\tf(\tg(\tz))$ has an expected derivation height of at least $\sum_{n
      = 0}^{\infty} \frac{1}{2^{n+1}} \cdot 2^n = \sum_{n = 0}^{\infty} \frac{1}{2}$,
    which diverges to infinity.

    On the other hand, we have $\SAST[\ito_{\PP_{13}^{(1)}}]$, 
    as $\PP_{13}^{(1)}$ is a PTRS with only trivial probabilities that corresponds to a terminating TRS.
    Moreover, $\SAST[\ito_{\PP_{13}^{(2)}}]$ holds as well, as the $\td$-rule 
    can  increase the number of $\tc$-symbols  in a term exponentially, 
    but those $\tc$-symbols will never be used. 
    Thus, $\SASTi$ is not modular for shared constructor unions.
\end{counterexample}

In contrast to the proof of \Cref{modularity-iSAST-disjoint},
we cannot use the disjoint union abstraction anymore to obtain a bound on the expected derivation height,
since symbols from $\Sigma^{\PP^{(2)}}$ can now ``generate'' constructor symbols of $\Sigma^{\PP^{(1)}}$.
However, for $\ASTi$ we can reuse the idea of the previous proof for \Cref{modularity-iAST-disjoint} 
to obtain a similar result for shared constructor unions.

\begin{restatable}[Modularity of $\ASTi$ for Shared Constructor Unions]{thm}{ModiASTCons}\label{modularity-iAST-constructor}
    Let $\PP^{(1)}$ and $\PP^{(2)}$ be PTRSs with
    $\Sigma_D^{\PP^{(1)}} \cap \Sigma_D^{\PP^{(2)}} = \emptyset$. Then we have:
    \begin{align*}
        \AST[\ito_{\PP^{(1)} \cup \PP^{(2)}}]
        &\; \Longleftrightarrow \; \AST[\ito_{\PP^{(1)}}] \; \text{ and } \; \AST[\ito_{\PP^{(2)}}]\!
     \end{align*}
\end{restatable}

\begin{proof}
    The proof for ``$\Longleftarrow$''
    is again via structural induction on the term $t$ in the initial
    multi-distribution
    $\{1:t\}$ and very similar to the proof of \Cref{modularity-iAST-disjoint}.
    The only difference is that in the induction step, 
    after performing rewrite steps below the root until all proper subterms are in normal form,
    if the root is from $\Sigma^{\PP^{(1)}}$, then
    we do not replace all maximal subterms with roots from $\Sigma^{\PP^{(2)}}$ by fresh variables
    but just maximal subterms with root symbols from $\Sigma^{\PP^{(2)}}_{D}$, as the constructor
    symbols may be used by rules of $\PP^{(1)}$.
    This, however, does not interfere with the proof idea.
\end{proof}

Again, we can reuse our results
from \Cref{Relating AST and its Restricted Forms} on
the relation between full and innermost rewriting to obtain
the following corollary for full rewriting.
Due to Counterex.\ \ref{example:modularity-full-disjoint}, this corollary
does not hold for general PTRSs.

\begin{restatable}[Modularity of $\ASTf$ for Shared Constructor Unions]{cor}{ModASTCons}\label{modularity-AST-constructor}
    Let $\PP^{(1)}$ and $\PP^{(2)}$ be PTRSs with
    $\Sigma_D^{\PP^{(1)}} \cap \Sigma_D^{\PP^{(2)}} = \emptyset$
    that are NO and linear. Then we have:
    \begin{align*}
        \AST[\fto_{\PP^{(1)} \cup \PP^{(2)}}]
        &\; \Longleftrightarrow \; \AST[\fto_{\PP^{(1)}}] \; \text{ and } \; \AST[\fto_{\PP^{(2)}}]\!
    \end{align*}
\end{restatable}

\subsection{Signature Extensions}\label{Sect:Signature Extensions}

Finally, we study signature extensions of PTRSs.
While $\PASTs$ is not closed under signature extensions by \Cref{thm:Sig-PAST}, we now 
consider $\ASTs$ and $\SASTs$.
Signature extensions can be seen as special cases of disjoint unions,
where the second PTRS $\PP^{(2)}$ contains only trivially terminating rules
over the new signature that we want to add to $\Sigma^{\PP^{(1)}}$.
Hence, for innermost rewriting,
\Cref{modularity-iAST-disjoint,modularity-iSAST-disjoint} already imply that $\ASTi$ and $\SASTi$
are closed under signature extensions.

For full rewriting,
\Cref{modularity-AST-disjoint} implies that
$\ASTf$ and $\SASTf$
are closed under signature extensions for non-overlapping and linear PTRSs.
We now show that this also holds for arbitrary PTRSs.
So let $\PP$ be an arbitrary PTRS over the signature $\Sigma^{\PP}$ for the rest of this section.
We consider two cases.

First, let $\Sigma^{\PP}$ contain only constants and unary symbols,
e.g., $\Sigma^{\PP} = \{\tf, \ta, \tb\}$ where 
$\tf$ is unary and $\ta, \tb$ are constants. 
If we extend
$\Sigma^{\PP}$ by a
signature $\Sigma'$ that may also contain symbols of other arities,
e.g., a symbol $\tc$ of arity $2$,
and consider terms from
$\TSet{\Sigma^{\PP} \cup \Sigma'}{\VSet}$
like $\tf(\tc(\tf(\ta), \tf(\tb)))$, then the fresh symbol $\tc$
``completely separates''
the function symbols
from $\Sigma^{\PP}$
occurring below and above $\tc$.
More precisely,
instead of $\tf(\tc(\tf(\ta), \tf(\tb)))$, it suffices to analyze the start
terms  
$\tf(\ta)$, $\tf(\tb)$, and $\tf(x_\tc)$
(where the $\tc$-subterm is replaced by a fresh variable $x_\tc$), which are all from 
$\TSet{\Sigma^{\PP}}{\VSet}$. The reason is that
rewriting above $\tc$ does not interfere with the terms below $\tc$ and vice versa.
Hence, if every rewrite sequence that starts with $\tf(\ta)$, $\tf(\tb)$, or $\tf(x_\tc)$ converges with probability $1$, then
so does the term $\tf(\tc(\tf(\ta), \tf(\tb)))$. In other words, if $\ASTf$ holds
over the signature $\Sigma^{\PP}$, then  $\ASTf$ also holds
over the signature $\Sigma^{\PP} \cup \Sigma'$.
Furthermore, the expected derivation height of $\tf(\tc(\tf(\ta), \tf(\tb)))$ is bounded by $3 \cdot C_{\max}$,
where \mbox{$C_{\max}$
is the maximum expected derivation height of the terms $\tf(\ta)$, $\tf(\tb)$, and $\tf(x_\tc)$, 
similar} as in the proof of \Cref{modularity-iSAST-disjoint}. Thus, if we have
$\SASTf$ over the signature $\Sigma^{\PP}$, then we also have $\SASTf$
over the signature $\Sigma^{\PP} \cup \Sigma'$.

Second, we consider the case where
$\Sigma^{\PP}$ itself already contains a function symbol $\tg$
that has at least arity $2$.
Again, the fresh symbols $\tc$ of $\Sigma'$  separate the function symbols of
$\Sigma^{\PP}$
occurring above and below
them. However, now the expected 
derivation height of a term like $\tf(\tc(\tf(\ta), \tf(\tb), \tf(y)))$
from $\TSet{\Sigma^{\PP} \cup \Sigma'}{\VSet}$
is not bounded by simply adding the expected derivation heights
of the corresponding terms 
$\tf(\ta)$, $\tf(\tb)$, $\tf(y)$, $\tf(x_\tc)$ from $\TSet{\Sigma^{\PP}}{\VSet}$
anymore. The reason is that $\PP$ might now duplicate subterms (e.g., there could be a
rule like  $\tf(x) \to \{1: \tg(x,x)\}$). 
However, for any term $t \in \TSet{\Sigma^{\PP} \cup \Sigma'}{\VSet}$, we can now construct
a term $t' \in \TSet{\Sigma^{\PP}}{\VSet}$ over the original signature $\Sigma^{\PP}$
that has (at least) the same expected derivation height and (at most) the same convergence probability.

The construction works as follows:
Let $\tg \in \Sigma^{\PP}$ be a symbol of arity $2$ (if its arity
is greater than two, then we use the term $\tg(\_, \, \_, \; x,..., x)$ instead, where $x \in \VSet$).
For example,
if we extend the signature $\Sigma^{\PP}$ by a symbol $\tc \in \Sigma'$ of arity $3$, then
instead of a term like 
$\tf(\tc(\tf(\ta), \tf(\tb), \tf(y)))$ from
$\TSet{\Sigma^{\PP} \cup \Sigma'}{\VSet}$,
we can  consider the term $\tf(\tg(\tf(\ta), \tg(\tf(\tb), \tg(\tf(y),x_\tc))))$  from
$\TSet{\Sigma^{\PP}}{\VSet}$ without the symbol $\tc$, 
where we do not rewrite the newly added $\tg$ symbols.
So here, we replaced $\tc(\_, \, \_, \, \_)$ by $\tg(\_, \, \tg(\_, \, \tg( \_, x_\tc)))$.
Note that this construction works for full but not for innermost rewriting,
since
it may create new redexes with the symbol $\tg$ that we may have to rewrite when using the
innermost strategy.
However, for innermost rewriting we have already proven closedness
under signature extensions by \Cref{modularity-iAST-disjoint,modularity-iSAST-disjoint}, as explained above.

To summarize, this leads to the following theorem.

\begin{restatable}[Signature Extensions for $\ASTs$ and $\SASTs$]{thm}{SignatureExtensions}\label{signature-extensions-AST-SAST}
  Let $\PP$ be a PTRS, $s \in \{ \mathbf{f}, \mathbf{i} \}$, and let $\Sigma'$ be some signature. Then we have:
  \begin{align*}
        \ASTs \text{ over } \Sigma^{\PP}
        &\; \Longleftrightarrow \; \ASTs \text{ over } \Sigma^{\PP} \cup \Sigma'\\
        \SASTs \text{ over } \Sigma^{\PP}
        &\; \Longleftrightarrow \; \SASTs \text{ over } \Sigma^{\PP} \cup \Sigma'\!
  \end{align*}
\end{restatable}

\section{Related Work on Verification of Probabilistic Programs}\label{Related Work}

In the previous sections, we already discussed the connection to related work in term
rewriting. 
However, verification of probabilistic programs has also been studied
extensively for imperative programs on numbers, 
and for different recursive programming languages like the lambda calculus.
Thus, in this section we discuss existing work on probabilistic
termination analysis outside term rewriting.

Its \emph{hardness} has been investigated in 
\cite{kaminski2019hardness,  majumdar2024PositiveAlmostSureTermination},
showing that analyzing almost-sure termination is even more difficult than ordinary
termination when considering the halting problem for given inputs. 
In our paper, $\SN$ and  $\AST$ refer to universal termination,
as we consider rewrite sequences starting with arbitrary terms.
While for non-probabilistic programs, universal termination is ``harder'' than the halting
problem for given inputs,  these problems are equally hard for probabilistic programs
\cite{kaminski2019hardness}.

There exist numerous approaches and
proof rules mainly based on martingales 
for different properties of probabilistic programs.
For example, there are techniques for proving $\AST$ and $\PAST$
\cite{ferrerfioritiProbabilisticTerminationSoundness2015,chatterjeeTerminationAnalysisProbabilistic2016,  
 agrawal2017LexicographicRankingSupermartingales,
 chatterjee2018AlgorithmicAnalysisQualitative,huang2018NewApproachesAlmostSure,
 mciver2017new, dblp:journals/pacmpl/huang0cg19,
  chatterjee2020TerminationAnalysisProbabilistic,abateLearningProbabilisticTermination2021,
  moosbruggerAutomatedTerminationAnalysis2021a, chatterjee2023LexicographicProofRules,
  majumdar2025SoundCompleteProof}, 
for proving bounds on the termination probability
\cite{chatterjeeStochasticInvariantsProbabilistic2017,
  kura2019TailProbabilitiesRandomized,
chatterjee2022SoundCompleteCertificates,DBLP:journals/pacmpl/MoosbruggerSBK22, feng2023LowerBoundsPossibly}, and for
upper and lower bounds on expected runtimes and costs
\cite{kaminski2018WeakestPreconditionReasoning,ngoBoundedExpectationsResource2018,
    fuTerminationNondeterministicProbabilistic2019,giesl2019ComputingExpectedRuntimes,
    wang2019CostAnalysisNondeterministic,  avanziniModularCostAnalysis2020,hark2020AimingLowHarder,
    kaminski2020ExpectedRuntimeAnalyis,
    meyerInferringExpectedRuntimes2021a,
DBLP:journals/pacmpl/DasWH23,
lommen2024ControlFlowRefinementComplexity}.
  Many of these approaches can be automated directly
 or by extending them with automatic invariant synthesis techniques, 
e.g., \cite{katoen2010LinearInvariantGenerationProbabilistic,
  barthe2016SynthesizingProbabilisticInvariants,
  bao2022DataDrivenInvariantLearning, 
kofnov2022MomentBasedInvariantsProbabilistic, batzProbabilisticProgramVerification2023}.
Moreover, one can also use such proof rules within a quantitative program verification infrastructure like \tool{Caesar} 
\cite{DBLP:journals/pacmpl/SchroerBKKM23}.
Some of these proof rules have already been adapted to term rewriting, 
e.g., the proof rule of \cite{mciver2017new} can be adapted to prove
$\AST$ of PTRSs via polynomial and matrix interpretations \cite{kassinggiesl2023iAST} 
and it is also used within the probabilistic DP framework \cite{kassinggiesl2023iAST,FLOPS2024,JPK60}.
Further proof rules regarding $\PAST$ and $\SAST$ have been adapted
in \cite{avanzini2020probabilistic} and again integrated in the corresponding DP framework \cite{kassing2025DependencyPairsExpected}. 

In particular, there also exist several
\emph{tools} to analyze $\AST$, $\PAST$,
and expected costs for imperative probabilistic programs,
e.g.,  \tool{Amber} \cite{moosbruggerAutomatedTerminationAnalysis2021a},
\tool{KoAT}
\cite{meyerInferringExpectedRuntimes2021a,lommen2024ControlFlowRefinementComplexity},
 \tool{Eco-Imp}
\cite{avanziniModularCostAnalysis2020}, \tool{Absynth}
\cite{ngoBoundedExpectationsResource2018}, and \tool{Pastry}
\cite{novozhilov2025AlmostSureTerminationProbabilistic}. Moreover,
higher moments for a loop's variables are analyzed automatically with the tool 
\tool{Polar} \cite{DBLP:journals/pacmpl/MoosbruggerSBK22}.
These tools mainly consider imperative programs with an innermost evaluation strategy.
Hence, our
results on the 
relation between the different rewrite strategies cannot be directly used for
these tools, while our results on modularity may in principle be of interest.
However, our results concern the functional recursive nature of term rewriting,
where one does not have a fixed control flow as in an imperative program.
Our comparison between $\PAST$ and $\SAST$ in \Cref{thm:PAST-vs-SAST-1} should
also hold for
imperative programs which allow for multiple executions in parallel that are considered simultaneously.

Compared to these tools, our implementation in \tool{AProVE} currently can prove $\AST$,
$\PAST$,
$\SAST$, and analyze the expected complexity of PTRSs.
For algorithms whose termination behavior relies only on numbers,
techniques for imperative programs with built-in support for arithmetic
are usually more powerful than
approaches based on term rewriting. The reason is that for
term rewriting, 
numbers have to be represented
via terms, e.g., $2$ can be represented by the term $\ts(\ts(\tz))$.
On the other hand,
term rewriting can handle
programs with non-trivial recursive structure and 
arbitrary user-defined data structures, as these structures
can easily be represented as terms. 
For example, the list $[2,1]$ can be represented by the term
$\tcons(\ts(\ts(\tz)), \tcons(\ts(\tz), \tnil)))$. Thus, tools based on term rewriting
are particularly suitable when 
analyzing programs whose
termination depends on, e.g., lists, trees, or graphs.
In the non-probabilistic setting, there are also techniques and tools for term rewriting with
integrated built-in numbers, e.g., for
\emph{integer term rewrite systems} \cite{fuhs2009ProvingTerminationInteger} or
\emph{logically constrained term rewrite systems} \cite{DBLP:conf/frocos/KopN13}.
Lifting these approaches to the probabilistic setting is an interesting direction for future work.

In addition to the related work on probabilistic loop programs, there are also
several approaches for probabilistic recursive programs,
e.g., to analyze the probabilistic lambda calculus
or other higher-order functional languages
based on types or martingales
\cite{avanzini2019TypebasedComplexityAnalysis,dallago2019ProbabilisticTerminationMonadic,beutner2021probabilistic,
  DBLP:conf/lics/Kenyon-RobertsO21,dallago2021IntersectionTypesPositive,rajani2024ModalTypeTheory, dallago2024AlmostSureTerminationBinary}.
There has also been work on functional languages where $\AST$ or $\PAST$ is (partly) decidable, 
e.g., probabilistic higher-order rewrite schemes
\cite{kobayashi2020TerminationProblemProbabilistic},
probabilistic pushdown automata
\cite{brazdil2013AnalyzingProbabilisticPushdowna, brazdil2015RuntimeAnalysisProbabilistic,DBLP:conf/lics/WinklerK23}, and
restricted probabilistic tree-stack automata \cite{li2022ProbabilisticVerificationContextFreeness}.
Many of the results on recursive languages fix a leftmost-innermost rewrite strategy to avoid non-determinism.
Here, our results may be helpful to extend these techniques to different rewrite strategies,
and also our modularity results may be of interest for the different recursive probabilistic languages.

Finally, probabilistic programs that allow for \emph{data structures} are analyzed 
in \cite{wang2020RaisingExpectationsAutomating, leutgebAutomatedExpectedAmortised2022, batzCalculusAmortizedExpected2023}.
While \cite{batzCalculusAmortizedExpected2023} uses pointers to represent data structures like tables and lists, \cite{ wang2020RaisingExpectationsAutomating, leutgebAutomatedExpectedAmortised2022} consider a 
probabilistic programming language with matching similar to term rewriting
and develop an automatic amortized resource analysis via fixed template potential functions.
However, these latter works are mostly targeted towards specific data structures, 
and we consider general term rewrite systems that can model arbitrary data structures.

\section{Conclusion}\label{Conclusion}

In this paper, we
presented the first results on the relationship between
$\ASTs$ of a PTRS $\PP$
for different rewrite strategies $s \in \IS$, including several criteria such that
$\ASTi$ implies $\ASTf$. 
Our results also hold for $\PASTs$, $\SASTs$, and expected complexity, 
and all of our criteria are suitable for automation (for spareness,
there exist sufficient conditions that can be checked automatically).
We implemented our criteria for the equivalence of $\ASTi$ and
$\ASTf$ in our termination prover \aprove, and demonstrated their
practical usefulness in an experimental evaluation. Moreover, we developed the first
modularity results for termination of PTRSs under unions and signature extensions.

In the paper,
we already mentioned several topics for future work, e.g.:
\begin{itemize}
\item Improving our current results (in
  \Cref{properties-eq-Der-iDer-1,properties-eq-iDer-liDer-complex})
  on the relationship between
  non-probabilistic innermost and full derivational complexity.
  \item Extending our modularity results to hierarchical unions or finding
    more specific classes where $\SASTi$ 
  is modular for shared constructor unions, or even classes where $\PASTi$ becomes modular.
\item Tackling the (hard) problem of adapting Newman's Lemma \cite{Newman42}
  to the probabilistic setting.
\end{itemize}

\medskip

\noindent
\textbf{Acknowledgements.}
We thank  Florian Frohn for his help on the earlier conference
paper~\cite{FoSSaCS2024} and 
Stefan Dollase for pointing us to~\cite{fuhs2019transformingdctorc}.

\bibliographystyle{alphaurl}
\bibliography{biblio}
\clearpage
\appendix
\section{Missing Proofs}\label{appendix}

In this appendix, we present all missing proofs for
our new contributions and observations.
Most of our proofs use $\to$-RSTs instead of $\liftto$-rewrite sequences.
Therefore, in \Cref{appendix:characterization} we start with the formal definitions
for all required notions via RSTs (where some of them were already mentioned in the main
part 
of the paper). 
Then, in \Cref{appendix:relating} we give the missing proofs for
the theorems and lemmas from 
\Cref{Relating AST and its Restricted Forms} and \Cref{Improving Applicability},
concerning the relation between different rewrite strategies.
In \Cref{appendix:modularity}, 
we prove the results regarding modularity and signature extensions from
\Cref{Modularity}.

\subsection{Characterization via RSTs}\label{appendix:characterization}

We start with the formal definition of RSTs.

\begin{defi}[Rewrite Sequence Tree (RST)] \label{def:rewrite-sequence-tree}
    Let $\to \; \subseteq \TT \times \FDist(\TT)$.
    A $\to$\emph{-rewrite sequence tree ($\to$-RST)} $\F{T}$ is a labeled tree $\F{T}\!=\!(N,E,L)$ such that

    \begin{enumerate}
    \item[(1)] $N \neq \emptyset$ is a possibly infinite set of nodes and $E \subseteq N \times N$ is a set of directed edges, such that $(N, E)$ is a (possibly infinite) directed tree where $vE = \{ w \mid (v,w) \in E \}$ is finite for every $v \in N$.
    \item[(2)]  $L : N \rightarrow (0,1] \times \TT$ labels every node $v$ by a probability $p_v$ and a term $t_v$.
    For the root $v \in N$ of the tree, we have $p_v = 1$.
    \item[(3)] For all $v \in N$: If $vE = \{w_1, \ldots, w_k\}$, then $t_v \to \{\tfrac{p_{w_1}}{p_v}:t_{w_1}, \ldots, \tfrac{p_{w_k}}{p_v}:t_{w_k}\}$.
    \end{enumerate}
    $\ctleaf$ denotes the set of leaves of the 
    RST and for a node $x\in N$, $d(x)$ 
    denotes the depth of node $x$ in $\F{T}$.
    Here, the root has depth $0$.
    We say that $\F{T}$ is \emph{fully evaluated} if for every $v \in \ctleaf$
    the corresponding term $t_v$ is a normal form w.r.t.\ $\to$, i.e., $t_v \in \NF_{\to}$.
\end{defi}

Whenever the referenced RST is ambiguous, we will explicitly indicate the corresponding tree.
For instance, for the probability $p_v$ of the node $v \in N$ of some RST $\F{T} =
(N,E,L)$, we may also write $p_v^{\F{T}}$, and $N^{\F{T}} = N$ is the set of nodes of the tree $\F{T}$.

\begin{defi}[$|\F{T}|$, Convergence Probability] \label{def:rewrite-sequence-tree-convergence-notation}
    Let $\to \; \subseteq \TT \times \FDist(\TT)$.
    For any $\to$-RST $\F{T}$ we define $|\F{T}| = \sum_{v \in \ctleaf} p_v$
    and say that the RST $\F{T}$ \emph{converges with probability} $|\F{T}|$.
\end{defi}

It is now easy to observe that we have $\AST[\to]$
(i.e., for all $\liftto$-rewrite
sequences $(\mu_n)_{n \in \IN}$ we have $\lim_{n \to \infty} |\mu_n|_{\to} = 1$) iff for all
$\to$-RSTs $\F{T}$ we have $|\F{T}| = 1$. 
To see this, note that every infinite $\liftto$-rewrite sequence $(\mu_n)_{n \in \IN}$ that begins
with a single start term (i.e., $\mu_0 = \{1:t\}$) can be represented by an
$\to$-RST $\F{T}$ that is fully evaluated such that $\lim_{n \to \infty} |\mu_n|_{\to} =
|\F{T}|$ and vice versa. So $\AST[\to]$ holds iff all fully evaluated 
$\to$-RSTs converge with probability 1.
Furthermore, note that for every $\to$-RST $\F{T}$, 
there exists a fully evaluated $\to$-RST $\F{T}'$ such that $|\F{T}| \geq |\F{T}'|$.
To get from $\F{T}$ to $\F{T}'$ we can simply perform arbitrary (possibly infinitely many)
rewrite steps at the leaves that are not in normal form to fully evaluate the tree.

It remains to prove that it suffices to only regard $\liftto$-rewrite sequences that start
with a single start term.

\begin{lem}[Single Start Terms Suffice for ${\AST[\to]}$]\label{lemma:PTRS-AST-single-start-term}
    Let $\to \; \subseteq \TT \times \FDist(\TT)$.
    If $\AST[\to]$ does not hold, 
    then there exists an infinite $\liftto$-rewrite sequence 
    $(\mu_n')_{n \in \mathbb{N}}$ with a single start term, 
    i.e., $\mu_0' = \{1:t\}$, that converges with probability~$<1$.
\end{lem}

\begin{proof}
    We prove the converse.
    Assume that all infinite $\liftto$-rewrite sequences 
    $(\mu'_n)_{n \in \mathbb{N}}$ with a single start term converge with probability $1$.
    We prove that then every infinite $\liftto$-rewrite sequence converges with probability $1$,
    hence we have $\AST[\to]$.

    Let $(\mu_n)_{n \in \mathbb{N}}$ be an infinite $\liftto$-rewrite sequence.
    Suppose that we have $\mu_0 = \{p_1:t_1, \ldots, p_k:t_k\}$.
    Let $(\mu^{(j)}_n)_{n \in \mathbb{N}}$ with $\mu^{(j)}_0 =
    \{1 : t_j\}$ denote the infinite $\liftto$-rewrite sequence that uses
    the same rules
    as $(\mu_n)_{n \in \mathbb{N}}$ does on the term $t_j$ for every $1 \leq j \leq k$.
    We obtain
    \begin{alignat*}{3} 
        &\textstyle\mathrel{\hphantom{=}}\lim_{n \to \infty}|\mu_{n}|_{\to} &&\textstyle= \lim_{n \to \infty} \sum_{j = 1}^{k} p_j \cdot |\mu^{(j)}_n|_{\to} &&\textstyle= \sum_{j = 1}^{k} \lim_{n \to \infty} p_j \cdot |\mu^{(j)}_n|_{\to}\\
        &\textstyle= \sum_{j = 1}^{k} p_j \cdot \lim_{n \to \infty} |\mu^{(j)}_n|_{\to} &&\textstyle= \sum_{j =
          1}^{k} p_j \cdot 1 &&\textstyle= \sum_{j = 1}^{k} p_j = 1 \tag*{\qedhere}
    \end{alignat*}
\end{proof}

We now obtain the following corollary.

\begin{cor}[Characterizing ${\AST[\to]}$ with RSTs]\label{cor:Characterizing AST with RSTs}
    Let $\to \; \subseteq \TT \times \FDist(\TT)$.
    Then $\AST[\to]$ holds iff for all $\to$-RSTs
    $\F{T}$ we have $|\F{T}| = 1$.
    This is equivalent to the requirement that for all fully evaluated 
    $\to$-RSTs
    $\F{T}$ we have $|\F{T}| = 1$.
    Moreover, $\wAST[\to]$ holds iff for every term $t \in \TT$ there exists a fully
    evaluated $\to$-RST
    $\F{T}$ whose root is labeled with $(1:t)$ such that $|\F{T}| = 1$.
\end{cor}

Next, we recapitulate how $\PAST[\to]$ and $\SAST[\to]$ can be formulated in terms of RSTs.

\begin{cor}[Characterizing ${\PAST[\to]}$ with RSTs]\label{thm:PAST-via-RSTs}
	Let $\to \; \subseteq \TT \times \FDist(\TT)$, and $\F{T}$ be an $\to$-RST.  
    By
    \[
        \edl(\F{T}) = \sum_{x \in N \setminus \ctleaf} p_x = \lim_{n \to \infty}
        \;\; \sum_{\substack{x \in N \setminus \ctleaf \\ d(x) \leq n}} p_x
    \]
    we define the \emph{expected derivation length} of $\F{T}$.
    We have $\PAST[\to]$ iff $\edl(\F{T})$ is finite
    for every $\to$-RST $\F{T}$.
    Similarly, $\wPAST[\to]$ holds iff for every term $t$ there exists a fully evaluated
    $\to$-RST $\F{T}$
    whose root is labeled with $(1:t)$ such that $\edl(\F{T})$ is finite.
\end{cor}

\begin{cor}[Characterizing ${\SAST[\to]}$ with RSTs]\label{thm:SAST-via-RSTs}
    Let $\to \; \subseteq \TT \times \FDist(\TT)$.
    We have $\SAST[\to]$ iff 
    $\sup\{\edl(\F{T}) \mid \F{T}$ is an $\to$-RST whose root is labeled with $(1:t)\}$
    is finite for all $t \in \TT$.
\end{cor}

Both of these corollaries are again easy to observe, similar to the characterization of $\AST[\to]$
with RSTs in \Cref{cor:Characterizing AST with RSTs}. 
First note that every infinite $\liftto$-rewrite sequence $\vec{\mu} = (\mu_n)_{n \in \IN}$ that begins
with a single start term can be represented by an infinite
$\to$-RST $\F{T}$ that is fully evaluated such that 
$\edl(\vec{\mu}) = \sum_{n = 0}^{\infty} (1 - |\mu_n|_{\to}) =
\sum_{x \in N \setminus \ctleaf} p_x^{\F{T}} = \edl(\F{T})$.
The reason is that 
whenever we reach a normal form after $n$ steps, this is both a normal form in
$\mu_n$ and a leaf at depth $n$ of the RST with the same probability.
Otherwise,
if we have a term $t$ in $\Supp(\mu_n)$ that is not in normal form, 
then there exists a (unique) inner node in $\F{T}$ at depth $n$
with the same probability.
Note that we do not need a version of \Cref{lemma:PTRS-AST-single-start-term} for $\PAST[\to]$ and $\SAST[\to]$,
since both of them already consider single start terms.

\subsection{Proofs for \Cref{Relating AST and its Restricted Forms} and \Cref{Improving Applicability}}\label{appendix:relating}

In this subsection, we present all missing proofs on the relation 
between different restricted forms of probabilistic termination.
To this end, we first recapitulate the notion of a \emph{rewrite sequence subtree} from \cite{reportkg2023iAST}.

\begin{defi}[Rewrite Sequence Subtree] \label{def:chain-tree-induced-sub}
    Let $\to \; \subseteq \TT \times \FDist(\TT)$, and let $\F{T} = (N,E,L)$ be an $\to$-RST.
	Let $W \subseteq N$ be non-empty, weakly connected, 
    and for all $x \in W$ we have $xE \cap W = \emptyset$ or $xE \cap W = xE$.
	Then, we define the $\to$-\emph{rewrite sequence subtree} (or simply \emph{subtree}) $\F{T}[W]$ 
    by $\F{T}[W] = (W,E \cap (W \times W),L^W)$.
	Let $w \in W$ be the root of $\F{T}[W]$.
	To ensure that the root of our subtree has the probability $1$ again,
	we use the labeling $L^W(x) = (\frac{p_{x}^{\F{T}}}{p_w^{\F{T}}}: t_{x}^{\F{T}})$ for all nodes $x \in W$.
\end{defi}

The property of being non-empty and weakly connected ensures that the resulting 
graph $(W,E \cap (W \times W))$ is a tree again.
The property that we either have $xE \cap W = \emptyset$ or $xE \cap W = xE$ ensures that the sum of 
probabilities for the successors of a node $x$ is equal to the probability for the node $x$ itself.

To prove the theorems in \Cref{Relating AST and its Restricted Forms} and \Cref{Improving Applicability}
we only have to prove the corresponding lemmas. Then 
the theorems are a direct consequence, similar to the proof of \Cref{properties-eq-AST-iAST-1} given in 
\Cref{Relating AST and its Restricted Forms}.
As mentioned earlier, we will always prove the lemmas using rewrite sequence trees.

\ASTAndIASTLemmaOne*

\begin{proof}
    Let $\PP$ be a PTRS that is non-overlapping and linear.
    Furthermore, let $\F{T}$ be an $\fto$-RST.
    We create a new $\ito$-RST $\F{T}^{(\infty)}$ such that $|\F{T}^{(\infty)}| \leq |\F{T}|$ and $\edl(\F{T}^{(\infty)}) \geq \edl(\F{T})$.
    W.l.o.g., at least one rewrite step in $\F{T}$ is performed at some node $x$ with a
    redex that is not an innermost redex (otherwise we can use $\F{T}^{(\infty)}
    = \F{T}$). 
    The core steps of the proof are the following:
    \begin{enumerate}[label=\arabic{enumi}.]
        \item We iteratively move innermost rewrite steps to a higher position in the tree using a construction $\Phi(\_)$. 
        The limit of this iteration, namely $\F{T}^{(\infty)}$, is an \emph{innermost}
        $\ito_{\PP}$-RST with $|\F{T}^{(\infty)}| \leq |\F{T}|$. For each
         step of the iteration:
        \begin{enumerate}[label=\arabic{enumi}.\arabic{enumii}]
            \item We formally define the construction $\Phi(\_)$ that replaces a certain subtree $\F{T}_x$ by a new subtree $\Phi(\F{T}_x)$, by moving an innermost rewrite step to the root of $\Phi(\F{T}_x)$.
            \item We show $|\Phi(\F{T}_x)| = |\F{T}_x|$.
            \item We show that $\Phi(\F{T}_x)$ is indeed a valid RST\@.
        \end{enumerate}
        \item We show that the same construction also guarantees $\edl(\F{T}^{(\infty)}) \geq \edl(\F{T})$. For this:
        \begin{enumerate}[label=\arabic{enumi}.\arabic{enumii}]
            \item We prove $\edl(\F{T}) \leq \edl(\F{T}^{(1)})  \leq \edl(\F{T}^{(2)}) \leq \ldots$
            \item We prove $\edl(\F{T}) \leq \edl(\F{T}^{(\infty)})$.
        \end{enumerate}
    \end{enumerate}

    \smallskip
          
    \noindent 
    \textbf{\underline{1. We iteratively move innermost rewrite steps to a higher position.}}

    \noindent 
    In $\F{T}$ there exists at least one rewrite step performed at some node $x$, which is not an innermost rewrite step.
    Furthermore, we can assume that this is the first such rewrite step in the path from the root to the node $x$ 
    and that $x$ is a node of minimum depth\footnote{If we allowed rules with infinite
    support, then there could be infinitely many nodes at minimum depth. Hence, then one
    would have to use a more elaborate enumeration of all nodes where a non-innermost step is performed.} with this property.
    Let $\F{T}_x$ be the subtree that starts at node $x$, i.e., $\F{T}_x = \F{T}[xE^*]$, 
    where $xE^*$ is the set of all reachable nodes (via the edge relation) from $x$.
    We then construct a new tree $\Phi(\F{T}_x)$ such that $|\Phi(\F{T}_x)| = |\F{T}_x|$,
    where we use an innermost rewrite step at the root
    node $x$ instead of the old one, 
    i.e., we pushed the first non-innermost rewrite step deeper into the tree.
    This construction only works because $\PP$ is non-overlapping and
    linear.\footnote{For the construction,
    non-overlappingness is essential (while
    a related construction could also be defined without linearity).
    However, linearity is needed to ensure that the probability of termination in the new tree
    is not larger than in the original one.}
    Then, by replacing the subtree $\F{T}_x$ with the new tree $\Phi(\F{T}_x)$ in $\F{T}$,
    we obtain an
 $\ito_{\PP}$-RST $\F{T}^{(1)}$ with $|\F{T}^{(1)}| = |\F{T}|$,
    where we use an innermost rewrite step at node $x$ instead of the old rewrite step, as desired.
    We can then do such a replacement iteratively for every use of a non-innermost rewrite step, i.e., we again replace the first non-innermost rewrite step in $\F{T}^{(1)}$ to obtain $\F{T}^{(2)}$ with $|\F{T}^{(2)}| = |\F{T}^{(1)}|$, and so on.
    In the end, the limit of all these RSTs $\lim_{i \to \infty} \F{T}^{(i)}$ is an $\ito_{\PP}$-RST, 
    that we denote by $\F{T}^{(\infty)}$ 
    such that $|\F{T}^{(\infty)}| \leq |\F{T}|$. 
    So while the termination probability remains the same in each step, it can
    decrease in the limit.\footnote{As an example, consider a tree $\F{T}$ which is just
    a finite path and its path length increases in each iteration by one. Then the limit $\F{T}^{(\infty)}$ is an infinite path and converges with probability $0$, while $\F{T}, \F{T}^{(1)}, \ldots$ all converge with probability $1$.}

    To see that $\F{T}^{(\infty)}$ is indeed a valid $\ito_{\PP}$-RST, note that in
    every iteration of the construction we turn a non-innermost rewrite step at minimum depth
    into an innermost one.
    Hence, for every depth $H$ of the tree, we eventually turned every non-innermost rewrite step up to
    depth $H$ into an innermost one. So the construction will not change the tree above depth $H$ anymore,\footnote{Again, for rules with infinite support, this would not hold anymore. Nevertheless, with a more elaborate enumeration, one should still be able to obtain a valid $\ito_{\PP}$-RST in the limit.} i.e.,
    there exists an $m_H$ such that $\F{T}^{(\infty)}$ and $\F{T}^{(i)}$ are the same trees up to depth $H$ for all $i \geq m_H$.
    This means that the sequence $\lim_{i \to \infty} \F{T}^{(i)}$ really converges into an $\ito_{\PP}$-RST.

    Next, we want to prove that we have $|\F{T}^{(\infty)}| \leq |\F{T}|$.
    By induction on $n$ one can prove that $|\F{T}^{(i)}| = |\F{T}|$ for all $1 \leq i \leq n$, 
    since we have $|\F{T}^{(i)}| = |\F{T}^{(i-1)}|$ for all $i \geq 2$ and $|\F{T}^{(1)}| = |\F{T}|$.
    Assume for a contradiction that $|\F{T}^{(\infty)}| > |\F{T}|$.
    Then there exists a depth $H \in \IN$ such that
    $\sum_{x \in \ctleaf^{\F{T}^{(\infty)}}\!, \, d^{\F{T}^{(\infty)}}\!(x) \leq H} p_x > |\F{T}|$.
    Again, let $m_H \in \IN$ such that $\F{T}^{(\infty)}$ and $\F{T}^{(m_H)}$ are the same trees up to depth $H$.
    But this would mean that
    $|\F{T}^{(m_H)}| \geq \sum_{x \in \ctleaf^{\F{T}^{(m_H)}}\!, \,
    d^{\F{T}^{(m_H)}}\!(x) \leq H} p_x^{\F{T}^{(m_H)}} = \sum_{x \in \ctleaf^{\F{T}^{(\infty)}}\!,\,
    d^{\F{T}^{(\infty)}}\!(x) \leq H} p_x^{\F{T}^{(\infty)}} > |\F{T}|$, which is a contradiction to
    $|\F{T}^{(m_H)}| = |\F{T}|$.
  
    \smallskip
          
    \noindent 
    \textbf{\underline{1.1 Construction of $\Phi(\_)$}}

    \noindent 
    It remains to define
    the construction $\Phi(\_)$ mentioned above.
    Let $\F{T}_x$ be an $\fto_{\PP}$-RST that performs a non-innermost rewrite step at the root node
    $x$. This step has the form $t_x^{\F{T}_x} \fto_{\PP} \{p_{y_1}^{\F{T}_x}:t_{y_1}^{\F{T}_x}, \ldots,
    p_{y_k}^{\F{T}_x}:t_{y_k}^{\F{T}_x}\}$ using the rule
    $\bar{\ell} \to \{ \bar{p}_1:\bar{r}_1, \ldots, \bar{p}_k:\bar{r}_k\}$, the
    substitution $\bar{\sigma}$, and the position $\bar{\pi}$ such that $t_x^{\F{T}_x}|_{\bar{\pi}} = \bar{\ell} \bar{\sigma}$.
    Then we have $t_{y_j}^{\F{T}_x} = t_x^{\F{T}_x}[\bar{r}_j \bar{\sigma}]_{\bar{\pi}}$ for all $1 \leq j \leq k$.
 
    \begin{figure}
      \centering
            \begin{tikzpicture}[scale=0.5]
                \begin{pgfonlayer}{nodelayer}
                    \node [style=target,inner sep=0.5pt,pin={[pin distance=0.05cm, pin edge={,-}] 140:\tiny \textcolor{blue}{$x$}}] (3) at (0, 3) {\tiny$\mathbf{f}$};
                    \node [style=none] (6) at (1.5, 0) {};
                    \node [style=none] (7) at (-1.5, 0) {};
                    \node [style=none] (9) at (-2, -1) {};
                    \node [style=none] (10) at (2, -1) {};
                    \node [style=moveBlock,inner sep=1pt] (13) at (0, 0) {\tiny$\mathbf{i}$};
                    \node [style=moveBlock,inner sep=1pt] (14) at (0.75, 0.25) {\tiny$\mathbf{i}$};
                    \node [style=moveBlock,inner sep=1pt] (15) at (-0.75, 0.75) {\tiny$\mathbf{i}$};
                    \node [style=none] (16) at (-0.5, 1.5) {};
                    \node [style=none] (17) at (0, 1.25) {};
                    \node [style=none] (18) at (0.5, 1.5) {};
                    \node [style=none] (19) at (-0.75, 0) {};
                    \node [style=none] (20) at (0, -0.75) {};
                    \node [style=none] (21) at (0.75, -0.5) {};
                \end{pgfonlayer}
                \begin{pgfonlayer}{edgelayer}
                    \draw (3) to (6.center);
                    \draw (3) to (7.center);
                    \draw [style=dotWithoutHead] (7.center) to (9.center);
                    \draw [style=dotWithoutHead] (6.center) to (10.center);
                    \draw [style=dotWithoutHead, in=15, out=-105, looseness=0.50] (3) to (16.center);
                    \draw [style=dotWithoutHead, in=120, out=-90, looseness=0.75] (3) to (17.center);
                    \draw [style=dotWithoutHead, in=135, out=-75] (3) to (18.center);
                    \draw [style=dotHead, in=90, out=-30, looseness=0.75] (18.center) to (14);
                    \draw [style=dotHead, in=90, out=-150, looseness=0.75] (16.center) to (15);
                    \draw [style=dotHead, in=90, out=-45] (17.center) to (13);
                    \draw [style=dashHead, bend right=75, looseness=2.00] (14) to (3);
                    \draw [style=dashHead, bend left=75, looseness=1.75] (15) to (3);
                    \draw [style=dashHead, bend right=105, looseness=2.75] (13) to (3);
                    \draw [style=dotWithoutHead] (15) to (19.center);
                    \draw [style=dotWithoutHead] (13) to (20.center);
                    \draw [style=dotWithoutHead] (14) to (21.center);
                \end{pgfonlayer}
            \end{tikzpicture}
    \caption{Tree $\F{T}_x$ in the Construction of $\Phi(\_)$}\label{fig:Tx}
                    \end{figure}  

   Instead of applying a non-innermost rewrite step at the root $x$
    we want to directly apply an innermost rewrite step.
    Let $\tau$ be the position of some innermost redex in $t_x^{\F{T}_x}$ below $\bar{\pi}$. The 
            construction creates a new $\fto_{\PP}$-RST $\Phi(\F{T}_x) = (N',E',L')$ whose root is
    labeled with $(1:t_x^{\F{T}_x})$ such that $|\Phi(\F{T}_x)| =
    |\F{T}_x|$, and that directly performs the first rewrite step at position
    $\tau$ in the original tree $\F{T}_x$
    (which is an innermost rewrite step) at the root of the tree, by pushing it
    from the original nodes in the tree $\F{T}_x$ to the root of the new tree  $\Phi(\F{T}_x)$. 
    (It could also be that this innermost redex was never reduced in $\F{T}_x$.)
    This can be seen in the diagram in \Cref{fig:Tx}, which depicts the tree $\F{T}_x$.
    The boxes represent \textbf{i}nnermost rewrite steps at position $\tau$ and the dashed lines indicate that we
    push this rewrite step to the root.
    This push only results in the same convergence
    probability due to our
    restriction that $\PP$ is linear.
    However, since we are allowed to rewrite above $\tau$ in the original tree $\F{T}_x$, the actual position of the innermost redex that was originally at position $\tau$ might change during the application of a rewrite step.
    Hence, we recursively define the position $\varphi_{\tau}(v)$ that contains precisely this redex
    for each node $v$ in $\F{T}_x$ until we rewrite at this position.
    Initially, we have $\varphi_{\tau}(x) = \tau$.
    Whenever we have defined $\varphi_{\tau}(v)$ for some node $v$, and we have
    $t_v^{\F{T}_x} \fto_{\PP} \{p_{w_1}^{\F{T}_x}:t_{w_1}^{\F{T}_x}, \ldots,
    p_{w_m}^{\F{T}_x}:t_{w_m}^{\F{T}_x}\}$ for the direct successors $vE
    = \{w_1, \ldots, w_m\}$, using the rule $\ell \to \{
    p_1:r_1, \ldots, p_m:r_m\}$, the substitution $\sigma$, and the position $\pi$, we do the following:
    If $\varphi_{\tau}(v) = \pi$, meaning that we rewrite this innermost redex, then we set $\varphi_\tau(w_j) = \bot$ for all $1 \leq j \leq m$ to indicate that we have rewritten the innermost redex.
    If we have $\varphi_{\tau}(v) \bot \pi$, meaning that the rewrite step takes place on a
    position that is parallel to $\varphi_{\tau}(v)$, then we set
    $\varphi_\tau(w_j) = \varphi_{\tau}(v)$ for all $1 \leq j \leq
    m$, as the position of the innermost redex did not change.
    Otherwise, we have $\pi < \varphi_{\tau}(v)$ (since we cannot rewrite below $\varphi_{\tau}(v)$ as it is an innermost redex), and thus there exists a $\chi \in \IN^+$ such that $\pi.\chi = \varphi_{\tau}(v)$.
    Since the rules of $\PP$ are non-overlapping, the redex must be completely ``inside''
    the used substitution $\sigma$, and we can find a position $\alpha_q$ of a variable $q$ in $\ell$ and another position $\beta$ such that $\chi = \alpha_q.\beta$.
    Furthermore, since the rule is linear, $q$ only occurs once in $\ell$ and at most once in $r_j$ for all $1 \leq j \leq m$.
    If $q$ occurs in $r_j$ at a position $\rho_q^j$, then we set $\varphi_\tau(w_j) = \rho_q^j.\beta$.
    Otherwise, we set $\varphi_\tau(w_j) = \top$ to indicate that the innermost redex was erased during the computation.
    Finally, if $\varphi_{\tau}(v) \in \{\bot, \top\}$, then we set $\varphi_\tau(w_j)
    = \varphi_{\tau}(v)$ for all $1 \leq j \leq m$ as well.
    So to summarize, $\varphi_{\tau}(v)$ is
    now either the position of the innermost redex in $t_{v}$, $\top$ to indicate that the redex was erased, or $\bot$ to indicate that we have rewritten the redex.
     
    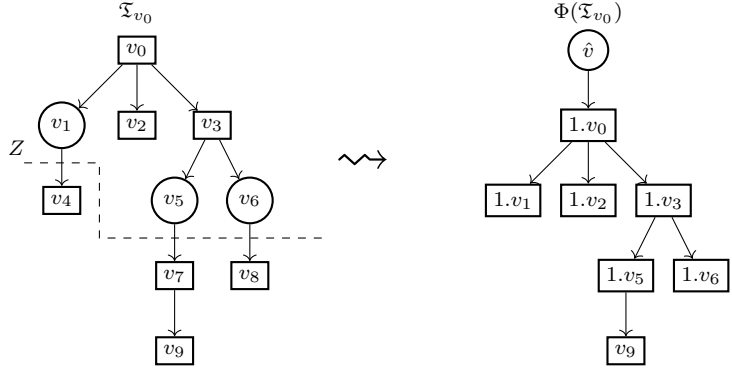
\begin{figure}
        \centering
        \scriptsize
		\begin{tikzpicture}
			\tikzstyle{adam}=[rectangle,thick,draw=black!100,fill=white!100,minimum size=3mm]
			\tikzstyle{empty}=[shape=circle,thick,minimum size=8mm]
			\tikzstyle{circle}=[shape=circle,draw=black!100,fill=white!100,thick,minimum size=3mm]
			
			\node[empty] at (-3.5, 0.5)  (name) {$\F{T}_{v_0}$};
			\node[adam] at (-3.5, 0)  (la) {$v_0$};

			\node[circle] at (-4.5, -1)  (lb1) {$v_1$};
			\node[adam] at (-3.5, -1)  (lb2) {$v_2$};
			\node[adam] at (-2.5, -1)  (lb3) {$v_3$};

			\node[adam] at (-4.5, -2)  (lc1) {$v_4$};

			\node[circle] at (-3, -2)  (ld1) {$v_5$};
			\node[circle] at (-2, -2)  (ld3) {$v_6$};
		
			\node[adam] at (-3, -3)  (lf1) {$v_7$};
			\node[adam] at (-2, -3)  (lf3) {$v_8$};

			\node[adam] at (-3, -4)  (lg1) {$v_9$};

			\draw (la) edge[->] (lb1);
			\draw (la) edge[->] (lb2);
			\draw (la) edge[->] (lb3);
			\draw (lb1) edge[->] (lc1);
			\draw (lb3) edge[->] (ld1);
			\draw (lb3) edge[->] (ld3);
			\draw (ld1) edge[->] (lf1);
			\draw (ld3) edge[->] (lf3);
			\draw (lf1) edge[->] (lg1);

			\node[empty] at (-5.1,-1.3)  (Z) {$Z$};

            \draw [dashed] (-5,-1.5) -- (-4,-1.5) -- (-4,-2.5) -- (-1,-2.5);

			\node[empty] at (-0.5, -1.5)  (lead) {\huge $\leadsto$};
			
			\node[empty] at (2.5, 0.5)  (name2) {$\Phi(\F{T}_{v_0})$};
			\node[circle] at (2.5, 0)  (a) {$\hat{v}$};

			\node[adam] at (2.5, -1)  (b1) {$1.v_0$};

			\node[adam] at (1.5, -2)  (c1) {$1.v_1$};
			\node[adam] at (2.5, -2)  (c2) {$1.v_2$};
			\node[adam] at (3.5, -2)  (c3) {$1.v_3$};

			\node[adam] at (3, -3)  (d1) {$1.v_5$};
			\node[adam] at (4, -3)  (d2) {$1.v_6$};

			\node[adam] at (3, -4)  (f1) {$v_9$};

			\draw (a) edge[->] (b1);
			\draw (b1) edge[->] (c1);
			\draw (b1) edge[->] (c2);
			\draw (b1) edge[->] (c3);
			\draw (c3) edge[->] (d1);
			\draw (c3) edge[->] (d2);
			\draw (d1) edge[->] (f1);
		\end{tikzpicture}
                \caption{Illustration of $\Phi(\F{T})$}\label{fig:illustration-phi}
                    \end{figure}

    In \Cref{fig:illustration-phi}, we illustrate the effect of
    $\Phi$, where the original tree $\F{T}_x$ with $x = v_0$
    is on the left and $\Phi(\F{T}_x)$ is on the right.
    The circled nodes represent the nodes where we perform a rewrite
    step at position $\varphi_{\tau}(v)$.
    We now define the $\fto_{\PP}$-RST $\Phi(\F{T}_x)$
    whose root $\hat{v}$ is labeled with $(1:t_x^{\F{T}_x})$ and 
    that directly performs the rewrite step $t_x^{\F{T}_x} = t_{\hat{v}}^{\Phi(\F{T}_x)} \ito_{\PP, \tau} \{\hat{p}_{1}:t_{1.x}^{\Phi(\F{T}_x)}, \ldots, \hat{p}_{h}:t_{h.x}^{\Phi(\F{T}_x)}\}$, 
    with the rule $\hat{\ell} \to \{ \hat{p}_1:\hat{r}_1, \ldots, \hat{p}_h:\hat{r}_h\} \in \PP$, a substitution $\hat{\sigma}$, 
    and the position $\tau$, at the new root $\hat{v}$. Here, we wrote ``$\ito_{\PP, \tau}$'' to
    make the position of the used redex explicit.
    We have $t_{\hat{v}}^{\Phi(\F{T}_x)}|_{\tau} = \hat{\ell} \hat{\sigma}$.
    Let $Z$ be the set of all nodes $z$ such that $\varphi_{\tau}(z) \neq \bot$, 
    i.e., at node $z$ we did not yet perform the rewrite step with the innermost redex.
    In the example we have $Z = \{v_0, \ldots, v_6\} \setminus \{v_4\}$.
    For each of these nodes $z \in Z$ and each $1 \leq e \leq h$,
    we create a new node $e.z \in N'$ with edges as in $\F{T}_x$ for the nodes in $Z$, e.g., for the node $1.v_3$ we create edges to $1.v_5$ and $1.v_6$.
    Furthermore, we add the edges from the new root $\hat{v}$ to the nodes $e.{x}$ for all $1 \leq e \leq h$.
    Note that $x$ was the root in the tree $\F{T}_x$ and has to be contained in $Z$.
    For example, for the node $\hat{v}$ we create an edge to $1.v_0$.
    We define the labeling of the nodes in $\Phi(\F{T}_x)$ as follows
    for all nodes $z$ in $Z$:        
    \begin{itemize}
        \item[] (T-1) $t_{e.z}^{\Phi(\F{T}_x)} = t_{z}^{\F{T}_x}[\hat{r}_e \hat{\sigma}]_{\varphi_{\tau}(z)}$ if $\varphi_{\tau}(z) \in \IN^*$ and $t_{e.z}^{\Phi(\F{T}_x)} = t_{z}^{\F{T}_x}$ if $\varphi_{\tau}(z) = \top$
        \item[] (T-2) $p_{e.z}^{\Phi(\F{T}_x)} = p_z^{\F{T}_x} \cdot \hat{p}_e$ 
    \end{itemize}
    Now, for a leaf $e.z' \in N'$ either $z' \in N$ is also a leaf (e.g., node $v_2$ in our example) 
    or we rewrite the innermost redex at
    position $\varphi_{\tau}(z')$ at node $z'$ in $\F{T}_x$ (e.g., node $v_1$ in our example).
    In the latter case, due to non-overlappingness, the same rule
    $\hat{\ell} \to \{ \hat{p}_1:\hat{r}_1, \ldots, \hat{p}_h:\hat{r}_h\}$ and
    the same substitution $\hat{\sigma}$ were used.
    If we rewrite $t_{z'}^{\F{T}_x} \ito_{\PP, \varphi_{\tau}(z')} \{\hat{p}_{1}:t_{w'_1}^{\F{T}_x}, \ldots, \hat{p}_{h}:t_{w'_h}^{\F{T}_x}\}$, 
    then we have $t_{w'_e}^{\F{T}_x} =
    t_{z'}^{\F{T}_x}[\hat{r}_e \hat{\sigma}]_{\varphi_{\tau}(z')} \stackrel{\mbox{\scriptsize (T-1)}}{=} t_{e.z'}^{\Phi(\F{T}_x)}$ 
    and $p_{w'_e}^{\F{T}_x} \stackrel{\mbox{\scriptsize (T-2)}}{=} p_{z'}^{\F{T}_x} \cdot \hat{p}_e =
    p_{e.z'}^{\Phi(\F{T}_x)}$.
    Thus, we can copy the rest of this subtree of $\F{T}_x$ to
    our newly generated tree $\Phi(\F{T}_x)$.
    In our example, $v_1$ has the only successor $v_4$, hence we can copy the subtree starting at node $v_4$, which is only the node itself, to the node $1.v_1$ in $\Phi(\F{T}_x)$.
    For $v_5$, we have the only successor $v_7$, hence we can copy the subtree starting at node $v_7$, which is the node itself together with its successor $v_9$, to the node $1.v_5$ in $\Phi(\F{T}_x)$.
    So essentially, we just had to define how to construct
    $\Phi(\F{T}_x)$ for
    the part of the tree before we reach nodes $v$ with $\varphi_{\tau}(v) = \bot$ in $\F{T}_x$.
    Now we have to show that $|\Phi(\F{T}_x)| = |\F{T}_x|$
    and that $\Phi(\F{T}_x)$ is indeed a valid $\fto_{\PP}$-RST 
    (i.e., that the edges between nodes $e.z$ with $z \in Z$ and its successors correspond to rewrite steps with $\PP$).
    \smallskip

    \noindent
    \underline{\textbf{1.2 We show $|\Phi(\F{T}_x)| = |\F{T}_x|$.}}

    \noindent
    Let $v$ be a leaf in $\Phi(\F{T}_x)$.
    If $v = e'.z$ for some node $z \in Z$ that is a leaf in $\F{T}_x$ (e.g., node $1.v_2$), then also $e.z$ must be a leaf in $\Phi(\F{T}_x)$ for every $1 \leq e \leq h$.
    Here, we get $\sum_{1 \leq e \leq h} p_{e.z}^{\Phi(\F{T}_x)} \stackrel{\text{(T-2)}}{=} \sum_{1 \leq e \leq h} p_z^{\F{T}_x} \cdot \hat{p}_e = p_z^{\F{T}_x} \cdot \sum_{1 \leq e \leq h} \hat{p}_e = p_z^{\F{T}_x} \cdot 1 = p_z^{\F{T}_x}$, and thus 
    
    \[
        \sum_{\substack{e.z \in \ctleaf^{\Phi(\F{T}_x)}\\ z \in Z,  z \in \ctleaf^{\F{T}_x}}}
        p_{e.z}^{\Phi(\F{T}_x)} = \sum_{\substack{z \in \ctleaf^{\F{T}_x}\\ z \in Z}} \left( \sum_{1 \leq e \leq h}
        p_{e.z}^{\Phi(\F{T}_x)} \right) = \sum_{\substack{z \in \ctleaf^{\F{T}_x}\\ z \in Z}} \!\! p_z^{\F{T}_x}
    \]

    If $v = e.z$ for some node $z \in Z$ that is not a leaf in $\F{T}_x$ (e.g., node $1.v_1$), then we know by construction 
    that the $e$-th successor $w_e$ of $z$ in $\F{T}_x$ is not contained in $Z$ and is a leaf of $\F{T}_x$.
    Let $(zE)_e$ denote the $e$-th successor of $z$.
    Here, we get $p_{e.z}^{\Phi(\F{T}_x)} \stackrel{\text{(T-2)}}{=} p_z^{\F{T}_x} \cdot \hat{p}_e = p_{w_e}^{\F{T}_x}$, and thus 
    
    \[
        \sum_{\substack{e.z \in \ctleaf^{\Phi(\F{T}_x)}\\z \in Z,  z \notin \ctleaf^{\F{T}_x}}} p_{e.z}^{\Phi(\F{T}_x)} 
        = \sum_{\substack{w_e \in \ctleaf^{\F{T}_x}\\ w_e = (zE)_e, w_e \not\in Z, z \in Z}} p_{e.z}^{\Phi(\F{T}_x)} 
        = \sum_{\substack{w_e \in \ctleaf^{\F{T}_x}\\ w_e = (zE)_e, w_e \not\in Z, z \in Z}} p_{w_e}^{\F{T}_x}
    \]

    Finally, if $v$ does not have the form $v = e.z$, then $v$ is also a leaf in $\F{T}_x$
    with $p_{v}^{\Phi(\F{T}_x)} = p_{v}^{\F{T}_x}$ and for both $v$ and its predecessor $u$ we have $v,u \not\in Z$, and thus 
    
    \[
        \sum_{\substack{v \in \ctleaf^{\Phi(\F{T}_x)}\\v \in \ctleaf^{\F{T}_x}}} p_v^{\Phi(\F{T}_x)} 
        = \sum_{\substack{v \in \ctleaf^{\F{T}_x}\\ v \in uE, u \not\in Z}} p_v^{\Phi(\F{T}_x)} 
        = \sum_{\substack{v \in \ctleaf^{\F{T}_x}\\ v \in uE, u \not\in Z}} p_v^{\F{T}_x}
    \]
    Note that these cases cover each leaf of $\F{T}_x$ exactly once.
    These three equations imply:
    {\small
    \allowdisplaybreaks
    \begin{align*}
        |\Phi(\F{T}_x)| &= \sum_{z \in \ctleaf^{\Phi(\F{T}_x)}} p_{z}^{\Phi(\F{T}_x)} = \sum_{\substack{e.z \in \ctleaf^{\Phi(\F{T}_x)}\\ z \in Z,  z \in \ctleaf^{\F{T}_x}}}
        p_{e.z}^{\Phi(\F{T}_x)} + \sum_{\substack{e.z \in \ctleaf^{\Phi(\F{T}_x)}\\
        z \in Z,  z \notin \ctleaf^{\F{T}_x}}} p_{e.z}^{\Phi(\F{T}_x)}
        + \sum_{\substack{v \in \ctleaf^{\Phi(\F{T}_x)}\\v \in \ctleaf^{\F{T}_x}}} p_v^{\Phi(\F{T}_x)} \\ 
        &= \sum_{\substack{z \in \ctleaf^{\F{T}_x}\\ z \in Z}} \!\! p_z^{\F{T}_x}
        + \sum_{\substack{w_e \in \ctleaf^{\F{T}_x}\\ w_e = (zE)_e, w_e \not\in Z, z \in Z}} p_{w_e}^{\F{T}_x}
        + \sum_{\substack{v \in \ctleaf^{\F{T}_x}\\ v \in uE, u \not\in Z}} p_v^{\F{T}_x} = \sum_{z \in \ctleaf^{\F{T}_x}} p_z^{\F{T}_x} = |\F{T}_x| \! 
    \end{align*}
    }
    \smallskip

    \noindent
    \underline{\textbf{1.3 We show that $\Phi(\F{T}_x)$ is indeed a valid RST\@.}}

    \noindent
    Finally, we prove that $\Phi(\F{T}_x)$ is a valid $\fto_{\PP}$-RST.
    Here, we only need to show that $t_{e.z}^{\Phi(\F{T}_x)} \fto_{\PP} \{p_{e.{w_1}}^{\Phi(\F{T}_x)}:t_{e.{w_1}}^{\Phi(\F{T}_x)}, \ldots, p_{e.{w_m}}^{\Phi(\F{T}_x)}:t_{e.{w_m}}^{\Phi(\F{T}_x)}\}$ for all $z \in Z$ as all the other edges and labelings were already present in $\F{T}_x$, which is a valid $\fto_{\PP}$-RST, and we have already seen that we have a valid innermost rewrite step at the new root $\hat{v}$.

    Let $z \in Z$.
    In the following, we distinguish between two different cases for a rewrite step at a
    node $z$ of $\F{T}_x$:
    \begin{enumerate}
        \item[(A)] We rewrite at a position
        parallel to $\varphi_{\tau}(z)$ or we have $\varphi_{\tau}(z) = \top$.
        \item[(B)] We rewrite at a position above $\varphi_{\tau}(z)$.
        Here, we need that $\PP$ is NO and linear.
    \end{enumerate}
    Let $t_z^{\F{T}_x} \fto_{\PP} \{p_z^{\F{T}_x} \cdot p_1:t_{w_1}^{\F{T}_x}, \ldots, p_z^{\F{T}_x} \cdot p_m:t_{w_m}^{\F{T}_x}\}$, 
    with a rule $\ell \to \{ p_1:r_1, \ldots, p_{m}:r_{m}\} \in \PP$, 
    a substitution $\sigma$, and a position $\pi$ such that $t_z^{\F{T}_x}|_{\pi} = \ell \sigma$. 
    We have $t_{w_j}^{\F{T}_x} = t_z^{\F{T}_x}[r_j \sigma]_{\pi}$ for all $1 \leq j \leq m$.
    \smallskip

    \noindent
    \textbf{(A)} We start with the case where
    we have $\pi \bot \varphi_{\tau}(z)$ or $\varphi_{\tau}(z) = \top$.
    By (T-1), we get $t_{e.z}^{\Phi(\F{T}_x)} = t_{z}^{\F{T}_x}[\hat{r}_e \hat{\sigma}]_{\varphi_{\tau}(z)}$ if $\varphi_{\tau}(z) \in \IN^*$ 
    and $t_{e.z}^{\Phi(\F{T}_x)} = t_{z}^{\F{T}_x}$ if $\varphi_{\tau}(z) = \top$.
    In both cases, we can rewrite $t_{e.z}^{\Phi(\F{T}_x)}$ using the same rule, 
    the same substitution, and the same position, 
    as we have $t_{e.z}^{\Phi(\F{T}_x)}|_{\pi} = t_{z}^{\F{T}_x}[\hat{r}_e \hat{\sigma}]_{\varphi_{\tau}(z)}|_{\pi} = t_z^{\F{T}_x}|_{\pi} = \ell \sigma$ 
    or directly $t_{e.z}^{\Phi(\F{T}_x)}|_{\pi} = t_z^{\F{T}_x}|_{\pi} = \ell \sigma$.

    It remains to show that $t_{e.{w_j}}^{\Phi(\F{T}_x)} = t_{e.z}^{\Phi(\F{T}_x)}[r_j \sigma]_{\pi}$ for all $1 \leq j \leq m$, i.e., that the labeling we defined for $\Phi(\F{T}_x)$ corresponds to this rewrite step.
    Let $1 \leq j \leq m$.
    If $\varphi_{\tau}(z) \in \IN^*$, then we have 
    $t_{e.z}^{\Phi(\F{T}_x)}[r_j \sigma]_{\pi} = t_{z}^{\F{T}_x}[\hat{r}_e \hat{\sigma}]_{\varphi_{\tau}(z)}[r_j \sigma]_{\pi} \stackrel{\varphi_{\tau}(z) \bot \pi}{=} t_{z}^{\F{T}_x}[r_j \sigma]_{\pi}[\hat{r}_e \hat{\sigma}]_{\varphi_{\tau}(z)} = t_{w_j}^{\F{T}_x}[\hat{r}_e \hat{\sigma}]_{\varphi_{\tau}(z)} = t_{e.{w_j}}^{\Phi(\F{T}_x)}$.
    If $\varphi_{\tau}(z) = \top$, then
    $t_{e.z}^{\Phi(\F{T}_x)}[r_j \sigma]_{\pi} = t_{z}^{\F{T}_x}[r_j \sigma]_{\pi} = t_{w_j}^{\F{T}_x} = t_{e.{w_j}}^{\Phi(\F{T}_x)}$.
    Finally, note that the probabilities of the labeling are correct, as we are using the same rule with the same probabilities.
    \smallskip

    \noindent
    \textbf{(B)} If we have $\pi < \varphi_{\tau}(z)$, then there exists a $\chi \in \IN^+$ such that $\pi.\chi = \varphi_{\tau}(z)$.
    Since the rules of $\PP$ are non-overlapping, the redex must be completely ``inside''
    the used substitution $\sigma$, and we can find a position $\alpha_q$ of a variable $q$ in $\ell$ and another position $\beta$ such that $\chi = \alpha_q.\beta$.
    Furthermore, since the rule is linear,
    $q$ only occurs once in $\ell$ and at most once in $r_j$ for all $1 \leq j \leq m$.
    Let $\rho_q^j$ be the position of $q$ in $r_j$ if it exists.
    By (T-1), we get $t_{z}^{\F{T}_x}[\hat{r}_e \hat{\sigma}]_{\varphi_{\tau}(z)} = t_{e.z}^{\Phi(\F{T}_x)}$.
    We can rewrite $t_{e.z}^{\Phi(\F{T}_x)}$ using the same rule, 
    the substitution $\sigma'$ with $\sigma'(q) = \sigma(q)[\hat{r}_e \hat{\sigma}]_{\beta}$ and $\sigma'(q') = \sigma(q')$ for all other variables $q' \neq q$, and the same position, 
    as we have $t_{e.z}^{\Phi(\F{T}_x)}|_{\pi} = t_{z}^{\F{T}_x}[\hat{r}_e \hat{\sigma}]_{\varphi_{\tau}(z)}|_{\pi} = t_{z}^{\F{T}_x}|_{\pi}[\hat{r}_e \hat{\sigma}]_{\chi} = \ell \sigma'$.

    It remains to show that $t_{e.{w_j}}^{\Phi(\F{T}_x)} = t_{e.z}^{\Phi(\F{T}_x)}[r_j \sigma']_{\pi}$ for all $1 \leq j \leq m$, i.e., that the labeling we defined for $\Phi(\F{T}_x)$ corresponds to this rewrite step.
    Let $1 \leq j \leq m$.
    If $\rho_q^j$ exists, then we have
    $t_{e.z}^{\Phi(\F{T}_x)}[r_j \sigma']_{\pi} = 
    t_{z}^{\F{T}_x}[\hat{r}_e \hat{\sigma}]_{\varphi_{\tau}(z)}[r_j \sigma']_{\pi} = 
    t_{z}^{\F{T}_x}[\hat{r}_e \hat{\sigma}]_{\pi.\alpha_q.\beta}[r_j \sigma']_{\pi} = 
    t_{z}^{\F{T}_x}[r_j \sigma]_{\pi}[\hat{r}_e \hat{\sigma}]_{\rho_q^j.\beta} = 
    t_{w_j}^{\F{T}_x}[\hat{r}_e \hat{\sigma}]_{\varphi_{\tau}(z)} = t_{e.{w_j}}^{\Phi(\F{T}_x)}$.
    Otherwise, we erase the precise redex and obtain
    $t_{e.z}^{\Phi(\F{T}_x)}[r_j \sigma]_{\pi} =
    t_{z}^{\F{T}_x}[\hat{r}_e \hat{\sigma}]_{\varphi_{\tau}(z)}[r_j \sigma']_{\pi} = 
    t_{z}^{\F{T}_x}[\hat{r}_e \hat{\sigma}]_{\pi.\alpha_q.\beta}[r_j \sigma']_{\pi} =
    t_{z}^{\F{T}_x}[r_j \sigma]_{\pi} = t_{w_j}^{\F{T}_x} = t_{e.{w_j}}^{\Phi(\F{T}_x)}$.
    Again, the probabilities of the labeling are correct, as we are using the same rule with the same probabilities.
    \smallskip

    \noindent
    \underline{\textbf{2. Analyzing the expected derivation length of $\F{T}^{(\infty)}$}}

    \noindent
    Our goal is to show that in the construction of the trees $\F{T}^{(1)}, \F{T}^{(2)},
    \ldots$, every leaf $v$ of $\F{T}$ is turned into leaves of 
    $\F{T}^{(1)}, \F{T}^{(2)}, \ldots$ whose probabilities sum up to $p^{\F{T}}_v$, 
    and whose depths are greater or equal
    than the original depth $d^{\F{T}}(v)$ of $v$ in $\F{T}$. 
    This implies $\edl(\F{T}) \leq \edl(\F{T}^{(1)}) \leq \edl(\F{T}^{(2)}) \leq \ldots$\
    From this observation, we then conclude 
    $\edl(\F{T}) \leq \edl(\F{T}^{(\infty)})$.

    \smallskip
          
    \noindent 
    \textbf{\underline{2.1 We prove $\edl(\F{T}) \leq \edl(\F{T}^{(1)})  \leq \edl(\F{T}^{(2)}) \leq \ldots$}}

    \noindent
    We start by considering how the leaves of $\F{T}$ correspond to
    the leaves of 
    $\F{T}^{(1)}$.
    Let $v$ be a leaf in $\F{T}$.
    If $v \not\in xE^*$, then $v$ is also a leaf in $\F{T}^{(1)}$, labeled with the same probability, and at the same depth.
    Otherwise, $v \in xE^*$, which means that either $v \in Z$, or $v \not\in Z$ but its
    predecessor is in $Z$, or neither $v$ nor its predecessor are in $Z$.

    If $v \in Z$ (like the node $v_2$ in \Cref{fig:illustration-phi}), then also $e.v$ must be a leaf in $\Phi(\F{T}_x)$ for every $1 \leq e \leq h$.
    Here, we again have $\sum_{1 \leq e \leq h} p_{e.v}^{\Phi(\F{T}_x)} = p_v^{\F{T}_x}$ as before.
    In addition, we also know that if $e.v$ is at depth $m$ in $\Phi(\F{T}_x)$, then $v$ is at depth $m-1$ in $\F{T}_x$.

    If $v \not\in Z$ but $v$ is the $e$-th successor of a node $z \in Z$
    (like the node $v_8$ in \Cref{fig:illustration-phi}), then $e.z$ is a leaf in $\Phi(\F{T}_x)$.
    Here, we again have $p_{e.z}^{\Phi(\F{T}_x)} = p_{v}^{\F{T}_x}$ as before.
    In addition, we also know that if $e.z$ is at depth $m$ in $\Phi(\F{T}_x)$, then $v$ is at depth $m$ in $\F{T}_x$.

    Finally, if neither $v$ nor its predecessor are in $Z$
    (like the node $v_9$ in \Cref{fig:illustration-phi}), 
    then $v$ is also a leaf in $\Phi(\F{T}_x)$ with $p_{v}^{\Phi(\F{T}_x)} = p_{v}^{\F{T}_x}$ and at the same depth.

    These cases cover each leaf of $\F{T}^{(1)}$ exactly once and all leaves in
    $\Phi(\F{T}_x)$ are at a depth greater or equal than the corresponding leaves in $\F{T}_x$, 
    implying that for each leaf $u$ in $\F{T}$ we can find a set of leaves $\Omega_u^{(1)}$ in $\F{T}^{(1)}$ 
    such that $d^{\F{T}}(u) \leq d^{\F{T}^{(1)}}(w)$ for each $w \in \Omega_u^{(1)}$,
    $\sum_{w \in \Omega_u^{(1)}} p^{\F{T}^{(1)}}_w = p^{\F{T}}_u$, 
    and $\ctleaf^{\F{T}^{(1)}} = \biguplus_{u \in \ctleaf^{\F{T}}} \Omega_u^{(1)}$.

    We can now use the same observation for each leaf $w$ in the tree $\F{T}^{(1)}$ in
    order to obtain a new set $\Xi_w^{(2)}$ of leaves in $\F{T}^{(2)}$ and define $\Omega_u^{(2)} = \bigcup_{w \in \Omega_u^{(1)}} \Xi_w^{(2)}$.
    Again, we get $d^{\F{T}}(u) \leq d^{\F{T}^{(2)}}(w)$ for each $w \in \Omega_u^{(2)}$,
    $\sum_{w \in \Omega_u^{(2)}} p^{\F{T}^{(2)}}_w = p^{\F{T}}_u$, 
    and $\ctleaf^{\F{T}^{(2)}} = \biguplus_{u \in \ctleaf^{\F{T}}} \Omega_u^{(2)}$.
    We now do this for each $i \in \IN$ to define the set $\Omega_u^{(i)}$ for each leaf $u$ in the original tree $\F{T}$.
    Overall, this implies  $\edl(\F{T}) \leq \edl(\F{T}^{(1)})  \leq \edl(\F{T}^{(2)}) \leq \ldots$.

    \smallskip

    \noindent
    \underline{\textbf{2.2 We prove $\edl(\F{T}^{(\infty)}) \geq \edl(\F{T})$.}}

    \noindent
    From the construction of the $\Omega_u^{(i)}$ above, in the end, we obtain
    \[\ctleaf^{\F{T}^{(\infty)}} = \biguplus_{u \in \ctleaf^{\F{T}}} \limsup_{i \to \infty} \Omega_u^{(i)},\] 
    where $\limsup_{i \to \infty} \Omega_u^{(i)} = \bigcap_{i' \in \IN} \bigcup_{i > i'} \Omega_u^{(i)} = \{w \mid w$ 
    is contained in infinitely many $\Omega_u^{(i)}\}$.
    To see this, let $v \in \ctleaf^{\F{T}^{(\infty)}}$.
    Remember that for every depth $H$ of the tree, 
    there exists an $m_H$ such that $\F{T}^{(\infty)}$ and $\F{T}^{(i)}$ are the same trees up to depth $H$ for all $i \geq m_H$.
    This means that the node $v$ must be contained in all trees $\F{T}^{(m)}$ with $m \geq m_{d^{\F{T}^{(\infty)}}(v)}$, 
    i.e., it is contained in $\biguplus_{u \in \ctleaf^{\F{T}}} \limsup_{i \to \infty} \Omega_u^{(i)}$.
    The other direction of the equality follows in the same manner.

    Since the probabilities of the leaves in $\F{T}^{(i)}$ always add up to the probability of
    the corresponding leaf in $\F{T}$, we have $\sum_{v \in \limsup_{i \to \infty}
    \Omega_u^{(i)}} p_v^{\F{T}^{(\infty)}} \leq p_u^{\F{T}}$ for all $u \in \ctleaf^{\F{T}}$.
    We now show $\edl(\F{T}) \leq \edl(\F{T}^{(\infty)})$ by considering the two cases where 
    $\sum_{v \in \limsup_{i \to \infty}
    \Omega_u^{(i)}} p_v^{\F{T}^{(\infty)}} < p_u^{\F{T}}$ for some $u \in \ctleaf^{\F{T}}$
    and where $\sum_{v \in \limsup_{i \to \infty}
    \Omega_u^{(i)}} p_v^{\F{T}^{(\infty)}} = p_u^{\F{T}}$ for all $u \in
    \ctleaf^{\F{T}}$.

    If we have $\sum_{v \in \limsup_{i \to \infty} \Omega_u^{(i)}} p_v^{\F{T}^{(\infty)}} < p_u^{\F{T}}$ for some $u \in \ctleaf^{\F{T}}$, then we obtain
    \[ 
        \begin{array}{rclcl}
            1 &\geq& |\F{T}|
            &=&
            \sum_{u \in \ctleaf^{\F{T}}} p_u^{\F{T}} \\
            &>&
            \sum_{u \in \ctleaf^{\F{T}}} \sum_{v \in \limsup_{i \to \infty} \Omega_u^{(i)}} p_v^{\F{T}^{(\infty)}}
            &=&
            \sum_{v \in \biguplus_{u \in \ctleaf^{\F{T}}} \limsup_{i \to \infty} \Omega_u^{(i)}} p_v^{\F{T}^{(\infty)}}\\
            &=&
            \sum_{v \in \ctleaf^{\F{T}^{(\infty)}}} p_v^{\F{T}^{(\infty)}}
            &=& 
            |\F{T}^{(\infty)}|
        \end{array}
    \]
    and from $|\F{T}^{(\infty)}| < 1$ we directly get $\edl(\F{T}^{(\infty)}) =
    \infty$.

    Otherwise, we have $\sum_{v \in \limsup_{i \to \infty} \Omega_u^{(i)}}
    p_v^{\F{T}^{(\infty)}} = p_u^{\F{T}}$ for all $u \in \ctleaf^{\F{T}}$, and thus we obtain
    
    {\small
    \begin{longtable}{C@{\;}C@{\;}L}
        &&\edl(\F{T}) = \sum_{n = 0}^{\infty} \sum_{\substack{u \in N^{\F{T}} \setminus \ctleaf^{\F{T}} \\ d^{\F{T}}(u) = n}} p_u^{\F{T}}\\
        (\textcolor{blue}{\sum_{\substack{u \in N \setminus \ctleaf \\ d^{\F{T}}(u) = n}} p_u^{\F{T}} = 1 - \sum_{\substack{u \in \ctleaf \\ d^{\F{T}}(u) \leq n}} p_u^{\F{T}}})&=&\sum_{n = 0}^{\infty} (1 - \sum_{\substack{u \in \ctleaf^{\F{T}} \\ d^{\F{T}}(u) \leq n}} p_u^{\F{T}})\\
        (\textcolor{blue}{p_u^{\F{T}} = \sum_{v \in \limsup_{i \to \infty} \Omega_u^{(i)}} p_v^{\F{T}^{(\infty)}}})&=&\sum_{n = 0}^{\infty} (1 - \sum_{\substack{u \in \ctleaf^{\F{T}} \\ d^{\F{T}}(u) \leq n}} \sum_{v \in \limsup_{i \to \infty} \Omega_u^{(i)}} p_v^{\F{T}^{(\infty)}})\\
        (\textcolor{blue}{\forall v \in \limsup_{i \to \infty} \Omega_u^{(i)}: d^{\F{T}}(u) \leq d^{\F{T}^{(\infty)}}(v)})&\leq&\sum_{n = 0}^{\infty} (1 - \sum_{u \in \ctleaf^{\F{T}} } \sum_{\substack{v \in \limsup_{i \to \infty} \Omega_u^{(i)}\\ d^{\F{T}^{(\infty)}}(v) \leq n}} p_v^{\F{T}^{(\infty)}})\\
        (\textcolor{blue}{\Omega_u^{(i)} \text{ all disjoint}})&=&\sum_{n = 0}^{\infty} (1 - \sum_{\substack{v \in \biguplus_{u \in \ctleaf^{\F{T}}} \limsup_{i \to \infty} \Omega_u^{(i)}\\ d^{\F{T}^{(\infty)}}(v) \leq n}} p_v^{\F{T}^{(\infty)}})\\
        (\textcolor{blue}{\ctleaf^{\F{T}^{(\infty)}} = \biguplus_{u \in \ctleaf^{\F{T}}} \limsup_{i \to \infty} \Omega_u^{(i)}})&=&\sum_{n = 0}^{\infty} (1 - \sum_{\substack{v \in \ctleaf^{\F{T}^{(\infty)}}\\ 
        d^{\F{T}^{(\infty)}}(v) \leq n}} p_v^{\F{T}^{(\infty)}})\\
        (\textcolor{blue}{\sum_{\substack{u \in N \setminus \ctleaf^{\F{T}^{(\infty)}} \\ 
        d^{\F{T}^{(\infty)}}(u) = n}} p_u^{\F{T}^{(\infty)}} = 1 - \sum_{\substack{v \in \ctleaf^{\F{T}^{(\infty)}} \\ 
        d^{\F{T}^{(\infty)}}(v) \leq n}} p_v^{\F{T}^{(\infty)}}})&=&\sum_{n = 0}^{\infty} \sum_{\substack{u \in N \setminus \ctleaf^{\F{T}^{(\infty)}} \\ d^{\F{T}^{(\infty)}}(u) = n}} p_u^{\F{T}^{(\infty)}}\\
        &=& \edl(\F{T}^{(\infty)})
    \end{longtable}}
    
    \noindent
    and therefore, $\edl(\F{T}^{(\infty)}) \geq \edl(\F{T})$.
\end{proof}

\ASTvswAST*

\begin{proof}
    We only have to prove the non-trivial direction ``$\Longleftarrow$''.
    Let $\PP$ be a PTRS that is non-overlapping, left-linear, and non-erasing.
    \smallskip
          
    \noindent 
    \textbf{\underline{$\ASTf \Longleftarrow \wASTf$:}}

    \noindent
    Assume for a contradiction that we have $\wASTf$ but not $\ASTf$.
    This means that there exists an
    $\fto_{\PP}$-RST $\F{T}$ such that $|\F{T}| = c$ for some $0 \leq c < 1$.
    Let $t \in \TT$ such that the root of $\F{T}$ is labeled with $(1:t)$.
    Since we have $\wASTf$, there exists another $\fto_{\PP}$-RST $\tilde{\F{T}} = (\tilde{N}, \tilde{E},
    \tilde{L})$ such that $|\tilde{\F{T}}| = 1$ and the root of $\tilde{\F{T}}$
    is also labeled with $(1:t)$.
    Hence, in $\F{T}$ at least one rewrite step is performed 
    at some node $x$ that is different to the rewrite step performed in $\tilde{\F{T}}$.
    The core steps of the proof are the same as for the proof of \Cref{lemma-eq-AST-iAST-1}.
    We iteratively push the rewrite steps that would be performed in $\tilde{\F{T}}$ at node $x$ to this node in $\F{T}$.
    Then, the limit of this construction is exactly $\tilde{\F{T}}$, which would mean that
    $|\tilde{\F{T}}| \leq c < 1$, which is the desired contradiction.
    For this, we have to adjust the construction $\Phi(\_)$.
    The rest of the proof remains completely the \pagebreak same.
    \smallskip
          
    \noindent 
    \textbf{\underline{1.1 Construction of $\Phi(\_)$}}
    \smallskip

    \noindent
    Let $\F{T}_x = \F{T}[xE^*]$ be an $\fto_{\PP}$-RST that performs a rewrite step at position $\zeta$ at the root node $x$, i.e., $t_x^{\F{T}_x} \fto_{\PP, \zeta} \{p_{y_1}^{\F{T}_x}:t_{y_1}^{\F{T}_x}, \ldots, p_{y_k}^{\F{T}_x}:t_{y_k}^{\F{T}_x}\}$ using the rule $\bar{\ell} \to \{ \bar{p}_1:\bar{r}_1, \ldots, \bar{p}_k:\bar{r}_k\}$, and the substitution $\bar{\sigma}$ such that $t_x^{\F{T}_x}|_{\zeta} = \bar{\ell} \bar{\sigma}$.
    Then $t_{y_j}^{\F{T}_x} = t_x^{\F{T}_x}[\bar{r}_j \bar{\sigma}]_{\zeta}$ for all $1 \leq j \leq k$.
    Furthermore, assume that
    $\tilde{T}_x = \tilde{\F{T}}[x\tilde{E}^*]$ rewrites at position $\tau$ (using a rule
    $\hat{\ell} \to \{ \hat{p}_1:\hat{r}_1, \ldots, \hat{p}_h:\hat{r}_h\} \in \PP$ and  
    the substitution $\hat{\sigma}$)
    with $\tau \neq \zeta$. (Note that
    if $\tau = \zeta$, then by non-overlappingness the rewrite step would be the same.)
    Instead of applying the rewrite step at position $\zeta$ at the root $x$
    we want to directly apply the rewrite step at position $\tau$.

    The construction creates a new RST $\Phi(\F{T}_x) = (N',E',L')$ such that
    $|\Phi(\F{T}_x)| = |\F{T}_x|$,
    and that directly performs a rewrite step at position $\tau$ at the root of the tree, 
    by pushing it from the original nodes in the tree $\F{T}_x$ to the root.
    This push only results in the same convergence probability due to our
    restriction that $\PP$ is linear and non-erasing.
    We need the non-erasing property now, because $\tau$ may be above $\zeta$, 
    which was not possible in the proof of \Cref{lemma-eq-AST-iAST-1}.

    Again, since we are allowed to rewrite above $\tau$ in the original tree $\F{T}_x$, the
    actual position of the redex that was originally at position $\tau$ might change
    during the application of a rewrite step.
    Hence, we recursively define the position $\varphi_{\tau}(v)$
    that contains precisely this redex for each node $v$ in $\F{T}_x$ until we rewrite at this position.
    Compared to the proof of \Cref{lemma-eq-AST-iAST-1}, 
    since the rules in $\PP$ are non-erasing, 
    we only have $\varphi_{\tau}(v) \in \IN^*$ or $\varphi_{\tau}(v) = \bot$. 
    The option $\varphi_{\tau}(v) = \top$ is not possible anymore.
    Initially, we have $\varphi_{\tau}(x) = \tau$.
    Whenever we have defined $\varphi_{\tau}(v)$ for some node $v$, and we have
    $t_v^{\F{T}_x} \fto_{\PP} \{p_{w_1}^{\F{T}_x}:t_{w_1}^{\F{T}_x}, \ldots, p_{w_m}^{\F{T}_x}:t_{w_m}^{\F{T}_x}\}$ 
    for the direct successors $vE = \{w_1, \ldots, w_m\}$, using the rule $\ell \to \{ p_1:r_1, \ldots, p_m:r_m\}$, 
    the substitution $\sigma$, and position $\pi$, we do the following:
    If $\varphi_{\tau}(v) = \pi$, then we set $\varphi_\tau(w_j) = \bot$ for all $1 \leq j \leq m$ to indicate that we have rewritten the redex.
    If we have $\varphi_{\tau}(v) \bot \pi$, meaning that the rewrite step takes place parallel to $\varphi_{\tau}(v)$, then we set $\varphi_\tau(w_j) = \varphi_{\tau}(v)$ for all $1 \leq j \leq m$, as the position did not change.
    If we have $\varphi_{\tau}(v) < \pi$, then we set $\varphi_\tau(w_j)
    = \varphi_{\tau}(v)$ for all $1 \leq j \leq m$ as well, as the position did not change
    either.
    If we have $\pi < \varphi_{\tau}(v)$, then there exists a $\chi \in \IN^+$ such that $\pi.\chi = \varphi_{\tau}(v)$.
    Since the rules of $\PP$ are non-overlapping, the redex must be completely ``inside'' the
    used substitution $\sigma$, and we can find a position $\alpha_q$ of a variable $q$ in $\ell$ and another position $\beta$ such that $\chi = \alpha_q.\beta$.
    Furthermore, since the rule is linear and non-erasing,
    $q$ only occurs once in $\ell$ and once in $r_j$ for all $1 \leq j \leq m$.
    Let $\rho_q^j$ be the position of $q$ in $r_j$.
    Here, we set $\varphi_\tau(w_j) = \rho_q^j.\beta$.
    Finally, if $\varphi_{\tau}(v) = \bot$, then we set $\varphi_\tau(w_j)
    = \varphi_{\tau}(v) = \bot$ for all $1 \leq j \leq m$ as well.

    Again, we now define the $\fto_{\PP}$-RST $\Phi(\F{T}_x)$ whose root is labeled with $(1:t_x^{\F{T}_x})$ such that $|\Phi(\F{T}_x)| = |\F{T}_x|$, 
    and that directly performs the rewrite step $t_{x}^{\F{T}_x} \fto_{\PP, \tau} \{\hat{p}_{1}:t_{1.x}^{\Phi(\F{T}_x)}, \ldots, \hat{p}_{h}:t_{h.x}^{\Phi(\F{T}_x)}\}$, 
    with the rule $\hat{\ell} \to \{ \hat{p}_1:\hat{r}_1, \ldots, \hat{p}_h:\hat{r}_h\} \in \PP$, 
    the substitution $\hat{\sigma}$, 
    and the position $\tau$, at the new root $\hat{v}$.
    Here, we have $t_x^{\F{T}_x}|_{\tau} = \hat{\ell} \hat{\sigma}$.
    Let $Z$ be the set of all nodes $v$ such that $\varphi_{\tau}(v) \neq \bot$.
    For each of these nodes $z \in Z$ and each $1 \leq e \leq h$,
    we create a new node $e.z \in N'$ with edges as in $\F{T}_x$ for the nodes in $Z$.
    Furthermore, we add the edges from the new root $\hat{v}$ to the nodes $e.{x}$ for all $1 \leq e \leq h$.
    Remember that $x$ was the root in the tree $\F{T}_x$ and has to be contained in $Z$.
    We define the labeling of the nodes in  $\Phi(\F{T}_x)$ as follows
    for all nodes $z$ in
    $Z$:
    \begin{itemize}
        \item[] (T-1) $t_{e.z}^{\Phi(\F{T}_x)} = t_{z}^{\F{T}_x}[\hat{r}_e \delta]_{\varphi_{\tau}(z)}$ for the substitution $\delta$ such that $t_{z}^{\F{T}_x}|_{\varphi_{\tau}(z)} = \hat{\ell} \delta$
        \item[] (T-2) $p_{e.z}^{\Phi(\F{T}_x)} = p_z^{\F{T}_x} \cdot \hat{p}_e$ 
    \end{itemize}
    Now, for a leaf $e.z' \in N'$ either $z' \in N$ is also a leaf
    or we rewrite at the position $\varphi_{\tau}(z')$ in node $z'$ in $\F{T}_x$.
    If we rewrite
    $t_{z'}^{\F{T}_x} \ito_{\PP, \varphi_{\tau}(z')} \{p_{w'_1}^{\F{T}_x}:t_{w'_1}^{\F{T}_x}, \ldots, p_{w'_h}^{\F{T}_x}:t_{w'_h}^{\F{T}_x}\}$,
    then we have $t_{w'_e}^{\F{T}_x} = t_{z'}^{\F{T}_x}[\hat{r}_e \delta]_{\varphi_{\tau}(z')} \stackrel{\mbox{\scriptsize (T-1)}}{=} t_{e.z'}^{\Phi(\F{T}_x)}$ for the substitution $\delta$ such that $t_{z'}^{\F{T}_x}|_{\varphi_{\tau}(z')} = \hat{\ell} \delta$
    and $p_{w'_e}^{\F{T}_x} \stackrel{\mbox{\scriptsize (T-2)}}{=} p_{z'}^{\F{T}_x} \cdot \hat{p}_e = p_{e.z'}^{\Phi(\F{T}_x)}$.
    Thus, we can copy the rest of this subtree of $\F{T}_x$ in
    our newly generated tree $\Phi(\F{T}_x)$.
    Again, we just had to define how to construct
    $\Phi(\F{T}_x)$ for
    the part of the tree before we reach the
    nodes $v$ with $\varphi_{\tau}(v) = \bot$ in $\F{T}_x$.
    As in the proof of \Cref{lemma-eq-AST-iAST-1}, we obtain $|\Phi(\F{T}_x)| =
    |\F{T}_x|$.
    We only need to show that $\Phi(\F{T}_x)$ is a valid $\fto_{\PP}$-RST, i.e., that
    $t_{e.z}^{\Phi(\F{T}_x)} \fto_{\PP} \{p_{e.{w_1}}^{\Phi(\F{T}_x)}:t_{e.{w_1}}^{\Phi(\F{T}_x)}, \ldots, p_{e.{w_m}}^{\Phi(\F{T}_x)}:t_{e.{w_m}}^{\Phi(\F{T}_x)}\}$ 
    for all $z \in Z$ as in the proof of \Cref{lemma-eq-AST-iAST-1}.
    \smallskip

    \noindent
    \underline{\textbf{1.3 We show that $\Phi(\F{T}_x)$ is indeed a valid RST\@.}}

    \noindent
    Let $z \in Z$.
    In the following, we distinguish three different cases for a rewrite step at node $z$ of $\F{T}_x$:
    \begin{enumerate}
        \item[(A)] We rewrite at a position parallel to $\varphi_{\tau}(z)$.
        \item[(B)] We rewrite at a position above $\varphi_{\tau}(z)$.
        \item[(C)] We rewrite at a position below $\varphi_{\tau}(z)$.
    \end{enumerate}
    The cases (A) and (B) are analogous to our earlier proof, but in Case (A) we cannot erase the redex anymore.
    We only need to look at the new Case (C).

    \noindent
    \textbf{(C)} If we have
    $t_z^{\F{T}_x} \fto_{\PP} \{p_z^{\F{T}_x} \cdot p_1:t_{w_1}^{\F{T}_x}, \ldots, p_z^{\F{T}_x} \cdot p_m:t_{w_m}^{\F{T}_x}\}$, 
    then there is a rule $\ell \to \{ p_1:r_1, \ldots, p_{m}:r_{m}\} \in \PP$, 
    a substitution $\sigma$, and a position $\pi$ with $t_z^{\F{T}_x}|_{\pi} = \ell \sigma$. 
    Then $t_{w_j}^{\F{T}_x} = t_z^{\F{T}_x}[r_j \sigma]_{\pi}$ for all $1 \leq j \leq m$.
    Additionally, we assume that $\varphi_{\tau}(z) < \pi$, and thus there exists a $\chi \in \IN^+$ such that $\pi = \varphi_{\tau}(z).\chi$.
    By (T-1), we get $t_{z}^{\F{T}_x}[\hat{r}_e \delta]_{\varphi_{\tau}(z)} = t_{e.z}^{\Phi(\F{T}_x)}$ for the substitution $\delta$ such that $t_{z}^{\F{T}_x}|_{\varphi_{\tau}(v)} = \hat{\ell} \delta$.
    Since the rules of $\PP$ are non-overlapping, including the rule $\hat{\ell} \to \{ \hat{p}_1:\hat{r}_1, \ldots, \hat{p}_h:\hat{r}_h\}$ that we use at the root of $\Phi(\F{T}_x)$, 
    the redex for the current rewrite step must be completely ``inside'' the substitution $\delta$, and we can find a variable position $\alpha_q$ of a variable $q$ in $\hat{\ell}$ and another position $\beta$ such that $\chi = \alpha_q.\beta$.
    Furthermore, since the rule is also linear
    and non-erasing, $q$ occurs exactly once in $\hat{\ell}$ and exactly once in $\hat{r}_j$ for all $1 \leq j \leq m$.
    Let $\rho_q^j$ be the position of $q$ in $r_j$.
    We can rewrite $t_{e.z}^{\Phi(\F{T}_x)}$ using the same rule, 
    the same substitution, and the same position, 
    as we have $t_{e.z}^{\Phi(\F{T}_x)}|_{\pi} =
    t_{z}^{\F{T}_x}[\hat{r}_e \delta]_{\varphi_{\tau}(z)}|_{\pi} = \hat{r}_e \delta|_{\chi}
    = \delta(z)|_{\beta} = t_{z}^{\F{T}_x}|_{\varphi_{\tau}(z)}|_{\alpha_q}|_{\beta} =
    t_{z}^{\F{T}_x}|_{\varphi_{\tau}(z).\alpha_q.\beta} = t_{z}^{\F{T}_x}|_{\varphi_{\tau}(z).\chi} =
    t_{z}^{\F{T}_x}|_{\pi} = \ell \sigma$.
    
    It remains to show that $t_{e.{w_j}}^{\Phi(\F{T}_x)} = t_{e.z}^{\Phi(\F{T}_x)}[r_j \sigma']_{\pi}$ for all $1 \leq j \leq m$, i.e., that the labeling we defined for $\Phi(\F{T}_x)$ corresponds to this rewrite step.
    Let $1 \leq j \leq m$.
    We have
    $t_{e.z}^{\Phi(\F{T}_x)}[r_j \sigma]_{\pi} = t_{z}^{\F{T}_x}[\hat{r}_e \delta]_{\varphi_{\tau}(z)}[r_j \sigma]_{\pi} = t_{z}^{\F{T}_x}[\hat{r}_e \delta']_{\varphi_{\tau}(z)} = t_{e.{w_j}}^{\Phi(\F{T}_x)}$ for the substitution $\delta'$ with $\delta'(q) = \delta(q)[r_j \sigma]_{\beta}$ and $\delta'(q') = \delta(q')$ for all other variables $q' \neq q$.
    With this new substitution, we get $t_{w_j}^{\F{T}_x}|_{\varphi_{\tau}(w_j)} = t_{z}^{\F{T}_x}[r_j \sigma]_{\varphi_{\tau}(z).\alpha_q.\beta}|_{\varphi_{\tau}(z)} = t_{z}^{\F{T}_x}|_{\varphi_{\tau}(z)}[r_j \sigma]_{\alpha_q.\beta} = (\hat{\ell} \delta)[r_j \sigma]_{\alpha_q.\beta} = \hat{\ell} \delta'$.
    Finally, note that the probabilities of the labeling are correct, as we are using the same rule with the same probabilities in both trees.

    \smallskip
          
    \noindent 
    \textbf{\underline{$\PASTf \Longleftarrow \wPASTf$:}}

    \noindent
    By the same construction as for $\wASTf$, we also get $\edl(\F{T}^{(\infty)}) \geq \edl(\F{T})$.
\end{proof}

\IASTAndLIASTLemma*

\begin{proof}
    The idea and the construction of this proof are completely analogous
    to the one of \Cref{lemma-eq-AST-iAST-1}.
    We iteratively move leftmost-innermost rewrite steps to a higher position in the innermost RST.
    Hence, the resulting tree is a leftmost $\lito_{\PP}$-RST.

    The construction of $\Phi(\_)$ is also analogous to the one
    in \Cref{lemma-eq-AST-iAST-1}.
    The only difference to the proof of \Cref{lemma-eq-AST-iAST-1} is that 
    the original tree $\F{T}_x$ is already an $\ito_{\PP}$-RST.
    This means that during our construction only Case (A) can occur, 
    as we cannot rewrite above a redex in an $\ito_{\PP}$-RST.
    And for Case (A), we only need the property of being non-overlapping.
\end{proof}

Next, we prove the new results from \cref{Improving Applicability}.

\ASTAndIASTLemmaThree*

\begin{proof}
    The proof is completely analogous to the one of \Cref{lemma-eq-AST-iAST-1}.
    We iteratively move the innermost rewrite steps to a higher position using the construction $\Phi(\_)$.
    Note that since $\PP$ is spare and left-linear, 
    in the construction of $\Phi(\_)$ and in the proof of
    Case (B), if the innermost redex is below a redex $\ell\sigma$ that is reduced next via a
    rule $\ell \to \{p_1:r_1, \ldots, p_m:r_m \}$, then the innermost redex is completely
    ``inside'' the used substitution $\sigma$, and it corresponds to a variable $q$ which 
    occurs only once in $\ell$ and at most once in $r_j$ for all $1 \leq j \leq m$,
    due to spareness of $\PP$ and the fact that we started with a basic term.
    Hence, we can use the same construction as in \Cref{lemma-eq-AST-iAST-1}.
\end{proof}

For \Cref{lemma:spareness-AST-proof-1} we need some more auxiliary functions from \cite{fuhs2019transformingdctorc} to
decode a basic term over $\Sigma \cup \Sigma_{\C{G}(\PP)}$ into the original term over $\Sigma$.

\begin{defi}[Constructor Variant, Basic Variant, Decoded Variant]
    Let $\PP$ be a PTRS over the signature $\Sigma$. 
    For a term $t \in \TT$, we define its \emph{constructor variant} $\cv(t)$ inductively as follows: 
    \begin{itemize}
        \item[$\bullet$] $\cv(x) = x$ for $x \in \VSet$
        \item[$\bullet$] $\cv(f(t_1, \ldots, t_n)) = f(\cv(t_1), \ldots, \cv(t_n))$ for $f \in \Sigma_C$
        \item[$\bullet$] $\cv(f(t_1, \ldots, t_n)) = \tcons_f(\cv(t_1), \ldots, \cv(t_n))$ for $f \in \Sigma_D$
    \end{itemize}
    For a term $t \in \TT$ with $t = f(t_1, \ldots, t_n)$, we define its \emph{basic variant} $\bv(f(t_1, \ldots, t_n)) = \tenc_f(\cv(t_1), \ldots, \cv(t_n))$.
    For a term $t \in \TSet{\Sigma \cup \Sigma_{\C{G}(\PP)}}{\VSet}$, we define its \emph{decoded variant} $\dv(t) \in \TT$ as follows: 
    \begin{itemize}
        \item[$\bullet$] $\dv(x) = x$ for $x \in \VSet$
        \item[$\bullet$] $\dv(\targenc(t)) = \dv(t)$
        \item[$\bullet$] $\dv(f(t_1, \ldots, t_n)) = g(\dv(t_1), \ldots, \dv(t_n))$ for $f \in \{g, \tcons_g, \tenc_g\}$ with $g \in \Sigma_D$
        \item[$\bullet$] $\dv(f(t_1, \ldots, t_n)) = f(\dv(t_1), \ldots, \dv(t_n))$ for $f \in \Sigma_C$
    \end{itemize}
\end{defi}

The only difference to the auxiliary functions from \cite{fuhs2019transformingdctorc} is that $\dv$ now also removes the $\targenc$ symbols from a term.
This was handled differently in \cite{fuhs2019transformingdctorc} in order to ensure
$|\!\dv(t)| = |t|$, but this is irrelevant for our proofs regarding $\PSNf$.

\BasicASTToAST*

\begin{proof}
    
    \noindent   
    ``$\Longleftarrow$''

    \noindent   
    Let $\F{T} = (N,E,L)$ be an $\fto_{\PP}$-RST
    whose root is labeled with $(1:t)$ for some $t \in \TT$.
    We construct an $(\PP \cup \C{G}(\PP))$-RST $\F{T}' = (N',E',L')$ whose root is labeled with $(1:\bv(t))$ where $|\F{T}'| = |\F{T}|$ and $\edl(\F{T}') \geq \edl(\F{T})$.
    
    In \cite{fuhs2019transformingdctorc} it was shown that 
    $\{1 : \bv(t)\} \iliftto_{\C{G}(\PP)}^* \{1: t \delta_t\}$.
    (PTRSs were not considered in
    \cite{fuhs2019transformingdctorc}, but since all the
    probabilities in the rules of $\C{G}(\PP)$ are trivial, the proof
    in \cite{fuhs2019transformingdctorc}
    directly translates to the probabilistic setting.)
    Here, for any term $t \in \TT$ the substitution $\delta_t$ is defined by $\delta_t(x) = \targenc(x)$ if $x \in \Var(t)$ and $\delta_t(x) = x$ otherwise.
    The $(\PP \cup \C{G}(\PP))$-RST $\F{T}'$ first performs these innermost rewrite steps to get from $\bv(t)$ to $t \delta_t$, and then we can mirror the rewrite steps from $\F{T}$.
    To be precise, we use the same underlying tree structure and an adjusted labeling such that $p_x^{\F{T}} = p_x^{\F{T}'}$ and $t_x^{\F{T}} \delta_t = t_x^{\F{T}'}$ for all $x \in N$.
    Since the tree structure and the probabilities are the same, we then obtain $|\F{T}| = |\F{T}'|$.
    \begin{center}
        \scriptsize
        \begin{tikzpicture}
            \tikzstyle{adam}=[thick,draw=black!100,fill=white!100,minimum size=4mm, shape=rectangle split, rectangle split parts=2,rectangle split
            horizontal]
            \tikzstyle{empty}=[rectangle,thick,minimum size=4mm]
            
            \node[adam] at (-3.5, 0)  (a) {$1$ \nodepart{two} $t$};
            \node[adam] at (-5, -0.8)  (b) {$p_1$ \nodepart{two} $t_{1}$};
            \node[adam] at (-2, -0.8)  (c) {$p_2$ \nodepart{two} $t_{2}$};
            \node[adam] at (-6, -1.6)  (d) {$p_3$ \nodepart{two} $t_3$};
            \node[adam] at (-4, -1.6)  (e) {$p_4$ \nodepart{two} $t_4$};
            \node[adam] at (-2, -1.6)  (f) {$p_5$ \nodepart{two} $t_5$};
            \node[empty] at (-6, -2.4)  (g) {$\ldots$};
            \node[empty] at (-4, -2.4)  (h) {$\ldots$};
            \node[empty] at (-2, -2.4)  (i) {$\ldots$};

            \node[empty] at (-0.5, -1.2)  (arrow) {\Huge $\leadsto$};
            
            \node[adam] at (3.5, 1.6)  (init) {$1$ \nodepart{two} $\bv(t)$};
            \node[empty] at (3.5, 0.8)  (init2) {$\ldots$};
            \node[adam] at (3.5, 0)  (a2) {$1$ \nodepart{two} $t \delta_t$};
            \node[adam] at (2, -0.8)  (b2) {$p_1$ \nodepart{two} $t_1 \delta_t$};
            \node[adam] at (5, -0.8)  (c2) {$p_2$ \nodepart{two} $t_2 \delta_t$};
            \node[adam] at (1, -1.6)  (d2) {$p_3$ \nodepart{two} $t_3 \delta_t$};
            \node[adam] at (3, -1.6)  (e2) {$p_4$ \nodepart{two} $t_4 \delta_t$};
            \node[adam] at (5, -1.6)  (f2) {$p_5$ \nodepart{two} $t_5 \delta_t$};
            \node[empty] at (1, -2.4)  (g2) {$\ldots$};
            \node[empty] at (3, -2.4)  (h2) {$\ldots$};
            \node[empty] at (5, -2.4)  (i2) {$\ldots$};
        
            \draw (a) edge[->] (b);
            \draw (a) edge[->] (c);
            \draw (b) edge[->] (d);
            \draw (b) edge[->] (e);
            \draw (c) edge[->] (f);
            \draw (d) edge[->] (g);
            \draw (e) edge[->] (h);
            \draw (f) edge[->] (i);

            \draw (init) edge[->] (init2);
            \draw (init2) edge[->] (a2);
            \draw (a2) edge[->] (b2);
            \draw (a2) edge[->] (c2);
            \draw (b2) edge[->] (d2);
            \draw (b2) edge[->] (e2);
            \draw (c2) edge[->] (f2);
            \draw (d2) edge[->] (g2);
            \draw (e2) edge[->] (h2);
            \draw (f2) edge[->] (i2);
        \end{tikzpicture}
    \end{center}

    \noindent
    ``$\implies$''

    \noindent 
    Let $\F{T} = (N,E,L)$ be an $(\PP \cup \C{G}(\PP))$-RST 
    whose root is labeled with $(1:t)$ for some term $t \in \TSet{\Sigma \cup \Sigma_{\C{G}(\PP)}}{\VSet}$.
    We construct an $\fto_{\PP}$-RST $\F{T}' = (N',E',L')$ inductively such that for all leaves
    $x$ of $\F{T}'$ during the construction there exists a node $\varphi(x)$ of $\F{T}$ such that
    \begin{equation}\label{claimBasic} 
        t_{x}^{\F{T}'} = \dv(t_{\varphi(x)}^{\F{T}}) \; \text{ and } \;
        p_{x}^{\F{T}'} = p_{\varphi(x)}^{\F{T}}.
    \end{equation}
    Here, $\varphi$ is injective, i.e., every
    leaf $x$ of $\F{T}'$ is mapped to a (unique) node $\varphi(x)$ of $\F{T}$.
    Furthermore, after this construction, if $x$ is still a leaf in $\F{T}'$, then $\varphi(x)$ is also a leaf in $\F{T}$.
    Hence, we obtain
    $|\F{T}| = \sum_{x \in \ctleaf^{\F{T}}}p_x^{\F{T}} \geq \sum_{x \in \ctleaf^{\F{T}'}}p_x^{\F{T}'} = |\F{T}'|$.
    Moreover, we only add rewrite steps in the beginning, so that $\edl(\F{T}) \leq \edl(\F{T}')$.
    
    We label the root of $\F{T}'$ with $(1:\dv(t))$.
    By letting $\varphi$ map the root of  $\F{T}'$  to the root of $\F{T}$, the
    claim \eqref{claimBasic} is clearly satisfied.
    As long as there is still a node $x$ in $\F{T}'$ such that $\varphi(x)$ is not a leaf in $\F{T}$, we do the following.
    If we perform a rewrite step with $\C{G}(\PP)$ at node $\varphi(x)$, i.e.,
    $t_{\varphi(x)}^{\F{T}} \fto_{\C{G}(\PP)} \{1: t_{y}^{\F{T}}\}$ for the only successor
    $y$ of $\varphi(x)$, then we have $\dv(t_{\varphi(x)}^{\F{T}}) = \dv(t_{y}^{\F{T}})$,
    i.e., we do nothing in $\F{T}'$ but simply
    change the definition of $\varphi$ such that 
    $\varphi(x)$ is now $y$.
    To see why $\dv(t_{\varphi(x)}^{\F{T}}) = \dv(t_{y}^{\F{T}})$ holds, note that we have $\dv(\ell) = \dv(r)$ for all rules $\ell \to \{1:r\} \in \C{G}(\PP)$, since $\dv(\tenc_f(x_1, \ldots, x_n)) = f(x_1, \ldots, x_n) = \dv(f(\targenc(x_1), \ldots, \targenc(x_n)))$, and analogously $\dv(\targenc(\tcons_f(x_1, \ldots, x_n))) = f(x_1, \ldots, x_n) = \dv(f(\targenc(x_1), \ldots, \targenc(x_n)))$.
    This can then be lifted to arbitrary rewrite steps.

    Otherwise, let $\varphi(x)E = \{y_1, \ldots, y_k\}$ be the set of successors of
    $\varphi(x)$ in $\F{T}$, and we have $t_{\varphi(x)}^{\F{T}} \fto_{\PP} \{\tfrac{p_{y_1}^{\F{T}}}{p_{\varphi(x)}^{\F{T}}}:t_{y_1}^{\F{T}}, \ldots, \tfrac{p_{y_k}^{\F{T}}}{p_{\varphi(x)}^{\F{T}}}:t_{y_k}^{\F{T}}\}$.
    Since $t_{x}^{\F{T}'} = \dv(t_{\varphi(x)}^{\F{T}})$ and $p_{x}^{\F{T}'} =
    p_{\varphi(x)}^{\F{T}}$ by the induction hypothesis, 
    then we also have
    $t_x^{\F{T}'} \fto_{\PP} \{\tfrac{p_{y_1}^{\F{T}'}}{p_{x}^{\F{T}'}}:t_{y_1}^{\F{T}'}, \ldots, \tfrac{p_{y_k}^{\F{T}'}}{p_{x}^{\F{T}'}}:t_{y_k}^{\F{T}'}\}$
    for terms $t_{y_j}^{\F{T}'}$ with  $t_{y_j}^{\F{T}'} = \dv(t_{y_j}^{\F{T}})$ and for
    $p_{y_j}^{\F{T}'} = p_{y_j}^{\F{T}}$. Thus, when defining $\varphi(y_j) = y_j$ for all $1 \leq j \leq k$, 
    \eqref{claimBasic} is satisfied for the new leaves $y_1,\ldots,y_k$ of $\F{T}'$.
\end{proof}

\BasicPASTToPAST*

\begin{proof} 
    Let $\F{T}$ be an $\fto_{\PP}$-RST
    whose root is labeled with $(1:t)$ for some $t \in \TT$.
    The construction from the previous proof yields
 an $(\PP \cup \C{G}(\PP))$-RST $\F{T}'$ whose root is labeled with $(1:\bv(t))$
 where $|\F{T}'| = |\F{T}|$.
     Since the tree structure and the probabilities are the same, 
    and we only add rewrite steps in the beginning, we get 
    $\edl(\F{T}') \geq
    \edl(\F{T})$.  
\end{proof}

\subsection{Proofs for \Cref{Modularity}}\label{appendix:modularity}

We start by proving the \emph{cutting lemma} that is needed to prove our
modularity results for $\ASTs$.
It states that if there exists an $\to$-RST $\F{T}$ that converges with probability~$<1$
and a partitioning of its inner nodes 
into two sets $N_1$ and $N_2$ such that
every subtree of $\F{T}$ that only contains inner nodes from $N_1$ converges with probability $1$,
then we can create a subtree of $\F{T}$ that converges with probability~$<1$ as well 
such that every infinite path contains an infinite number of $N_2$ nodes 
(i.e., it does not contain any
infinite path with eventually only nodes from $N_1$).

\begin{lem}[Cutting Lemma]\label{lemma:p-partition}
    Let $\to \; \subseteq \TT \times \FDist(\TT)$,
    and let $\F{T}$ be an $\to$-RST with $|\F{T}| < 1$.
    Suppose we can partition its inner nodes into
    $N_1 \uplus N_2$ 
    such that $|\F{T}_1| = 1$ holds for every subtree $\F{T}_1$ of $\F{T}$
    which only contains inner nodes from $N_1$.
    Then there exists a subtree $\F{T}'$ of $\F{T}$ with $|\F{T}'| < 1$ 
    such that every infinite path of $\F{T}'$ has an infinite number of nodes from $N_2$.
\end{lem}
  
\begin{proof}
    Let $\F{T} = (N,E,L)$ be an $\to$-RST with $|\F{T}| = c < 1$ for some $c \in \IR$.
    Since we have $0 \leq c < 1$, there is an $\varepsilon > 0$ such that $c + \varepsilon < 1$.
    Remember that the formula for the geometric series is:
    \[
        \sum_{n = 1}^{\infty} \left(\tfrac{1}{d}\right)^n = \tfrac{1}{d-1}, \text{ for all } d \in \IR \text{ such that } \tfrac{1}{|d|} < 1
    \]
    Let $d = \frac{1}{\varepsilon} + 2$. Then we have $\frac{1}{d} = \frac{1}{\frac{1}{\varepsilon} + 2} < 1$ and:
    \begin{equation}\label{eq:setting-d}
        \tfrac{1}{\varepsilon} + 1 < \tfrac{1}{\varepsilon} + 2 \Leftrightarrow \tfrac{1}{\varepsilon} + 1 < d \Leftrightarrow \tfrac{1}{\varepsilon} < d-1 \Leftrightarrow \tfrac{1}{d-1} < \varepsilon \Leftrightarrow \sum_{n = 1}^{\infty} \left(\tfrac{1}{d}\right)^n < \varepsilon
    \end{equation}
    We will now construct a subtree $\F{T}' = (N',E',L')$ such that every infinite path has an infinite number of $N_2$ nodes and such that 
    \begin{equation}\label{eq:sum-after-all-cuts}
        |\F{T}'| \leq |\F{T}| + \sum_{n = 1}^{\infty} \left(\tfrac{1}{d}\right)^n
    \end{equation}
    and then, we finally have
    \[
        |\F{T}'| \stackrel{\eqref{eq:sum-after-all-cuts}}{\leq} |\F{T}| + \sum_{n = 1}^{\infty} \left(\tfrac{1}{d}\right)^n = c + \sum_{n = 1}^{\infty} \left(\tfrac{1}{d}\right)^n \stackrel{\eqref{eq:setting-d}}{<} c + \varepsilon < 1
    \]
    
    The idea of this construction is that we cut infinite subtrees of pure $N_1$ nodes as soon as the probability for normal forms is high enough.
    In this way, one obtains paths where after finitely many $N_1$ nodes, there is a $N_2$ node, or we reach a leaf.

    The construction works as follows.
    For any node $x \in N$, let $\ctlevelTwo(x)$ be the number of $N_2$ nodes in the path from the root to $x$.
    Furthermore, for any set $W \subseteq N$ and $k \in \IN$, let
    $\ctlevelTwowithborder(W,k) = \{x \in W \mid \ctlevelTwo(x) \leq k \lor (x
    \in N_2 \land \ctlevelTwo(x) \leq k+1)\}$ be the set of all nodes in $W$ that have at most $k$ nodes from $N_2$ in the path from the root to its predecessor.
    So if $x \in \ctlevelTwowithborder(W,k)$ is not in $N_2$, then we have at most $k$ nodes from $N_2$ in the
    path from the root to $x$
    and if $x \in \ctlevelTwowithborder(W,k)$ is in $N_2$, 
    then we have at most $k+1$ nodes from $N_2$ in the path from the root to $x$.
    We will inductively define a set $U_k \subseteq N$ such that $U_k \subseteq
    \ctlevelTwowithborder(N,k)$ and then define the subtree as
    $\F{T}' = \F{T}[\bigcup_{k \in \IN} U_k]$.

    We start by considering the subtree $\F{T}_0 = \F{T}[\ctlevelTwowithborder(N,0)]$.
    This tree only contains inner nodes from $N_1$.
    While the node set $\ctlevelTwowithborder(N,0)$ itself may contain nodes from $N_2$, 
    they can only occur at the leaves of $\F{T}_0$.
    Using the prerequisite of the lemma, we get $|\F{T}_0|=1$.
    In \cref{Possibilities for Te} one can see the different possibilities for $\F{T}_0$.
    Either $\F{T}_0$ is finite or $\F{T}_0$ is infinite.
    In the first case, we can add all the nodes to $U_0$ since there is no infinite path of pure $N_1$ nodes.
    Hence, we define $U_0 = \ctlevelTwowithborder(N,0)$.
    In the second case, we have to cut the tree at a specific depth once the probability of leaves is high enough.
    Let $\ctdepth_{0}(y)$ be the depth of the node $y$ in the tree $\F{T}_0$.
    Moreover, let $D_{0}(k) = \{x \in \ctlevelTwowithborder(N,0) \mid \ctdepth_{0}(y) \leq k\}$ be the set of nodes in $T_0$ that have a depth of at most $k$.
    Since $|\F{T}_0|=1$ and $|\_|$ is monotonic w.r.t.\ the depth of
    the tree $\F{T}_0$, we can find an $N_{0} \in \IN$ such that
    \[
        \sum_{x \in \ctleaf^{\F{T}_0}, d_{0}(x) \leq N_{0}} p_x^{\F{T}_0} \geq  1 - \frac{1}{d}
    \]
            
    We include all nodes from $D_{0}(N_{0})$ in $U_0$ and delete every other node of $\F{T}_0$.
    In other words, we cut the tree after depth $N_{0}$.
    This cut can be seen in \cref{Possibilities for Te}, indicated by the dotted line.
    We now know that this cut may increase the probability of leaves by at most $\frac{1}{d}$.
    Therefore, we define $U_0 = D_{0}(N_{0})$ in this case.

    \begin{figure}
        \centering
        \begin{subfigure}[b]{0.4\textwidth}
            \centering
            \begin{tikzpicture}
                \tikzstyle{adam}=[circle,thick,draw=black!100,fill=white!100,minimum size=3mm]
                \tikzstyle{empty}=[circle,thick,minimum size=3mm]
                
                \node[adam] at (0, 0)  (a) {};
                \node[adam, label=center:{\tiny $N_1$}] at (2, -3)  (b) {};
                \node[adam, label=center:{\tiny $N_2$}] at (-2, -3)  (c) {};
                \node[adam, label=center:{\tiny $N_2$}] at (0.75, -3)  (d) {};
                \node[adam, label=center:{\tiny $N_1$}] at (-0.75, -3)  (e) {};
                \node[empty, label=center:{\small $N_1$}] at (0, -1.5)  (middleA) {};
                
                \node[adam, label=center:{\tiny $\NF$}] at (-1.5, -5)  (nf1) {};
                \node[adam, label=center:{\tiny $\NF$}] at (0, -5)  (nf2) {};
                
                \node[adam, label=center:{\tiny $\NF$}] at (2.75, -5)  (bb) {};
                \node[adam, label=center:{\tiny $N_2$}] at (1.25, -5)  (bc) {};
                \node[adam, label=center:{\tiny $N_2$}] at (2, -5)  (bd) {};
                \node[empty, label=center:{\small $N_1$}] at (2, -4)  (middleA) {};
    
                \node[empty, label=center:{\small $N_1$}] at (-0.75, -4)  (middleA) {};
                
                \node[empty] at (0, -6)  (stretch) {};
            
                \draw (a) edge[-] (b);
                \draw (b) edge[-] (d);
                \draw (d) edge[-] (e);
                \draw (e) edge[-] (c);
                \draw (a) edge[-] (c);
                
                \draw (b) edge[-] (bb);
                \draw (bb) edge[-] (bd);
                \draw (bd) edge[-] (bc);
                \draw (b) edge[-] (bc);
                
                \draw (e) -- (nf1) -- (nf2) -- (e);
                
                \begin{scope}[on background layer]
                \fill[green!20!white,on background layer] (0, 0) -- (-2, -3) -- (2, -3);
                \fill[green!20!white,on background layer] (2, -3) -- (1.25, -5) -- (2.75, -5);
                \fill[green!20!white,on background layer] (-0.75, -3) -- (-1.5, -5) -- (0, -5);
                \end{scope}
            \end{tikzpicture}
            \caption{$\F{T}_x$ finite}
        \end{subfigure}
        \hspace{30px}
        \begin{subfigure}[b]{0.4\textwidth}
            \centering
            \begin{tikzpicture}
                \tikzstyle{adam}=[circle,thick,draw=black!100,fill=white!100,minimum size=3mm]
                \tikzstyle{empty}=[circle,thick,minimum size=3mm]
                
                \node[adam] at (0, 0)  (a) {};
                \node[adam, label=center:{\tiny $N_1$}] at (2, -3)  (b) {};
                \node[adam, label=center:{\tiny $N_2$}] at (-2, -3)  (c) {};
                \node[adam, label=center:{\tiny $N_2$}] at (0.75, -3)  (d) {};
                \node[adam, label=center:{\tiny $N_1$}] at (-0.75, -3)  (e) {};
                \node[empty, label=center:{\small $N_1$}] at (0, -1.5)  (middleA) {};
                
                \node[adam, label=center:{\tiny $\NF$}] at (-1.5, -5)  (nf1) {};
                \node[adam, label=center:{\tiny $\NF$}] at (0, -5)  (nf2) {};
                
                \node[adam, label=center:{\tiny $N_1$}] at (2.75, -5)  (bb) {};
                \node[empty] (bbi) at (2.75,-6)  {};
                \node[adam, label=center:{\tiny $N_2$}] at (1.25, -5)  (bc) {};
                \node[adam, label=center:{\tiny $N_1$}] at (2, -5)  (bd) {};
                \node[empty] (bdi) at (2,-6)  {};
                \node[empty, label=center:{\small $N_1$}] at (2, -4)  (middleA) {};
    
                \node[empty, label=center:{\small $N_1$}] at (-0.75, -4)  (middleA) {};

                \node[empty, label=center:{\small \textcolor{red}{$N_{x}$}}] at (3, -3.7)  (cut) {};
            
                \draw (a) edge[-] (b);
                \draw (b) edge[-] (d);
                \draw (d) edge[-] (e);
                \draw (e) edge[-] (c);
                \draw (a) edge[-] (c);
                
                \draw (b) edge[-] (bb);
                \draw (bb) edge[-] (bd);
                \draw (bd) edge[-] (bc);
                \draw (b) edge[-] (bc);
                
                \draw (e) -- (nf1) -- (nf2) -- (e);
                
                \draw[] (bb) edge ($(bb)!0.3cm!(bbi)$) edge [dotted] ($(bb)!0.7cm!(bbi)$);
                \draw[] (bd) edge ($(bd)!0.3cm!(bdi)$) edge [dotted] ($(bd)!0.7cm!(bdi)$);

                \draw[] (-2, -3.7) edge [red, dotted] (cut);
                
                \begin{scope}[on background layer]
                \fill[green!20!white,on background layer] (0, 0) -- (-2, -3) -- (2, -3);
                \fill[green!20!white,on background layer] (2, -3) -- (1.25, -5) -- (2.75, -5);
                \fill[green!20!white,on background layer] (-0.75, -3) -- (-1.5, -5) -- (0, -5);
                \end{scope}
            \end{tikzpicture}
            \caption{$\F{T}_x$ infinite}
        \end{subfigure}
        \caption{Possibilities for $\F{T}_x$}\label{Possibilities for Te}
    \end{figure}
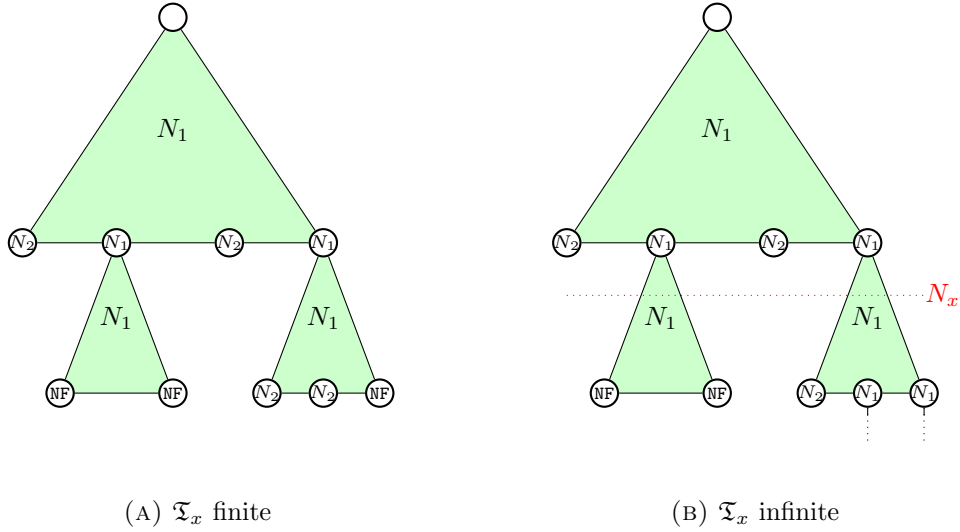

    For the induction step, assume that we have already defined a subset $U_i \subseteq \ctlevelTwowithborder(N,i)$.
    Let $H_i = \{x \in U_i \mid x \in N_2, \ctlevelTwo(x) = i+1\}$ be the set of leaves in $\F{T}[U_i]$ that are in $N_2$.
    For each $x \in H_i$, we consider the subtree that starts at $x$ until we
    reach the next node from $N_2$, including the node itself.
    Everything below such a node will be cut.
    To be precise, we regard the tree $\F{T}_x = (N_x,E_x,L_x) =
    \F{T}[\ctlevelTwowithborder(xE^*,i+1)]$. Here, $xE^*$ is the set of all nodes that
    are reachable from $x$ by arbitrary many steps.
    
    First, we show that $|\F{T}_x| = 1$.
    For every direct successor $y$ of $x$, the subtree $\F{T}_y = \F{T}_x[y E_x^*]$ of $\F{T}_x$ 
    that starts at $y$ does not contain any inner nodes from $N_2$.
    Hence, we have $|\F{T}_y| = 1$ by the prerequisite of the lemma, and hence
    \[
        |\F{T}_x| = \sum_{y \in xE} p_y \cdot |\F{T}_y| = \sum_{y \in xE} p_y \cdot 1 = \sum_{y \in xE} p_y = 1.
    \]
    For the construction of $U_{i+1}$, we have the same cases as before, see \Cref{Possibilities for Te}.
    Either $\F{T}_x$ is finite or $\F{T}_x$ is infinite.
    Let $Z_{x}$ be the set of nodes that we want to add to our node set $U_{i+1}$ from the tree $\F{T}_x$.
    In the first case we can add all the nodes again and set $Z_{x} = N_x$.
    In the second case, we once again cut the tree at a specific depth once the probability for leaves is high enough.
    Let $\ctdepth_{x}(z)$ be the depth of the node $z$ in the tree $\F{T}_x$.
    Moreover, let $D_{x}(k) = \{x \in N_x \mid \ctdepth_{x}(z) \leq k\}$ be the set of nodes in $\F{T}_x$ that have a depth of at most $k$.
    Since $|\F{T}_x|=1$ and $|\_|$ is monotonic w.r.t.\ the depth of the tree $\F{T}_x$, we can find an $N_x \in \IN$ such that
    \[
        \sum_{y \in \ctleaf^{\F{T}_x}, d_{x}(y) \leq N_{x}} p_y^{\F{T}_x} \geq 1 - \left(\frac{1}{d}\right)^{i+1} \cdot \frac{1}{|H_i|}
    \]
    We include all nodes from $D_{x}(N_{x})$ in $U_{i+1}$ and delete every other node of $\F{T}_x$.
    In other words, we cut the tree after depth $N_{x}$.
    We now know that this cut may increase the probability of leaves by at most $\left(\frac{1}{d}\right)^{i+1} \cdot \frac{1}{|H_i|}$.
    Therefore, we set $Z_{x} = D_{x}(N_{x})$.
    
    We do this for each $x \in H_i$ and in the end, we set
        $U_{i+1} = U_i \cup \bigcup_{x \in H} Z_{x}$.

    It is straightforward to see that $\bigcup_{k \in \IN} U_k$ satisfies the conditions of \Cref{def:chain-tree-induced-sub}, as we only cut after certain nodes in our construction.
    Hence, $\bigcup_{k \in \IN} U_k$ is non-empty and weakly connected, and for each
    of its nodes, it either contains no or all successors.
    Furthermore, $\F{T}' = \F{T}[\bigcup_{k \in \IN} U_k]$ is a subtree which
    does not contain an infinite path of pure $N_1$ nodes as we cut every such
    path after a finite depth.
        
    It remains to prove that $|\F{T}'| \leq |\F{T}| + \sum_{n = 1}^{\infty} \left(\frac{1}{d}\right)^n$ holds.
    During the $i$-th iteration of the construction, we may increase the value of $|\F{T}|$ by the sum of all probabilities corresponding to the new leaves resulting from the cuts.
    As we cut at most $|H_i|$ trees in the $i$-th iteration and for each such tree, we added at most a total probability of $\left(\frac{1}{d}\right)^{i+1} \cdot \frac{1}{|H_i|}$ for the new leaves, the value of $|\F{T}|$ might increase by 
    \[
        |H_i| \cdot \left(\frac{1}{d}\right)^{i+1} \cdot \frac{1}{|H_i|} = \left(\frac{1}{d}\right)^{i+1}
    \]
    in the $i$-th iteration, and hence in total, we then \pagebreak[3] obtain
    \[
        |\F{T}'| \leq |\F{T}| + \sum_{n = 1}^{\infty} \left(\frac{1}{d}\right)^n,
    \]
    as desired (see \eqref{eq:sum-after-all-cuts}).
\end{proof}

With the cutting lemma, we can now prove the following lemma regarding the parallel execution of 
rewrite sequences that are $\ASTs$.

\begin{lem}[Parallel Execution Lemma for $\ASTs$]\label{lemma:parallel-execution-lemma-AST}
    Let $\PP$ be a PTRS and $s \in \IS$.\@
    Furthermore, let $q_1, \ldots, q_n \in \TT$ be terms such that for every $\sto_{\PP}$-RST $\F{T}_i$ that starts with $(1:q_i)$
    for some $1 \leq i \leq n$ we have $|\F{T}_i| = 1$. 
    Then, every $\sto_{\PP}$-RST $\F{T}$ that starts with $(1: \tc(q_1, \ldots, q_n))$ for some symbol $\tc$, 
    where we do not rewrite at the root position,
    converges with probability $1$.
\end{lem}

\begin{proof}
    By $\snotrootto_{\PP}$ we denote the restriction of $\sto_{\PP}$ that does not perform any rewrite steps at the root position.
    We prove that every $\snotrootto_{\PP}$-RST which starts with $(1:\tc(q_1, \ldots, q_n))$ for some symbol $\tc$
    converges with probability $1$.
    Note that if we rewrite a term $q_i$ to $\{p_1:q_{i,1}, \ldots, p_k:q_{i,k}\}$, then we obtain a distribution $\{p_1:\tc(q_1, \ldots, q_{i,1}, \ldots, q_n), \ldots, p_k:\tc(q_1, \ldots, q_{i,k}, \ldots, q_n)\}$.
	Now, the terms $q_j$ with $j \neq i$ occur multiple times in this distribution, and we may use different rules to rewrite them.
	Hence, the order in which we rewrite the different $q_i$ matters and cannot be
        chosen arbitrarily (as seen in Counterex.\ \ref{example:liAST-vs-iAST}). 
    Therefore, the proof is much more complex than in the non-probabilistic setting for termination.
    (This is not the case for leftmost-innermost rewriting, where an easier induction step
    than
    the following one would also be possible.)
    Let $\tored{e}{}{\PP}$ be the restriction of $\snotrootto_{\PP}$ where we can only
    rewrite on or
    below a position $1 \leq i \leq e$.
    Note that $\tored{n}{}{\PP} \;=\; \snotrootto_{\PP}$.
    By induction on $e$, we prove
that every $\tored{e}{}{\PP}$-RST which starts with $(1:\tc(q_1, \ldots, q_n))$ converges
with probability 1.

    In the base case we have $e=1$ and only allow rewrite steps on or below position $1$.
    Obviously, since every $\sto_{\PP}$-RST that starts with $(1:q_1)$ converges with probability $1$ by our assumption,
    so does every $\tored{1}{}{\PP}$-RST that starts with $(1:\tc(q_1, \ldots, q_n))$.

    In the induction step, we assume that the statement holds for $e-1 < n$, 
    i.e., every $\tored{e-1}{}{\PP}$-RST which starts with $(1:\tc(q_1, \ldots, q_n))$
    converges with probability $1$.
    Now consider an arbitrary $\tored{e}{}{\PP}$-RST $\F{T}$ that starts with $(1:\tc(q_1, \ldots, q_n))$.
    Assume for a contradiction that it converges with a probability~$<1$.
    The rest of the induction step has the following structure:
    \begin{enumerate}[label=\arabic{enumi})]
        \item We first use the cutting lemma (\Cref{lemma:p-partition}) to get rid of infinite paths that only
        use rewrite steps on or below position $e$, resulting in an $\tored{e}{}{\PP}$-RST
        tree $\F{T}'$,
which uses infinitely many $\tored{e-1}{}{\PP}$-steps in each infinite path.
        
        \item Then, for each finite height $H \in \IN$ we split the tree $\F{T}'_{H}$,
          which consists of the first $H$ layers of the tree $\F{T}'$, into multiple
          (finitely many) $\tored{e-1}{}{\PP}$-RSTs
 of height at most $H$
          that all start with $(1:\tc(q_1, \ldots, q_n))$
          Furthermore, we show that for each $H \in \IN$, at least one of these
           $\tored{e-1}{}{\PP}$-RSTs converges with low enough probability.
        \item We create an infinite, finitely-branching tree whose nodes
are labeled with RSTs. More precisely, its
nodes          at depth $H \in \IN$ are labeled with
 $\tored{e-1}{}{\PP}$-RSTs with low enough probability.
        \item Finally, we use König's Lemma to obtain an infinite path in this tree, which corresponds to an infinite 
        $\tored{e-1}{}{\PP}$-RST that starts with $(1:\tc(q_1, \ldots, q_n))$ and converges with probability~$<1$, 
        which is our desired contradiction.
    \end{enumerate}
	\medskip

    \noindent
	\textbf{\underline{1) Use the cutting lemma}}

    \noindent
    We can partition the inner
    nodes of our RST $\F{T}$ into the sets
    \begin{itemize}
      \item[$\bullet$] $N_1 := \{x \in N^{\F{T}}\setminus\ctleaf^{\F{T}} \mid$ the rewrite step at $x$ is on or below position $e\}$
      \item[$\bullet$] $N_2 := N \setminus N_1 = \{x \in N^{\F{T}}\setminus\ctleaf^{\F{T}} \mid$ the rewrite step at $x$
        is on or below position $k$ with $1 \leq k < e\}$
    \end{itemize}
    We know that every $\sto_{\PP}$-RST that starts with $(1:q_e)$ converges with probability $1$ by assumption.
    Hence, also every subtree of $\F{T}$ that only contains nodes from $N_1$ as inner
    nodes
    converges with probability $1$.
Thus, we can apply the cutting lemma (\Cref{lemma:p-partition}) and obtain a subtree $\F{T}'$ of $\F{T}$
    that converges with probability~$<1$ and contains infinitely many nodes of $N_2$ in each infinite path.
    
	\medskip

    \noindent
	\textbf{\underline{2) Split the tree $\F{T}_{H}$ for every $H \in \IN$}}

    \noindent
    Let $H \in \IN$ and let $\F{T}'_H$ be the finite tree consisting of the first $H$ layers of $\F{T}'$,
    where the subterm on position $e$ remains the same in every node and is equal to the
    subterm on
    position $e$ in the root of $\F{T}'$.
    Since $\F{T}'$ converges with probability~$<1$, there exists an $0 < \alpha \leq 1$
    such that $|\F{T}'| = \lim_{H \to \infty} \sum_{x \in \ctleaf^{\F{T}'} \land
      \ctdepth(x) \leq H} \, p_x^{\F{T}'} < 1 - \alpha$, 
    and hence $\sum_{x \in \ctleaf^{\F{T}'} \land \ctdepth(x) \leq H} \, p_x^{\F{T}'} =
    \sum_{x \in \ctleaf^{\F{T}'} \land x \in \ctleaf^{\F{T}'_H}} \, p_x^{\F{T}'_H}  < 1 - \alpha$ for every $H \in \IN$.
    Note that $\F{T}'_H$ is not a valid $\tored{e}{}{\PP}$-RST, as we now have steps where no term changes, i.e., where we would have performed a rewrite step below position $e$.

 From $\F{T}'_H$ we generate a set $\mathbb{T}_H$  of
 pairs
$(p_{\F{t}}, \F{t})$ where
 $p_{\F{t}} \in (0,1]$ is a probability and
$\F{t}$ is an $\tored{e-1}{}{\PP}$-RST  of height at most $H$
    (i.e., we do not perform any rewrite steps on or below position $e$ anymore).   
   The set $\mathbb{T}_H$ contains all $\tored{e-1}{}{\PP}$-RSTs $\F{t}$ resulting from
   $\F{T}'_H$ when skipping the nodes of $N_1$ (where the rewrite step is on or below the
   position $e$).
   Instead, one uses one of its child nodes. Moreover,
$p_{\F{t}}$ is the probability for choosing the RST $\F{t}$.
   We will give more intuition on
 $\mathbb{T}_H$ below.
   Since
a rewrite step on or below the position $e$ may create several children (one for each term
in the support of the multi-distribution on the right-hand side of the applied rewrite
rule),  the set $\mathbb{T}_H$ may contain several pairs.
    $\mathbb{T}_H$ will satisfy the following two properties:
    \begin{enumerate}[label=(Prop-\arabic{enumi})]
		\item\label{itm:const-prop-1} We have $\sum_{(p_{\F{t}},\F{t}) \in \mathbb{T}_H} p_{\F{t}} = 1$. This means that the sum of the probabilities for all trees in $\mathbb{T}_H$ sum up to one.
		
		\item\label{itm:const-prop-2} For all $x \in \ctleaf^{\F{T}'_H} \cup N_2$
                  we have $p_x^{\F{T}'_H} =
                  \sum_{(p_{\F{t}},\F{t}) \in \mathbb{T}_H} p_{\F{t}} \cdot
                  p_{x}^{\F{t}}$. As mentioned above, $p_{\F{t}}$ represents the
                  probability that the RST $\F{t}$ is chosen and $p_{x}^{\F{t}}$ is the
                  probability of the node $x$ in the RST $\F{t}$. Whenever $x$ is not a
                  node of $\F{t}$ (i.e., $x \notin N^{\F{t}}$), then we define
                  $p_{x}^{\F{t}} = 0$.
                  	This means that the probability for a leaf $x \in \ctleaf^{\F{T}'_H}$ or
                an inner node $x \in N_2$ in our cut tree $\F{T}'_H$ is equal to the sum
                over all trees $\F{t}$ that contain $x$, where we multiply the probability
                $p_{\F{t}}$ of the tree $\F{t}$ by the probability of node $x$ in $\F{t}$.
	\end{enumerate}

    Assume now that for every $(p_{\F{t}},\F{t}) \in \mathbb{T}_H$ we have $\sum_{x \in
      \ctleaf^{\F{t}} \land x \in \ctleaf^{\F{T}'}} \,  p_x^{\F{t}} \geq 1 - \alpha$.
    Then, with the two properties
\ref{itm:const-prop-1} and \ref{itm:const-prop-2} we obtain \pagebreak[3]
the following:
    \begin{equation}\label{eq:array-AST-parallel}
      \begin{array}{ccl}
        &  & \sum_{x \in \ctleaf^{\F{T}_H'} \land x \in \ctleaf^{\F{T}'}} \, p_x^{\F{T}_H'} \\
        \text{\textcolor{blue}{(by \ref{itm:const-prop-2})}} & = & \sum_{x \in
          \ctleaf^{\F{T}_H'} \land x \in \ctleaf^{\F{T}'}} \sum_{(p_{\F{t}},\F{t}) \in
          \mathbb{T}_{H}} \, p_{\F{t}} \cdot p_x^{\F{t}} \\
        & = & \sum_{(p_{\F{t}},\F{t}) \in \mathbb{T}_{H}} \sum_{x \in \ctleaf^{\F{T}_H'}
          \land x \in \ctleaf^{\F{T}'}} \, p_{\F{t}} \cdot p_x^{\F{t}}
        \\
        & = & \sum_{(p_{\F{t}},\F{t}) \in \mathbb{T}_{H}} \, p_{\F{t}} \cdot \sum_{x \in
          \ctleaf^{\F{T}_H'} \land x \in \ctleaf^{\F{T}'}} \, p_x^{\F{t}}\\
                & = & \sum_{(p_{\F{t}},\F{t}) \in \mathbb{T}_{H}} \, p_{\F{t}} \cdot
        \sum_{x \in \ctleaf^{\F{t}} \land x \in \ctleaf^{\F{T}'}}\,  p_x^{\F{t}} \\
        \text{\textcolor{blue}{(by assumption)}} & \geq & \sum_{(p_{\F{t}},\F{t}) \in
          \mathbb{T}_{H}} \, p_{\F{t}} \cdot (1 - \alpha) \\
        & = & (1 - \alpha) \cdot \sum_{(p_{\F{t}},\F{t}) \in \mathbb{T}_{H}} \, p_{\F{t}} \\
        \text{\textcolor{blue}{(by \ref{itm:const-prop-1})}} & = & (1 - \alpha) \cdot 1 \\
        & = & 1 - \alpha
      \end{array}
	  \end{equation}
    which is a contradiction.
    Hence, for each $H \in \IN$ there exists a $(p_{\F{t}},\F{t}) \in \mathbb{T}_{H}$ such that 
    \begin{equation}\label{eq:t-small-probability-every-H}
        \sum_{x \in \ctleaf^{\F{t}} \land x \in \ctleaf^{\F{T}'}} p_x^{\F{t}} < 1 - \alpha
    \end{equation}
    We will use this in Step 3) and 4) to generate our infinite $\tored{e-1}{}{\PP}$-RST that converges with probability $<1$.
    But first we have to define the set $\mathbb{T}_{H}$ that satisfies
    \ref{itm:const-prop-1} and \ref{itm:const-prop-2}. 

	  \medskip

    \noindent
	  \textbf{\underline{Idea of $\mathbb{T}_H$}}

    \noindent
    The idea of the split can be seen in \Cref{fig:example_ct_witness_proof_after_cut_with_prob,fig:example_ct_witness_proof_after_cut_corr_trees_in_split}.
    There, the probabilities of the nodes are indicated by the \textcolor{blue}{small
      numbers}
     and the probabilities for the edges are indicated by the \textcolor{red}{big numbers}.
    Whenever we encounter a node from $N_1$ with successors $y_1,\ldots,y_m$, 
    meaning that we rewrite on or below position $e$,
    then we split the tree into $m$ different trees, 
    since the terms below positions $1$ to $e-1$ did not change,
    and we may use different rules on the different resulting terms afterwards.
	  \begin{figure}
      \centering
      \begin{tikzpicture}
          \tikzstyle{adam}=[circle,thick,draw=black!100,fill=white!100,minimum size=3mm]
          \tikzstyle{empty}=[shape=rectangle,draw=black!100,thick,minimum size=4mm]
          
          \node[empty,pin={[pin distance=0.1cm, pin edge={-}] 135:\tiny \textcolor{blue}{$1$}}] at (0, 0)  (a) {};

          \node[empty,pin={[pin distance=0.1cm, pin edge={-}] 135:\tiny \textcolor{blue}{$\tfrac{1}{3}$}}] at (-3, -1)  (aa) {$e$};
          \node[empty,pin={[pin distance=0.1cm, pin edge={-}] 45:\tiny \textcolor{blue}{$\tfrac{2}{3}$}}] at (3, -1)  (ab) {};

          \node[empty,pin={[pin distance=0.1cm, pin edge={-}] 135:\tiny \textcolor{blue}{$\tfrac{1}{12}$}}] at (-4.5, -2)  (aaa) {};
          \node[empty,pin={[pin distance=0.1cm, pin edge={-}] 45:\tiny \textcolor{blue}{$\tfrac{1}{4}$}}] at (-1.5, -2)  (aab) {};
        
          \draw (a) edge[->] node[above] {$\textcolor{red}{\tfrac{1}{3}}$} (aa);
          \draw (a) edge[->] node[above] {$\textcolor{red}{\tfrac{2}{3}}$} (ab);
          \draw (aa) edge[->] node[above] {$\textcolor{red}{\tfrac{1}{4}}$} (aaa);
          \draw (aa) edge[->] node[above] {$\textcolor{red}{\tfrac{3}{4}}$} (aab);
      \end{tikzpicture}
      \caption{Example tree $\F{T}'_2$ with a step on or below position $e$ in the left node}
      \label{fig:example_ct_witness_proof_after_cut_with_prob}
	\end{figure}
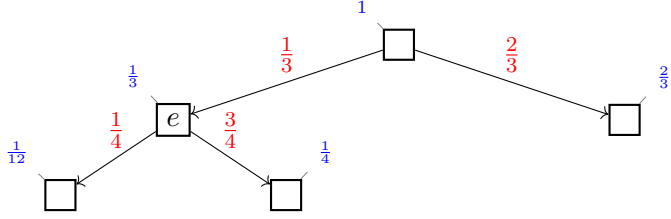
	\begin{figure}
        \centering
        \begin{subfigure}[b]{0.4\textwidth}
            \centering
            \begin{tikzpicture}
				\tikzstyle{adam}=[circle,thick,draw=black!100,fill=white!100,minimum size=3mm]
				\tikzstyle{empty}=[shape=rectangle,draw=black!100,thick,minimum size=4mm]
				\tikzstyle{realempty}=[shape=rectangle,draw=black!00,thick,minimum size=4mm]
				
				\node[empty,pin={[pin distance=0.1cm, pin edge={-}] 135:\tiny \textcolor{blue}{$1$}}] at (0, 0)  (a) {};

				\node[empty] at (-1.5, -1)  (aa) {$e$};
				\node[empty,pin={[pin distance=0.1cm, pin edge={-}]
                                    45:\tiny \textcolor{blue}{$\tfrac{2}{3}$}}] at (1.5,
                                -1)  (ab) {}; 

				\node[realempty] at (-3, 0)  (prop1) {$\tfrac{1}{4}$:};

				\node[empty,pin={[pin distance=0.1cm, pin edge={-}] 135:\tiny \textcolor{blue}{$\tfrac{1}{3}$}}] at (-2, -2)  (aaa) {};
				\node[empty] at (-1, -2)  (aab) {};
			
				\draw (a) edge[->] (aaa);
				\draw (a) edge[->] (ab);
			\end{tikzpicture}
            \caption{$\F{t}_{1}$}
        \end{subfigure}
        \hspace{70px}
        \begin{subfigure}[b]{0.4\textwidth}
            \centering
            
            \begin{tikzpicture}
				\tikzstyle{adam}=[circle,thick,draw=black!100,fill=white!100,minimum size=3mm]
				\tikzstyle{empty}=[shape=rectangle,draw=black!100,thick,minimum size=4mm]
				\tikzstyle{realempty}=[shape=rectangle,draw=black!00,thick,minimum size=4mm]
				
				\node[empty,pin={[pin distance=0.1cm, pin edge={-}] 135:\tiny \textcolor{blue}{$1$}}] at (0, 0)  (a) {};

				\node[empty] at (-1.5, -1)  (aa) {$e$};
				\node[empty,pin={[pin distance=0.1cm, pin edge={-}] 45:\tiny \textcolor{blue}{$\tfrac{2}{3}$}}] at (1.5, -1)  (ab) {};

				\node[realempty] at (-3, -0)  (prop1) {$\tfrac{3}{4}$:};
	
				\node[empty] at (-2, -2)  (aaa) {};
				\node[empty,pin={[pin distance=0.1cm, pin edge={-}] 45:\tiny \textcolor{blue}{$\tfrac{1}{3}$}}] at (-1, -2)  (aab) {};
			
				\draw (a) edge[->] (aab);
				\draw (a) edge[->] (ab);
			\end{tikzpicture}
            \caption{$\F{t}_{2}$}
        \end{subfigure}
        \caption{The set $\mathbb{T}_2 =  \{(\tfrac{1}{4},\F{t}_1), (\tfrac{3}{4}, \F{t}_2)\}$}
		\label{fig:example_ct_witness_proof_after_cut_corr_trees_in_split}
    \end{figure}
	
	\medskip
        \pagebreak[3]
        
    \noindent
	\textbf{\underline{Constructing $\mathbb{T}_H$}}

  \noindent
  To construct $\mathbb{T}_H$, we
  repeatedly remove all nodes of $N_1$ from the tree $\F{T}'_H$ 
  and create a set $M$ inductively that satisfies the following properties.
	\begin{enumerate}[label=(Ind-\arabic{enumi})]
		\item\label{itm:const-induction-1} $\sum_{(p_{\F{t}},\F{t}) \in M} p_{\F{t}} = 1$.
		\item\label{itm:const-induction-2} For all $(p_{\F{t}},\F{t}) \in M$ and
                  all inner nodes $x \in N^{\F{t}}$ with successors $y_1, \ldots, y_k$ in $\F{t}$,
                  the edge relation either represents a valid rewrite step on or
                  below a position $1 \leq j \leq e-1$, or the subterms
on or below the positions $1 \leq j \leq e-1$
                  remain the same, 
    but the probabilities of the successors still sum up to the probability of node $x$, 
    i.e., $\sum_{j=1}^{k} p_{y_j}^{\F{t}} = p_{x}^{\F{t}}$.
		\item\label{itm:const-induction-3} For all $x \in \ctleaf^{\F{T}'_H} \cup N_2$
    we have $p_x^{\F{T}'_H} = \sum_{(p_{\F{t}},\F{t}) \in M} p_{\F{t}} \cdot p_x^{\F{t}}$.
	\end{enumerate}
  We stop the construction once every $\F{t} \in M$ is a valid $\tored{e-1}{}{}$-RST,
  i.e., once we have
  removed all nodes
  where the terms do not change as we would have performed a rewrite step on or below position $e$.
  In the end, we result in a set $M$ that satisfies \ref{itm:const-induction-1}, \ref{itm:const-induction-2}, and \ref{itm:const-induction-3}, 
  and  $\F{t}$ is a valid $\tored{e-1}{}{}$-RST for every $(p_{\F{t}},\F{t}) \in M$.
	Then $M$ is our desired set, and we define $\mathbb{T}_H = M$.
        Moreover, in the end, \Cref{itm:const-prop-1} and \Cref{itm:const-prop-2} follow from
        \ref{itm:const-induction-1} and
        \ref{itm:const-induction-3}.
        
	We start with $M = \{(1,\F{T}'_H)\}$.
	Here, clearly all  three properties
        \ref{itm:const-induction-1}-\ref{itm:const-induction-3} are satisfied.
	Now, assume that there is still a pair $(p_{\F{l}}, \F{l}) \in M$ such that $\F{l}$ 
  contains a node $v \in N_1$ that is not a leaf in $\F{l}$.
	We will now split $\F{l}$ into multiple trees that do not contain $v$ 
  anymore but move directly to one of its children, as illustrated in \Cref{fig:split-construction-definition}.

	First, assume that $v$ is not the root of $\F{l}$.
	Let $vE^\F{l} = \{w_1, \ldots, w_m\}$ be the direct successors of $v$ in $\F{l}$ 
    and let $z$ be the predecessor of $v$ in $\F{l}$.
	Instead of one tree $\F{l}$ with the edges $(z,v), (v,w_1), \ldots, (v,w_m)$, 
    we split the tree into $m$ different trees $\F{l}_1, \ldots, \F{l}_m$ 
    such that for every $1 \leq h \leq m$, the tree $\F{l}_h$ contains a direct edge from $z$ to $w_h$.
	In addition to that, the unreachable nodes are removed, 
    and we also have to adjust the probabilities of all (not necessarily direct) successors of $w_h$ 
    (including $w_h$ itself) in $\F{l}_h$.
	\begin{figure}
        \centering
		\begin{tikzpicture}
			\tikzstyle{adam}=[circle,thick,draw=black!100,fill=white!100,minimum size=3mm]
			\tikzstyle{empty}=[shape=rectangle,draw=black!100,thick,minimum size=4mm]
			\tikzstyle{realempty}=[shape=rectangle,draw=black!00,thick,minimum size=4mm]
			
			\node[realempty] at (-3, 0)  (a) {$z$};

			\node[realempty] at (-3, -1)  (aa) {$v$};

			\node[realempty] at (-4, -2)  (aaa) {$w_1$};
			\node[realempty] at (-3, -2)  (aab) {$\ldots$};
			\node[realempty] at (-2, -2)  (aac) {$w_m$};

			\node[realempty] at (0, -1)  (mid) {\LARGE $\leadsto$};
			
			\node[realempty] at (3, 0)  (x) {$z$};

			\node[realempty] at (3, -1)  (xx) {$v$};

			\node[realempty] at (2, -2)  (xxx) {$w_1$};
			\node[realempty] at (3, -2)  (xxy) {$\ldots$};
			\node[realempty] at (4, -2)  (xxz) {$w_m$};

			\node[realempty] at (5.5, -1)  (dots) {$\ldots$};
			
			\node[realempty] at (8, 0)  (tx) {$z$};

			\node[realempty] at (8, -1)  (txx) {$v$};

			\node[realempty] at (7, -2)  (txxx) {$w_1$};
			\node[realempty] at (8, -2)  (txxy) {$\ldots$};
			\node[realempty] at (9, -2)  (txxz) {$w_m$};

			\draw (a) edge[->] (aa);
			\draw (aa) edge[->] (aaa);
			\draw (aa) edge[->] (aab);
			\draw (aa) edge[->] (aac);

			\draw (x) edge[->] (xxx);

			\draw (tx) edge[->] (txxz);
		\end{tikzpicture}
		\caption{Skipping inner node $v \in N_2$ to create $m$ different trees}
		\label{fig:split-construction-definition}
    \end{figure}
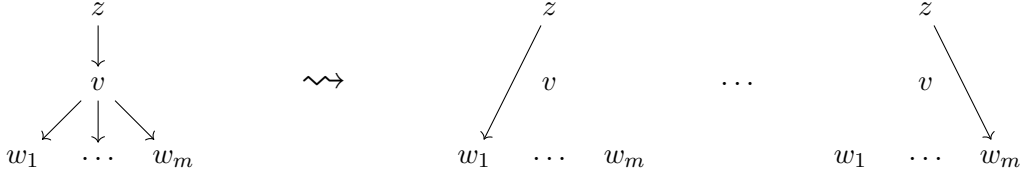
More precisely, we set $\F{l}_h = (N^{\F{l}_h}, E^{\F{l}_h}, L^{\F{l}_h})$, with 
	\begin{align*}
		N^{\F{l}_h} &= (N^\F{l} \setminus v(E^\F{l})^*) \cup w_h(E^\F{l})^*\\
		E^{\F{l}_h} &= (E^\F{l} \setminus (v(E^\F{l})^* \times v(E^\F{l})^*)) \cup \{(z,w_h)\} \cup (E^\F{l} \cap (w_h(E^\F{l})^* \times w_h(E^\F{l})^*))
	\end{align*}
	Furthermore, let $p_h = \tfrac{p_{w_h}^{\F{l}}}{p_{v}^{\F{l}}}$.
	Then, the labeling is defined by
	\[
		L^{\F{l}_h}(x) = \begin{cases}
			(\frac{1}{p_h} \cdot p_{x}^{\F{l}},t_x^{\F{l}}) & \text{ if } x \in w_h(E^\F{l})^*\\
			(p_x^{\F{l}},t_x^{\F{l}}) & \text{ otherwise}
		\end{cases}
	\]
	Note that 
	\begin{equation} \label{p_h sum is one}
		\sum_{1 \leq h \leq m} p_h = \sum_{1 \leq h \leq m} \tfrac{p_{w_h}^{\F{l}}}{p_v^{\F{l}}} = \tfrac{1}{p_v^{\F{l}}} \cdot \sum_{1 \leq h \leq m} p_{w_h}^{\F{l}} \stackrel{\Cref{itm:const-induction-2}}{=} \tfrac{1}{p_v^{\F{l}}} \cdot p_v^{\F{l}} = 1
	\end{equation}

	If $v$ is the root of $\F{l}$, then we use the same construction, but we have no predecessor $z$ of $v$ and directly start with the node $w_h$ as the new root.
	Hence, we have to use the node set $N^{\F{l}_h} = w_h(E^\F{l})^*$ and edge relation $E^{\F{l}_h} = (E \cap (w_h(E^\F{l})^* \times w_h(E^\F{l})^*))$ and the rest stays the same.

	In the end, we set 
	\[
		M' = M \setminus \{(p_\F{l},\F{l})\} \cup \{(p_\F{l} \cdot p_1,\F{l}_1), \ldots, (p_\F{l} \cdot p_m,\F{l}_m)\}
	\]
	This construction is exactly what we did in
        \cref{fig:example_ct_witness_proof_after_cut_with_prob} and
        \cref{fig:example_ct_witness_proof_after_cut_corr_trees_in_split} for the only
        node where we would have rewritten on or below position $e$.
	It remains to prove that
 \ref{itm:const-induction-1}, \ref{itm:const-induction-2}, and \ref{itm:const-induction-3}
 are still satisfied for $M'$.
	\begin{description}
		\item[\ref{itm:const-induction-1}]
		We have
		\[
			\begin{array}{ccl}
				&   & \sum_{(p_{\F{t}},\F{t}) \in M'} \, p_{\F{t}} \\
				& = & \sum_{(p_{\F{t}},\F{t}) \in M \setminus
                            \{(p_\F{l},\F{l})\}} \, p_{\F{t}} + \sum_{(p_{\F{t}},\F{t})
                            \in \{(p_{\F{l}} \cdot p_1, \F{l}_1), \ldots, (p_{\F{l}} \cdot
                            p_m,\F{l}_m)\}} \, p_{\F{t}} \\
				& = & \sum_{(p_{\F{t}},\F{t}) \in M \setminus
                            \{(p_\F{l},\F{l})\}} \, p_{\F{t}} + \sum_{1 \leq h \leq m} \, p_\F{l} \cdot p_h \\
				& = & \sum_{(p_{\F{t}},\F{t}) \in M \setminus
                            \{(p_\F{l},\F{l})\}} \, p_{\F{t}} + p_\F{l} \cdot \sum_{1 \leq
                            h \leq m} \, p_h \\
                \textcolor{blue}{\text{(by \eqref{p_h sum is one})}} & = &
                \sum_{(p_{\F{t}},\F{t}) \in M \setminus \{(p_\F{l},\F{l})\}} \, p_{\F{t}} + p_\F{l} \cdot 1\\ 
				& = & \sum_{(p_{\F{t}},\F{t}) \in M \setminus
                  \{(p_\F{l},\F{l})\}} \, p_{\F{t}} + p_\F{l}\\
				& = & \sum_{(p_{\F{t}},\F{t}) \in M} \, p_{\F{t}}\\
				& \stackrel{IH}{=} & 1
			\end{array}
		\]

		\item[\ref{itm:const-induction-2}]
		Let $1 \leq h \leq m$.
    We only split nodes whenever we would have
    rewritten on or below position $e$.
    Hence, we only have to prove that in this case the probabilities of the successors for a node add up to the probability for the node itself with our new labeling.
		We constructed $\F{l}_h$ by skipping the node $v$ and directly moving from $z$ to $w_h$ (or starting with $w_h$ if $v$ was the root node).
		Hence, $(N^{\F{l}_h}, E^{\F{l}_h})$ is still a finitely branching tree.
		Let $x \in N^{\F{l}_h}$ with $xE^{\F{l}_h} \neq \emptyset$.
		If $x \in w_h(E^\F{t})^*$, then $xE^{\F{l}_h} = xE^{\F{l}}$ and thus
		\[
			\sum_{y \in xE^{\F{l}_h}} p_y^{\F{l}_h}
			= \sum_{y \in xE^{\F{l}_h}} \tfrac{1}{p_h} \cdot p_{y}^{\F{l}}
			= \tfrac{1}{p_h} \cdot \sum_{y \in xE^{\F{l}_h}} p_{y}^{\F{l}}
			= \tfrac{1}{p_h} \cdot \sum_{y \in xE^{\F{l}}} p_{y}^{\F{l}}
			= \tfrac{1}{p_h} \cdot p_{x}^{\F{l}}
			= p_x^{\F{l}_h}
		\]
		If $v$ was not the root and $x = z$, then $zE^{\F{l}_h} = (zE^{\F{l}} \setminus \{v\}) \cup \{w_h\}$ and thus
		\begin{align*}
			\sum_{y \in zE^{\F{l}_h}} p_y^{\F{l}_h}
			&= \sum_{y \in (zE^{\F{l}} \setminus \{v\}) \cup \{w_h\}} p_{y}^{\F{l}_h}
			= \sum_{y \in (zE^{\F{l}} \setminus \{v\})} p_{y}^{\F{l}_h} + p_{w_h}^{\F{l}_h}
			= \sum_{y \in (zE^{\F{l}} \setminus \{v\})} p_{y}^{\F{l}} + \tfrac{1}{p_h} \cdot p_{w_h}^{\F{l}}\\
			&= \sum_{y \in (zE^{\F{l}} \setminus \{v\})} p_{y}^{\F{l}} + \tfrac{p_{v}^{\F{l}}}{p_{w_h}^{\F{l}}} \cdot p_{w_h}^{\F{l}}
			= \sum_{y \in (zE^{\F{l}} \setminus \{v\})} p_{y}^{\F{l}} + p_{v}^{\F{l}}
			= \sum_{y \in zE^{\F{l}}} p_{y}^{\F{l}}
			= p_z^{\F{l}}	= p_z^{\F{l}_h}
		\end{align*}
		Otherwise, we have $x \in N^{\F{l}} \setminus (v(E^\F{t})^* \cup \{z\})$.
		This means $p_y^{\F{l}_h} = p_y^{\F{l}}$ for all $y \in xE^{\F{l}_h}$ and $xE^{\F{l}_h} = xE^{\F{l}}$, and thus
		\[
			\sum_{y \in xE^{\F{l}_h}} p_y^{\F{l}_h} 
			= \sum_{y \in xE^{\F{l}_h}} p_{y}^{\F{l}} 
			= \sum_{y \in xE^{\F{l}}} p_{y}^{\F{l}} 
			= p_x^{\F{l}}
			= p_x^{\F{l}_h}
		\]
		For the last property
($p_x^{\F{l}}	= p_x^{\F{l}_h}$), note that if $v$ is not the root in $\F{l}$, then the
                root and its labeling did not change, so that we have
                $p_{\ctroot^{\F{l}_h}}^{\F{l}_h} = p_{\ctroot^{\F{l}}}^{\F{l}} = 1$, where
                $\ctroot^{\F{l}_h} = \ctroot^{\F{l}}$ is the root of $\F{l}_h$ and $\F{l}$.
		If $v$ was the root, then $w_h$ is the new root with 
		\[
			p_{w_h}^{\F{l}_h} = \frac{1}{p_h} \cdot p_{w_h}^{\F{l}} =
                        \tfrac{p_v^{\F{l}}}{p_{w_h}^{\F{l}}} \cdot p_{w_h}^{\F{l}} =
                        p_{v}^{\F{l}} = 1
		\]
		
		\item[\ref{itm:const-induction-3}]
		Let $x \in \ctleaf^{\F{T}'_H} \cup N_2$.
		If we have $x \not\in N^{\F{l}}$, then also $x \not\in N^{\F{l}_h}$ for all $1 \leq h \leq m$ and thus
		\[
			\begin{array}{cl}
				& \sum_{(p_{\F{t}},\F{t}) \in M'} \, p_{\F{t}} \cdot p_x^{\F{t}} \\
				= & \sum_{(p_{\F{t}},\F{t}) \in M \setminus
                                  \{(p_\F{l},\F{l})\} \cup \{(p_\F{l} \cdot p_1,\F{l}_1),
                                  \ldots, (p_\F{l} \cdot p_m,\F{l}_m)\}} \, p_{\F{t}} \cdot p_x^{\F{t}} \\
				= & \sum_{(p_{\F{t}},\F{t}) \in M \setminus
                                  \{(p_\F{l},\F{l})\}} \, p_{\F{t}} \cdot p_x^{\F{t}} +
                                \sum_{(p_{\F{t}},\F{t}) \in \{(p_\F{l} \cdot p_1,\F{l}_1),
                                  \ldots, (p_\F{l} \cdot p_m,\F{l}_m)\}} \, p_{\F{t}} \cdot p_x^{\F{t}} \\
				= & \sum_{(p_{\F{t}},\F{t}) \in M \setminus
                                  \{(p_\F{l},\F{l})\}} \, p_{\F{t}} \cdot p_x^{\F{t}} +
                                \sum_{(p_{\F{t}},\F{t}) \in \{(p_\F{l} \cdot p_1,\F{l}_1),
                                  \ldots, (p_\F{l} \cdot p_m,\F{l}_m)\}} \, p_{\F{t}} \cdot 0 \\
				= & \sum_{(p_{\F{t}},\F{t}) \in M \setminus
                                  \{(p_\F{l},\F{l})\}} \, p_{\F{t}} \cdot p_x^{\F{t}} \\
				= & \sum_{(p_{\F{t}},\F{t}) \in M \setminus
                                  \{(p_\F{l},\F{l})\}} \, p_{\F{t}} \cdot p_x^{\F{t}} + p_{\F{t}} \cdot 0\\
				= & \sum_{(p_{\F{t}},\F{t}) \in M \setminus
                                  \{(p_\F{l},\F{l})\}} \, p_{\F{t}} \cdot p_x^{\F{t}} + p_{\F{t}} \cdot p_x^{\F{l}}\\
				= & \sum_{(p_{\F{t}},\F{t}) \in M} \, p_{\F{t}} \cdot p_x^{\F{t}}\\
				\stackrel{IH}{=} & p_x^{\F{T}'_H}
			\end{array}
		\]
		If we have $x \in N^{\F{l}}$ and $x \not\in v(E^{\F{l}})^*$, then $x \in
                N^{\F{l}_h}$ for all $1 \leq h \leq m$ and $p_x^{\F{l}} = p_x^{\F{l}_h}$. Hence
		\[
			\begin{array}{cl}
				& \sum_{(p_{\F{t}},\F{t}) \in M'} \, p_{\F{t}} \cdot p_x^{\F{t}} \\
				= & \sum_{(p_{\F{t}},\F{t}) \in M \setminus
                                  \{(p_\F{l},\F{l})\} \cup \{(p_\F{l} \cdot p_1,\F{l}_1),
                                  \ldots, (p_\F{l} \cdot p_m,\F{l}_m)\}} \, p_{\F{t}} \cdot p_x^{\F{t}} \\
				= & \sum_{(p_{\F{t}},\F{t}) \in M \setminus
                                  \{(p_{\F{l}},\F{l})\}} \, p_{\F{t}} \cdot p_x^{\F{t}} +
                                \sum_{(p_{\F{t}},\F{t}) \in \{(p_{\F{l}} \cdot
                                  p_1,{\F{l}}_1), \ldots, (p_{\F{l}} \cdot
                                  p_m,{\F{l}}_m)\}} \, p_{\F{t}} \cdot p_x^{\F{t}} \\
				= & \sum_{(p_{\F{t}},\F{t}) \in M \setminus
                                  \{(p_{\F{l}},\F{l})\}} \, p_{\F{t}} \cdot p_x^{\F{t}} +
                                \sum_{1 \leq h \leq m} \, p_{\F{l}} \cdot p_h \cdot p_x^{\F{l}_h} \\
				= & \sum_{(p_{\F{t}},\F{t}) \in M \setminus
                                  \{(p_{\F{l}},\F{l})\}} \, p_{\F{t}} \cdot p_x^{\F{t}} +
                                \sum_{1 \leq h \leq m} \, p_{\F{l}} \cdot p_h \cdot p_x^{\F{l}} \\
				= & \sum_{(p_{\F{t}},\F{t}) \in M \setminus
                                  \{(p_{\F{l}},\F{l})\}} \, p_{\F{t}} \cdot p_x^{\F{t}} +
                                p_{\F{l}} \cdot p_x^{\F{l}} \cdot \sum_{1 \leq h \leq m}
                                \, p_h\\
				= & \sum_{(p_{\F{t}},\F{t}) \in M \setminus
                                  \{(p_{\F{l}},\F{l})\}} \, p_{\F{t}} \cdot p_x^{\F{t}} + p_{\F{l}} \cdot p_x^{\F{l}} \cdot 1\\
				= & \sum_{(p_{\F{t}},\F{t}) \in M} \, p_{\F{t}} \cdot p_x^{\F{t}}\\
				\stackrel{IH}{=} & p_x^{\F{T}'_H}
			\end{array}
		\]
		Otherwise we have $x \in N^{\F{l}}$ and $x \in v(E^{\F{l}})^*$.
		This means that we have $x \in N^{\F{l}_h}$ and $x \in w_h(E^{\F{l}})^*$ for some $1 \leq h \leq m$ and $x \not\in N^{\F{l}_{h'}}$ for all $h' \neq h$.
		Furthermore, we have $p_x^{\F{l}_h} = \frac{1}{p_h} \cdot p_x^{\F{l}}$, and hence
		\[
			\begin{array}{cl}
				& \sum_{(p_{\F{t}},\F{t}) \in M'} \, p_{\F{t}} \cdot p_x^{\F{t}} \\
				= & \sum_{(p_{\F{t}},\F{t}) \in M \setminus
                                  \{(p_\F{l},\F{l})\} \cup \{(p_\F{l} \cdot p_1,\F{l}_1),
                                  \ldots, (p_\F{l} \cdot p_m,\F{l}_m)\}} \, p_{\F{t}} \cdot p_x^{\F{t}} \\
				= & \sum_{(p_{\F{t}},\F{t}) \in M \setminus
                                  \{(p_{\F{t}},\F{t})\}} \, p_{\F{t}} \cdot p_x^{\F{t}} +
                                \sum_{(p_{\F{t}},\F{t}) \in \{(p_\F{l} \cdot p_1,\F{l}_1),
                                  \ldots, (p_\F{l} \cdot p_m,\F{l}_m)\}} \, p_{\F{t}} \cdot p_x^{\F{t}} \\
				= & \sum_{(p_{\F{t}},\F{t}) \in M \setminus
                                  \{(p_{\F{t}},\F{t})\}} \,  p_{\F{t}} \cdot p_x^{\F{t}} + p_\F{l} \cdot p_h \cdot p_x^{\F{l}_h} \\
				= & \sum_{(p_{\F{t}},\F{t}) \in M \setminus
                                  \{(p_{\F{t}},\F{t})\}} \, p_{\F{t}} \cdot p_x^{\F{t}} + p_\F{l} \cdot p_h \cdot \frac{1}{p_h} \cdot p_x^{\F{l}} \\
				= & \sum_{(p_{\F{t}},\F{t}) \in M \setminus
                                  \{(p_{\F{t}},\F{t})\}} \, p_{\F{t}} \cdot p_x^{\F{t}} + p_\F{l} \cdot p_x^{\F{l}}\\
				= & \sum_{(p_{\F{t}},\F{t}) \in M} \, p_{\F{t}} \cdot p_x^{\F{t}}\\
				\stackrel{IH}{=} & p_x^{\F{T}'_H}
			\end{array}
		\]
	\end{description}
  \medskip

  \noindent
	\textbf{\underline{3) Create the finitely-branching infinite tree of RSTs}}

    \noindent
    Next, we create a tree $\F{F}$ whose nodes are labeled with RSTs. More precisely, its
    nodes
    at depth $H$ 
    represent $\tored{e-1}{}{\PP}$-RSTs $\F{t} \in \mathbb{T}_{H}$ 
    that converge with a small enough probability,
    i.e., $\sum_{x \in \ctleaf^{\F{t}} \land x \in \ctleaf^{\F{T}'}} \, p_x^{\F{t}} < 1 - \alpha$.
    The root of $\F{F}$ is labeled with the subtree $\F{T}'_0$ of $\F{T}'$ that only
    consists of its root.
	  $\F{T}'_0$ is a finite subtree and 
    $\sum_{x \in \ctleaf^{\F{T}'_0} \land x \in \ctleaf^{\F{T}'}} \, p_x^{\F{T}'_0} = 0 \leq 1 - \alpha$, 
    since the root of $\F{T}'$ cannot be a leaf in $\F{T}'$.

	Let $H \in \IN$ with $H > 1$.
	For the tree $\F{F}$, we have a node at depth $H$ for every tree $\F{t} \in \mathbb{T}_{H}$ 
    such that $\sum_{x \in \ctleaf^{\F{t}} \land x \in \ctleaf^{\F{T}'}} \, p_x^{T} < 1 - \alpha$.
	We draw an edge from a node $X$ at depth $H-1$ to the node $Y$ at depth $H$ 
    if the corresponding RSTs in the labels are the same, or
    if the RST for node $Y$ is an extension of the RST for node $X$ 
    (i.e., we evaluate the leaves in the RST of $X$ further to obtain $Y$).
	
    Each $\mathbb{T}_{H}$ for every $H \in \IN$ is finite, so $\F{F}$ is finitely branching.
    Furthermore, there is a node in each layer of the tree $\F{F}$, since for every $H \in \IN$, 
    we can find an $(p_{\F{t}},\F{t}) \in \mathbb{T}_{H}$ 
    such that $\sum_{x \in \ctleaf^{\F{t}} \land x \in \ctleaf^{\F{T}'}} \, p_x^{\F{t}} < 1 - \alpha$, 
    see \eqref{eq:t-small-probability-every-H}.
    \medskip

    \noindent
	\textbf{\underline{4) Use König's Lemma}}

    \noindent
    Now we have an infinite tree $\F{F}$ that is finitely branching, 
    which means that the tree has an infinite path by König's Lemma.
	This path represents an $\tored{e-1}{}{\PP}$-RST that does not converge with probability $1$, 
    which is our desired contradiction.
    To see this, let $\F{t}_1, \F{t}_2, \ldots$ be the finite $\tored{e-1}{}{\PP}$-RSTs in
    the labels 
    of the nodes in the infinite path in $\F{F}$, where $\F{t}_i$ is a prefix of
    $\F{t}_{i+1}$ for all $i \geq 1$. Hence, we can  
 define the tree $\F{t}_{\lim} = \lim_{i \to \infty} \F{t}_i$.
    Furthermore, there exists no infinite path in $\F{F}$ such that the sequence $\F{t}_1, \F{t}_2, \ldots$
    eventually stays the same finite tree forever, due to the fact that there are no infinite paths in $\F{T}'$
    that eventually only contain nodes from $N_1$ (due to the application of the cutting
    lemma to create $\F{T}'$).
Finally, we have $|\F{t}_{\lim}| = \lim_{i \to \infty} \sum_{x \in \ctleaf^{\F{t}_i} \land
  x \in \ctleaf^{\F{T}'}} \, p_x^{\F{t}_i} < 1 - \alpha$ since $\sum_{x \in
  \ctleaf^{\F{t}_i} \land x \in \ctleaf^{\F{T}'}} \, p_x^{\F{t}_i} < 1 - \alpha$ holds for all $i \in \IN$.
\end{proof}

We also obtain a corresponding parallel execution lemma for $\SASTs$.
Note that the parallel execution lemma
does not hold for $\PASTs$, see \Cref{example:Past-binary-function}.

\begin{lem}[Parallel Execution Lemma for $\SASTs$]\label{lemma:parallel-execution-lemma-SAST}
    Let $\PP$ be a PTRS and $s \in \IS$.
    Furthermore, let $q_1, \ldots, q_n \in \TT$ be terms and $C_1, \ldots, C_n \in \IR$ constants 
    such that for every $1 \leq i \leq n$ and every $\sto_{\PP}$-RST $\F{T}_i$ that starts with $(1:q_i)$
    we have $\edl(\F{T}_i) \leq C_i$. 
    Then, every $\sto_{\PP}$-RST $\F{T}$ that starts with $(1:\tc(q_1, \ldots, q_n))$ for some symbol $\tc$, 
    where we do not use rewrite steps at the root position,
    has a finite expected derivation length,
    which is bounded by $\sum_{i=0}^{n} C_i \in \IR$, i.e., $\edl(\F{T}) \leq \sum_{i=0}^{n} C_i$.
\end{lem}

\begin{proof}
    We use the same induction on $e$ as in the proof for the parallel execution lemma for $\ASTs$ 
    (\Cref{lemma:parallel-execution-lemma-AST}).
    Note that we now need to prove an upper bound (the expected derivation length is smaller than $\sum_{i=0}^{n} C_i$),
    while for $\ASTi$ we had to prove a lower bound 
    (probability of convergence is (at least) $1$).
    Therefore, we do not use the cutting lemma and the proof is a bit different in the induction step, 
    while the base case is again trivial.

    In the induction step, we assume that the statement holds for $e-1$, 
    i.e., there exists a bound $C_{e-1}' = \sum_{i=0}^{e-1} C_i \in \IR$ 
    such that for every $\tored{e-1}{}{\PP}$-RST $\F{T}$ that starts with $(1:\tc(q_1, \ldots, q_n))$
    we have $\edl(\F{T}) \leq C'_{e-1}$.
    Now consider an arbitrary $\tored{e}{}{\PP}$-RST $\F{T}$.
    We prove that its expected derivation length is bounded by $C'_{e-1} + C_{e}$.

    Let $H \in \IN$ and let $\F{T}_{H}$ be the tree consisting of the first $H$ layers
    of $\F{T}$.
    We partition the inner nodes of $\F{T}_{H}$ into $N_1$ and $N_2$, analogous to the
    proof of  the parallel execution lemma for $\ASTs$ 
    (\Cref{lemma:parallel-execution-lemma-AST}). As in that proof, for each $H \in \IN$
    we split the tree $\F{T}_{H}$ 
    into multiple (finitely many) $\tored{e-1}{}{\PP}$-RSTs
of height at most $H$
    that all start with 
    $(1:\tc(q_1, \ldots, q_n))$. This again leads to a set of pairs $\mathbb{T}_H$ such that:

    \begin{enumerate}[label=(Prop-\arabic{enumi})]
      \item\label{itm:const-prop-1-SAST} $\sum_{(p_{\F{t}},\F{t}) \in \mathbb{T}_H} p_{\F{t}} = 1$
      
      \item\label{itm:const-prop-2-SAST} For all $x \in N_2$ we have
        $p_x^{\F{T}_H} = \sum_{(p_{\F{t}},\F{t}) \in \mathbb{T}_H} p_{\F{t}} \cdot p_x^{\F{t}}$
	  \end{enumerate}
    using the same notation as before.
    Furthermore, we can use the same construction to create a set $\mathbb{T}_H^{e}$ containing 
    (finitely many) pairs of probabilities and $\sto_{\PP}$-RSTs that start with $(1:q_e)$ of height at most $H$,
    by simply switching when to split the tree and when to perform the rewrite step.
    To be precise, we split the tree if we encounter a rewrite step on or below a position $1
    \leq j \leq e-1$,
    and perform the rewrite step if it is on or below position $e$.
    Again, we get
    \begin{enumerate}[label=(Prop-\arabic{enumi}-e)]
      \item\label{itm:const-prop-1-SAST-e} $\sum_{(p_{\F{t}},\F{t}) \in \mathbb{T}_H^e} p_{\F{t}} = 1$
      
      \item\label{itm:const-prop-2-SAST-e} For all $x \in N_1$ we have
        $p_x^{\F{T}_H} = \sum_{(p_{\F{t}},\F{t}) \in \mathbb{T}_H^e} p_{\F{t}} \cdot p_x^{\F{t}}$
	  \end{enumerate}

    Now we can bound the expected runtime after height $H$ by $C'_{e-1} + C_{e}$ for every $H \in \IN$,
    as we have
    \[
		\begin{array}{ccl}
            &  & \edl(\F{T}_H) \\
            & = & \sum_{x \in N^{\F{T}_H}\setminus\ctleaf^{\F{T}_H}} \, p_x^{\F{T}_H} \\
      \text{\textcolor{blue}{(since $N_1 \uplus N_2 = N^{\F{T}_H} \setminus
          \ctleaf^{\F{T}_H}$)}}      & = &
      \sum_{x \in N_1} \, p_x^{\F{T}_H} \; + \; \sum_{x \in N_2} \, p_x^{\F{T}_H} \\
			\text{\textcolor{blue}{(by \ref{itm:const-prop-2-SAST-e} and \ref{itm:const-prop-2-SAST})}} & = & 
            \sum_{x \in N_1}  \sum_{(p_{\F{t}},\F{t}) \in \mathbb{T}_H^e} \, p_{\F{t}} \cdot
            p_x^{\F{t}}\\
            & & + \; \sum_{x \in N_2}  \sum_{(p_{\F{t}},\F{t}) \in \mathbb{T}_H} \, p_{\F{t}} \cdot p_x^{\F{t}}\\
			\text{\textcolor{blue}{(as in \eqref{eq:array-AST-parallel})}} & =
                        & \sum_{(p_{\F{t}},\F{t}) \in \mathbb{T}_{H}^e} \, p_{\F{t}} \cdot
                        \sum_{x \in N^{\F{t}} \setminus \ctleaf^{\F{t}}} \, p_x^{\F{t}} \\
            &  & + \; \sum_{(p_{\F{t}},\F{t}) \in \mathbb{T}_{H}} \, p_{\F{t}} \cdot
                        \sum_{x \in N^{\F{t}} \setminus \ctleaf^{\F{t}}} \, p_x^{\F{t}} \\
			\text{\textcolor{blue}{(by assumption and induction hypothesis)}}
                        & \leq & \sum_{(p_{\F{t}},\F{t}) \in \mathbb{T}_{H}^e} \, p_{\F{t}} \cdot C'_{e-1} \\
            &  & + \; \sum_{(p_{\F{t}},\F{t}) \in \mathbb{T}_{H}} \, p_{\F{t}} \cdot C_{e} \\
			& = & C'_{e-1} \cdot \sum_{(p_{\F{t}},\F{t}) \in \mathbb{T}_{H}^e}
                        \, p_{\F{t}} \\
            &  & + \; C_n \cdot \sum_{(p_{\F{t}},\F{t}) \in \mathbb{T}_{H}} \, p_{\F{t}} \\
			\text{\textcolor{blue}{(by \ref{itm:const-prop-1-SAST-e} and \ref{itm:const-prop-1-SAST})}} & = & C'_{e-1} \cdot 1 + C_e \cdot 1 \\
			& = & C'_{e-1} + C_e
		\end{array}
	\]
    This gives us $\edl(\F{T}) = \lim_{H \to \infty} \edl(\F{T}_H) \leq C'_{e-1} + C_{e}$,
    as desired.
\end{proof}

With the parallel execution lemmas for $\ASTi$ and $\SASTi$
we can finally prove our modularity results. We start with $\ASTi$.

\ModiASTDisjoint*

\begin{proof}
  The direction ``$\Longrightarrow$'' is trivial and thus, we only prove
  ``$\Longleftarrow$''. So let  $\PP = \PP^{(1)} \cup \PP^{(2)}$, where
  both $\AST[\ito_{\PP^{(1)}}]$ and $\AST[\ito_{\PP^{(2)}}]$ hold.

      By \Cref{lemma:PTRS-AST-single-start-term}, for $\ASTi$ it suffices to regard only rewrite sequences that start
    with multi-distributions of the form $\{1:t\}$.
    Thus, we show by
    structural induction on the term structure that for every 
    $t \in \TSet{\Sigma^{\PP}}{\VSet}$, all
    $\iliftto_{\PP}$-rewrite sequences starting with $\{1:t\}$
    converge with probability 1.

      If $t \in \VSet$, then $t$ is in normal form.
    If $t$ is a constant, then w.l.o.g.\ let $t \in \PP^{(1)}$.
    Since we have $\AST[\ito_{\PP^{(1)}}]$, $\{1:t\}$ cannot start an infinite $\iliftto_{\PP}$-rewrite sequence that converges with probability~$<1$.

    Now we regard the induction step, and consider the case where
      $t = f(q_1, \ldots, q_n)$.
    By the induction hypothesis, every $\ito_{\PP}$-RST $\F{T}$ 
    that starts with $(1:q_i)$ for some $1 \leq i \leq k$ converges with probability $1$.
    Let $\F{T}$ be a fully evaluated $\ito_{\PP}$-RST 
    that starts with $(1:f(q_1, \ldots, q_n))$.
	We prove that for every $0 < \delta < 1$ we can find an $M \in \IN$ 
	such that $\sum_{x \in \ctleaf^{\F{T}}, d(x) \leq N}\, p_x^{\F{T}} > 1 - \delta$, 
    which means that $|\F{T}| = \lim_{k \to \infty} \sum_{x \in \ctleaf^{\F{T}}, d(x) \leq k} \, p_x^{\F{t}} = 1$.
    (Recall that for proving $\ASTi$, it suffices to consider only fully
    evaluated RSTs, see \Cref{cor:Characterizing AST with RSTs}.)

	Let $0 < \delta < 1$.
    By the induction hypothesis and the
    parallel execution lemma (\Cref{lemma:parallel-execution-lemma-AST}),
    the maximal subtree $\F{T}_{\lnot \varepsilon}$  of $\F{T}$ that starts
    with $\F{T}$'s root node and
    only performs rewrite steps at  non-root positions   converges with probability
    $1$.
    Since $|\F{T}_{\lnot \varepsilon}| = 1$,
    there exists a depth $H$ such that $\sum_{x \in \ctleaf^{\F{T}_{\lnot
    \varepsilon}}, d(x) \leq H} \, p_x^{\F{T}} \geq  \beta$ with $\beta = 
    \sqrt{1-\delta}$. 
    Let $\F{T}_H$ be the tree resulting from
    cutting the tree $\F{T}$ at depth $H$, 
    and let $Z^{\F{T}_H}$ be the set of leaves in $\F{T}_H$ that were already leaves in $\F{T}_{\lnot \varepsilon}$.
    So we have $\sum_{x \in Z^{\F{T}_H}} \, p_x^{\F{T}} 
    = \sum_{x \in \ctleaf^{\F{T}_{\lnot \varepsilon}}, d(x) \leq H} \, p_x^{\F{T}} 
    \geq  \beta$.   
	For each leaf $x \in Z^{\F{T}_H}$, every proper subterm of $t_x$ is in normal form
        w.r.t.\ $\PP$.
This is due to the fact that we use an innermost rewrite strategy and the RST $\F{T}$ is
fully evaluated.
 	Let us look at the induced subtree $\F{T}_x$ of $\F{T}$ that starts at $x$ (i.e., $\F{T}_x = \F{T}[xE^*]$).
    W.l.o.g., let the root symbol $f$ of $t_x$ be from $\Sigma^{\PP^{(1)}}$.
    Let $\F{T}^{(1)}$ result from $\F{T}_x$ by labeling the root with 
    $t_x'$,
where  $t_x'$ results from $t_x$ by
replacing all its maximal (i.e., topmost)
subterms with root symbols from $\Sigma^{\PP^{(2)}}$ 
by fresh variables (using the same variable for the same subterm).
Obviously, since both PTRSs have disjoint signatures
and all proper subterms of $t_x$ are in normal form,
we can still apply the same rules as in $\F{T}_x$,
such that $\F{T}^{(1)}$ is
an $\PP^{(1)}$-RST  with $|\F{T}^{(1)}| = |\F{T}_x|$.
Since every $\PP^{(1)}$-RST converges with probability $1$,
we obtain $1 = |\F{T}^{(1)}| = |\F{T}_x|$.

For any node $y$ of $\F{T}_x$, 
let $\ctdepth_{x}(y)$ be the depth of the node $y$ in the tree $\F{T}_x$.
	Moreover, let $D_{x}(k) = \{y \in N^{\F{T}_x} \mid \ctdepth_{x}(y) \leq k\}$ be the set of nodes in $\F{T}_x$ that have a depth of at most $k$.
	Since $|\F{T}_x|=1$ and
          $|\_|$ increases weakly monotonically
        with the depth of the tree, we can find an $M_{x} \in \IN$ such that $|\F{T}_x[D_{x}(M_{x})]| >  \beta$.
	
	Note that $Z^{\F{T}_H}$ is finite and thus $M_{\max} = \max\{M_x \mid x \in Z^{\F{T}_H}\}$ exists.
	Now, for $N = H + M_{\max}$ we finally have 
	\[
		\begin{array}{rl}
			& \sum_{x \in \ctleaf^{\F{T}}, d^{\F{T}}(x) \leq N} \, p_x^{\F{T}} \\
	        =& \sum_{x \in \ctleaf^{\F{T}}, d^{\F{T}}(x) \leq H + M_{\max}} \, p_x^{\F{T}} \\
            \geq & \sum_{x \in Z^{\F{T}_H}} \sum_{y \in \ctleaf^{\F{T}_x}
                \land \ctdepth_x(y) \leq M_{\max}} \, p_{y}^{\F{T}}\\
			= & \sum_{x \in Z^{\F{T}_H}} \sum_{y \in \ctleaf^{\F{T}_x} 
                \land \ctdepth_x(y) \leq M_{\max}} \, p_{x}^{\F{T}} \cdot p_{y}^{\F{T}_x}\\
			= & \sum_{x \in Z^{\F{T}_H}} \, p_{x}^{\F{T}} \cdot 
                \sum_{y \in \ctleaf^{\F{T}_x} \land \ctdepth_x(y) \leq M_{\max}} \, p_{y}^{\F{T}_x}\\
			\geq & \sum_{x \in Z^{\F{T}_H}} \, p_{x}^{\F{T}} \cdot 
                \sum_{y \in \ctleaf^{\F{T}_x} \land \ctdepth_x(y) \leq M_{x}} \, p_{y}^{\F{T}_x}\\
			> & \sum_{x \in Z^{\F{T}_H}} \, p_{x}^{\F{T}} \cdot  \beta\\
			= & \beta \cdot \sum_{x \in Z^{\F{T}_H}} \, p_{x}^{\F{T}}\\
			\geq &  \beta \cdot \beta\\
			= & 1-\delta \\
		\end{array}
	\]
	The first inequality holds since every leaf in $\F{T}_x$ with a depth of at most $M_{\max}$ (in $\F{T}_x$) for some $x \in Z^{\F{T}_H}$ 
	must also be a leaf in $\F{T}$ with a depth of at most $H + M_{\max}$, since $x$ is at a depth of at most $H$.
\end{proof}

Before we can prove modularity for disjoint unions w.r.t.\ $\SASTi$,
we have to give some more definitions regarding the disjoint union abstraction.
In addition to \Cref{def:term-abstraction}, we also want to label the function symbols in the abstraction,
to indicate from which position the function symbol originated from in the original
term. For a term $t$, let
$\pos_\VSet(t)$ be the set of all its variable positions and 
let $\pos_\Sigma(t)$ be the set of all those positions of $t$ where
$t$ has function symbols instead of variables.

\begin{defi}[Labeled Disjoint Union Abstraction]\label{def:term-abstraction-labeled}
    Let $\PP^{(1)}, \PP^{(2)}$ be PTRSs with $\Sigma^{\PP^{(1)}} \cap \Sigma^{\PP^{(2)}} = \emptyset$.
    For any $d \in \{1,2\}$, $t \in \TSet{\Sigma^{\PP^{(1)}} \cup
      \Sigma^{\PP^{(2)}}}{\VSet}$, 
    and position $\pi \in \pos_\Sigma(t)$,
 $\caphterm_d^{\pi}(t)$ and $\capterm_d(t)$ are multisets of terms    
    from $\TSet{L(\Sigma^{\PP^{(d)}})}{\VSet}$, which are defined as follows.  Here, we
    use the signature $L(\Sigma^{\PP^{(d)}}) = \Sigma^{\PP^{(d)}} \times \pos(t)$.    
    {\small
    \[
        \begin{array}{ll}
          \caphterm_d^{\pi}(y) & = \{x\} \text{, if } y \in \VSet
          \text{, where
                    $x$  is always a new fresh variable } \\
        	\caphterm_d^{\pi}(f(t_1, \ldots, t_k)) & = \{f^{\pi}(q_1, \ldots, q_k) \mid  q_1 \in
                \caphterm_d^{\pi.1}(t_1), \ldots, q_k \in
                \caphterm_d^{\pi.k}(t_k)
                \} \text{, if } f \in \Sigma^{\PP^{(d)}}\\ 
        	\caphterm_d^{\pi}(f(t_1, \ldots, t_k)) & =  \{x\} \cup \caphterm_d^{\pi.1}(t_1) \cup
                \ldots \cup \caphterm_d^{\pi.k}(t_k)\text{,  otherwise, where
                    $x$  is always a new variable }\!
    \end{array}
        \]}
    
    \noindent
    So $\caphterm_d^{\pi}(t)$ is always a linear term, i.e., it never contains multiple occurrences of the
    same variable.
 
    For any function $\varphi: X \to X$ with $X \subseteq \VSet$, 
    let $\sigma_{\varphi}$ be the substitution that replaces
    every variable $x \in X$ by $\varphi(x) \in X$
    and leaves all other variables unchanged, i.e., $\sigma_{\varphi}(x) =
    \varphi(x)$ if $x \in X$ and $\sigma_X(x) = x$ otherwise. 
    Then we define
    \[\capterm_d(t) = \{ \sigma_{\varphi}(q) \mid q \in \caphterm_d^{\varepsilon}(t),
    \varphi: \VSet(q) \to \VSet(q) \}\]
    and $\capterm(t) = \capterm_1(t) \cup \capterm_2(t)$.
    The \emph{(labeled) disjoint union abstraction} of $t$ is the multiset $\capterm_1(t)
    \cup \capterm_2(t)$.
\end{defi}

\begin{exa}\label{example:SAST-modularity-proof-abstraction-labeled}
    Consider the PTRS $\PP_{14}^{(1)}$ with the rules $\ta \to \{1:\tf(\tb)\}$, $\tb \to
    \{1:\tz\}$, and $\th(y,y) \to \{1:y\}$, the PTRS
  $\PP_{14}^{(2)}$ with the rule $\tg(y) \to \{1:\tc\}$,
    and the term $\th(\tg(\ta), \tg(x))$.
    We obtain\footnote{To be precise,
    $\caphterm^{\varepsilon}_1(\th(\tg(\ta), \tg(x)))$ contains two additional terms
    $\th^{\varepsilon}(\ta^{1.1}, z')$ and
      $\th^{\varepsilon}(y,z')$ since $\caphterm^{2}_1(\tg(x)) = \{z, z'\}$. However, to
    ease readability, we disregarded them here.}
    \[ \begin{array}{ll}
      \caphterm_1^{\varepsilon}(\th(\tg(\ta), \tg(x))) & = \{\th^{\varepsilon}(\ta^{1.1}, z),
      \th^{\varepsilon}(y,z)\}\\
      \capterm_1(\th(\tg(\ta), \tg(x))) & = \{\th^{\varepsilon}(\ta^{1.1}, z),
      \th^{\varepsilon}(z,y), \th^{\varepsilon}(y,z), \th^{\varepsilon}(z,z), \th^{\varepsilon}(y,y)
      \}\\
      \caphterm_2^{\varepsilon}(\th(\tg(\ta), \tg(x))) & = \{x', \tg^{1}(y), \tg^{2}(z)\}\\
      \capterm_2(\th(\tg(\ta), \tg(x))) & = \{x', \tg^{1}(y), \tg^{2}(z)\}\!
    \end{array}\]
\end{exa}

Furthermore, we will define
which terms $q \in \capterm(t)$ ``cover'' which of the 
nodes from $N^{\F{T}}$ 
for an arbitrary $\ito_{\PP^{(1)} \cup \PP^{(2)}}$-RST $\F{T}$ starting with $(1:t)$.
The idea is that every  step in a rewrite sequence starting with $t$
(corresponding to an inner
node of $\F{T}$)
can also be performed when starting with a suitable $q \in \capterm(t)$. 
For this, we first define the notion of an \emph{origin graph}.
For every rewrite sequence starting with $t$, the graph indicates which subterm
``originates'' from which subterm of $t$.
Moreover, every function symbol in the origin graph is labeled by the position of the
corresponding symbol on the right-hand side of the rule that created it.

\begin{defi}[Origin Graph]\label{def:orig}
    Let $\PP$ be a PTRS
    and let $\F{T}$ be an $\ito_{\PP}$-RST.
    The \emph{origin graph} for $\F{T}$ is a labeled graph with the nodes $(x,\pi)$ 
    for all $x \in N^{\F{T}}$ and all $\pi \in \pos_\Sigma(t_x)$,
    where the edges and labels are defined as follows:
    For the root $\ctroot$ of $\F{T}$ we label the node $(\ctroot,\pi)$ by $\varepsilon$.
    For $x \in N^{\F{T}}$,
    let the rewrite step  $t_x \ito_{\PP} \{p_1:t_{y_1}, \ldots, p_k:t_{y_k}\}$ be performed using the rule
    $\ell \to \{p_1:r_{y_1}, \ldots, p_k:r_{y_k}\}$, the position $\tau$, and the substitution $\sigma$,
    i.e., $t_x|_\tau = \ell \sigma$ and $t_{y_j} = t_x[r_j \sigma]_\tau$ for all
 $1 \leq j \leq k$.    
    Let $\pi \in \pos_\Sigma(t_x)$.
    \begin{description}
        \item[(a)] If $\pi < \tau$ or $\pi \bot \tau$ (i.e., $\pi$ is above or parallel to $\tau$),        
        then there is an edge from $(x,\pi)$ to $(y_j,\pi)$.
        If $(x,\pi)$ was labeled by $\gamma$, then $(y_j,\pi)$ is labeled by $\gamma$ as well.
        \item[(b)] For $\pi=\tau$,
          there is an edge from $(x,\pi)$ to $(y_j,\pi.\alpha)$ for all
          $\alpha \in \pos_{\Sigma}(r_j)$.
        If $(x,\pi)$ was labeled by $\gamma$, then $(y_j,\pi.\alpha)$ is labeled by $\gamma.\alpha$.
        \item[(c)] For every variable position
        $\alpha_\ell \in \pos_\VSet(\ell)$, for all
        positions $\alpha_{r_j} \in \pos_{\VSet}(r_j)$ with
        $r_j|_{\alpha_{r_j}} = \ell|_{\alpha_\ell}$,
        and for all $\beta \in \IN^*$ with
        $\alpha_\ell.\beta \in \pos_\Sigma(\ell\sigma)$, there is an edge from 
        $(x,\tau.\alpha_\ell.\beta)$ to $(y_j,\tau.\alpha_{r_j}.\beta)$.
        \item[(d)] For all other positions $\pi \in \pos(t_x)$, 
        i.e., the positions that are inside the redex $\ell \sigma$
        but neither at the root of $\ell$ nor inside the substitution $\sigma$, 
        there is no outgoing edge from the node $(x, \pi)$.
    \end{description}
\end{defi}

\begin{exa}\label{graph1}
    Reconsider the PTRSs $\PP_{14}^{(1)}$ 
    with the rules $\ta \to \{1:\tf(\tb)\}$, $\tb \to \{1:\tz\}$, and $\th(y,y) \to \{1:y\}$, 
   the PTRS  $\PP_{14}^{(2)}$ with the rule $\tg(y) \to \{1:\tc\}$,
    and the term $\th(\tg(\ta), \tg(x))$ from \Cref{example:SAST-modularity-proof-abstraction-labeled}.
    Furthermore, consider the following RST, where we omitted the (trivial) probabilities,
    and numbered each rewrite step.     
    \[\th(\tg(\underline{\ta}), \tg(x)) \stackrel{(1)}{\longrightarrow} \th(\tg(\tf(\underline{\tb})), \tg(x)) \stackrel{(2)}{\longrightarrow} \th(\underline{\tg(\tf(\tz))}, \tg(x)) \stackrel{(3)}{\longrightarrow} \th(\tc, \underline{\tg(x)}) \stackrel{(4)}{\longrightarrow} \underline{\th(\tc,\tc)} \stackrel{(5)}{\longrightarrow} \tc\]

    This RST yields the following
    origin graph.
    
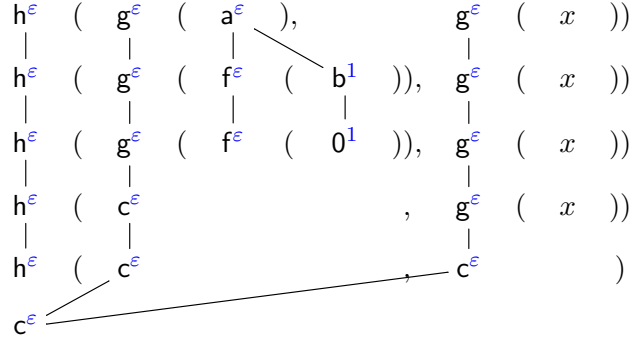
\[ \xymatrix @-1.5pc {
        \th^{\textcolor{blue}{\varepsilon}}\ar@{-}[d] &( & \tg^{\textcolor{blue}{\varepsilon}} \ar@{-}[d]&(&\ta^{\textcolor{blue}{\varepsilon}}\ar@{-}[d] \ar@{-}[drr]&),&& & \tg^{\textcolor{blue}{\varepsilon}}\ar@{-}[d] & ( &x &)) \\
        \th^{\textcolor{blue}{\varepsilon}}\ar@{-}[d] &( & \tg^{\textcolor{blue}{\varepsilon}} \ar@{-}[d]&(&\tf^{\textcolor{blue}{\varepsilon}}\ar@{-}[d]&(&\tb^{\textcolor{blue}{1}} \ar@{-}[d]&)), & \tg^{\textcolor{blue}{\varepsilon}}\ar@{-}[d] & ( &x &)) \\
        \th^{\textcolor{blue}{\varepsilon}}\ar@{-}[d]&( & \tg^{\textcolor{blue}{\varepsilon}}\ar@{-}[d]&(&\tf^{\textcolor{blue}{\varepsilon}}&(&\tz^{\textcolor{blue}{1}}&)), & \tg^{\textcolor{blue}{\varepsilon}}\ar@{-}[d] & ( &x &)) \\
        \th^{\textcolor{blue}{\varepsilon}}\ar@{-}[d]&( & \tc^{\textcolor{blue}{\varepsilon}}\ar@{-}[d]&&&&&, & \tg^{\textcolor{blue}{\varepsilon}}\ar@{-}[d] & ( &x &)) \\
        \th^{\textcolor{blue}{\varepsilon}}&( & \tc^{\textcolor{blue}{\varepsilon}}\ar@{-}[dll]&&&&&, & \tc^{\textcolor{blue}{\varepsilon}}\ar@{-}[dllllllll] &&&) \\
        \tc^{\textcolor{blue}{\varepsilon}}&&&&&&&&&&& \!
    }
    \]
\end{exa}

Note that the root of every subterm that is not in normal form
is reachable from exactly one node $(\ctroot, \pi)$.
Only nodes $(x,\tau)$ where $t_x|_{\tau}$ is in normal form may have multiple incoming transitions.

Our goal is to split up the RST that started with $(1:t)$ into RSTs that start with
$(1:q)$ for $q \in \capterm(t)$. The reason is that these terms $q$ only contain symbols
from either $\Sigma^{\PP^{(1)}}$ or $\Sigma^{\PP^{(2)}}$ and hence, there is a bound on
the expected derivation length of all these RSTs.

The labels in the origin graph and the labels in the terms $q \in \capterm(t)$ can now be
used to  construct the new RSTs that start with $(1:q)$ for $q \in \capterm(t)$
from the original RST that starts with $(1:t)$. Then every rewrite step in the original
RST corresponds to at least one step in one of these new RSTs.
Let us illustrate this with our running example.

\begin{exa}\label{example:cover}
  From the $\ito_{\PP_{14}}$-RST in \Cref{graph1} we obtain
  the following $\ito_{\PP_{14}^{(1)}}$-RSTs
    and $\ito_{\PP_{14}^{(2)}}$-RSTs\footnote{In addition to the RST starting with
    $\th^{\varepsilon}(z, z)$ we also obtain the corresponding RST starting with
    $\th^{\varepsilon}(y, y)$, but omitted it here for readability.} \pagebreak[3]
    \[\begin{array}{lclcl}
        \th^{\varepsilon}(\underline{\ta^{1.1}}, z) &\stackrel{(1')}{\longrightarrow}& \th(\tf(\underline{\tb}), z) &\stackrel{(2')}{\longrightarrow}& \th(\tf(\tz), z)\\
        \underline{\th^{\varepsilon}(z, z)} &\stackrel{(5')}{\longrightarrow}& z &&\\
        \underline{\tg^{1}(z)} &\stackrel{(3')}{\longrightarrow}& \tc \\
        \underline{\tg^{2}(z)} &\stackrel{(4')}{\longrightarrow}& \tc \!
    \end{array}\]
    where the labels are ignored for the rewrite steps, but they are used
     to determine which start term from $\capterm(t)$ to use for which step.
    The rewrite step $(i)$ in the original RST corresponds to the rewrite step $(i')$ in
    our new RSTs.
In the original RST, one first performs rewrite steps for the symbols $\ta$ and $\tb$
from  $\Sigma^{\PP_{14}^{(1)}}$, then one rewrites the symbol $\tg$ from
$\Sigma^{\PP_{14}^{(2)}}$ above, and finally one rewrites the top symbol $\th$ from
$\Sigma^{\PP_{14}^{(1)}}$. In contrast, these rewrite steps are now separated such that in
each of the above RSTs, one either only rewrites symbols from $\Sigma^{\PP_{14}^{(1)}}$ or
only symbols from $\Sigma^{\PP_{14}^{(2)}}$.

    Let us explain how to detect that the rewrite step $(2)$ that is performed at position $\pi = 1.1.1$ 
    in $\th(\tg(\tf(\tb)), \tg(x))$ in the original RST 
    should be applied at position $1.1$ in $\th(\tf(\tb), z)$ for rewrite step $(2')$.
    We consider the (unique) predecessor $(t_\ctroot,1.1)$
(i.e., $\ta^\varepsilon$)
    of $(x_2,1.1.1)$ (i.e., $\tb^1$), where $\F{r}$ is the root of the original RST
and $x_2$ denotes the second node of the original RST containing the term
    $\th(\tg(\tf(\tb)), \tg(x))$
This indicates that the rewrite step $(2)$ in the original RST corresponds to a rewrite
step in a new RST that starts with  $(1:q)$
for a term $q$ containing
    $\rootsym(t_\ctroot|_{1.1})^{1.1} =
\rootsym(\th(\tg(\ta), \tg(x))|_{1.1})^{1.1} = \ta^{1.1}$. Hence, we have to consider the
new RST starting with  $\th^{\varepsilon}(\underline{\ta^{1.1}}, z)$.

To find the actual rewrite step that corresponds to  $(2)$
in this new RST, we determine 
the position of the symbol
    $\ta^{1.1}$
    within $\th^{\varepsilon}(\ta^{1.1}, x)$, which is $\gamma = 1$,
    and the label of $(x_2,1.1.1)$ (i.e., $\tb^1$), which is $\chi = 1$. This 
    indicates that in the new RST we need to rewrite at position $\gamma.\chi = 1.1$.
\end{exa}

Now, we define this covering formally in order to
determine which term $q$ from $\capterm(t)$ to use for the current rewrite step.
More precisely, for any $q \in \capterm(t)$ we define the \emph{abstraction cover}  $AC_q \subseteq N^{\F{T}}$ which
contains all inner nodes
of the original RST whose rewrite step can be simulated by a rewrite step in the new RST for $q$.
 Moreover, for every
 node $x$, we define the term $\psi_x(q)$ which we would obtain instead of $t_x$ if we
  had started the RST with $q$ instead of $t$.
 Thus, if there is an edge from the node $x$ to the node $y$ in the original RST $\F{T}$, then 
  $\psi_{x}(q)$ rewrites to $\psi_{y}(q)$ or they are equal. Therefore, for the root
  $\ctroot$ of $\F{T}$, the term $\psi_{\ctroot}(q)$ rewrites (in zero or more steps) to 
 $\psi_{x}(q)$ for every node of $\F{T}$.

\begin{defi}[Abstraction Cover]\label{def:cover}
  Let $\PP = \PP^{(1)} \cup \PP^{(2)}$ be a PTRS
with $\Sigma^{\PP^{(1)}} \cap \Sigma^{\PP^{(2)}} = \emptyset$
  and let $\F{T}$ be an $\ito_{\PP}$-RST that starts with $(1:t)$.
     We define the \emph{abstraction cover} with $\bigcup_{q \in \capterm(t)} AC_q =
     N^{\F{T}}\setminus\ctleaf^{\F{T}}$ recursively,
         where we additionally define a term
       $\psi_x(q)$ for each node $x \in N^{\F{T}}$
    that corresponds
to $t_x$ in the RST starting with $(1:q)$. 
So to reach $\psi_x(q)$ from $q$, one performs the same 
rewrite steps as in the path from the root to $x$ in $\F{T}$ whenever possible,
where the appropriate position of the new rewrite step is 
    indicated by the labels of the symbols in $q$ and the labels in the origin graph.
   For instance, in \Cref{example:cover} for $q = \th(\ta, z)$ we have $\psi_{x_2}(q) = \th(\tf(\tb), z)$ 
    and $\psi_{x_i}(q) = \th(\tf(\tz), z)$ for all $i \in \{3,4,5,6\}$ (we remove the
    labels in $\psi_x(q)$).

   For the root node $\ctroot$ of $\F{T}$, we initially set $\psi_{\ctroot}(q) = q$.

Now consider the rewrite step at an arbitrary node $x \in N^{\F{T}}$, 
    which uses the rule $\ell \to \{\ldots, p_j:r_j, \ldots\}$ at position $\pi$ in $t_x$ with substitution $\delta$.
    Hence, we have $t_{y_j} = t_x[r_j \delta]_{\pi}$.
Let the origin graph of $\F{T}$ contain a path from  $(\ctroot, \tau)$ to $(x, \pi)$ for
some position $\tau$ and let
 $Q \subseteq \capterm(t)$ be the set of
all terms from $\capterm(t)$ that contain a
    function
    symbol labeled with $\tau$.
    Moreover, let $(x, \pi)$ be labeled with position $\chi$ in the origin graph of $\F{T}$.
    Then
we add $x$ to $AC_q$
    for all those $q \in Q$ where  there exists a substitution $\delta'$ with
    $\psi_x(q)|_{\tau.\chi} = \ell \delta'$. For these $q$,
    we  set $\psi_{y_j}(q) = \psi_{x}(q)[r_j\delta']_{\tau.\chi}$. Note that now indeed,
    $\psi_x(q)$ rewrites to $\psi_{y_j}(q)$.
    For all other $q \in \capterm(t)$ we set $\psi_{y_j}(q) = \psi_{x}(q)$.
\end{defi}

\begin{lem}[$AC_q$ Covers all Inner Nodes]\label{CoveringComplete}
  Let $\PP = \PP^{(1)} \cup \PP^{(2)}$ be a PTRS
  with $\Sigma^{\PP^{(1)}} \cap \Sigma^{\PP^{(2)}} = \emptyset$ and let $\F{T}$ be an $\ito_{\PP}$-RST that starts with $(1:t)$.
    Then $N^{\F{T}}\setminus\ctleaf^{\F{T}} = \bigcup_{q \in \capterm(t)}
    AC_q$.
\end{lem}

\begin{proof}
    In order to show that $\bigcup_{q \in \capterm(t)} AC_q$ is indeed a cover of $N^{\F{T}}\setminus\ctleaf^{\F{T}}$,
    we have to show that for every node $x \in N^{\F{T}}\setminus\ctleaf^{\F{T}}$ there
    exists a $q \in \capterm(t)$ with $x \in AC_q$. 
    
    In other words, we have to show that there exists at least one term $q \in Q$ 
    such that there is a substitution $\delta'$ with $\psi_x(q)|_{\tau.\chi} = \ell \delta'$.
    W.l.o.g., let $\rootsym(\ell) \in \Sigma^{\PP^{(1)}}$.
    Let $C_1$ be the maximal context
containing no 
    symbols from $\Sigma^{\PP^{(2)}}$ such that there exists another context $C$
    and terms $s_1, \ldots, s_m$ with
 $\rootsym(s_i) \in \Sigma^{\PP^{(2)}}$ for all $1 \leq i \leq m$ such that
    $t_x = C[C_1[s_1, \ldots, s_m]]$ 
    and the position $\pi$ is
    within the context $C_1$.
    
    Let $\Phi$ be the set of positions in $\pos_\Sigma(t_x)$ that are within the context
    $C_1$.
 There exists a term $q \in \capterm(t)$ 
 that contains all function symbols that are labeled with  positions $\varphi^{-1}$ where
 there is a path in the origin graph from $(\ctroot,\varphi^{-1})$ to $(x, \varphi)$
 for some $\varphi \in \Phi$. This also holds if there is a normal form at position
 $\varphi$ which may be reached from several positions $\varphi^{-1}$. 
    Moreover, let $\Psi = \{\psi_1, \ldots, \psi_m\}$ be the set of root positions 
    of the $s_1, \ldots, s_m$ within $\pos(t_x)$, i.e.,
    the positions of the holes in $C_1$ within $t_x$.
We can choose $q$ 
in such a way that the fresh variables
for symbols from   $\Sigma^{\PP^{(2)}}$
at positions $\chi$ and $\chi'$ are the same whenever
there are paths in the origin graph from $(\ctroot,\chi)$ to $(x, \psi_i)$
and from  $(\ctroot,\chi')$ to $(x, \psi_{i'})$, and  $s_i = s_{i'}$.

    Finally, for this specific $q$ we have $\psi_x(q)|_{\tau.\chi} = \ell \delta'$ where
    $\delta'$ is
    like $\delta$, 
    but we replace every occurrence of the terms $s_i$ with the corresponding fresh variables.
\end{proof}

\begin{exa}
    As a final example, consider the rewrite step $(5)$, i.e., $\th(\tc,\tc) \to \tc$.
    Following the notation from the proof above, we have $C_1 = \th(\square, \square)$, $s_1 = \tc$, and $s_2 = \tc$.
    Furthermore, we have $\psi_1 = \chi = 1$, $\psi_2 = \chi' = 2$, 
    and the position $\varphi$ of the $\th$ in the context $C_1$ is $\varphi = \varepsilon$ with $\varphi^{-1} = \varepsilon$.
    Hence, we choose  $\th^{\varepsilon}(z,z) \in \capterm(t)$, 
    as this term contains the function symbol $\th$ labeled with $\varphi^{-1} = \varepsilon$,
    and uses the same variable $z$ for both subterms, as they are equal ($s_1 = s_2 = \tc$).
\end{exa}

\ModiSASTDisjoint*

\begin{proof}
  Let  $\PP = \PP^{(1)} \cup \PP^{(2)}$ and assume
that both $\SAST[\ito_{\PP^{(1)}}]$ and $\SAST[\ito_{\PP^{(2)}}]$ hold.
    Let $\F{T} = (N,E,L)$ be an arbitrary $\ito_{\PP}$-RST that starts with $(1:t)$.
    We prove that $\edl(\F{T})$ is bounded by some constant
    which does not depend on $\F{T}$ but just on $t$.
    This proves $\SASTi$.

    In \Cref{def:cover}
    we defined sets $AC_q \subseteq N^{\F{T}}$ for each $q \in \capterm(t)$
    such that $x \in AC_q$ if
   the rewrite step at node $x \in N^{\F{T}}$ is performed at some position $\pi$ in $t_x$,
    $q$ contains a function symbol labeled with $\tau$ where there is a path from $(\ctroot,
    \tau)$ to $(x,\pi)$ in the origin graph of $\F{T}$, 
 $(x,\pi)$ is labeled with some position $\chi$ in the origin  graph, and  we can
perform the same rewrite step on $\psi_x(q)$  at position $\tau.\chi$.
In \Cref{CoveringComplete} we showed that 
$\bigcup_{q \in \capterm(t)} = N^{\F{T}}\setminus\ctleaf^{\F{T}}$.

    With the definition of $AC_q$ we can now prove the upper bound on the expected derivation height of $t$.
    Let $H \in \IN$ and let $\F{T}_{H}$ be the tree consisting of the first $H$ layers of
    $\F{T}$.
    As in the proof of the parallel execution lemma for $\SASTs$ 
    (\Cref{lemma:parallel-execution-lemma-SAST}), for each $H \in \IN$ 
    and each $q \in \capterm(t)$
    we split the tree $\F{T}_{H}$, 
    into  (finitely many) sets $\mathbb{T}_H^{q}$ 
    containing (finitely many) pairs of $\ito_{\PP^{(1)}}$-RSTs 
    and $\ito_{\PP^{(2)}}$-RSTs $\F{t}$ with certain probabilities $p_{\F{t}}$.
    Each of these RSTs $\F{t}$ starts with $(1:q)$,
    has height at most $H$, and 
    \begin{enumerate}[label=(Prop-\arabic{enumi})]
		\item\label{itm:const-prop-1-SAST-mod} $\sum_{(p_{\F{t}},\F{t}) \in \mathbb{T}_H^q} p_{\F{t}} = 1$
		
		\item\label{itm:const-prop-2-SAST-mod} For all $x \in
                  AC_q$ we have $p_x^{\F{T}_H} = \sum_{(p_{\F{t}},\F{t}) \in
                    \mathbb{T}_H^q} \, p_{\F{t}} \cdot p_x^{\F{t}}$. 
        
                \item\label{itm:const-prop-3-SAST-mod} For all $\F{t} \in \mathbb{T}_H^q$, $\F{t}$ starts with
$(1:q)$ and is a valid $\ito_{\PP^{(d)}}$-RST for some $d \in \{1,2\}$.
	\end{enumerate}
    The construction is analogous to
    the ones in \Cref{lemma:parallel-execution-lemma-AST,lemma:parallel-execution-lemma-SAST}.
Now we keep all nodes that are in $AC_q$ and split a tree into multiple ones if
    the node is not in $AC_q$, i.e., if the rewrite step required a different starting term in the beginning.

    For example, consider the PTRS $\PP^{(1)}$ containing the rule $\tf(x,x) \to \{1:\tf(\tb,\tc)\}$,
 $\PP^{(2)}$ containing the rule $\tg(x) \to \ta$,
    and the innermost rewrite sequence $\tf(\tg(\tb), \tg(\tc)) \to \tf(\ta, \tg(\tc)) \to \tf(\ta, \ta) \to \tf(\tb,\tc)$,
    where we wrote $\to$ instead of $\iliftto_{\PP^{(1)} \cup \PP^{(2)}}$ and omitted the probabilities for readability.
    We split this sequence (RST) into the three sequences $\tg(x_1) \to \ta$, $\tg(x_2) \to \ta$, and $\tf(x_3,x_3) \to \tc$,
    where all start terms are contained in $\capterm(\tf(\tg(\tb), \tg(\tc)))$.
    
    Since $|\capterm(t)| = K$ is finite 
    and the expected derivation lengths of  all $\ito_{\PP^{(d)}}$-RSTs with $d \in
    \{1,2\}$
    that start with $(1:q)$ for a term $q \in \capterm(t)$ are bounded by some constant $C_q < \omega$,
    there is a $C_{\max} < \omega$ such that for all RSTs $\F{t}$ with
    $(\F{t},p_{\F{t}}) \in \mathbb{T}_H^q$ we have $\edl(\F{t}) \leq C_{\max}$ by \ref{itm:const-prop-3-SAST-mod}.
    Hence, we obtain for each $H \in \IN$:
    \[
		\begin{array}{ccl}
            &  & \edl(\F{T}_H) \\
            & = & \sum_{x \in N^{\F{T}_H}\setminus\ctleaf^{\F{T}_H}} \, p_x^{\F{T}_H} \\
            \text{\textcolor{blue}{(since $\bigcup_{q \in \capterm(t)} AC_q =
                N^{\F{T}_H}\setminus
                \ctleaf^{\F{T}}$)}} & \leq & \sum_{q \in \capterm(t)} \sum_{x \in AC_q} \, p_x^{\F{T}_H} \\
			\text{\textcolor{blue}{(by \ref{itm:const-prop-2-SAST-mod})}} & = & 
            \sum_{q \in \capterm(t)} \sum_{x \in AC_q} \sum_{(p_{\F{t}},\F{t}) \in
              \mathbb{T}_H^q} \, p_{\F{t}} \cdot p_x^{\F{t}}\\
			\text{\textcolor{blue}{(as in \eqref{eq:array-AST-parallel})}} & =
                        & \sum_{q \in \capterm(t)} \sum_{(p_{\F{t}},\F{t}) \in
                          \mathbb{T}_H^q} \, p_{\F{t}} \cdot \sum_{x \in AC_q} \, p_x^{\F{t}} \\
			\text{\textcolor{blue}{(by assumption)}} & \leq & \sum_{q \in
                          \capterm(t)} \sum_{(p_{\F{t}},\F{t}) \in \mathbb{T}_H^q} \, p_{\F{t}} \cdot C_{\max} \\
			& = & \sum_{q \in \capterm(t)} \, C_{\max} \cdot
                        \sum_{(p_{\F{t}},\F{t}) \in \mathbb{T}_H^q} \, p_{\F{t}} \\
			\text{\textcolor{blue}{(by \ref{itm:const-prop-1-SAST-mod})}} & =
                        & \sum_{q \in \capterm(t)}   \, C_{\max} \cdot 1 \\
			& = & \sum_{q \in \capterm(t)} \, C_{\max}\\
			& = & K \cdot C_{\max}\!
		\end{array}
	\]
    Thus, we have $\edl(\F{T}) = \lim_{H \to \infty} \edl(\F{T}_H) \leq K \cdot C_{\max}$.
\end{proof}

Finally, we prove the theorem on signature extensions.

\SignatureExtensions*

\begin{proof}
    We only prove the non-trivial direction ``$\Longrightarrow$'' and
    consider the following three cases:
    \begin{enumerate}
        \item For innermost rewriting, the theorem is implied by
        modularity of $\ASTi$ and $\SASTi$
        for disjoint unions
        (\Cref{modularity-iAST-disjoint,modularity-iSAST-disjoint}).

        \item For full rewriting, we first consider
        the case where $\Sigma^{\PP}$ only contains constants and unary symbols.
        Let $\Sigma'$ be another signature containing fresh symbols, i.e.,
        w.l.o.g.\ we have 
        $\Sigma^{\PP} \cap \Sigma' = \emptyset$.
        For any term from $\TSet{\Sigma^{\PP} \cup \Sigma'}{\VSet}$ we now compute
        a multiset of terms from $\TSet{\Sigma^{\PP}}{\VSet}$ which can be regarded instead.
        Thus, we define a corresponding mapping $\caphterm$ from $\TSet{\Sigma^{\PP} \cup
        \Sigma'}{\VSet}$ to 
        multisets of terms from $\TSet{\Sigma^{\PP}}{\VSet}$.
        For this definition, we need two auxiliary mappings.
        The mapping
        $\caphcterm: \TSet{\Sigma^{\PP} \cup \Sigma'}{\VSet} \to \TSet{\Sigma^{\PP}}{\VSet}$
        replaces all topmost subterms with a root $f$ from $\Sigma'$ by the fresh variable $x_f$
        and the mapping $\caphbterm$ maps any term $t \in \TSet{\Sigma^{\PP} \cup \Sigma'}{\VSet}$
        to the multiset that unites $\caphterm(r)$ for the topmost subterms $r$ with $\rt(r) \in
        \Sigma^{\PP}$ occurring below a symbol from $\Sigma'$ in $t$.
               \[
        \begin{array}{ll}
        \caphcterm(x) &= x \text{, if } x \in \VSet \\
        \caphcterm(f(t_1, \ldots, t_k)) & = f(\caphcterm(t_1), \ldots,
        \caphcterm(t_k))\text{, if } f \in \Sigma^{\PP}\\
        \caphcterm(f(t_1, \ldots, t_k)) & = x_f\text{, if } f \in \Sigma'\\
        \\
        \caphterm(x) &= \{x\}\text{, if } x \in \VSet \\
        \caphterm(f(t_1, \ldots, t_k)) & = \{f(\caphcterm(t_1), \ldots,
        \caphcterm(t_k))\} \cup
        \caphbterm(t_1) \cup \ldots \cup \caphbterm(t_k)        \text{, if } f \in \Sigma^{\PP}\\
        \caphterm(f(t_1, \ldots, t_k)) & =\caphterm(t_1) \cup \ldots \cup \caphterm(t_k)        \text{,
        if } f \in \Sigma'\\
        \\
        \caphbterm(x) &= \emptyset\text{, if } x \in \VSet \\
        \caphbterm(f(t_1, \ldots, t_k)) & =\caphbterm(t_1) \cup \ldots \cup \caphbterm(t_k)        \text{,
        if } f \in \Sigma^{\PP}\\
        \caphbterm(f(t_1, \ldots, t_k)) & =\caphterm(t_1) \cup \ldots \cup \caphterm(t_k)        \text{,
        if } f \in \Sigma'
        \end{array}
        \]
        Let $t \in \TSet{\Sigma^{\PP} \cup \Sigma'}{\VSet}$ and let $\caphterm(t) = \{q_1, \ldots, q_n\}$.
        Let $\tc$ be a fresh $n$-ary constructor symbol. Then instead of $t$, we consider the term $\tc(q_1, \ldots, q_n)$.
        Every $\fto_{\PP}$-RST $\F{T}$ that starts with $(1:t)$ gives rise to an 
        $\fto_{\PP}$-RST $\F{T}'$ that starts with $(1:\tc(q_1, \ldots, q_n))$ with $|\F{T}| = |\F{T}'|$
        and $\edl(\F{T}) = \edl(\F{T}')$.
        To see this, suppose that
        $t \fto_{\PP} \{p_1:s_1, \ldots, p_k:s_k\}$,
$A(t) = \{q_1, \ldots, q_n\}$, and $A(s_i) = \{q_1^i, \ldots, q_{m_i}^i\}$ for all $1 \leq i
        \leq k$.
        Then
        there exists a $1 \leq j \leq n$ with $q_j \fto_{\PP}
        \{p_1:u_1, \ldots, p_k:u_k\}$ where $u_i \in
        A(s_i)$ for all $1 \leq i \leq k$.      
For every
$v_i \in A(s_i) \setminus
          \{u_i\}$  there exists a $j' \neq j$ with  $1 \leq j' \leq n$ 
   such that $v_i  = q_{j'}$ for all $1 \leq i \leq k$.
        Note that $m_i \leq n$, and we might even have $m_i < n$ if we use an erasing rule.
        Hence, every step in $\F{T}$ in a term $t$ can be mirrored by a step in $\F{T}'$
        in a term $\tc(q_1, \ldots, q_n)$.

        Since the $q_i$ are terms from $\TSet{\Sigma^{\PP}}{\VSet}$, every $\fto_{\PP}$-RST
        that starts with  $(1:q_i)$ converges with probability 1 (resp.\ has
        bounded expected derivation length). Hence, by the 
        parallel execution lemmas
        (\Cref{lemma:parallel-execution-lemma-AST,lemma:parallel-execution-lemma-SAST}) this also
        holds for 
        every    $\fto_{\PP}$-RST $\F{T}'$
        that starts with $(1:\tc(q_1, \ldots, q_n))$ (note that there cannot be any rewrite steps
        at the root since $\tc$ is a constructor). But then this also holds for 
        every $\fto_{\PP}$-RST $\F{T}$ that starts with $(1:t)$.
       
        \item Now we consider the case where
        $\Sigma^{\PP}$ contains at least one symbol $\tg$ of arity at least $2$. If $\tg$ has
        an arity greater than two, then we use
        the term $\tg(\_, \, \_, \; x,..., x)$ instead, where $x \in \VSet$.
        Again,        let $\Sigma'$ be another signature containing fresh symbols, i.e.,
        w.l.o.g.\ we have $\Sigma^{\PP} \cap \Sigma' = \emptyset$. 

        We now define a function $\phi : \TSet{\Sigma^{\PP} \cup \Sigma'}{\VSet} \to \TSet{\Sigma^{\PP}}{\VSet}$
        such that $\edh_{\fto_{\PP}}(t) \leq \edh_{\fto_{\PP}}(\phi(t))$:
                \[
            \begin{array}{@{\qquad}ll}
            \phi(x) &= x \text{, if } x \in \VSet \\
            \phi(f(t_1, \ldots, t_k)) & = f(\phi(t_1), \ldots,
            \phi(t_k))\text{, if } f \in \Sigma^{\PP}\\
            \phi(f) & = x_f\text{, if } f \in \Sigma' \text{ has arity 0}\\
            \phi(f(t)) & = \tg(\phi(t),x_f)\text{, if } f \in \Sigma' \text{ has arity 1}\\
            \phi(f(t_1, \ldots, t_k)) & = \tg(\phi(t_1), \tg(\phi(t_2), \ldots \tg(\phi(t_k),
            x_f)\ldots))\text{, if } f \in \Sigma' \text{ has arity $k > 1$}
            \end{array}
        \]
        Every $\fto_{\PP}$-RST that starts with $(1:t)$ gives rise to an 
        $\fto_{\PP}$-RST that starts with $(1:\phi(t))$
        using exactly the same rules, leading to the same convergence probability and the
        same expected derivation length. 
        To see this, note that whenever $t \fto_{\PP} \{p_1:s_1, \ldots, p_k:s_k\}$,
        then also $\phi(t) \fto_{\PP} \{p_1:\phi(s_1), \ldots, p_k:\phi(s_k)\}$. \qedhere
    \end{enumerate}
\end{proof}

\end{document}